\documentclass[final]{memo-l}
\usepackage{amssymb,amsrefs}
\usepackage{graphicx}
\usepackage{bm}
\usepackage{hyperref} 
\usepackage{enumerate}

\numberwithin{equation}{chapter}

\newtheorem{theorem}{Theorem}[chapter]
\newtheorem{corollary}[theorem]{Corollary}
\newtheorem{lemma}[theorem]{Lemma}
\newtheorem{proposition}[theorem]{Proposition}

\theoremstyle{definition}
\newtheorem{definition}{Definition}[chapter]

\theoremstyle{remark}
\newtheorem{remark}{Remark}[chapter]

\newenvironment{proofof}{\trivlist\item[]}%
{\unskip\nobreak\hskip 1em plus 1fil\nobreak$\Box$
\parfillskip=0pt%
\endtrivlist}

\providecommand{\abs}[1]{\left\lvert#1\right\rvert}
\providecommand{\norm}[1]{\left\lVert#1\right\rVert}
\providecommand{\inner}[1]{\left\langle#1\right\rangle}
\providecommand{\ee}{\text{e}}
\providecommand{\oo}{\text{o}}

\providecommand{\D}{\partial}
\providecommand{\R}{\mathbb{R}}
\providecommand{\C}{\mathbb{C}}
\providecommand{\Z}{\mathbb{Z}}
\providecommand{\eps}{\varepsilon}
\providecommand{\nit}{\noindent}
\providecommand{\nn}{\nonumber}
\DeclareMathOperator{\sgn}{sgn}

\providecommand{\lamsharp}{{\lambda_{\sharp}}}
\providecommand{\thetasharp}{{\vartheta_{\sharp}}}

\makeatletter
\newcommand{\vast}{\bBigg@{3}}
\newcommand{\Vast}{\bBigg@{4}}

\begin{document}

\frontmatter

\title{Topologically Protected States in One-Dimensional Systems}

\author{C. L. Fefferman}
\address{Department of Mathematics, Princeton University, Princeton, NJ, USA}
\email{cf@math.princeton.edu}
\thanks{The first author was supported in part by NSF grant DMS-1265524.}

\author{J. P. Lee-Thorp}
\address{Department of Applied Physics and Applied Mathematics, Columbia University, New York, NY, USA}
\email{jpl2154@columbia.edu}

\author{M. I. Weinstein}
\address{Department of Applied Physics and Applied Mathematics and Department of Mathematics, Columbia University, New York, NY, USA}
\email{miw2103@columbia.edu}
\thanks{The second and third authors were supported in part by NSF grants: DMS-10-08855, DMS-1412560 and the Columbia Optics and Quantum Electronics IGERT  DGE-1069420.}

\date{\today}




\maketitle

\tableofcontents

\begin{abstract}  
We study a class of periodic Schr\"odinger operators, which in distinguished cases can be proved to have linear band-crossings or ``Dirac points". We then show that the introduction of an ``edge'',  via adiabatic modulation of these periodic potentials by a domain wall, results in the bifurcation of  spatially localized ``edge states''. These bound states are associated with the topologically protected zero-energy mode of an   asymptotic one-dimensional Dirac operator. 
Our model captures many aspects of the phenomenon of topologically protected edge states for two-dimensional bulk structures such as the honeycomb structure of graphene. The states we construct can be realized as highly robust TM- electromagnetic modes for a class of photonic waveguides with a phase-defect.
\end{abstract}

\markboth{Topologically protected bound states in 1D}{C.L. Fefferman, J.P. Lee-Thorp, M.I. Weinstein}

\mainmatter

\chapter{Introduction and Outline}\label{introduction}
Energy localization in {\it surface modes} or {\it edge states} at the interface between dissimilar media has
been explored, going back to the 1930's, as a vehicle for localization and transport of energy 
\cites{Tamm:32,Shockley:39,RH:08,Fan-etal:08,Soljacic-etal:08,Chen-etal:09,Rechtsman-etal:13,Rechtsman-etal:13a}. These
phenomena can be exploited in, for example, quantum, electronic or photonic systems.
An essential property for applications is robustness; the localization properties of such surface states needs to be
stable with respect to distortions of or imperfections along the interface. 

A class of structures, which have attracted great interest since about 2005, are topological insulators
\cite{Kane-Mele:05}. 
 In certain energy ranges, such structures behave as insulators in their bulk (this is associated with an energy gap in
the spectrum of the bulk Hamiltonian), but have boundary conducting states with energies in the bulk energy gap; these
are states which propagate along the boundary and are localized transverse to the boundary. Some of these states may be
{\it topologically protected}; they persist under deformations of the interface which preserve the bulk spectral gap,
{\it e.g.} localized perturbations of the interface.  In  honeycomb structures, {\it e.g.} graphene,  where a bulk gap
is opened at a ``Dirac point'' by breaking time-reversal symmetry \cites{RMP-Graphene:09,HK:10,Katsnelson:12,FW:12,FW:14},  protected edge
states are uni-directional and furthermore do not backscatter in the presence of interface perturbations
\cites{RH:08,Fan-etal:08,Soljacic-etal:08,Rechtsman-etal:13a}. An early well-known instance of topological protected
states  in condensed matter physics are the chiral edge states observed in the quantum Hall effect. In tight-binding
models, discrete lattice
models which arise, for example, as infinite contrast limits, this property can be
understood in terms of topological invariants associated with the  band structure of the bulk periodic structure
\cites{W:81,TKNN:82,H:82,Hatsugai:93,Graf-Porta:13,Avila-etal:13}.

In this article we introduce a one-dimensional continuum model, the Schr\"odinger equation with a periodic potential modulated by a  {\it domain wall}, for which we rigorously study the bifurcation of topologically protected edge states as a parameter {\it lifts} a Dirac point degeneracy.  This model,  which has many of the features of the above examples, is motivated by the study of photonic edge states in honeycomb structures in \cite{RH:08}.  The bifurcation we study is governed by the existence of a topologically protected zero-energy eigenstate of a one-dimensional Dirac operator, $\mathcal{D}$; see  \eqref{dirac_op}. The zero-mode of this operator plays a role in electronic excitations in coupled scalar - spinor fields \cite{JR:76} and polymer chains \cite{Su-Schrieffer-Heeger:79}.
There are numerous studies of edge states for {\it tight-binding} models; see, for example, the above citations. The present work considers the far less-explored setting of edge states in the underlying partial differential equations; see also \cites{FW:12,FW:14}. A summary of  our results is given in \cite{FLW-PNAS:14}.

\section{Motivating example - Dimer model with a phase defect}\label{sec:motivation}

 In this section we build up a family of Schr\"odinger Hamiltonians with the properties outlined above. 
 The construction has several steps: (i) We introduce a family of  one-periodic potentials which, due to an additional translation symmetry, have spectral bands which touch at a ``Dirac points''. The spectrum of this periodic Schr\"odinger operator is continuous; all states are extended. (ii) We observe that by maintaining periodicity but breaking the additional translation symmetry 
 (``dimerizing''),   gaps in the continuous spectrum are opened about the Dirac points. (iii)  Finally, we introduce a potential which interpolates between 
 different ``dimerizations'' at $+\infty$ and at $-\infty$, having a common spectral gap. The latter Schr\"odinger Hamiltonian is one of a general class for  which we will prove in Theorem  \ref{thm:validity} that localized eigenstates exist with eigenvalues  approximately located mid-gap.
 
Start with a real-valued $1-$ periodic function, $Q(x)$, which is even: $Q(x)=Q(-x)$. Introduce the one-parameter family of potentials $\mathcal{Q}(x;s)$, a sum of translates of $Q$:
\begin{equation}
\mathcal{Q}(x;s) = Q\left(x+\frac{s}2\right) + Q\left(x-\frac{s}2\right),\ \ 0\le s\le1.
\label{Qxs}
\end{equation}
Clearly, $\mathcal{Q}(x;s)$ is  $1-$ periodic. Moreover, $\mathcal{Q}\left(x;s=1/2\right)$ has minimal period equal to
$1/2$, {\it i.e.} $\mathcal{Q}(x;1/2)$ has an additional translation symmetry. 

An example of particular interest is obtained by first taking $q_0(x)$ to be a real-valued, even, smooth and  rapidly
decaying function on $\R$, $Q(x)$ to be a sum over its integer translates, {\it i.e.} $Q(x)=\sum_{n\in\Z}q_0(x+n)$,  
 and $\mathcal{Q}(x;s)$ to be the {\it superlattice potential}, 
concentrated on staggered sub-lattices $\Z-s/2$ and $\Z+s/2$, given by \eqref{Qxs}; see Figure \ref{fig:double}.

The function $\mathcal{Q}(x;s)$ may be expressed as a Fourier series
\begin{align}
\mathcal{Q}(x;s)\ &=\ 4\sum_{m\in\Z_+}\widehat{Q}(m)\cos(\pi ms)\ \cos(2\pi mx) ,
 \end{align}
 and we begin by considering the family of operators 
 \[ H(s)\equiv -\D_x^2+{\mathcal Q}(x;s).\]
For $s=1/2$,  $\mathcal{Q}(x;1/2)$ reduces to an 
even-index cosine series:
\[ \mathcal{Q}(x;1/2) = \sum_{m\in2\Z_+} \mathcal{Q}_m \cos(2\pi m x) \equiv V_\ee(x). \]

The operator $ H(1/2) = -\D_x^2 + V_\ee(x)$ 
 can be shown to have {\it Dirac points}.  These are quasi-momentum / energy pairs, $(k_\star,E_\star)$, where
$k_\star$ is a point with ``high-symmetry'' (see Chapter \ref{1dperiodic&dirac}, Definition \ref{dirac-pt-gen}) at which
neighboring spectral bands touch and at which dispersion curves cross linearly; see Figure \ref{fig:unperturbedspectrum}.

Below we construct an appropriately modulated dimer structure, by allowing the parameter $s$ to vary slowly with $x$. This modulated structure, which adiabatically transitions between periodic structures at $\pm \infty$, that shall be the focus of this paper.
 
\medskip

\begin{figure}
\includegraphics[width=\textwidth]{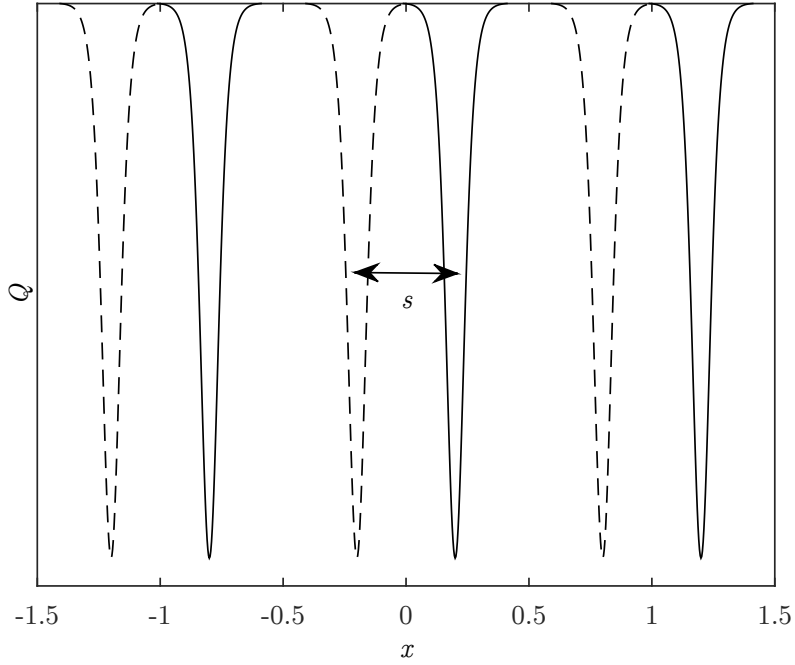}
\caption{$\mathcal{Q}(x;s) ={Q}(x+s/2)+{Q}(x-s/2)$, a periodic potential, consisting of wells
supported on two sub-lattices. ${Q}(x+s/2)$ is given by the dashed curve and ${Q}(x-s/2)$ is given by the solid curve.
$\mathcal{Q}(x;s)$ has minimal period 1 for $s\neq1/2$ and consists of {\it dimers}; each period cell contains
a double-well. For $s=1/2$, $\mathcal{Q}(x;1/2)$ has minimal period  $1/2$.}
\label{fig:double}
\end{figure}

\nit {\bf Dimer model:}\ For $s\ne 1/2$, $\mathcal{Q}(x;s)$ may be viewed as a superposition of {\it dimers},
double-well potentials in each period cell; see Figure \ref{fig:double}. Set $s^\delta=\frac12+\delta\kappa_0$, with
$0<|\delta|\ll1$ small and $0\ne\kappa_0\in\R$
fixed. Then, $\mathcal{Q}(x;s^\delta)$ 
is of the form
\begin{equation}
 \mathcal{Q}(x;s^\delta) = \sum_{m\in2\Z_+} \mathcal{Q}^\delta_m\ \cos(2\pi m x) + 
 \delta\kappa_0\ \sum_{m\in2\Z_++1}\mathcal{W}^\delta_m\ \cos(2\pi m x), \label{dimer-pot}
 \end{equation}
where $\mathcal{Q}^\delta_m,\ \mathcal{W}^\delta_m\ =\ \mathcal{O}(1)$ as $\delta\to0$.
 The operator 
 \[ H(1/2+\delta\kappa_0)\ =\ -\D_x^2+\mathcal{Q}(x;s^\delta)  \]
  is a Hill's operator \cites{magnus2013hill,Eastham:74}. The character of its spectrum is well-known. Two examples are
displayed in the top panels of Figures \ref{perturbed-hills} and \ref{perturbed-hills-eps5}. For $\delta$ fixed and small, 
the gap widths are of order $\mathcal{O}(\delta)$; see Appendix \ref{band_splitting}.
\medskip

\nit {\bf Dimer model with phase defect:} Now fix a constant $\kappa_\infty>0$ and let 
\[ s^\delta(x)= \frac12+\delta\kappa(\delta x),\]
where $\kappa(X)\to\pm\kappa_\infty$ as $X\to\pm\infty$. Consider the operator
 \begin{align}
  H(s^\delta(x)) &=-\D_x^2+\mathcal{Q}(x;s^\delta(x)).\label{dimer_Hdelta-def} 
  \end{align}
Figure \ref{fig:double_wells_s_delta} displays a typical potential $\mathcal{Q}(x;s^\delta(x))$ (bottom panel) along with its phase-shifted asymptotics (top and middle panels).
  
 \begin{figure}
\includegraphics[width=\textwidth]{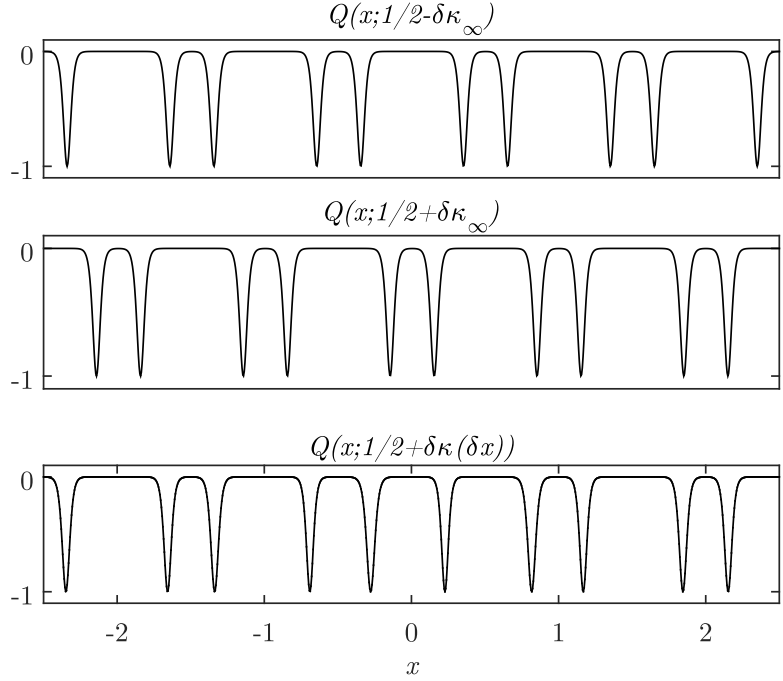}
\caption{Dimerized potential, $\mathcal{Q}(x;s) ={Q}(x+s/2)+{Q}(x-s/2)$, for  $s=s^\delta\ne1/2$.
 {\bf Top panel:} $\mathcal{Q}(x;1/2-\delta\kappa_\infty)$. {\bf Middle panel:} 
 $\mathcal{Q}(x;1/2+\delta\kappa_\infty)$. {\bf Bottom panel:} Domain wall modulated structure, $\mathcal{Q}(x;s^\delta(x))$, where  $s^\delta(x)=1/2+\delta\kappa(\delta x)$ with  $\kappa(X)\to\pm\kappa_\infty$ as $X\to\infty$. 
Top and middle panels display phase-shifted asymptotics of $\mathcal{Q}(x;s^\delta(x))$, as $x\to\pm\infty$.}
\label{fig:double_wells_s_delta}
\end{figure} 
  
 Expanding the potential, we have
  \begin{equation}
 \label{Qexpand}
 \mathcal{Q}(x;s^\delta(x))\approx \mathcal{Q}(x;1/2) + (s(x)-1/2) \frac{\partial}{\partial s}\mathcal{Q}(x;1/2) \equiv
V_{\ee}(x) +
\delta\kappa(\delta x) W_{\oo}(x),
\end{equation} 
where $V_\ee(x)$ denotes an even-index cosine series and $W_\oo(x)$ denotes an odd-index cosine series;
compare with  \eqref{dimer-pot}.
 This motivates our study of the family of operators:
\begin{equation}
H_\delta \ =\  -\D_x^2 + V_{\ee}(x) +
\delta\kappa(\delta x) W_{\oo}(x),
\label{Hdelta-def}\end{equation}
parametrized by $\delta$,  which interpolates adiabatically, through
a {\it domain wall},
 between the operators
 \begin{align}
 H(1/2-\kappa_\infty\delta)&\approx H_{\delta,-}\equiv H(1/2)-\delta \kappa_\infty W_\oo(x)\ \ {\rm at}\ \  
 x=-\infty,\ \ and \label{H-delta}\\ 
 H(1/2+\kappa_\infty\delta)&\approx H_{\delta,+}\equiv H(1/2)+\delta \kappa_\infty W_\oo(x)\ \ {\rm at}\ \  x=+\infty.
 \label{H+delta}\end{align}
 The operator $H_\delta$ has,  for $|\delta|\ll1$,  a
spectral gap of width $\mathcal{O}(\delta)$ about $E_\star$; see Proposition \ref{kappa_spec}.
\medskip

{\it The main results of this paper concern the bifurcation of a branch of simple eigenvalues,
$E^\delta=E_\star+\mathcal{O}(\delta^2)$, located approximately mid-gap,
with corresponding spatially localized eigenstate, $\Psi^\delta$. The bifurcating eigenpair $\delta\mapsto (E^\delta,\Psi^\delta)$ is
topologically stable; that is, it persists under spatially localized
perturbations of the domain wall, $\kappa(X)$.  The existence and stability of this bifurcation are related to a
topologically protected zero mode of a limiting one-dimensional Dirac equation; see \cites{JR:76,RH:08}.}
\medskip

Note that this bifurcation is associated with a {\it non-compact} perturbation of $H_0=-\D_x^2+V_\ee$, a phase change or phase
defect across the structure, which at once changes the essential spectrum and spawns a bound state.  Bifurcations from
the edge of continuous spectra, arising from localized perturbations have been studied extensively; see, for example,
\cites{Simon:76,Deift-Hempel:86,Figotin-Klein:97,Borisov-Gadylshin:08,DVW:12,DVW-CMS:14}. A class of edge bifurcations
due to a non-compact perturbation is studied in \cite{BR:11}.

Our model also captures many aspects of the phenomenon of
topologically protected edge states, for two-dimensional bulk structures, such as the honeycomb structure of graphene. We explore edge states in the two-dimensional setting of honeycomb structures in forthcoming publications.

Finally, we remark that the bound states we construct can be realized as TM- electromagnetic modes for a class of
photonic waveguides with a
phase-defect; we explore this in  \cite{experiment:14}. 
 
 \section{Outline of main results}\label{setup-results}
Let  $V_\ee(x)$ be a sufficiently smooth,  even-index Fourier cosine-series:
 \begin{equation}
V_\ee(x)\ =\ \sum_{p\in2\Z_+} v_p\cos(2\pi px)\ ,\quad v_p\in\R,  \label{Ve-def}
\end{equation} 
and introduce the Schr\"odinger operator, $H_0$:
\begin{equation}
H_0\ =\ -\D_x^2 + V_\ee(x).
\label{H0def}\end{equation}

Recall that the spectrum of a Schr\"odinger operator with a periodic potential is given by the union of closed
intervals,
the spectral bands \cites{RS4,Eastham:74}.
\begin{figure}
\includegraphics[width=\textwidth]{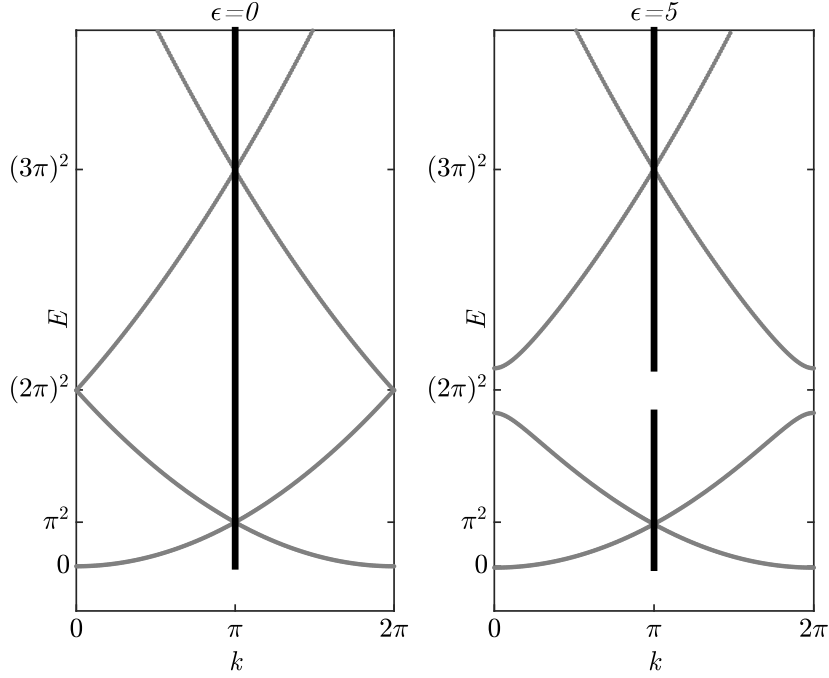}
\caption{Dispersion curves (gray), $E_b$ vs $k$, $k\in\mathcal{B}=[0,2\pi]$, and spectrum (black) for two choices
of the unperturbed Hamiltonian
$H^{(\eps)}_0=-\D_x^2+\eps V_{\ee}(x)$, given by \eqref{H0eps-def}, with $V_{\ee}(x) =
2\cos(2(2\pi x))$. {\bf Left panel:} $\varepsilon=0$, $H^{(0)}_0=-\D_x^2$.
Bands cross at $E_m=m^2\ \pi^2,\ m=1,2,3,\dots$. {\bf Right panel:} $\varepsilon\ne0$. Band crossings at
$(k=0,E=(2m)^2\pi^2)$ (equivalently,
$k=2\pi$) are ``lifted'' and spectral gaps open; see Appendix \ref{degeneracies_lifted}.
  For $\eps\ne0$, crossings at {\it Dirac points}:
$(k=\pi,E^{(\eps)}_{\star,m}\approx(2m-1)^2\pi^2),\ m=1,2,\dots$ persist.
$E_{\star,1}^{(\eps=0)}=\pi^2\approx9.87$
 and $E_{\star,1}^{(\eps=5)}\approx9.45$.}
\label{fig:unperturbedspectrum}
\end{figure}
\medskip

 We outline our main results:
\begin{itemize}
\item[(1)]  Theorems \ref{thm:dirac-pt}-\ref{thm:dirac-pt-gen},\ {\it Generic existence of ``Dirac points'' for a class
of 1D periodic structures}:
    Consider the one-parameter family of periodic Hamiltonians:
\begin{equation}
H_0^{(\eps)}\ =\ -\D_x^2 + \eps V_\ee(x),\ \ \eps\in\R.\label{H0eps-def}
\end{equation}
We prove that for all $\eps\in\R$, except possibly a discrete set of values contained in $\R\setminus(-\eps_0,\eps_0)$, with $ \eps_0>0$,
the  $H^{(\eps)}$ has a Dirac point  or linear band crossing at $(k_\star=\pi,E_\star^{(\eps)})$,  in the sense of 
Definition \ref{dirac-pt-gen}.
That is, $E=E_\star$ is a degenerate $k_\star-$ pseudo-periodic  eigenvalue
 with corresponding two-dimensional eigenspace given by ${\rm span}\{\Phi_1,\Phi_2\}$. Furthermore,  there exist $\lambda_{\sharp}\ne0$ and Floquet-Bloch eigenpairs 
\begin{align}
\label{Phi_pm}
& k\ \mapsto\ (\Phi_+(x;k),E_+(k))\ \ {\rm and}\ \   k\ \mapsto\ (\Phi_-(x;k),E_-(k)), \\
&\Phi_1(x)\ =\ \Phi_-(x;k_\star)\in L^2_{k_\star,\ee}\ ,\ \Phi_2(x)\ =\ \Phi_+(x;k_\star))\in L^2_{k_\star,\oo},\nn
 \end{align}
  and such that
 $E_\pm(k)-E_\star\ =\ \pm\ \lambda_\sharp\ \left(k-k_\star\right)\ \Big(1+ o(k-k_\star)\Big)$
as $k\to k_\star$. See Chapter \ref{introduction}, Section \ref{1d-toy} for the definitions of $L^2_{k_\star,\ee}$ and $L^2_{k_\star,\oo}$ ($k_\star=\pi$) and the decomposition: $L^2_{k_\star}=L^2_{k_\star,\ee}\oplus L^2_{k_\star,\oo}$.

\medskip Figure \ref{fig:unperturbedspectrum} shows the dispersion curves and spectrum for
$H_0^{(\eps)}=-\D_x^2+\eps V_\ee(x)$ in the cases:  $\eps=0$ (left) and $\eps\ne0$ (right). One sees that for
$\eps\ne0$,
the band touchings at quasi-momentum $k=0$, or equivalently $k=2\pi$, are ``lifted'' (Appendix
\ref{degeneracies_lifted})  while those occurring for
quasi-momentum $k=\pi$, the Dirac points, persist. The persistence is associated with the  $k_\star$ being a high-symmetry quasi-momentum
 for  $H^{(\eps)}$; see Definition \ref{dirac-pt-gen}. The analogous picture was established for 
honeycomb structures defined by honeycomb lattice potentials; see  \cites{FW:12,FW:14}.
 \medskip

\item[(2)]  Theorem  \ref{thm:validity},\  {\it Topologically protected edge states in perturbations of 1D
systems with Dirac points:}\ Consider the periodic Hamiltonian $H_0=-\D_x^2+V_\ee(x)$,  
assumed to have a 
Dirac point, {\it e.g.} $V_\ee=\eps \tilde{V}_\ee$, where $\tilde{V}_\ee$ is of the form \eqref{Ve-def} 
 and $\eps$ is as in (1), {\it i.e.} $\eps\in\R$, except possibly a discrete set of values contained in $\R\setminus(-\eps_0,\eps_0)$, with $ \eps_0>0$.

\nit Define the perturbed Hamiltonian:
\begin{align}
 H_\delta&=  -\D_x^2+V_\ee(x)+\delta\kappa(\delta x)W_\oo(x) ,
\label{Hdelta}
\end{align}
where $W_\oo$ is a sufficiently smooth \underline{odd}-index cosine series:
\begin{equation}
 W_\oo(x)= \sum_{p\in2\Z_++1} w_p\cos(2\pi px), \quad w_p\in\R .
 \label{Woo-def}
 \end{equation}

For $\kappa\equiv\kappa_\infty$, a non-zero constant, the spectrum of $H_\delta$ is continuous
with a spectral gap of order $\delta|\kappa_\infty|$ about the Dirac point, $(k_\star,E_\star)$, of $H_0$; see Proposition \ref{kappa_spec}.\medskip

Now let 
$\kappa(X)$ be a sufficiently smooth ``domain wall'', {\it i.e.} 
$\kappa(X)$ tends to $\pm\kappa_{\infty}$ as $ X\to\pm\infty, \ \ \kappa_{\infty}>0$. More precisely,
\begin{equation}
 \kappa(X)-\sgn(X)\kappa_{\infty}\ \ {\rm and}\ \ \kappa'(X)\ \  \textrm{tend to zero sufficiently rapidly as}\ \
X\to\pm\infty.\label{kappa-def}
 \end{equation}
 (We believe our results extend to the case where $\kappa(X)$ has the
more general asymptotic behavior, $\kappa(X)\to\pm\kappa_{\pm\infty}$ as $|X|\to\infty$, where
$\kappa_\infty\times\kappa_{-\infty}>0$, but have assumed \eqref{kappa-def}  to simplify
an aspect of the proofs.)

\nit Figure \ref{fig:domain_wall_potentials} displays two examples of the potential $V_\ee(x)+\delta\kappa(\delta
x)W_\oo(x)$ of $H_\delta$.  \medskip
 
 \begin{figure}
\includegraphics[width=\textwidth]{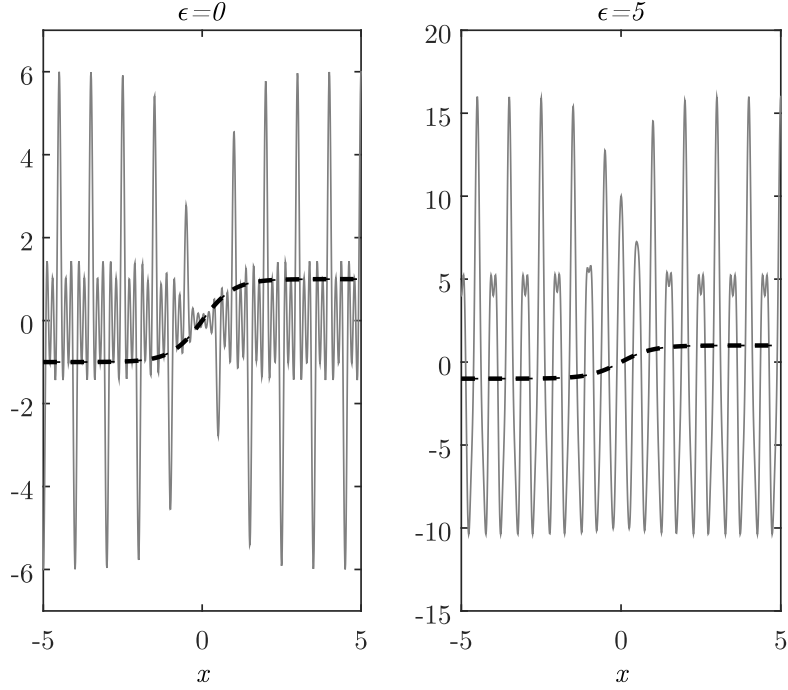}
\caption{Domain wall modulated potentials 
$2\varepsilon\cos(2(2 \pi x))+2\delta\tanh(\delta x)\sum_{j\in\{1,3,5\}}\cos(2\pi j x)$ for $(\eps,\delta)=(0,1)$ ({\bf left panel})
and  $(\eps,\delta)=(5,1)$  ({\bf right panel}). Both the potential $U_{\delta}(x)$ (solid gray) and the domain wall profile
$\kappa(\delta x)$ (dashed black) are plotted.
Spectra of $H_\delta$, as $\delta$ is varied, are displayed in the bottom panels of Figures \ref{perturbed-hills} ($\eps=0$) and 
\ref{perturbed-hills-eps5}
($\eps=5$). }
\label{fig:domain_wall_potentials}
\end{figure}

Let  $(k_\star,E_\star)$ denote a Dirac point of $H_0$, with corresponding $k_\star$- pseudo-periodic eigenspace 
spanned by $\Phi_1(x)$ and $\Phi_2(x)$.
From Theorem  \ref{thm:validity} we have:\medskip

\begin{itemize}
 \item [(A)]  There exists $\delta_0>0$ such that for all $0<\delta<\delta_0$, the perturbed Schr\"odinger operator,
$H_\delta$, 
has an eigenpair $(\Psi^{\delta},E^\delta)$. The energy, $E^\delta$, falls within a gap in the essential spectrum of $H_\delta$ of width $\mathcal{O}(\delta)$, and $\Psi^\delta\in H^2(\R_x)$ is exponentially decaying. \smallskip

 \item [(B)] To leading order, $\Psi^\delta(x)$ is a slow and localized modulation of the degenerate (Dirac) eigenspace:
\begin{align}
\Psi^{\delta}(x)&\approx \alpha_{\star,1}(\delta x)\Phi_1(x)+\alpha_{\star,2}(\delta x)\Phi_2(x)\ \ \textrm{in}\ \
H^2(\R_x),\label{Psi-delta-leading}\\
E^\delta&\approx E_\star +\mathcal{O}(\delta^2).\nn
\end{align}
The vector of amplitudes, 
$\alpha_\star=
\left(\ \alpha_{\star,1}(X),\alpha_{\star,2}(X)\ \right)$, is the topologically protected zero-energy
eigenstate of the system of Dirac equations: $\mathcal{D} \alpha_\star=0$, where $\mathcal{D}$  is a one-dimensional Dirac operator given by:
 \begin{equation}
\mathcal{D}\ \equiv
i\lamsharp\sigma_3\partial_{X}+\thetasharp\kappa(X)\sigma_1\ .
\label{D-ms}\end{equation}

\end{itemize}
\medskip

\begin{figure}
\includegraphics[width=\textwidth]{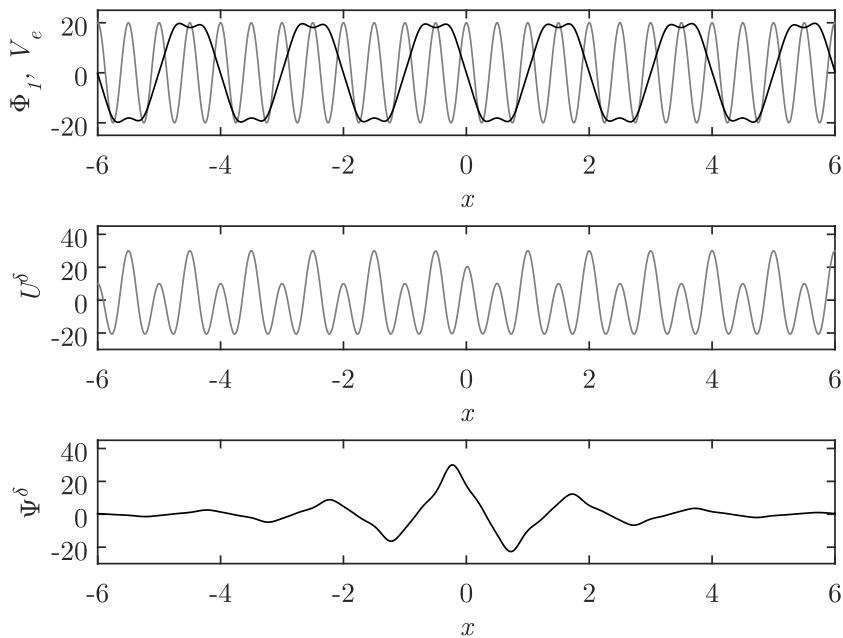}
\caption{Potentials and modes for $H^{\delta}=-\D_x^2+U_\delta(x)$, where $U_\delta(x)=20\cos(2(2 \pi
x))+2\delta\tanh(\delta x)\cos(2(2 \pi
x))$. {\bf Top panel:}  Floquet-Bloch eigenmode $\Phi_1(x)\in L^2_{\pi,\ee}$ (black) corresponding to the degenerate
eigenvalue
$E_{\star}$ at the Dirac point of unperturbed Hamiltonian, $H_{\delta=0}$, superimposed on a plot of $U_0(x)$ (gray). Here,
$\Phi_2(x)=\Phi_1(-x)\in L^2_{\pi,\oo}$. {\bf Middle panel:} Domain wall modulated periodic structure,
$U_{\delta}(x)$,\ 
$\delta=5$. {\bf Bottom panel:}  Localized mid-gap eigenmode $\Psi^{\delta}(x)$ of $H_\delta$, $\delta=5$. }
\label{fig:efunc_plots}
\end{figure}

\begin{figure}
\includegraphics[width=\textwidth]{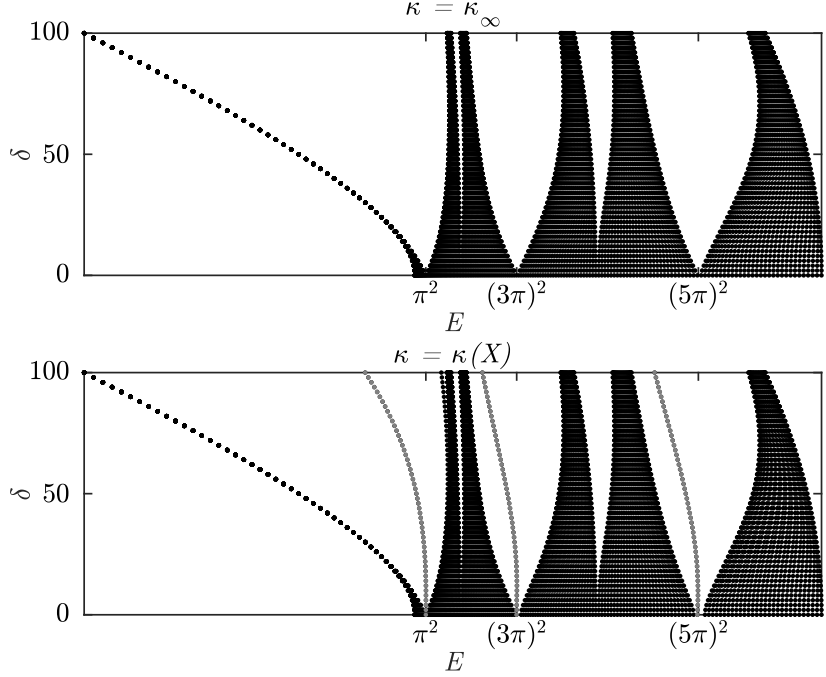}
\caption{{\bf Top panel:} Spectrum of Hill's operator with $V_\ee(x)\equiv0$:
$-\D_x^2+\delta\kappa_\infty\sum_{j\in\{1,3,5\}}\cos(2\pi
j x)$. {\bf Bottom panel:} Spectrum of the
domain-wall modulated periodic potential shown in the left panel of Figure \ref{fig:domain_wall_potentials}:
$-\D_x^2+2\delta\kappa(\delta x) \sum_{j\in\{1,3,5\}} \cos(2\pi j
x)$ with $\kappa(X)=\tanh(X)$. Bifurcating branches of topologically protected ``edge states'', predicted by Theorem
\ref{thm:validity},  are the solid (gray) curves emanating from
the points: $(E_{\star,m},\delta=0)$ where $(k_\star=\pi,E_{\star,m}\equiv(2m+1)^2\pi^2)$ are Dirac points.}
\label{perturbed-hills}
\end{figure}

\begin{figure}
\includegraphics[width=\textwidth]{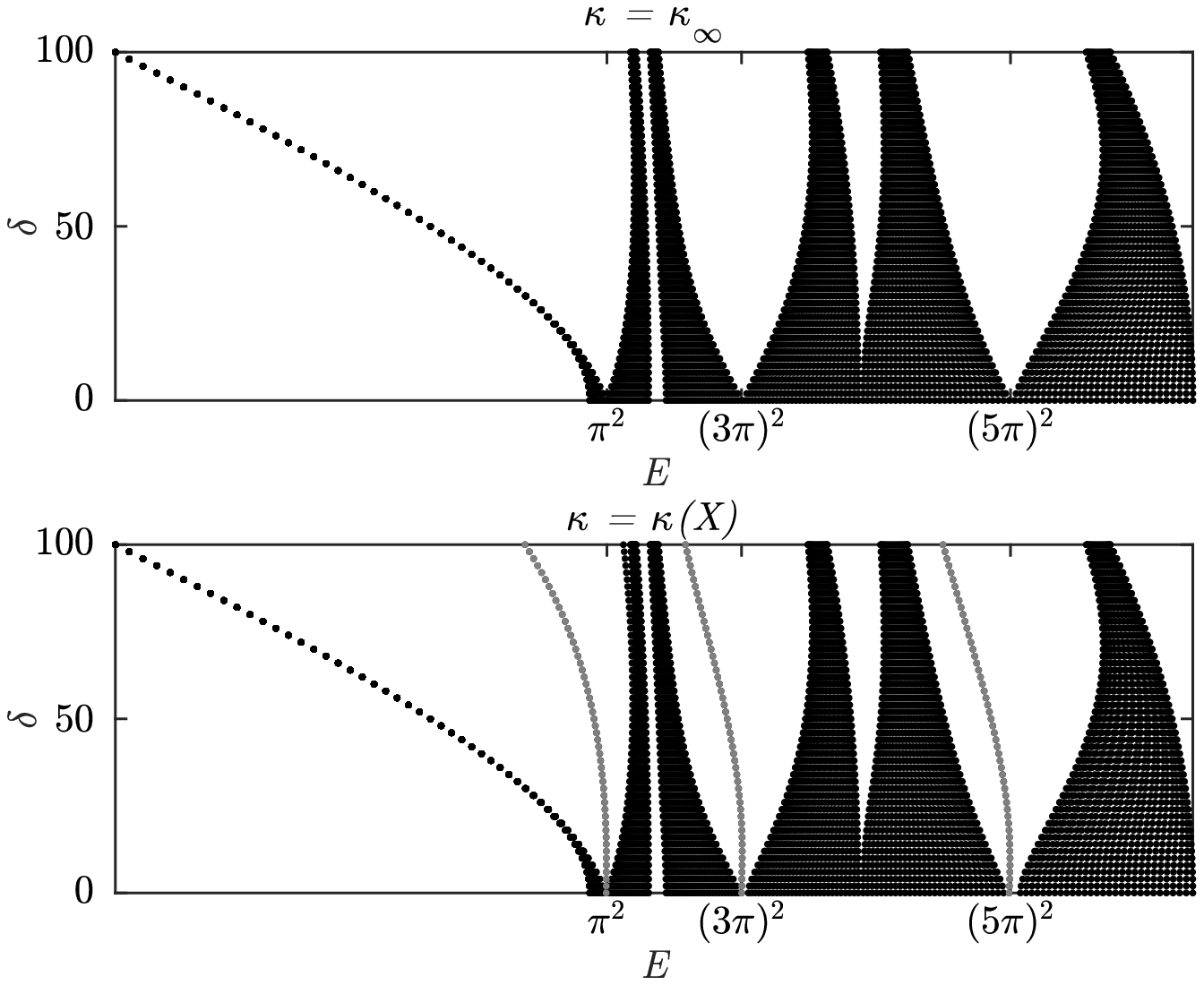}
\caption{
{\bf Top panel:} Spectrum of Hill's operator with $V_\ee(x)\neq0$: $-\D_x^2+2\varepsilon\cos(2(2 \pi x))+2\delta\kappa_\infty
\sum_{j\in\{1,3,5\}} \cos(2\pi j x)$, for $\eps=5$. {\bf Bottom panel:} Spectrum of the domain-wall modulated
periodic potential shown in the right panel of Figure \ref{fig:domain_wall_potentials}: $-\D_x^2+2\varepsilon\cos(2(2
\pi x))+2\delta\kappa(\delta x) \sum_{j\in\{1,3,5\}} \cos(2\pi j x)$ with
$\kappa(X)=\tanh(X)$.}
\label{perturbed-hills-eps5}
\end{figure}

Here, $\lambda_\sharp=2i\left\langle\Phi_1,\D_x\Phi_1\right\rangle_{L^2[0,1]}$  and $\vartheta_\sharp=\left\langle\Phi_1,W_\oo\Phi_2\right\rangle_{L^2[0,1]}$. We assume  $\lambda_\sharp\times\vartheta_\sharp\ne0$. 
 This
non-degeneracy assumption holds for generic $V_\ee$ and $W_\oo$. \\

Figure \ref{fig:efunc_plots} displays an unperturbed periodic structure ($\delta=0$) with a Dirac point and the Floquet-Bloch mode $\Phi_1(x)$ (top panel), the perturbed (domain wall modulated) potential (middle panel) and a localized edge mode (bottom panel), ensured by Theorem \ref{thm:validity}.

\medskip
Bifurcation curves of mid-gap modes are depicted for the operators
$H_\delta=-\D_x^2+\delta\kappa(\delta x) W_\oo (x)$ in Figure \ref{perturbed-hills} and 
$H_\delta=-\D_x^2+ 10\cos(2(2 \pi x))+ W_\oo(x)$ 
 in Figure \ref{perturbed-hills-eps5}. Here, $W_\oo(x)=2\sum_{j\in\{1,3,5\}} \cos(2\pi j x)$  is even
and $1-$ periodic, and contains only the $1-,3-$ and $5-$ harmonics. Theorem \ref{thm:validity} ensures these
bifurcations since, by Remark \ref{on-theta-sharp}, $\vartheta_\sharp\ne0$ for each of the first three Dirac points
$(k_\star=\pi,E_\star=\pi^2(2m+1)^2),\ m=0,1,2$.
 For a
discussion of bifurcations in the case where $\vartheta_\sharp=0$, see Remark \ref{on-theta-sharp-eq0}.
\end{itemize}

\section{Outline} We briefly comment on the contents of each section.

\nit$\bullet$\ In Chapter \ref{floquet-bloch} we review the spectral (Floquet-Bloch) theory of one-dimensional
Schr\"odinger (Hill's) operators, $H_Q=-\D_x^2+Q(x)$, where $ Q(x+1)=Q(x)$. We also present a formulation of the Poisson
summation formula in $L^2_{\rm loc}$, which we require due to the non-compactness of the domain wall perturbation.\\
$\bullet$\  In Chapter \ref{1dperiodic&dirac} we define what it means for an energy / quasi-momentum pair, $(E,k)$, to
be a Dirac point of $H_Q$ (Definition \ref{dirac-pt-gen}). We discuss the spectral theory of $H^{(\eps)}=-\D_x^2+\eps
V_\ee(x),\ \eps\in\R$ (see \eqref{H0def}, \eqref{Ve-def}) and, in particular, show that such Hamiltonians have Dirac
points for all real  $\eps$ outside of a discrete set (Theorems \ref{thm:dirac-pt} and \ref{thm:dirac-pt-gen}).
\\
$\bullet$\   In Chapter \ref{sec:bound_states} we embark on a study of topologically protected bound states of
$H_\delta$; see \eqref{Hdelta}. In Section \ref{subsec:multiscale_analysis}, we give a systematic, but formal,
multiple-scale expansion of
these bound states and, in particular, derive an effective Dirac equation with ``potential'' $\kappa(X)$, governing the
spatial envelope of the bound state. This Dirac equation has a spectral gap of width $\mathcal{O}(|\kappa_\infty|)$, centered at zero.  
  In Section \ref{subsec:dirac_bound_state} we show that for {\it any} $\kappa(X)$ of domain-wall-type
   (see \eqref{kappa-def}) these  Dirac equations have a zero energy mode. Correspondingly, the operator  $H_\delta$ has
a localized mode, approximately mid-gap: $E^\delta= E_\star+\mathcal{O}(\delta^2)$. This zero mode is topologically
stable; it
persists for any deformation of $\kappa(X)$ which respects the imposed asymptotic conditions at $X=\pm\infty$.
 This zero mode plays an important role in \cites{JR:76,RH:08}. 
\\
$\bullet$\   In Chapter \ref{sec:main_result}  we state our main result (Theorem \ref{thm:validity}) on the existence,
for $\delta$ small,  of a
bifurcating branch of eigenpairs of the Hamiltonian $H_\delta$, with eigenvalues located approximately mid-gap. This bifurcation is topologically protected in the sense that it persists under  arbitrary spatially localized perturbation in the domain wall, $\kappa(X)$.
\\ $\bullet$\   In Chapter \ref{sec:proof-exists-mode} we present the proof of Theorem \ref{thm:validity}.
The proof is based on a Lyapunov-Schmidt reduction of the spectral problem for $H_\delta$ to one for the Dirac
operator.  A rough strategy of the proof is given in Section \ref{rough-strategy} and
a more detailed 
sketch in Section \ref{subsec:freq_decomp}. 
 \\ There are several appendices.\\ 
 $\bullet$\  In Appendix \ref{PSF} we prove a variant of the Poisson Summation
Formula in $L^2_{\rm loc}$, stated in Chapter \ref{floquet-bloch}.\\
 $\bullet$\  In Appendix  \ref{linear-band-crossing} we give the
proof of Theorem
\ref{conditions-for-dirac}, which gives sufficient conditions for the  existence of a Dirac point.\\
$\bullet$\ In  Appendix \ref{dirac-pt-small-eps} we use this characterization to prove
  Theorem \ref{thm:dirac-pt}, the existence of Dirac points
of $H^{(\eps)}$ for all $\eps$ in some interval, $(-\eps_0,\eps_0)$, about zero. \\
$\bullet$\ In Appendix \ref{tC-discrete!} we study the existence of Dirac points of $H^{(\eps)}$ without  restrictions on the size of $\eps$. Using an extension of the methods of \cite{FW:12} we 
prove Theorem \ref{thm:dirac-pt-gen}:\ $H^{(\eps)}$ has Dirac points for all $\eps\in\R$ except possibly a discrete set in $\R\setminus(-\eps_0,\eps_0)$. Dirac points occur as well in two-dimensional honeycomb structures \cite{FW:12}. In Remark \ref{honey-discrete} we indicate how to adapt the arguments of Appendix  \ref{tC-discrete!}   to prove that the exceptional set, outside of which  the two-dimensional Schr\"odinger equation with honeycomb lattice potential has Dirac points, is discrete. The more restrictive result, that the exceptional set is countable and closed, was proved in \cite{FW:12}. See also Remark \ref{TexasA&M} concerning related recent work by G. Berkolaiko and A. Comech.\\
$\bullet$\  In Appendix \ref{degeneracies_lifted} we show that 
band crossings of $H^{(\eps)}=-\partial_x^2+\eps V_\ee$,
occurring at quasi-momentum $k=0$ for $\eps=0$ are ``lifted'' for $\eps$ non-zero and small; see Figure
\ref{fig:unperturbedspectrum}.\\
 $\bullet$\ Appendix \ref{band_splitting}  presents the key step
 in the proof of Proposition \ref{kappa_spec}, establishing that the operators $H_{\delta,\pm}$ and $H_\delta$ have a
gap in their essential spectrum of width $\mathcal{O}(\delta)$.\\
$\bullet$\  Appendix \ref{psi_bound_proof} provides useful bounds on functions
comprising the leading order expansion of the bifurcating bound state. \\
$\bullet$\ Finally, in Appendix \ref{near_freq_limits} we
present explicit computations needed to justify the {\it bifurcation equation}, arising in the Lyapunov-Schmidt
reduction.

\section{Notation\label{subsec:notation}}

\begin{enumerate}[(1)]
 \item $\overline{z}$ denotes the complex conjugate of $z\in\mathbb{C}$.
 \item $\bm{x},\bm{y}\in\mathbb{C}^n$, $\inner{\bm{x},\bm{y}}=\overline{\bm{x}}\cdot\bm{y} =
\overline{x_1}y_1+\ldots+\overline{x_n}y_n$.
  \item $\inner{f,g} = \int\overline{f}g$.
  \item $\mathcal{B}\equiv[0,2\pi]$ denotes the Brillouin Zone centered at $\pi$.
  \item $x\lesssim y$ if and only if there exists $C>0$ such that $x \leq Cy$. $x \approx y$ if and only if $x \lesssim
y$ and $y \lesssim x$.
  \item $\ell^2\equiv\ell^2(\mathbb{Z})=\left\{\bm{x}=\{x_j\}_{j\in\mathbb{Z}}:\|\bm{x}\|_{_{\ell^2}}=
\sum_{j\in\mathbb{Z}}\abs{{x}_j} ^2<\infty \right\}$.
  \item For $s\ge0$, $H^s_{\text{per}}([0,1])$ denotes the space of $H^s$ functions which are periodic with period 1.
    \item $L^2_k=L^2_k([0,1])$ denotes the subspace of $L^2([0,1])$ functions, which satisfy the pseudo-periodic boundary
condition: $f(x+1) = f(x)e^{ik}$, for $x\in\R$
  \item $W^{k,p}(\R)$, $k\geq1$, $p\geq1$, denotes the Sobolev space of functions with derivatives up to order
$k$ in $L^p(\R)$.
  \item $L^{p,s}(\R)$ is the space of functions $F:\R\rightarrow\R$ such that
$(1+\abs{\cdot}^2)^{s/2}F\in L^p(\R)$, endowed with the norm
  \begin{align*}
  \norm{F}_{L^{p,s}(\R)} &\equiv \norm{(1+\abs{\cdot}^2)^{s/2}F}_{L^p(\R)} \approx
\sum_{j=0}^{s}\norm{\abs{\cdot}^jF}_{L^p(\R)} < \infty,~~~ 1\leq p\leq \infty.
  \end{align*}
  \item For $f,g\in L^2(\R)$, the Fourier transform and its inverse are given by
  \begin{equation}
  \mathcal{F}\{f\}(\xi)\equiv\widehat{f}(\xi)=\frac{1}{2\pi}\int_{\R}e^{-iX\xi}f(X)dX,~~~
  \mathcal{F}^{-1}\{g\}(X)\equiv\check{g}(X)=\int_{\R}e^{ iX\xi}g(\xi)d\xi.
  \label{FT-def}
  \end{equation}
The Plancherel relation states:
  \begin{equation}
  \int_\R f(x)\overline{g(x)} dx = 2\pi\ \int_\R \widehat{f}(\xi)\overline{\widehat{g}(\xi)} d\xi .
  \label{plancherel}\end{equation}
  \item $\chi_a$ and $\overline{\chi}_a$ denote the characteristic functions defined by
  \begin{equation}
  \chi_a(\xi)=\chi(\abs{\xi}\leq a)\equiv \left\{
  \begin{array}{cc}
   1,&\abs{\xi}\leq a \\0,&\abs{\xi}> a
  \end{array}\right.,~~~
  \overline{\chi_a}(\xi)=\overline{\chi}(\abs{\xi}\leq a)\equiv1-\chi(\abs{\xi}< a). \label{chi-def}
  \end{equation}
  \item $\sigma_j$, $j=1,2,3$, denote the Pauli matrices, where
  \begin{equation}\label{Pauli-sigma}
  \sigma_1 = \begin{pmatrix}0&1\\1&0\end{pmatrix},~~
  \sigma_2 = \begin{pmatrix}0&-i\\i&0\end{pmatrix},~~\text{and}~~
  \sigma_3 = \begin{pmatrix}1&0\\0&-1\end{pmatrix}.
  \end{equation}
\end{enumerate} Further notations are introduced in the text.
\medskip

\section{Acknowledgements}
 The authors thank P. Deift, J.B. Keller, J. Lu, A. Millis, M. Rechtsman and M. Segev for stimulating discussions. We acknowledge a very helpful discussion  with J. K$\acute{o}$llar  in connection with Appendix  \ref{tC-discrete!}.

\chapter{Floquet-Bloch and Fourier Analysis}\label{floquet-bloch}

\section{Floquet-Bloch theory\ -\ 1D}

Let $Q\in C^\infty$ denote a one-periodic real-valued potential, {\it i.e.} $Q(x+1)=Q(x),\ x\in\R$. 
In this section we outline the spectral theory of the Schr\"odinger operator:
\begin{equation} H_Q=-\D_x^2+Q(x).\label{HQdef}\end{equation}

\begin{definition} The space of $k-$ pseudo-periodic $L^2$ functions is given by:
\begin{equation}
L^2_k\ =\ \left\{f\in L^2_{\rm loc}\ :\ f(x+1;k)=e^{ik}f(x;k)\right\}.
\label{L2k}
\end{equation}
\end{definition}

Since the $k-$ pseudo-periodic boundary condition is invariant under $k\mapsto k+2\pi$, it is natural to work with a
fundamental dual period cell or Brillouin zone, which we take to be: 
\begin{equation}
\mathcal{B}\equiv[0,2\pi].
\label{BZ-def}\end{equation}
Typically the Brillouin zone is taken to be $[-\pi,\pi]$, which is symmetric about $k=0$. Our choice is motivated by 
the presentation being simplified if the quasi-momentum $k=\pi$ is an interior point of $\mathcal{B}$.

We next consider a one-parameter family of Floquet-Bloch eigenvalue problems, parametrized by $k\in\mathcal{B}$:
\begin{equation}
H_Q\Phi=E\Phi,\ \ \Phi(x+1;k)=e^{ik}\Phi(x;k).
\label{FB-evp}
\end{equation}
The  eigenvalue problem \eqref{FB-evp} is self-adjoint on $L^2_k$ and has a discrete set of eigenvalues
\begin{equation}
E_1(k)\le E_2(k)\le\cdots\le E_j(k)\le\cdots,
\nn\end{equation}
listed with repetitions. The maps $k\in\mathcal{B}=[0,2\pi]\mapsto E_j(k)$, $j\ge1$,  are Lipschitz continuous functions;
see, for example, Appendix A of \cite{FW:14}. Furthermore,  these maps satisfy the following  properties:
\begin{align}
& \textrm{Symmetry:}\qquad E_j(k) = E_j(2\pi-k),\ k\in[0,\pi]\ ,\nn\\
&  \textrm{Monotonicity:}\qquad k\mapsto E_j(k)\ \textrm{is monotone for}\ k\in[0,\pi]\ \textrm{and for}\ k\in[\pi,2\pi]\ ;\label{Ej-sym-mon}
\end{align}
see \cite{RS4}. Figure \ref{fig:unperturbedspectrum} illustrates these properties for two different potentials, $Q(x)$.\medskip

An equivalent formulation is obtained by setting $\Phi(x;k)=e^{ikx}p(x;k)$. This yields  the 
periodic eigenvalue problem
for an eigenpair, $(p(x;k),E)$:
 \begin{align}
H_Q(k)p &=Ep,\ \ p(x+1;k)=p(x;k),
\label{kFB-evp}\\
\textrm{where}\ \ H_Q(k)&=-(\D_x+ik)^2+Q(x)\label{HQk}.
\end{align}
For each $k\in\mathcal{B}$,
there is a sequence of eigenvalues, $\{E_j(k)\}_{ j\ge1}$ with corresponding  orthonormal sequence of eigenstates $\{p_j(x;k)\}_{j\ge1}$, which are complete in $L^2[0,1]$:
\begin{equation}
f\in L^2[0,1]\ \implies\ f(x)= \frac{1}{2\pi} \ \sum_{j\ge1} \left\langle
p_j(\cdot;k),f(\cdot)\right\rangle_{L^2[0,1]}\ p_j(x;k).
\label{pj-complete}
\end{equation}

For $g\in L^2(\R)$, define the Floquet-Bloch coefficients
\begin{equation}
\tilde{g}_b(k)\ \equiv\ \left\langle
\Phi_b(\cdot;k),g(\cdot)\right\rangle_{L^2(\R)},\ \ b\ge1.
\label{g-tilde}\end{equation}
The family of states $\{\Phi_b(x;k):\ b\ge1,k\in\mathcal{B}\}$ is complete in $L^2(\R)$: 
\begin{equation}
g\in L^2(\R)\ \implies\ g(x)= \frac{1}{2\pi} \ \sum_{b\ge1}\ \int_{\mathcal{B}} \tilde{g}_b(k)\ \Phi_b(x;k)\ dk .
\label{uj-complete}
\end{equation}
Furthermore,  the Parseval relation holds:
\begin{equation}
 \norm{g}_{L^2(\R)}^2\ =\ \frac{1}{(2\pi)^2} \ 
\sum_{b=1}^{\infty}\int_{\mathcal{B}}\abs{\inner{\Phi_b(\cdot,k),g(\cdot)}_{L^2(\R)}}^2 dk\ =\ \frac{1}{(2\pi)^2} \ 
\sum_{b=1}^{\infty}\int_{\mathcal{B}} |\tilde{g}_b(k)|^2 dk\ . \label{parseval}
\end{equation}

\begin{remark}\label{welldefinedonL2}
Note that $\tilde{g}_b(k)$, set equal to the expression $\left\langle\Phi_b(\cdot;k),g(\cdot)\right\rangle_{L^2(\R)}$,
in \eqref{g-tilde}
 requires interpretation, since $\Phi_b(\cdot,k)\notin L^2(\R)$. This inner product and many other inner products in
this paper are interpreted in the following sense:
\medskip

Let $f(x,k)$ and $g(x,k)$ be functions defined on $\R\times\mathcal{B}$. Assume 
\[ f, g\in L^2([-M_1,M_2]\times\mathcal{B})\]
  for each positive and finite $M_1$ and $M_2$ . We say that 
\begin{equation}
\textrm{``\ $\left\langle f(\cdot,k),g(\cdot,k)\right\rangle_{L^2(\R)}$\ is well-defined\ ''}
\nn\end{equation}
if the functions $\mathcal{I}_{M_1M_2}(k)=\int_{-M_1}^{M_2}\overline{f(x,k)} g(x,k) dx$ belong to $L^2(\mathcal{B})$ 
and converge in  $L^2(\mathcal{B})$   as $M_1, M_2\to\infty$.
\medskip

\nit Returning to the particular case of $\left\langle\Phi_b(\cdot;k),g(\cdot)\right\rangle_{L^2(\R)}$, we note that
this inner product is well-defined in the sense just discussed for $g\in L^2(\R)$.
\end{remark}
\medskip

By elliptic theory, $p_j(x;k)$ and $\Phi_j(x;k)$ are smooth as functions of  $x$, for  each fixed $k\in\mathcal{B}$. 
 Sobolev regularity can also be measured in terms of the Floquet-Bloch coefficients:
\begin{lemma}
 \label{lemma11} For $g(x)\in H^s(\R)$, with $s\in\mathbb{N}$,
 \begin{equation}
  \label{far10}
\norm{g}^2_{H^s(\R)} \approx
\int_{\mathcal{B}}\sum_{b=1}^{\infty}(1+b^2)^s\abs{\inner{\Phi_b(\cdot,k),g(\cdot)}_{L^2(\R)}}^2dk =
\int_{\mathcal{B}}\sum_{b=1}^{\infty}(1+b^2)^s\ |\tilde{g}_b(k)|^2 dk \ .
\end{equation} 
\end{lemma}

{\it Proof.}
Let $H=-\D_x^2+V_\ee(x)$. 
For some constant $g$, positive and  sufficiently large, we have by elliptic theory that 
$\norm{f}^2_{H^s(\R)} \approx \norm{(I+H+g)^{s/2}f}_{L^2(\R)}^2$. 
Using the Weyl asymptotics: $C_1b^2\leq E_b(k)\leq C_2b^2$ for all $k\in\mathcal{B}$ and
$b\gg1$ we have:
 \begin{align}
  \label{far13}
\norm{f}^2_{H^s(\R)} &\approx \norm{(I+H+g)^{s/2}f}_{L^2(\R)}^2 \nonumber \\
&= \frac{1}{(2\pi)^2} \ \sum_{b=1}^{\infty}\int_{\mathcal{B}}
\abs{\inner{\Phi_b(\cdot,k),(I+H+g)^{s/2}f(\cdot)}_{L^2(\R)}}^2dk \nonumber\\
&= \frac{1}{(2\pi)^2} \  \sum_{b=1}^{\infty}\int_{\mathcal{B}}\abs{\inner{1+E_b(k)+g)^{s/2}
\Phi_b(\cdot,k),f(\cdot)}_{L^2(\R)}}^2dk \nonumber\\
&= \frac{1}{(2\pi)^2} \  \sum_{b=1}^{\infty}\int_{\mathcal{B}}\abs{1+E_b(k)+g}^{s}
\abs{\inner{\Phi_b(\cdot,k),f(\cdot)}_{L^2(\R)}}^2dk \nonumber\\
&\approx \sum_{b=1}^{\infty}\int_{\mathcal{B}}(1+b^2+g)^{s}
\abs{\inner{\Phi_b(\cdot,k),f(\cdot)}_{L^2(\R)}}^2dk. \qquad \Box
\end{align} 

\begin{remark}\label{L2H2}
Note that if $g$ is a function for which only finitely many $\widetilde{g}_b(k)$ are non-zero, then
$\|g\|_{H^s(\R)}\lesssim \|g\|_{L^2(\R)}$.
\end{remark}

\section{Poisson summation in $L^2_{\rm  loc}$}\label{poisson-summation}
We shall require a variant of the Poisson summation formula which holds in $L^2_{\rm loc}$.
 
\begin{theorem}\label{psum-L2} (See also Theorem \ref{psum-L2a})
Let $\Gamma(x,X)$ be a function defined for $(x,X)\in\R\times\R$. Assume that  $x\mapsto\Gamma(x,X)$ is $H^2_{\rm
periodic}([0,1])$ with respect to $x$ with values in $L^2(\R_X)$, {\it i.e.}
\begin{align}\label{Gamma-cond1}
&\Gamma(x+1,X)\ =\ \Gamma(x,X),\\
&\sum_{j=0}^2\ \int_0^1\ \left\|\D_x^j\Gamma(x,\cdot)\right\|_{L^2(\R_X)}^2\ dx\ <\ \infty .
\label{Gamma-cond2}\end{align}
We denote this Hilbert space of functions by $\mathbb{H}^2$ with norm-squared, $\|\cdot\|_{\mathbb{H}^2}^2$, given in
\eqref{Gamma-cond2}.
Denote by $\widehat{\Gamma}(x,\omega)$ the Fourier transform of $\Gamma(x,X)$ with respect to $X$
given by
\begin{equation}
\widehat{\Gamma}(x,\omega)\ \equiv\ \lim_{N\uparrow\infty}\ \frac{1}{2\pi}\int_{|x|\le N}e^{- i\omega X}\Gamma(x,X)dX ,
\label{Gammahat-def}
\end{equation}
where the limit is taken in $L^2([0,1]_x\times\R_\omega)$. \ Fix an arbitrary $\zeta_{\rm max}>0$. Then, 
\begin{align}
\label{psum-Gofx-X} 
\sum_{n\in\Z} e^{- i\zeta(x+n)}\Gamma(x,x+n) &= 2\pi\ \sum_{n\in\Z} 
e^{2\pi i n x}\widehat{\Gamma}\left(x,2\pi n+\zeta \right) \\
&\qquad\qquad \textrm{in $L^2([0,1]\times[-\zeta_{\rm max},\zeta_{\rm max}];dxd\zeta)$.} \nonumber
\end{align}
\end{theorem}
\medskip

Results of this type (in $L^p$) were obtained in \cite{Boas:46}. 
 We give a different, self-contained, proof in Appendix \ref{PSF}.

%
\chapter{Dirac Points of 1D Periodic Structures}\label{1dperiodic&dirac}

In this section we define what we mean by a Dirac point of a one-dimensional periodic structure or loosely, a ``1D Dirac
point''.  We then introduce a concrete one-parameter family of operators, $H(s)$, which for $s=\frac12$ has a 1D Dirac
point.
\medskip

\begin{definition}\label{dirac-pt-gen}
Consider the Schr\"odinger operator $H=-\D_x^2+Q(x)$, where $Q(x+1)=Q(x)$. 
 We say that a \underline{linear band crossing of Dirac type} occurs at the quasi-momentum
$k_\star\in\mathcal{B}$ and energy $E_\star$, or loosely \underline{$(k_\star,E_\star)$ is a 1D Dirac point}, if the
following holds:
\begin{enumerate}
\item There exists $b_{\star}\geq1$ such that $E_{\star}=E_{b_{\star}}=E_{b_{\star}+1}$.
\item $E_\star$ is an $L^2_{k_\star}$ eigenvalue of multiplicity $2$.
\item $H_{k_\star}^2=H^2_A\oplus H^2_B$, where $H:H^2_A\to L^2_A$ and $H:H^2_B\to L^2_B$.
\item There is an operator $\mathcal{S}:H^2_A\to H^2_B$ and $\mathcal{S}:H^2_B\to H^2_A$,\
$\mathcal{S}\circ\mathcal{S}=I$,\ such that $\mathcal{S}$ commutes with $H(k_\star)\equiv e^{-ik_\star x}He^{ik_\star x}$. That is, 
$[H(k_\star),\mathcal{S}]\equiv H(k_\star)\mathcal{S}-\mathcal{S}H(k_\star)$ vanishes. 
\item The $L^2_{k_\star}-$ nullspace of $H-E_\star I$ is spanned by
 \[\left\{\ \Phi_1(x)\ ,\ \Phi_2(x)\equiv \mathcal{S}\left[\Phi_1\right](x)\ \right\},\quad 
\left\langle\Phi_a,\Phi_b\right\rangle_{L^2[0,1]}\ =\ \delta_{ab},\ a,b=1,2. \]
\item There exist $\lambda_{\sharp}\ne0$, $\zeta_0>0$ and Floquet-Bloch eigenpairs 
\[ (\Phi_+(x;k),E_+(k))\ \ {\rm and}\ \   (\Phi_-(x;k),E_-(k))\ ,\]
and smooth functions $\eta_\pm(k)$, with $\eta_\pm(0)=0$,
defined for $|k-k_\star|<\zeta_0$ and such that
\begin{equation}
 E_\pm(k)-E_\star\ =\ \pm\ \lambda_\sharp\ \left(k-k_\star\right)\ \Big(1+\eta_\pm(k-k_\star)\Big).
\label{1d-crossing} \end{equation}
\end{enumerate}
\end{definition}

\begin{remark} As remarked earlier, the eigenvalue maps $k\mapsto E_j(k)$ are, in general, only
Lipschitz continuous \cite{FW:14}.
Here, in 1D, the dispersion locus near a Dirac point is the union of smooth, transversely intersecting, curves. 
Let $(k_\star,E_\star)$ denote a Dirac point in the sense of Definition \ref{dirac-pt-gen}.
Then,
\begin{align}
\label{E_b_star_defn}
E_{b_\star}(k) = 
\begin{cases} 
E_+(k),\ &\mbox{if } k\in\mathcal{B},\ k_\star-\zeta_0 <k\le k_\star \\
E_-(k),\ &\mbox{if } k\in\mathcal{B},\  k_\star\le k < k_\star+\zeta_0
\end{cases}
\end{align}
and 
\begin{align}
\label{E_b1_star_defn}
E_{b_\star+1}(k) = 
\begin{cases} 
E_-(k),\ &\mbox{if } k\in\mathcal{B},\ k_\star-\zeta_0<k\le k_\star \\
E_+(k),\ &\mbox{if } k\in\mathcal{B},\  k_\star\le k<k_\star+\zeta_0
\end{cases}\ .
\end{align}
Figure \ref{fig:unperturbed_vs_kappa} (left panel) illustrates Dirac points
of a periodic potential, $V_\ee$ of the type plotted in Figure \ref{fig:efunc_plots}  (top panel). 
\end{remark}
\begin{figure}
\includegraphics[width=\textwidth]{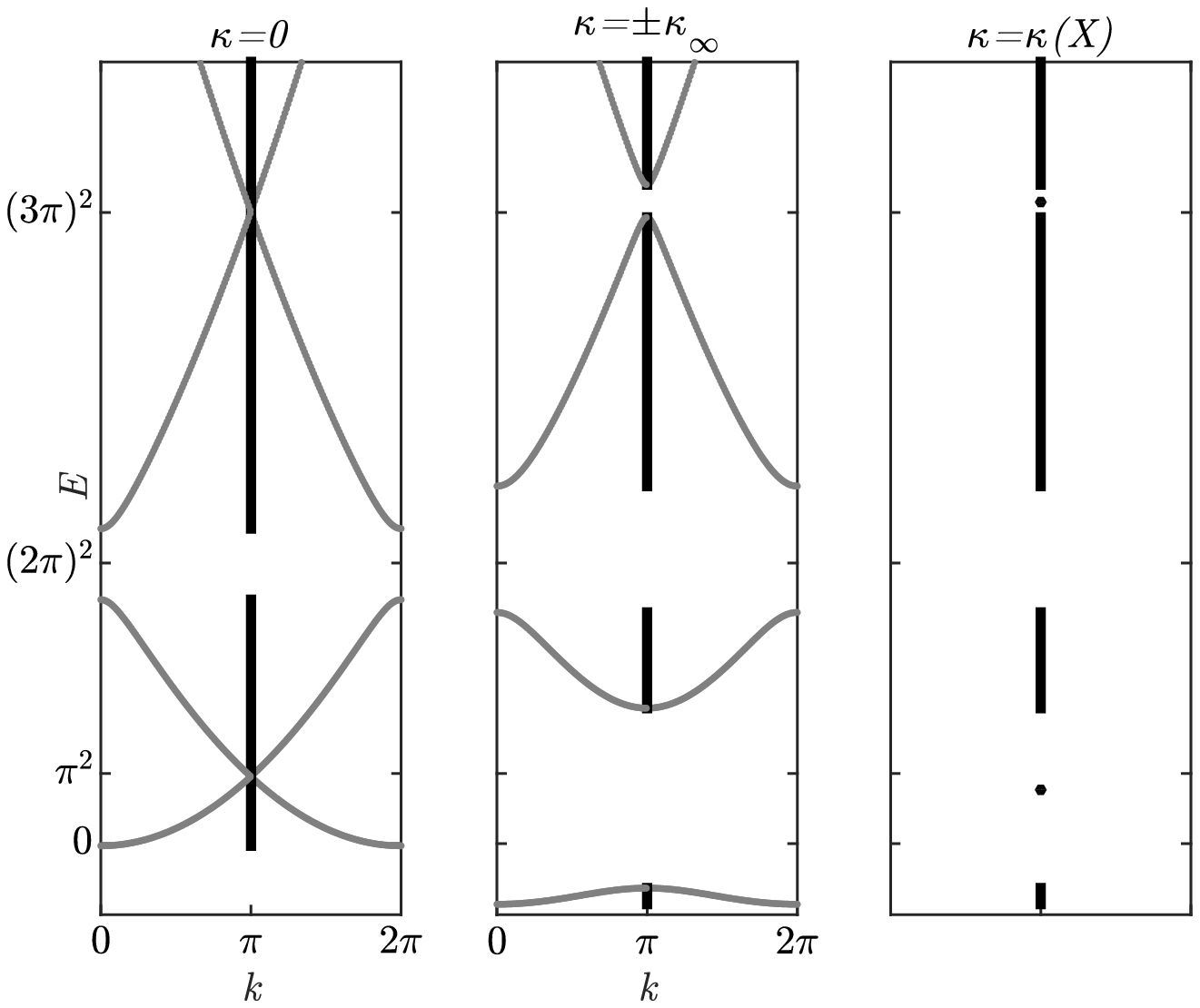}
\caption{Spectra (black) for $H_{\delta=15} = -\D_x^2+10\cos(2(2 \pi x))+30\kappa(X) \cos(2\pi x)$ for different choices
of $\kappa(X)$. {\bf Left panel:} $\kappa(X)\equiv0$. Plotted are the first four Floquet-Bloch dispersion curves
$k\mapsto E_b(k)$ (gray) and the Dirac
points $(k_\star=\pi,E_\star)$ approximately at $E_\star=\pi^2, (3\pi)^2$. {\bf Middle panel:} 
$\kappa(X)\equiv\kappa_\infty$, a non-zero constant.  Shown are open gaps about Dirac points
of the unperturbed potential and smooth dispersion curves (gray).  {\bf Right panel:}\  
Mid-gap eigenvalues (black dots) are shown for the periodic potential modulated by domain-wall:  $\kappa(X)=\tanh(X)$.}
\label{fig:unperturbed_vs_kappa}
\end{figure}
\medskip

\begin{remark}\label{2D-lip-surf} In 2D, the dispersion locus of honeycomb structures with Dirac points is locally
conical, the union of Lipschitz surfaces \cite{FW:12}.
\end{remark}
\medskip

\begin{remark}\label{moser:81}
In \cite{Moser:81},  Dirac points play a central role in the construction of 1-dimensional almost periodic potentials for 
which the Schr\"odinger operator has nowhere dense spectrum.
\end{remark}

\section{The family of Hamiltonians, $H(s)$, and its Dirac points for $s=\frac12$}\label{1d-toy}
We consider a  family of Hamiltonians depending on a real parameter, $s$, discussed in Section   \ref{sec:motivation} of Chapter \ref{introduction}:
\begin{align}
H(s)\ &=\ -\D_x^2+\mathcal{Q}(x;s), 
\label{Hofs}\\
\mathcal{Q}(x;s)\ &=\ \ \mathcal{Q}_0+\sum_{p\ge1} \mathcal{Q}_p\ \cos(\pi p s)\ \cos(2\pi p x), 
\label{Vxs}
\end{align}
where $\{\mathcal{Q}_p\}_{p\in\Z}$ is a real sequence 
which tends to zero rapidly as $p\to\infty$.  For each real $s$, 
(i)\ $x\mapsto\mathcal{Q}(x;s)$  is real-valued and smooth,  (ii)\  $1-$ periodic, $\mathcal{Q}(x+1;s)=\mathcal{Q}(x;s)$
and 
(iii)\ even, $\mathcal{Q}(-x;s)=\mathcal{Q}(x;s)$ . \medskip

\nit Since $\mathcal{Q}(x;s)$, is even,  we have the following:
\begin{proposition}\label{inversion}
For any $s\in[0,1]$, 
$H(s) = -\D_x^2+\mathcal{Q}(x;s)$\ commutes with the inversion operator\ $\mathcal{I}$\  defined by
\begin{align}
\mathcal{I}[f](x)\ &=\ f(-x)\ .\label{eqn-inversion}
\end{align}
Therefore if $H(s)\Phi=\mu\ \Phi$, then $H(s)\mathcal{I}\left[\Phi\right] = \mu\ \mathcal{I}\left[\Phi\right]$\ .
\end{proposition}
\medskip

\section{$H\left(\frac12\right )=-\D_x^2+\mathcal{Q}\left(x;\frac12\right)$ has an additional translation symmetry}
By \eqref{Vxs}  we have for  $s=1/2$ that $
\mathcal{Q}\left(x+\frac12;\frac12\right)\ =\ \mathcal{Q}\left(x;\frac12\right),\ \textrm{for all}\ x\in\R.
$ Hence, $\mathcal{Q}(x;\frac12)$ is an even-index cosine series:
\[ \mathcal{Q}\left(x;\frac12\right)\ =\ \sum_{m\in 2\mathbb{Z}_+} \mathcal{Q}_m\cos(2\pi mx)\ .\]

\medskip

To explore the consequences of this extra translation symmetry, we  introduce the following two key subspaces of
$L^2_k$.
\begin{definition}\label{L2eL2o}
\begin{enumerate}
\item $L^2_{k,\ee}$ is the subspace of $L^2_k$ consisting of functions of the form 
\[e^{ikx}P_\ee(x)\ =\ e^{ikx}\ \sum_{m\in2\Z}p(m) e^{2\pi imx},\ \ \sum_{m\in2\Z}|p(m)|^2<\infty,\]
{\it i.e.}  $P_\ee(x)$ is an even-index $1-$ periodic Fourier series.
\item $L^2_{k,\oo}$ is the subspace of $L^2_k$ consisting of functions of the form 
\[e^{ikx}P_\oo(x)\ =\ e^{ikx}\ \sum_{m\in2\Z+1}p(m) e^{2\pi imx},\ \ \sum_{m\in2\Z+1}|p(m)|^2<\infty,\]
{\it i.e.}  $P_\oo(x)$ is an odd-index $1-$ periodic Fourier series.
\item Sobolev spaces,
\[ H^M_{k,\ee} \ {\rm and}\  H^M_{k,\oo},\quad M=0,1,2,\dots,\]  are defined in the natural way. 
\item For $k=0$, we shall used the simplified notations: 
\[ L^2_\sigma=L^2_{0,\sigma}\  {\rm and}\  H^M_\sigma=H^M_{0,\sigma},\ \ \sigma=\ee,\oo.\]
\end{enumerate}
$L^2_\ee-$ functions are one-periodic even-index cosine series and 
$L^2_\oo-$ functions are one-periodic odd-index cosine series. 
\end{definition}
\medskip

\nit Clearly, we have for any $M\ge0$:\ $H^M_k\ =\ H^M_{k,\ee}\ \oplus\ H^M_{k,\oo}.$
\medskip

The quasi-momentum, $k_\star=\pi$, is distinguished by the following property of  $\mathcal{I}$:
\medskip

\begin{proposition}\label{Imaps}
Let $\Phi\in L^2_{\pi,\ee}$ with expansion
\begin{equation}
\Phi(x) = e^{\pi ix}\sum_{m\in2\mathbb{Z}}c(m)e^{2\pi i m x}\ .
\label{Psi-pi-ee}\end{equation}
Then, $\mathcal{I}\left[\Phi\right]\in  L^2_{\pi,\oo}$ with expansion
\begin{equation}
 \label{Psi-pi-oo}
 \mathcal{I}\left[\Phi\right](x)=\Phi(-x) = e^{\pi ix}\sum_{m\in2\mathbb{Z}+1}c(-m-1)e^{2\pi i m x}.
 \end{equation}
 \end{proposition}
 \medskip
 
 \nit {\it Proof of Proposition \ref{Imaps}:}\ 
 \begin{equation}
 \label{inverse_soln}
 \begin{split}
 \Phi(-x)&= e^{-\pi ix}\sum_{m\in2\mathbb{Z}}c(m)e^{-2\pi i m x} = e^{\pi ix}\sum_{n\in2\mathbb{Z}+1}c(-n-1)e^{2\pi i n
x}\ .\qquad \Box
 \end{split}
\end{equation} 

\section{The action of $-\D_x^2+V_\ee\left(x\right)$ on $L^2_{k_\star=\pi}$}
For potentials which are one-periodic even-index cosine series we have

\begin{proposition} \label{H-action-on-L_e}
$V_\ee(x)\in L^2_\ee$ and therefore, for $M\ge2$, 
\begin{align*}
-\D_x^2+V_\ee\left(x\right)\ &:\ H^M_{k,\ee}\ \to\ H^{M-2}_{k,\ee}\ , \\
-\D_x^2+V_\ee\left(x\right)\ &:\ H^M_{k,\oo}\ \to\ H^{M-2}_{k,\oo}\ .
\end{align*}
It follows that  $\Phi(x)$ is an $ L^2_{\pi,\ee}$ eigenfunction of $-\D_x^2+V_\ee\left(x\right)$, 
with eigenvalue $E$ if and only if  $\mathcal{I}\left[\Phi\right](x)$ is a linearly independent 
$ L^2_{\pi,\oo}$ eigenfunction of $-\D_x^2+V_\ee\left(x\right)$, 
with eigenvalue $E$.
\end{proposition}

\medskip

In the coming subsections we show that $-\D^2_x+V_\ee(x)$, with $V_\ee\in L^2_\ee$, smooth and generic, has 1D Dirac points $(k_\star=\pi,E_\star)$ in the sense of Definition \ref{dirac-pt-gen}.

\section{Spectral properties of $H^{(\eps=0)}=-\D_x^2$ in $L^2_k$\label{eps_zero}}
For $\eps=0$ the family of Floquet-Bloch eigenvalue problems reduces to 
\begin{align}\label{FBeps0}
\left(-\partial_x^2-\mu\right)\Phi(x) &= 0,\ \  x\in\R,\\
\Phi(x+1) &= e^{i k}\Phi(x),
\nn\end{align}
where $k\in\mathcal{B}=[0,2\pi]$. Eigensolutions are of the form
\begin{equation}
\Phi(x)\ =\ e^{ikx}\ e^{2\pi imx},\ \ m\in\Z,
\label{Psim}
\end{equation}
with corresponding eigenvalues
\begin{equation}
\mu_m(k)\ =\ (k+2m\pi)^2,\ \ m\in\Z.
\label{mu-m}
\end{equation}
We seek $k\in[0,2\pi]$ for which there are degenerate eigenvalues. These occur when $\mu_m(k)=\mu_n(k)$, with
$m\ne n$. That is, $(k+2m\pi)^2=(k+2n\pi)^2$ or equivalently
\begin{equation}
 (m-n)\left(\ k\ +\ \pi(m+n)\ \right)\ =\ 0\ .
 \label{kmn}\end{equation}
Since $m\ne n$ and $k\in\mathcal{B}$, we have that either
\begin{equation}
k=0\ \ {\rm and}\ m=-n
\nonumber\end{equation}
or 
\begin{equation}
k=k_\star\equiv \pi \ \ {\rm and}\ m+n+1=0 \ .
\label{kstardef}\end{equation}
\medskip

\begin{remark}\label{k0_remark}
A Lyapunov-Schmidt reduction argument (see Section \ref{conds4dirac-appendix} of Appendix \ref{linear-band-crossing}) shows that the 
degeneracies at quasi-momentum $k=0$ are ``lifted'' for
$\varepsilon>0$. In contrast, the degeneracies at $k=\pi$ persist for $\eps\ne0$; see Section \ref{generic-Dirac}.
\end{remark}
\medskip

Focusing on the case $k_\star=\pi$ and $m=-n-1$, \eqref{kstardef}:
\begin{equation}
E^{(0)}_{\star,m}\ \equiv\ \mu_m(\pi)\ =\ \mu_{-m-1}(\pi)\ =\ (2m+1)^2\ \pi^2,\qquad m=0,1,2,\dots \ ,
\label{mu-starm}
\end{equation}
is a sequence of multiplicity two eigenvalues of the Floquet-Bloch eigenvalue problem \eqref{FBeps0} for quasi-momentum
$k=k_\star=\pi$. See the left panel of Figure \ref{fig:unperturbedspectrum}, where two such crossings, for
$k=k_\star=\pi$, are shown.

Corresponding to $E^{(0)}_{\star,m}$ is  the two-dimensional $L^2_{k_\star}-$ null space of $H^{(0)}$:
\begin{equation}
{\rm span}\left\{ \Phi^{(0)}_m(x)=e^{i\pi x}\ e^{2\pi i m x},\ \ e^{i\pi x}\ e^{2\pi i (-m-1) x}\ =\ \Phi^{(0)}_m(-x)\
\right\} .
\label{nullspace-m}
\end{equation}
Furthermore, recalling the decomposition $L^2_{k_\star}=L^2_{k_\star,\ee}\oplus L^2_{k_\star,\oo}$, we have that for
each $m\ge0$:
\begin{enumerate}
\item $E^{(0)}_{\star,m}$ is a simple $L^2_{k_\star,\ee}-$ eigenvalue, with eigenspace
\[ {\rm span}\left\{ \Phi^{(0)}_m(x)=e^{i\pi x}\ e^{2\pi i m x}\right\}\subset L^2_{\pi,\ee},
\quad \textrm{if $m$ is even},  \]
  and 
eigenspace \[ {\rm span}\left\{ \Phi^{(0)}_m(-x)=e^{i\pi x}\ e^{2\pi i (-m-1) x}\right\}\subset L^2_{\pi,\ee},\quad
\textrm{ if $m$ is odd.}\]

\item $E^{(0)}_{\star,m}$ is a simple $L^2_{k_\star,\oo}-$ eigenvalue, 
with eigenspace \[ {\rm span}\left\{ \Phi^{(0)}_m(-x)=e^{i\pi x}\ e^{2\pi i (-m-1) x}\right\}\subset L^2_{\pi,\oo},\quad
\ \textrm{ if $m$ is even, }\]
and 
eigenspace  \[ {\rm span}\left\{ \Phi^{(0)}_m(x)=e^{i\pi x}\ e^{2\pi i m x}\right\}\subset L^2_{\pi,\oo},
\quad \textrm{if $m$ is odd.}\] 
\end{enumerate}

\section{Sufficient conditions for occurrence of a 1D Dirac point}

The following result gives sufficient conditions to be established for $(k_{\star},E_\star)$ to be a Dirac point in the sense of
Definition \ref{dirac-pt-gen}.  The proof, given in Section  \ref{conds4dirac-appendix} of Appendix \ref{linear-band-crossing}, closely follows that of Theorem
4.1 of \cite{FW:12}.
\medskip

 \begin{theorem}\label{conditions-for-dirac}
 Consider $H=-\D_x^2+V_\ee$, where $V_\ee\in L^2_\ee$  and is sufficiently smooth. 
 Thus, $V$ has minimal period $1/2$. Let $k_\star=\pi$ and assume that $E_\star$ is a double eigenvalue, lying at the
intersection of the $b_{\star}^{th}$ and $(b_{\star}+1)^{st}$ spectral bands:
 \[ E_\star\ =\ E_{b_{\star}}(k_\star)\ =\ E_{b_{\star}+1}(k_\star)\ .\]
 Assume the following conditions:
\begin{enumerate}[I.]
 \item $E_{\star}$ is a simple $L^2_{k_{\star},\ee}$- eigenvalue of $H$ with 
1-dimensional
eigenspace \[{\textrm span}\{\Phi_1(x)\}\subset L^2_{k_{\star},\ee}.\]
 \item $E_{\star}$ is a simple $L^2_{k_{\star},\oo}$- eigenvalue of $H$ with 
1-dimensional
eigenspace \[{\textrm span}\Big\{\Phi_2(x)=\mathcal{I}\left[\Phi_1\right](x)=\Phi_1(-x)\Big\}\subset
L^2_{k_{\star},\oo}.\]
 \item Non-degeneracy condition:
 \begin{equation}
  0\ne\lamsharp \equiv 2i\left\langle\Phi_1,\D_x\Phi_1\right\rangle\ = 
-2\pi\left\{2\sum_{m\in2\mathbb{Z}}m\abs{c_1(m)}^2+1\right\}.\ \label{lambda-sharp}
 \end{equation} Here, $\{c_1(m)\}_{m\in\Z}$ denote the $L^2_{k_\star,\ee}-$ Fourier coefficients of
$\Phi_1(x)$.
\end{enumerate}
\medskip

Then, $(k_{\star}=\pi,E_\star)$ is a Dirac point in the sense of Definition \ref{dirac-pt-gen} with 
\begin{itemize}
\item Symmetry:\ $\mathcal{S}=\mathcal{I}$.
\item Decomposition: $H^s_{k_\star}= H^s_{k_\star,\ee}\oplus H^s_{k_\star,\oo}$.
\end{itemize}
\end{theorem}
\medskip

\section{Expansion of Floquet-Bloch eigenfunctions near a Dirac point}
\label{sec:flo-blo-dirac}


\begin{proposition}\label{flo-blo-dirac}
Let $(k_\star,E_\star)$ denote a 1D Dirac point in the sense of Definition \ref{dirac-pt-gen}.
Recall the basis $\{\Phi_1(x),\Phi_2(x)\}$ of the $L^2_{k_\star}-$ nullspace of $H-E_\star I$ in  Definition
\ref{dirac-pt-gen} and introduce the periodic functions
 \begin{equation}
 p_1(x)=e^{-ik_\star x}\Phi_1(x),\ \ p_2(x)=e^{-ik_\star x}\Phi_2(x).
 \label{pj-def}\end{equation}
 
\noindent Further, let $(\Phi_\pm(x;k),E_\pm(k))$ denote $k-$ pseudo-periodic eigenpairs as in part 6 of 
 Definition \ref{dirac-pt-gen}.   Introduce the periodic functions $p_\pm(x;k)$ by
 \begin{equation}
 \Phi_\pm(x;k)\ =\ e^{ikx}\ p_\pm(x;k),\ \ \left\langle p_a(\cdot;k),p_b(\cdot;k)\right\rangle=\delta_{ab} ,\ \ 
 a,b\in\{+,-\}.
 \label{p_pm-def}
 \end{equation}
\medskip
 
 \noindent Let $k=k_\star+k'$. Let $\lambda_{\sharp}\in\R$ be as in \eqref{lambda-sharp}.
 Then, there is a constant $\zeta_0>0$ such that  for $0<|k'|<\zeta_0$  we have the following:
 \begin{enumerate}
 \item The mapping $k\mapsto E_\pm(k)$ is smooth near $k_\star$ and 
 \begin{equation}
 E_\pm(k_\star+k')=E_\star \pm \lambda_\sharp\ k' + \mathcal{O}( |k'|^2 ).
 \label{Epm-expand}\end{equation}
 Here, $E_\pm'(k_\star)=\pm\lambda_\sharp$, where  $\lambda_\sharp$ is given by \eqref{lambda-sharp}. Note also that 
 \begin{equation}
 \lambda_\sharp = 2i\left\langle\Phi_1,\D_x\Phi_1\right\rangle\ =\ -2i\left\langle\Phi_2,\D_x\Phi_2\right\rangle .
 \label{pm-lambda_sharp}\end{equation}
  \item $p_\pm(x;k)$, a priori defined up to an arbitrary complex multiplicative constant of absolute value $1$, can be
chosen so that 
 \begin{align}
 p_-(x;k_\star+k')\ &=\ c_{-}(k')\left(\ p_1(x) + \varphi_-(x;k')\ \right), \nn \\
 p_+(x;k_\star+k')\ &=\ c_{+}(k')\left(\ p_2(x) + \varphi_+(x;k')\ \right).
 \label{p-pm-expand}
 \end{align} 
 The maps $k'\mapsto\varphi_{\pm}(x;k'),\ c_{\pm}(k')$ are smooth with $c_{\pm}(k')=1+\mathcal{O}_\pm(k')$  and $\varphi(x;k')=\mathcal{O}(k')$ for $x\in[0,1]$ and $|k'|<\zeta_0$.
 \end{enumerate}
\end{proposition}

The proof of Proposition \ref{flo-blo-dirac} is given in Section  \ref{flo-blo-appendix} of Appendix \ref{linear-band-crossing}. 

\section{Genericity of Dirac points at $k=k_\star$}
\label{generic-Dirac}
Fix $V\in L^2_\ee$, a one-periodic and sufficiently smooth potential.  
We consider the one-parameter family of Hamiltonians
\begin{equation}
H^{(\eps)}\ =\ -\D_x^2+\eps V(x),
\label{Heps}
\end{equation}
where $\eps\in\R$.  
We denote by $\Omega=[0,1]$ the unit period cell of $V(x)$ and  $\mathcal{B}=[0,2\pi]$ is a choice of fundamental dual
period cell, chosen so that the quasi-momentum, $k_{\star}=\pi$, is an interior point.

\medskip

\begin{theorem}[Dirac points / linear band crossings for small $\eps$]\label{thm:dirac-pt}

Fix  $n\ge1$.  There exists $\eps_0=\eps_0(n)>0$ and  a real-valued  continuous function, defined on $(0,\eps_0)$:
 \[ \eps\mapsto E^{(\eps)}_{\star,n},\ \ {\rm with}\
\left.E^{(\eps)}_{\star,n}\right|_{\eps=0}=\pi^2(2n+1)^2, \]
  such that
  $(k_\star=\pi,E^{(\eps)}_{\star,n})$ is a Dirac point in the sense of Definition \ref{dirac-pt-gen}
 with symmetry $\mathcal{S}=\mathcal{I}$ and decomposition $H^2_{k_\star}=H^2_{k_\star,\ee}\oplus H^2_{k_\star,\oo}$.
\end{theorem}
\smallskip

The proof of Theorem \ref{thm:dirac-pt} is given in Appendix \ref{dirac-pt-small-eps}.
\medskip

\begin{theorem}[Dirac points / linear band crossings for generic $\eps$]\label{thm:dirac-pt-gen}

Let $\eps_0(n)$ be as in Theorem \ref{thm:dirac-pt}. Then, for all real $\eps$ except possibly a discrete set outside $(-\eps_0,\eps_0)$,   the
operator $H^{(\eps)}$ has a Dirac point at some $\left(k_\star=\pi,E_\star^{(\eps)}\right)$ in the sense of Definition
\ref{dirac-pt-gen}
 with symmetry $\mathcal{S}=\mathcal{I}$ and decomposition $H^2_{k_\star}=H^2_{k_\star,\ee}\oplus H^2_{k_\star,\oo}$.
\end{theorem}
\smallskip

\nit Theorem  \ref{thm:dirac-pt-gen} is proved in Appendix \ref{tC-discrete!}.
 The proof in fact yields infinitely many Dirac points.
\medskip

 \chapter{Domain Wall Modulated Periodic Hamiltonian and Formal Derivation of Topologically Protected Bound States}\label{sec:bound_states} 

Let $V_\ee(x)$ denote a potential which has a Dirac point $(E_\star,k_\star=\pi)$
 and consider the smooth perturbed potential:
\begin{equation}
 U_\delta(x) = V_{\ee}(x)+\delta\kappa(\delta x)W_{\oo}(x)\ .
 \label{domain-wall-pot}
\end{equation} 
$V_\ee$ and $W_\oo$ denote, respectively, even- and odd- index Fourier series:
\begin{align*}
 V_\ee(x) &= \sum_{p\in2\mathbb{Z}_{_+}} v_p\cos(2\pi px), \\
 W_{\oo}(x) &= \sum_{p\in2\mathbb{Z}_{_+}+1} w_p\cos(2\pi px),
\end{align*} 
where  $\kappa(X)$ is $C^{\infty}(\R)$ with constant
asymptotic values $\pm\kappa_\infty$:
\begin{equation}
\lim_{X\rightarrow\pm\infty} \kappa(X) = \pm\kappa_{\infty},\ \ \ \   \kappa_{\infty}>0\ .
\label{kappa-type}
\end{equation}

\nit Our next result concerns the essential spectrum of $H_\delta$ and the (formal) asymptotic operators:
\begin{equation}
H_{\delta,+} = -\D_x^2 + V_\ee(x)+\delta\kappa_\infty W_\oo(x)\ \ {\rm and}\ \ 
H_{\delta,-} = -\D_x^2 + V_\ee(x)-\delta\kappa_\infty W_\oo(x).
\label{Hpm-def}
\end{equation}

\begin{proposition}\label{kappa_spec}
Let $V_\ee\in L^2_\ee$ and let $(k_\star,E_\star)$ denote a Dirac point of $-\D_x^2+V_\ee$.
Assume that $\vartheta_\sharp=\inner{\Phi_1,W_0\Phi_2}_{L^2[0,1]}\ne0$; see also Remark \ref{on-theta-sharp}.
Fix $c$ less than but arbitrarily close to $1$ and define the real interval
 \[\mathcal{I}_\delta\equiv
\Big(\ E_\star-c\delta\kappa_\infty|\thetasharp|\ ,\ E_\star+c\delta\kappa_\infty|\thetasharp|\ \Big)\ .
\]
 Then there exists $\delta_0>0$, such that for all
$0<\delta<\delta_0$, we have that
 \begin{align}
&\mathcal{I}_\delta\cap\sigma_{\rm ess}(H_\delta),\ \ \mathcal{I}_\delta\cap\sigma_{\rm ess}(H_{\delta,-})
\ \ {\rm and}\ \ \mathcal{I}_\delta\cap\sigma_{\rm ess}(H_{\delta,+})\ \textrm{ are all empty sets}. \nn
\end{align}
\end{proposition}

\nit{\it Proof of Proposition \ref{kappa_spec}:} We first prove the assertion for the periodic Hill's operators,
$H_{\delta,+}$ and 
$H_{\delta,-}$. Consider $H_{\delta,+}$. The proof for $H_{\delta,-}$ is identical. By Appendix \ref{band_splitting}
(\eqref{E_k_cond} and \eqref{E_Epm_defn}), for any $c\in(0,1)$, there exists $\delta_0>0$ such that for all
$0<\delta<\delta_0$ 
there are dispersion curves $E_{\delta,-}(k)\le E_{\delta,+}(k)$ such that 
for $k$ satisfying $|k-k_\star|=|k-\pi|<\delta$, we have:
\[\max_{|k-\pi|\le\frac{\delta}{2}} 
E_{\delta,-}(k) = E_\star - c \ \delta \kappa_\infty|\thetasharp|, \qquad
\min_{|k-\pi|\le\frac{\delta}{2}} 
E_{\delta,+}(k) = E_\star + c \ \delta \kappa_\infty|\thetasharp|\ .\]

To show that there is a gap in the spectrum of width $\approx \mathcal{O}(\delta\kappa_\infty)$, it suffices to show
that for $k'\in\mathcal{B}\cap\left\{k:|k-\pi|\ge\frac{\delta}{2}\right\}$, that $E_{\delta,-}(k')<E_\star
-c\delta\kappa_\infty|\thetasharp|$ and $E_{\delta,-}(k')>E_\star + c\delta\kappa_\infty|\thetasharp|$. Indeed,
this follows from
the symmetry and monotonicity properties of $k\mapsto E_{\delta,\pm}(k)$ on $[0,\pi]$ and on $[\pi,2\pi]$; see
\eqref{Ej-sym-mon}.
Therefore, $H_{\delta,\pm}$ has a gap in its spectrum containing an interval $\mathcal{I}_\delta$. 

Finally, the assertion concerning $H_\delta$ follows from  Theorem 6.28 (page 267) of \cite{RofeBeketov-Kholkin:05}
using the above results on the asymptotic operators, $H_{\delta,\pm}$. This completes the proof of Proposition   
\ref{kappa_spec}.
\medskip

\section{Formal multiple scale construction of ``edge states''}\label{subsec:multiscale_analysis}

Consider the eigenvalue problem:
\begin{align}
 \label{perturbed_schro_prob}
H_\delta\Psi(x)&\equiv \left(-\partial_x^2+U_{\delta}(x)\right)\Psi(x)= E\ \Psi(x), ~~~x\in\R,\\
\Psi(x) &\rightarrow 0 ~\text{as}~x\rightarrow \pm\infty,\nn
\end{align} 
where $U_\delta(x)=V_\ee(x) +\delta\kappa(\delta x)W_\oo(x)$; see  \eqref{domain-wall-pot}.
\medskip

In this section we use a formal multiple scale expansion in $\delta$ to study the eigenvalue problem 
\eqref{perturbed_schro_prob}. In particular, we shall see how the hypothesized asymptotic conditions of the domain wall function,  $\kappa(X)$,
$\kappa(X)\to\pm\kappa_\infty$ as $|X|\to\infty$, ensures the bifurcation, from a Dirac point, of a family of  localized
eigenstates of
\eqref{perturbed_schro_prob}, for $\delta\ne0$, .  In Chapter \ref{sec:proof-exists-mode} we prove the
validity of this  expansion. 
\medskip

We seek a solution to \eqref{perturbed_schro_prob} as a two-scale  expansion in the small quantity $\delta$:
\begin{align}
E^{\delta}&=E^{(0)}+\delta E^{(1)}+\delta^2E^{(2)}+\ldots, \label{formal-E}\\
\Psi^{\delta}&= \psi^{(0)}(x,X)+\delta\psi^{(1)}(x,X)+\delta^2\psi^{(2)}(x,X)+\ldots\ \label{formal-psi}.
\end{align} 
where $X=\delta x$ is a slow spatial scale. 
 For each $i=0,1,2,\dots$, we impose the boundary conditions:
\begin{enumerate}
\item   $\psi^{(i)}(x,X)$ is $k_{\star}=\pi$-pseudo-periodic in $x$.
\item  $\psi^{(i)}(x,X)$ is spatially localized in $X$.
\end{enumerate}

Treating $x$ and $X$ as independent variables, equation \eqref{perturbed_schro_prob} may be written as
\begin{align*}
 &-\left(\partial_{x}^2+2\delta\partial_{x}\partial_{X}+\delta^2\partial_{X}^2+\ldots\right)
\left(\psi^{(0)}+\delta\psi^{(1)}+\delta^2\psi^{(2)}+\ldots\right) \\
&~~~+\left(V_{\ee}(x)+\delta\kappa(X)W_{\oo}(x)\right)
\left(\psi^{(0)}+\delta\psi^{(1)}+\delta^2\psi^{(2)}+\ldots\right) \\
&~~~-\left(E^{(0)}+\delta E^{(1)}+\delta^2E^{(2)}+\ldots\right)
\left(\psi^{(0)}+\delta\psi^{(1)}+\delta^2\psi^{(2)} +\ldots\right)=0.
\end{align*}
Equating terms of equal order in $\delta^j,\ j\ge0$ yields a hierarchy of equations governing $\psi^{(j)}(x,X)$.
  
  At order $\delta^0$ we have that $(E^{(0)},\psi^{(0)})$ satisfy
\begin{equation}
 \label{perturbed_schro_delta0}
 \begin{split}
 &\left(-\partial_{x}^2+V_{\ee}(x)-E^{(0)}\right)\psi^{(0)}(x,X)=0, \\
 &\psi^{(0)}(x+1,X) = e^{ik_{\star}}\psi^{(0)}(x,X).
 \end{split}
\end{equation} 
Equation \eqref{perturbed_schro_delta0} has the solution
\begin{equation}
E^{(0)}=E_{\star},\ \  \psi^{(0)}(x,X) = \alpha_1(X)\Phi_1(x)+\alpha_2(X)\Phi_2(x),
\label{psi0-soln}
\end{equation}
 where $\{\Phi_1(x),\Phi_2(x)\}$ is the basis of the $L^2_{k_{\star}}-$ nullspace of $H-E_{\star}$ in Definition
\ref{dirac-pt-gen}.

Proceeding to order $\delta^1$ we find that $(\psi^{(1)},E^{(1)})$ satisfies
\begin{align}
 \label{perturbed_schro_delta1}
 \left(-\partial_{x}^2+V_{\ee}(x)-E_{\star}\right)\psi^{(1)}(x,X)&=G^{(1)}(x, X;\psi^{(0)}) + 
 E^{(1)} \psi^{(0)},\\
 \psi^{(1)}(x+1,X) &= e^{ik_{\star}}\psi^{(1)}(x,X) \nonumber,
 \end{align}
 where
 \begin{align}
 &G^{(1)}(x, X;\psi^{(0)})\ \equiv\ \left(2\partial_{x}\partial_{X}
-\kappa(X)W_{\oo}(x)\right)\psi^{(0)}(x,X)\nonumber\\
&=\sum_{j=1}^2\Big[2\partial_{x}\Phi_j(x)\partial_{X}\alpha_j(X)
-\kappa(X)W_{\oo}(x)\alpha_j(X)\Phi_j(x)\Big]\ .\label{G1def}
\end{align} 
Equation \eqref{perturbed_schro_delta1} has a solution $\psi^{(1)}$ if and only if 
$G^{(1)}(x, X;\psi^{(0)})$, is $L_x^2[0,1]-$ orthogonal to the $L^2_{k_\star}-$ kernel of $(H_0-E_{\star})$,
which is spanned by $\Phi_1$ and $\Phi_2$. Thus we require that 
 $G^{(1)}(x, X;\psi^{(0)})$ be $L^2_x[0,1]-$ orthogonal to the normalized eigenfunctions $\Phi_1(x)$ and
$\Phi_2(x)$:
 \begin{equation}
  \left\langle \Phi_j(\cdot),G^{(1)}(\cdot, X;\psi^{(0)})\right\rangle_{L^2([0,1])}\ =\ 0,\ \ j=1,2\ \ ,
  \label{G1orthog}\end{equation}
where $G^{(1)}(x,X;\psi^{(0)})$ is displayed in \eqref{G1def}. 
The two orthogonality conditions \eqref{G1orthog} reduce to the eigenvalue problem for a 1D Dirac operator:
\begin{equation} 
 \label{dirac_eqn1}
 \left(\mathcal{D}-E^{(1)}\right)\alpha(X) = 0,\ \ \alpha(X)\to0, \ \ X\to\pm\infty.
\end{equation} 
Here,  $\alpha(X)=(\alpha_{1}(X),\alpha_{2}(X))^T$ and $\mathcal{D}$ denotes the one-dimensional Dirac operator:
\begin{equation}
  \label{dirac_op}
\mathcal{D}\equiv i\lamsharp\sigma_3\partial_{X}+\thetasharp\kappa(X)\sigma_1 \ .
\end{equation} 
The constants $\lambda_\sharp,\ \thetasharp$ are real and are given by:
\begin{align}
\lamsharp&=2i\inner{\Phi_1,\partial_x\Phi_1}_{L^2([0,1])}\ =\ -2i\inner{\Phi_2,\partial_x\Phi_2}_{L^2([0,1])},\label{thetasharp_defn}\\
\thetasharp&=\inner{\Phi_1,W_{\oo}\Phi_2}_{L^2([0,1])}\ =\ \inner{\Phi_2,W_{\oo}\Phi_1}_{L^2([0,1])}. \nn
\end{align} 
 Note: $\lamsharp$ is real by self-adjointness of $i\D_x$ and $\thetasharp$ is real since $W_\oo$ is real,
$W_\oo(x)=W_\oo(-x)$ and $\Phi_2(x)=\Phi_1(-x)$. 
\medskip

\begin{remark}\label{lambda_theta_generic}
$\lambda_\sharp$ is generically non-zero; see Theorem \ref{thm:dirac-pt} and Theorem
\ref{thm:dirac-pt-gen}.
We assume that $\thetasharp\ne0$; see Remark \ref{on-theta-sharp}.
\end{remark}
\medskip

In Section \ref{subsec:dirac_bound_state} of Chapter \ref{sec:bound_states} we prove  that for $\kappa(X)$ satisfying 
\[ \lim_{x\rightarrow\pm\infty}
\kappa(\delta x) = \pm\kappa_{\pm\infty},\ \ \kappa_{\pm\infty}>0,\]
 that the eigenvalue problem \eqref{dirac_eqn1} for the Dirac operator has an
  exponentially localized
eigenfunction $\alpha_{\star}(X)$ with corresponding (mid-gap) eigenvalue $E^{(1)}=0$.
 Moreover, this eigenvalue is simple (multiplicity one). We impose the normalization: $\|\alpha_\star\|_{L^2}=1$. 
\medskip

Fix $(E^{(1)},\alpha)=(0,\alpha_\star)$ as above.
Then, the solvability conditions \eqref{G1orthog} are satisfied and we may invert 
 $(H_0-E_{\star})$ on $G^{(1)}(x,X ;\psi^{(0)})$ obtaining:
\begin{align}
 \psi^{(1)}(x,X) &= \left(R(E_{\star})G^{(1)}\right)(x,X) + \psi^{(1)}_h(x,X)
 \equiv \psi^{(1)}_p(x,X) + \psi^{(1)}_h(x,X),\label{psi1p-def}
\end{align} where the resolvent,
\begin{equation}
R(E_{\star})\ \equiv\ (H_0-E_{\star})^{-1}:\ L^2_{k_\star}(\R_x)\to H^2_{k_\star}(\R_x)\ \ \textrm{is bounded.}
\label{REstar}
\end{equation} 
 $\psi^{(1)}_p$ is a
particular solution, and
\begin{equation*}
 \psi^{(1)}_h(x,X) = \alpha^{(1)}_{1}(X)\Phi_1(x)+\alpha^{(1)}_{2}(X)\Phi_2(x)
\end{equation*} is a homogeneous solution with coefficients $\alpha^{(1)}_{j}$ to be determined. 

We now proceed to $\mathcal{O}(\delta^2)$ in the perturbation hierarchy. Equating terms at order $\delta^2$ yields
\begin{align}
 \label{perturbed_schro_delta2}
& \left(-\partial_{x}^2+V_{\ee}(x)-E_{\star}\right)\psi^{(2)}(x,X)\\
  &\quad =\left(2\partial_{x}\partial_{X}
-\kappa(X)W_{\oo}(x)\right)\psi^{(1)}(x,X)\  +\left(\partial_X^2+E^{(2)}\right)\psi^{(0)}(x,X) \nonumber\\
&\quad = \left(2\partial_{x}\partial_{X}
-\kappa(X)W_{\oo}(x)\right)\psi_h^{(1)} + G^{(2)}(x,X;\psi^{(0)},\psi_p^{(1)}) +E^{(2)}\psi^{(0)}\nn\\
&\psi^{(2)}(x+1,X) = e^{ik_{\star}}\psi^{(2)}(x,X). \nonumber
\end{align}
where
\begin{equation}
 G^{(2)}(x,X;\psi^{(0)},\psi_p^{(1)})
 =  \left(2\partial_{x}\partial_{X}
-\kappa(X)W_{\oo}(x)\right)\psi_p^{(1)} + \D_X^2\psi^{(0)}\ . \label{G2def}
\end{equation}

As before, \eqref{perturbed_schro_delta2} has a solution if and only if the right hand side is $L^2_x[0,1]-$ orthogonal
 to the functions $\Phi_j(x)$, $j=1,2$.
Written out, this solvability condition reduces to the inhomogeneous Dirac system:
\begin{align}
 \mathcal{D}\alpha^{(1)}(X) &=
\mathcal{G}^{(2)}\left(X\right)+E^{(2)}\alpha_{\star}(X), \label{solvability_cond}\\
  \alpha^{(1)}(X) &\to0, \ {\rm as}\ X\to\pm\infty,\ \ {\rm where} \nn \\
 \mathcal{G}^{(2)}(X) &= 
 \left( \begin{array}{c}
 \left\langle \Phi_1(\cdot),G^{(2)}(\cdot,X;\psi^{(0)},\psi_p^{(1)}) \right\rangle_{L^2([0,1])} \\ 
  \left\langle \Phi_2(\cdot),G^{(2)}(\cdot,X;\psi^{(0)},\psi_p^{(1)}) \right\rangle_{L^2([0,1])}
 \end{array} \right)\ . \label{ipG}
\end{align} 
Solvability of the system \eqref{solvability_cond} is ensured by imposing $L^2_X-$ orthogonality
of  the right hand side of \eqref{solvability_cond} to $\alpha_\star(X)$. This yields:
\begin{equation}
 \label{solvability_cond_E2}
E^{(2)} =
-\inner{\alpha_{\star},\mathcal{G}^{(2)}}_{L^2(\R_X)}.
\end{equation}

Proceeding as earlier, we obtain $\psi^{(2)}=\psi_p^{(2)}+\psi_h^{(2)}$, where $\psi_p^{(2)}$ is a particular solution
of \eqref{perturbed_schro_delta2} and 
\begin{equation*}
 \psi^{(2)}_h(x,X) = \alpha^{(2)}_{1}(X)\Phi_1(x)+\alpha^{(2)}_{2}(X)\Phi_2(x)
\end{equation*} is a homogeneous solution.

\nit This systematic expansion procedure may be carried out to arbitrary order in $\delta$ yielding
$\psi^{(\ell)}(x,X)$, $E^{(l)}$, $\ell\geq0$ and the formal solution \eqref{formal-E}-\eqref{formal-psi}.

\section{Spectrum of the  1D  Dirac operator and its topologically protected zero
energy eigenstate}\label{subsec:dirac_bound_state}
 The formal multiscale analysis in Section \ref{subsec:multiscale_analysis}, applied to the  Schr\"odinger eigenvalue
problem with domain-wall potential, shows that the leading order in $\delta$ behavior 
 of the eigenvalue problem is governed by a one-dimensional Dirac operator,
 \begin{equation}
\mathcal{D}\equiv i\lamsharp\sigma_3\partial_{X}+\thetasharp\kappa(X)\sigma_1. \label{Ddef}
 \end{equation}
 In this section we characterize the essential spectrum of $\mathcal{D}$ and prove, under mild conditions on
$\kappa(X)$, that $\mathcal{D}$ has a simple zero energy eigenvalue, with corresponding ($L^2(\R_X)$- normalized) eigenstate $\alpha_\star(X)$. This zero energy state is topologically stable in the sense that $\mathcal{D}$  has a zero energy eigenstate  for arbitrary spatially localized perturbations of $\kappa(X)$.

 Substitution of this state, $\alpha_\star(X)$ (see \eqref{dirac-m-soln})  into \eqref{psi0-soln} yields a formal leading order localized eigen-solution of the Schr\"odinger eigenvalue problem, for $\delta\ne0$ and small.
 We establish this rigorously  in Chapter \ref{sec:proof-exists-mode}. \medskip

\begin{theorem}
 \label{thm:dirac_bound_state}[{\bf Spectrum of $\mathcal{D}$}]\smallskip
  
\nit Assume $\vartheta_\sharp\ne0$ and let $\kappa_{\pm\infty}$ denote nonzero real constants. Assume $\kappa(X)$ is
bounded and that 
\[ \kappa-\kappa_{+\infty}\in L^1[0,\infty]\ \textrm{ and}\  \kappa+\kappa_{-\infty}\in L^1[-\infty,0].\]
 Then, the following holds:
\begin{enumerate}
\item $\mathcal{D}$ has essential spectrum equal to $(-\infty, -a]\cup 
[a, \infty)$, \\
where $a=\min\{ \abs{\thetasharp\kappa_\infty} , \abs{\thetasharp\kappa_{-\infty}}\}$.
\item Assume $\kappa_\infty\times\kappa_{-\infty}>0$.  Then, the eigenvalue problem for the one-dimensional Dirac
operator \eqref{Ddef}:
\begin{align}
 \label{dirac_eqn}
& \left(\mathcal{D}-\Omega\right)\alpha(X) = 0,\ \ \alpha(X)\to0\ \ {\rm as}\ \ X\to\pm\infty \ ,
\end{align}
has a simple eigenvalue at $\Omega=0$ with 
 localized eigenstate $\alpha_\star(X)$, which we may take  to be normalized $\|\alpha_\star\|_{L^2(\R)}\ =\ 1$.
\item If $\kappa'(X)\in\mathcal{S}(\R)$ then $\alpha_\star\in
H^{s}(\R)$, for all $s\in\mathbb{N}$. In fact, $\alpha_\star(X)$ and all its derivatives with respect to $X$ are
exponentially decaying functions of $X$.
\item Assume $\kappa_{+\infty}\times\kappa_{-\infty}<0$, {\it i.e.} $\sgn(\kappa(+\infty))\ =\sgn(\kappa(-\infty))$. 
Then $\Omega=0$ is not an eigenvalue of \eqref{dirac_eqn}. 
\end{enumerate}
\end{theorem}
\medskip

\begin{figure}
\includegraphics[width=\textwidth]{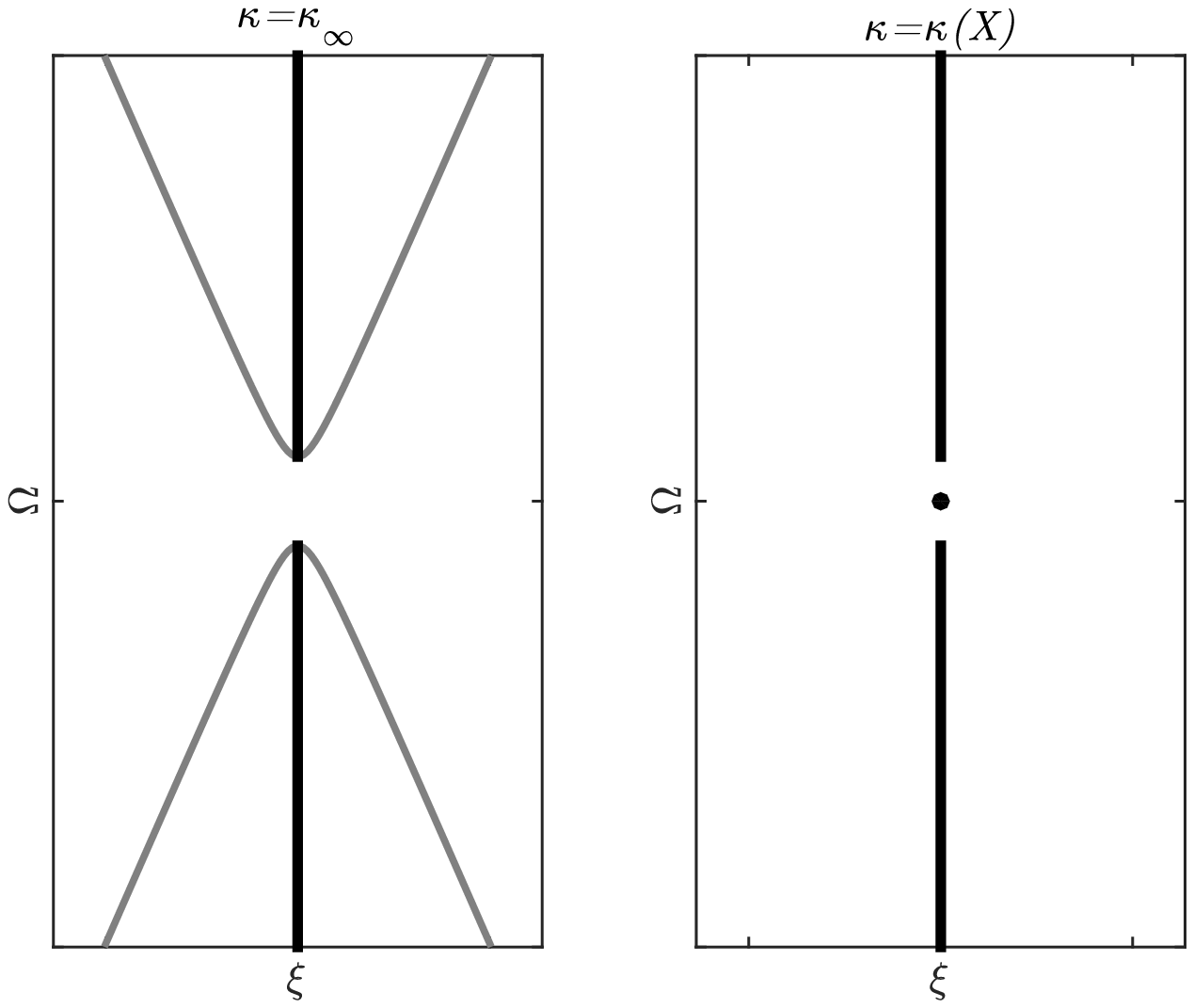}
\caption{{\bf Left panel:} Spectrum of the Dirac operator $\mathcal{D}$ for $\kappa\equiv\pm\kappa_\infty$
constant. Light gray curves are dispersion curves $\pm\Omega(\xi)$.  Continuous spectrum shown in
black. {\bf Right panel:} Spectrum for $\kappa\equiv\kappa(X)=\tanh(X)$, a domain wall. Black dot marks
the mid-gap eigenvalue $\Omega=0$.}
\label{dirac_kappa_spectrum}
\end{figure}

\begin{remark}[Topological Stability]
The eigenpair with eigenvalue zero  (part 2 of Theorem \ref{thm:dirac_bound_state}) is ``topologically stable'' or
``topologically protected''. Indeed,
 the zero-eigenvalue persists when 
perturbing $\kappa(X)$ arbitrarily within a bounded region, while preserving the  asymptotic behavior of $\kappa$:\ 
$\kappa(X)\to\pm\kappa_{\pm\infty}$ as $X\to\pm\infty$.
\end{remark}
\medskip

\begin{proofof} \textit{Proof of Theorem \ref{thm:dirac_bound_state}}.
(1) The essential spectrum of $\mathcal{D}$ can be computed by noting that at $\pm\infty$, $\mathcal{D}$ is a compact
perturbation of the operator $\mathcal{D}_{\pm}=
i\lamsharp\sigma_3\partial_{X}\pm\thetasharp\kappa_{\pm\infty}\sigma_1$.
\\
(2) Without loss of generality, assume that $\kappa_\infty>0$ and $\kappa_{-\infty}>0$. The eigenvalue problem
\eqref{dirac_eqn} with eigenvalue parameter $\Omega=0$ is:
\begin{equation}
\label{bound_state_proof_eqn1}
 \partial_{X}\alpha(X) = \frac{1}{\lamsharp}
 \begin{pmatrix}
  0&i\vartheta_\sharp\kappa(X)\\-i\vartheta_\sharp\kappa(X)&0
 \end{pmatrix}
\alpha(X).
\end{equation} 
Recall that $\lambda_\sharp\ne0$ and that we have assumed $\vartheta_\sharp\ne0$; see \eqref{thetasharp_defn} and Remark \ref{lambda_theta_generic}.

The matrix on the right hand side of \eqref{bound_state_proof_eqn1} has
eigenvalues $\Lambda_{\pm}(X)=$ \\ $\pm\thetasharp\kappa(X)/\lamsharp$ with corresponding \emph{constant} eigenvectors
$
(1, -i)^T~\leftrightarrow~ \Lambda_{+}(X)$ and  $(1,i)^T~\leftrightarrow~ \Lambda_{-}(X)$.
Let 
\[ \alpha(X)=S\gamma(X),\ \ {\rm where }\ \ 
 S = \begin{pmatrix}1&1\\-i&i\end{pmatrix}.\] 
Then,
\begin{equation}
\label{bound_state_proof_eqn2}
 \partial_{X}\gamma(X) = \frac{\thetasharp\kappa(X)}{\lamsharp}
 \begin{pmatrix}
  1&0\\0&-1
 \end{pmatrix}
\gamma(X).
\end{equation}
Equation \eqref{bound_state_proof_eqn2} has solutions
\begin{equation}
\label{gamma_solns}
 \begin{pmatrix}\gamma_1(X)\\ \gamma_2(X)\end{pmatrix}=
\begin{pmatrix}\gamma_{10}e^{(\thetasharp/\lamsharp)\int^{X}_a\kappa(s)ds}\\ 
\gamma_{20}e^{-(\thetasharp/\lamsharp)\int^{X}_a\kappa(s)ds}\end{pmatrix}.
\end{equation} We seek localized solutions. That is, we seek $\gamma(X)\rightarrow0$ as $X\rightarrow\pm\infty$. We
then have two cases to consider: either $\thetasharp/\lamsharp>0$ or $\thetasharp/\lamsharp<0$. We will consider the
case with $\thetasharp/\lamsharp>0$; the case $\thetasharp/\lamsharp<0$  is handled similarly. 

Since $(\thetasharp/\lamsharp)\kappa(s)>0$ for  $s$ positive and sufficiently large, we set $\gamma_{10}=0$. 
By the hypotheses on $\kappa(X)$, we have for $X\to\infty$, 
\begin{align*}
e^{-(\thetasharp/\lamsharp)\int^{X}_a\kappa(s)ds}&=
e^{-(\thetasharp/\lamsharp)\int^{X}_a(\kappa(s)-\kappa_{\infty})ds}e^{-(\thetasharp/\lamsharp)\kappa_{\infty}
(X-a)} =\mathcal{O}\left(e^{-(\thetasharp/\lamsharp)\kappa_{\infty}X}\right).
\end{align*}
Similarly,  for  $X<a$, $e^{-(\thetasharp/\lamsharp)\int^{X}_a\kappa(s)ds}
=\mathcal{O}\left(e^{(\thetasharp/\lamsharp)\kappa_{-\infty}X}\right) ~\text{as}~ X\rightarrow -\infty$.

Thus $\Omega=0$ is a simple eigenvalue with corresponding one-dimensional nullspace spanned by
\begin{align}
 \alpha_\star(X) &\equiv \begin{pmatrix} \alpha_{\star,1}(X)\\ \alpha_{\star,2}(X) \end{pmatrix}\ =\ S\gamma(X)\ =\  \gamma_{0} \begin{pmatrix}1\\i\end{pmatrix}
 e^{-(\thetasharp/\lamsharp)\int^{X}_a\kappa(s)ds},
 \label{dirac-m-soln}
\end{align} 
where the constant $\gamma_{0}\in\C$ may be chosen to satisfy the normalization $\|\alpha_\star\|_{L^2}=1$. 
\\
(3)
If $\kappa\in C^{\infty}(\R)$, clearly $\alpha_\star(X)$ is in $C^{\infty}(\R) \cap H^{s}(\R)$  for all
$s\in\mathbb{N}$.
\\
(4) If $\kappa_{+\infty}\times\kappa_{-\infty}<0$, then it follows from \eqref{gamma_solns} that $\Omega=0$ is not an eigenvalue of \eqref{dirac_eqn}.
\end{proofof}

\chapter{Main Theorem\ - \ Bifurcation of Topologically Protected States}\label{sec:main_result}

%
Let $V_\ee\in L^2_\ee\cap C^\infty$ and $W_\oo\in L^2_\oo\cap C^\infty$. 
 Let $H_0=-\D_x^2+V_\ee$  and assume $(E_\star,k_\star)$ is a Dirac point 
in the sense Definition \ref{dirac-pt-gen}.

 In this section we give precisely state  the main result of this paper on the  bifurcation of topologically
protected bound state solutions of the Schr\"odinger operator $H_\delta=-\D_x^2+U_\delta$:
\begin{align}
 \label{perturbed_schro_prob1}
H_\delta\Psi(x)&\equiv \left(-\partial_x^2+U_{\delta}(x)\right)\Psi(x)= E\ \Psi(x), \quad x\in\R, \quad \Psi \in L^2(\R), \\
U_\delta(x)&=V_\ee(x)+\delta\kappa(\delta x)W_\oo(x) \nn .
\end{align} 
We express the bifurcating family $\delta\mapsto (\Psi^\delta,E^\delta)$ as a two-term truncation of the formal multiscale expansion of Chapter
\ref{sec:bound_states} plus a
corrector:

\begin{equation}
 \label{main_result_ansatz}
 \begin{split}
 \Psi^\delta(x) &= \delta^{1/2}\psi^{(0)}(x,X)+\delta^{3/2}\psi_p^{(1)}(x,X)+\delta^{3/2}\eta(x),\ \ X=\delta x \ , \\
 E^\delta &= E_\star+\delta^2\mu.
 \end{split}
\end{equation} 
Here, from Section \ref{subsec:multiscale_analysis} of Chapter \ref{sec:bound_states} we have
\begin{align}
 \psi^{(0)}(x,X) &= \alpha_{\star,1}(X)\Phi_1(x)+\alpha_{\star,2}(X)\Phi_2(x)\ \ {\rm and}\label{main_result_psi0}\\
  \psi^{(1)}_p(x,X) &= \left(R(E_{\star})G^{(1)}\right)\left(x,X\right), \label{main_result_psi1}
  \end{align} 
  where
  \begin{align}
 &G^{(1)}(x,X;\psi^{(0)},\psi_p^{(1)}) \nonumber \\
 &\qquad= \sum_{j=1}^2\Big[2\partial_{x}\Phi_j(x)\partial_{X} \alpha_{\star,j}(X)\Phi_j(x)
-\kappa(X)W_{\oo}(x)\alpha_{\star,j}(X)\Phi_j(x)\Big]\nn\\
 &\qquad= e^{ik_\star x}\ \Big[ 2(\D_x+ik_\star)\D_X-\kappa(X)W_o(x)\Big]\ \sum_{j=1}^2\alpha_{\star,j}(X)p_j(x;k_\star) .
\label{Fdefn}\end{align}
$\psi^{(1)}_p(x,X) $ is a particular solution of  \eqref{perturbed_schro_delta1}-\eqref{G1def} and
 $\alpha_\star=\left(\alpha_{\star,1},\alpha_{\star,2}\right)$ is a $L^2(\R_X)-$ normalized eigenstate of the Dirac
operator, $\mathcal{D}$ with eigenvalue $\Omega=0$:
\begin{equation}
 \label{main_result_dirac_eqn}
 \mathcal{D}\alpha_{\star}(X) = 0,\ \ \alpha_\star(X)\to0 \ \ {\rm as}\ \ X\to\pm\infty.
\end{equation} 

 We shall solve for $\eta=\eta^\delta(x)$ and
$\mu=\mu(\delta)$ for $0<\delta\ll1$ with appropriate bounds on $\eta^\delta$ and thereby establish our main result:
  \medskip

\begin{theorem}
 \label{thm:validity}
[{\bf Bifurcation of  topologically protected gap modes}]\\
Consider the eigenvalue problem \eqref{perturbed_schro_prob1}.
Assume $-\D_x^2+V_\ee$ has a Dirac point in the sense of Definition \ref{dirac-pt-gen}. It follows, in particular, that
\begin{equation}
\lamsharp=2i\inner{\Phi_1,\partial_x\Phi_1}_{L^2([0,1])}\ne0 . \nn
\end{equation}
Furthermore, assume 
\begin{equation}
\thetasharp=\inner{\Phi_1,W_{\oo}\Phi_2}_{L^2([0,1])}\ne0;
\label{theta-ne0}
\end{equation}
The condition $\lambda_\sharp\times\vartheta_\sharp\ne0$ holds for generic $V_\ee$ and $W_\oo$.

Assume that  $\kappa(X)\rightarrow\pm\kappa_{\infty}$ as $X\to\pm\infty$ and moreover that 
$\kappa^2(X)-\kappa_\infty^2$ and its derivatives decay rapidly at infinity, {\it e.g.} Schwartz class.
\footnote{ From the analysis of Section \ref{analysis-blDirac} of Chapter \ref{sec:proof-exists-mode} we see that the proof goes through under the assumptions
that the functions
$
\Upsilon_1(X)=\kappa^2(X)-\kappa_\infty^2\ \ {\rm and}\ \ \Upsilon_1(X)=\kappa'(X)
$ satisfy:
\begin{equation}
\int_\R (1+|X|)^a |\Upsilon_\ell(X)| dX<\infty\ \ \textrm{for some $a>5/2$ \ and }\ \ \int_\R |\D_X\Upsilon_\ell(X)|
dX<\infty,
 \  \ell=1,2.
\label{kappa-hypotheses}
\end{equation}
These hypotheses are required for the boundedness of {\it wave operators} on $W^{k,2}(\R),\ k=0,2$; see Theorem
\ref{thm7} and the discussion following it. The conditions \eqref{kappa-hypotheses} are easily satisfied if 
$\kappa^2(X)-\kappa^2_\infty$ is a Schwartz class function, {\it e.g.} $\kappa(X)=\tanh(X)$, 
$\kappa^2(X)-\kappa_\infty^2=-{\rm sech}^2(X)$.
}
Then, the following holds:
\begin{enumerate}
\item There exists  $\delta_0>0$ and a branch of solutions:
\[\delta\mapsto (E^\delta,\Psi^\delta)\in\mathcal{I}_\delta\times H^2(\R),\ \ 0<\delta<\delta_0, \]
  of the eigenvalue problem for $H_\delta$. 
\item This branch  bifurcates from the band intersection energy, $E_\star$, at $\delta=0$ into the gap
$\mathcal{I}_\delta$ (Proposition \ref{kappa_spec}). 
\item Furthermore, $\Psi^\delta$ is well-approximated by a slowly varying and spatially decaying modulation of the
degenerate Floquet-Bloch modes $\Phi_1$ and $\Phi_2$:
\begin{align}
&\left\|\ \Psi^\delta(\cdot)\ -\ \delta^{\frac12}\psi^{(0)}(\cdot,\delta\cdot)\ \right\|_{H^2(\R)}\nn\\
&=\left\|\ \Psi^\delta(\cdot)\ -\
\delta^{\frac12}\left[\alpha_{\star,1}(\delta\cdot)\Phi_1(\cdot)+\alpha_{\star,2}(\delta\cdot)\Phi_2(\cdot)\right]\
\right\|_{H^2(\R)}\
\lesssim\ \delta\ ,
\label{Psi-error}\\
&E^\delta = E_\star\ +\ E^{(2)}\delta^2+o(\delta^2). \label{E-error}
\end{align}
\item The amplitude vector, $\alpha_\star(X)=\left(\alpha_{\star,1}(X),\alpha_{\star,2}(X)\right)$, is an  $L^2(\R_X)$- 
normalized topologically protected 
$0$-energy eigenstate of the $1D$ Dirac operator, $\mathcal{D}\equiv i\sigma_3\lambda_\sharp\D_X+
\vartheta_\sharp\kappa(X)\sigma_2$; see \eqref{dirac_eqn}. The $2^{\rm nd}$ order eigenfrequency corrector, $E^{(2)}$, is
given by \eqref{solvability_cond_E2}.
\end{enumerate}
\end{theorem}
\medskip

\begin{remark}\label{on-theta-sharp}
As discussed in Section \ref{generic-Dirac} of Chapter \ref{1dperiodic&dirac},  $\lambda_\sharp\ne0$ holds for generic $V_\ee\in L^2_\ee$.
{\it What about condition \eqref{theta-ne0}: $\vartheta_\sharp\ne0$?} Consider the case where 
\[ V_\ee(x)\equiv0\ \  \textrm{and}\ \ W_\oo(x)= \sum_{p\in 2\Z_+1}w_p \cos(2\pi px)=\frac12\sum_{p\in 2\Z+1}w_p
e^{2i\pi px}.\]
Suppose we wish to study the bifurcation of localized states from the Dirac point $(k_\star,E_{\star,n})$, where
$E_{\star,n}=\pi^2(2n+1)^2$. In this case, the associated degenerate subspace is spanned by the functions
$\Phi_1(x)=e^{i\pi x}e^{2\pi inx}$ and $\Phi_2(x)=e^{i\pi x}e^{-2\pi i (n+1)x}$. A simple calculation gives
\begin{equation}
\vartheta_{\sharp,n} = \frac12 w_{2n+1}.\label{theta-sharp-n}
\end{equation}
Therefore, if $ w_{2n+1}\ne0$, there is a bifurcation from $(k_\star,E_{\star,n})$ seeded by the topologically
protected zero mode of the one-dimensional Dirac operator. Furthermore, if all Fourier modes are active in $W_\oo$, {\it
i.e.} $w_{2p+1}\ne0$ for all $p\in\Z$, then there are such bifurcations from all Dirac points of $\D_x^2$.

Figure \ref{perturbed-hills} in the Introduction corresponds to a case where $V_\ee=0$, $w_1=w_3=w_5=1\ne0$ and
$w_{2p+1}=0$ for all $p\ge3$. Thus, $\vartheta_{\sharp,j}\ne0,\ j=1,2,3$ and Theorem \ref{thm:validity} explains the
observed bifurcations.
\medskip

\begin{figure}
\includegraphics[width=\textwidth]{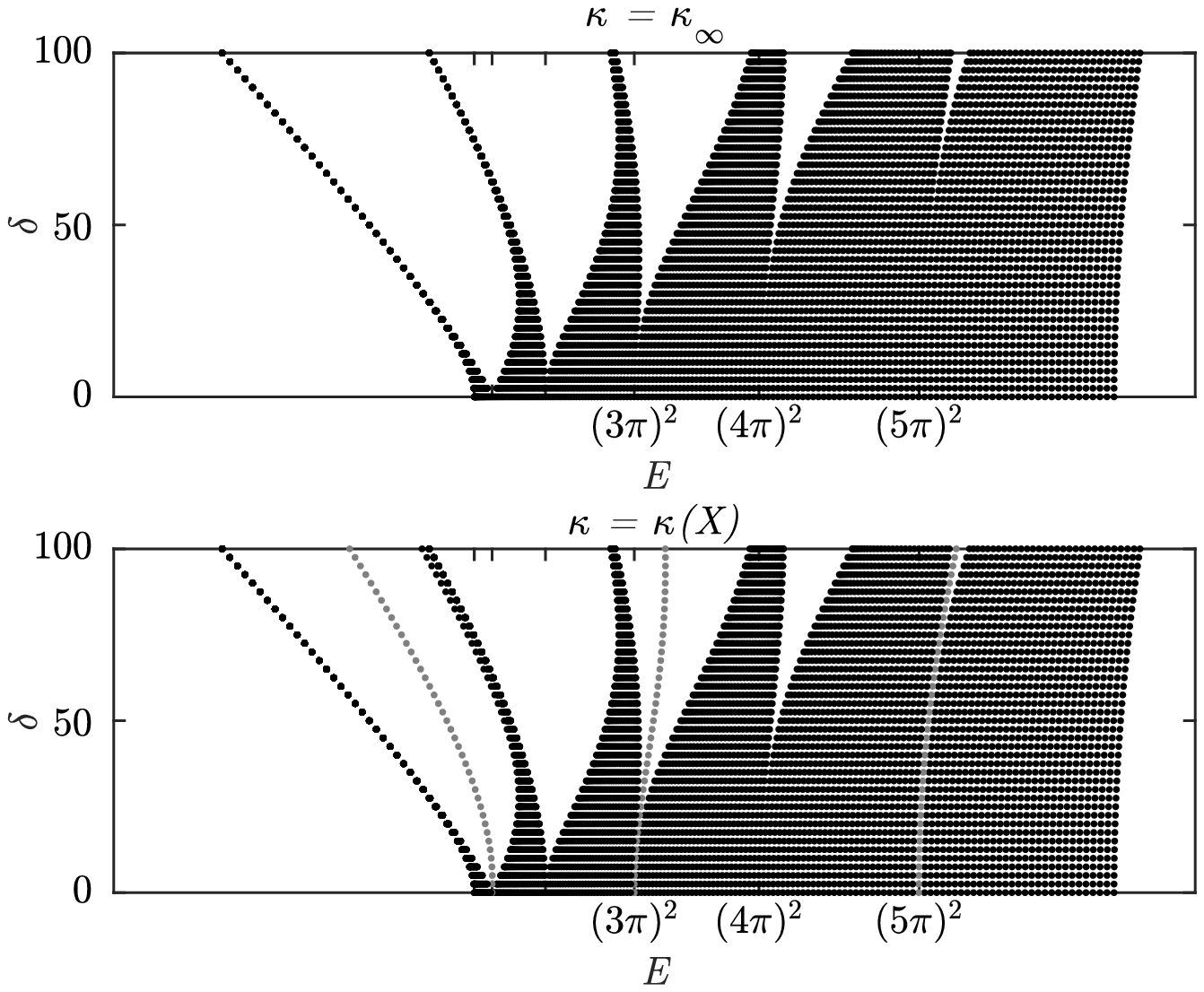}
\caption{{\bf Top panel:} Spectrum of Hill's operator: $-\D_x^2+\delta\kappa_\infty\cos(2\pi
x)$; see, for example
\cites{Meixner-Schaefke:54,magnus2013hill,WK:87}. {\bf Bottom panel:} Spectrum of the domain-wall modulated
periodic potential: $-\D_x^2+\delta\kappa(\delta x) \cos(2\pi
x)$ with $\kappa(X)=\tanh(X)$. Bifurcating branches of edge states are the (gray) curves emanating from
the points: $(E_{\star,m},\delta=0)$ where $(k_\star=\pi,E_{\star,m}\equiv(2m+1)^2\pi^2)$ are Dirac points. Theorem \ref{thm:validity} applies the bifurcation from the $m=0$ Dirac point. For $m\ge1$, $\vartheta_{\sharp,m}=0$ and the bifurcation appears to occur at higher order; see Remark \ref{on-theta-sharp-eq0}.}
\label{perturbed_hills_V0}
\end{figure}
\end{remark}

\begin{remark}\label{on-theta-sharp-eq0} {\it What if $\vartheta_\sharp=0$?} Figure \ref{perturbed_hills_V0} displays
spectra in the case  where $V_\ee=0$ and $W_\oo$ has only one non-zero Fourier cosine mode, the first. Hence, $w_1\ne0$
but $w_{2p+1}=0$ for all $p\ge1$. Therefore, Theorem \ref{thm:validity} explains the first bifurcation from
$E_\star=\pi^2$ at $\delta=0$, but the condition $\vartheta\ne0$ is violated for all higher
energy Dirac points, $E_\star=(3\pi)^2, (5\pi)^2,\dots$. Nevertheless we see
bifurcations. A higher order analysis, not pursued in this article, is necessary. Note that since the $n^{th}$ gap
width is                                                                                                                
   $\mathcal{O}(\delta^{\rho(n)})$ with $ \rho(n)\uparrow$ as $n\uparrow$,
\cites{Meixner-Schaefke:54,magnus2013hill,WK:87}, we expect that an expansion of the
bifurcating eigenvalue from the $n^{th}$ gap  to be of the form: 
 $E_n^\delta\approx E_{n,\star}+o(\delta^{\rho(n)})$.
\end{remark}

\medskip

\begin{remark}[Higher order expansion]
\label{higher-order}
The validity of the multiple scale expansion \eqref{formal-E}-\eqref{formal-psi} to any specified \underline{finite}
order in $\delta$ can be proved by the same methods used to prove Theorem \ref{thm:validity}. We omit this
generalization for ease of presentation.
\end{remark}
\medskip

\begin{remark}\label{H1bounds-psi01} We shall prove Theorem \ref{thm:validity} by showing that $\Psi^\delta$ 
can be expanded as in \eqref{main_result_ansatz}, where
\begin{align*}
&\|\delta^{\frac12}\psi^{(0)}(\cdot,\delta\cdot)\|_{H^1(\R)} =
\|
\delta^{\frac12}\left[\alpha_{\star,1}(\delta\cdot)\Phi_1(\cdot)+\alpha_{\star,2}(\delta\cdot)\Phi_2(\cdot)\right]\|_{
H^1}=
\mathcal{O}(1),\\
& \norm{\delta^{3/2}\psi^{(1)}_p(\cdot,\delta\cdot)}_{H^1(\R)} \leq C\delta,\qquad
  \norm{\delta^{3/2}\eta}_{H^1(\R)} \leq C\delta.
  \end{align*}
\end{remark}

\chapter{Proof of the Main Theorem}\label{sec:proof-exists-mode}

In this section we prove Theorem \ref{thm:validity}. 
 We begin with bounds on the first two terms in the expansion
\eqref{main_result_ansatz}, and then proceed to the study of the equation for the corrector $(\mu,\eta(x))$
 in \eqref{main_result_ansatz}.

\begin{remark}
 \label{regularity}
 We shall make frequent use of the regularity of the solutions $\Phi\equiv(\Phi_1(x),\Phi_2(x))^T$
and $\alpha_{\star}(X)\equiv(\alpha_{\star,1}(X),\alpha_{\star,2}(X))^T$ throughout the proof:
\begin{itemize}
 \item[(a)]  $\Phi\in C^{\infty}(\R)$, which follows from $V_{\ee} \in C^{\infty}(\R)$ and elliptic
regularity, and
 \item[(b)]  $\alpha_\star(X)$ is in Schwartz class, $\mathcal{S}$. Furthermore,  $\alpha_\star(X)$ and its derivatives
with respect to $X$ are all exponentially decaying as $|X|\to\infty$; see Theorem \ref{thm:dirac_bound_state}.          
     
\end{itemize}    
\end{remark}
\medskip

The following lemma, proved in Appendix \ref{psi_bound_proof},  
lists $H^s$ bounds on $\psi^{(0)}$ and $\psi^{(1)}$, which will be used in the proof of Theorem \ref{thm:validity};
see also Remark \ref{H1bounds-psi01}.\medskip

\begin{lemma}[$H^s$ bounds on $\psi^{(0)}(x,X)$ and $\psi^{(1)}(x,X)$] 
 \label{lemma:psi_bounds} For all $s\in\mathbb{N}$, there exists $\delta_0>0$, such that for all $0<\delta<\delta_0$,
the leading order expansion terms $\psi^{(0)}(x,X)$ and $\psi^{(1)}_p(x,X)$ displayed in \eqref{main_result_psi0} and
\eqref{main_result_psi1} satisfy the bounds:
\begin{align}
 \norm{\psi^{(0)}(\cdot,\delta\cdot)}_{H^s(\R)} + \norm{\partial_X^2\psi^{(0)}(x,X)\Big|_{X=\delta
x}}_{L^2(\R_x)}\ &\lesssim \delta^{-1/2}, ~~~ \label{psi0_bdds}\\
 \norm{\psi^{(1)}_p(\cdot,\delta\cdot)}_{H^s(\R)} &\lesssim \delta^{-1/2}, \label{psi_H_bdds} \\
 \norm{\partial_X^2\psi^{(1)}_p(x,X)\Big|_{X=\delta x}}_{L^2(\R_x)}\ +\ 
 \norm{\partial_x\partial_X\psi^{(1)}_p(x,X)\Big|_{X=\delta x}}_{L^2(\R_x)} &\lesssim \delta^{-1/2}.
\label{psi1_specific_bdds}
\end{align}
\end{lemma}

\section{Rough strategy}\label{rough-strategy}
Substitution of expansion \eqref{main_result_ansatz} into the eigenvalue problem
\eqref{perturbed_schro_prob1}, yields an equation for the corrector $\eta(x)$, which depends on $\mu$:
\begin{align}
\label{eta_eqn}
&\left(-\partial_x^2+V_{\ee}(x)-E_{\star}\right)\eta(x) +\delta\kappa(\delta x)W_{\oo}(x)\eta(x)\\
&~~~=\delta\left(2\partial_x\partial_X-\kappa(X)W_{\oo}(x)\right)\psi^{(1)}_p(x,X)\Big|_{X=\delta x}\nn\\
&+\delta \mu\psi^{(0)}(x,\delta x) + \delta^2\mu\psi^{(1)}(x,\delta x)
+\delta^2\mu\eta(x) +\delta\D_X^2\psi^{(0)}+\delta^2\D_X^2\psi^{(1)}_p. \nonumber
\end{align}
To prove Theorem \ref{thm:validity}, we 
prove that \eqref{eta_eqn} has a solution, 
$(\eta^\delta,\mu^\delta)$, with $\eta^\delta\in H^2(\R)$ satisfying the bound
\begin{equation}
\label{etabound}
 \norm{\delta^{3/2}\eta^\delta}_{H^2(\R)} \leq C\delta\ .
 \end{equation}

Recall that if $f\in L^2(\R)$, then the Floquet-Bloch coefficients 
\[\widetilde{f}_b(k)=\inner{\Phi_b(\cdot,k),f(\cdot)}_{L^2(\R)}
=\lim_{N\to\infty} \int_{-N}^N \overline{\Phi_b(x,k)} f(x) dx\] 
are well-defined as an $L^2(\R)$ limit. Furthermore, we have the completeness relation:\begin{align}
\label{f_exp_floquet_bloch}
f(x) &= \frac{1}{2\pi}\  \sum_{b=1}^{\infty}\int_{\mathcal{B}}
\inner{\Phi_b(\cdot,k),f(\cdot)}_{L^2(\R)} \Phi_b(x,k) dk = \frac{1}{2\pi}\ \sum_{b=1}^{\infty}\int_{\mathcal{B}}\
\widetilde{f}_b(k)\ \Phi_b(x,k) dk \ ,
\end{align}
where equality holds in the sense of the $L^2-$ limit of partial sums.
\medskip

Through a systematic, but unfortunately long, calculation we shall derive a system of equations  
for $\{\widetilde{\eta}_b(k)\}_{b\ge1}$, which is formally equivalent to system
\eqref{eta_eqn}. This is the band-limited Dirac system of Proposition \ref{near_freq_compact}. We then prove
that the latter system has a solution, which is then used to construct a solution to \eqref{eta_eqn}.
We now embark on this derivation.

Recalling that $\widetilde{f}_b(k)=\inner{\Phi_b(\cdot,k),f(\cdot)}$, we take the inner product of \eqref{eta_eqn} with
$\Phi_b(x,k),\ b\ge1$ to obtain
\begin{align}
\label{eta_eqn_floquet_bloch}
&\left(E_b(k)-E_{\star}\right)\widetilde{\eta}_b(k) +\delta\inner{\Phi_b(\cdot,k),\kappa(\delta
\cdot)W_{\oo}(\cdot)\eta(\cdot)}_{L^2(\R)}\nn\\
&\qquad\quad = \delta F_b[\mu,\delta](k) + \delta^2\mu\  \widetilde\eta_b(k), \qquad b\ge1.
\end{align}
Here, for $b\geq1$, 
\begin{align}
& F_b[\mu,\delta](k)\equiv
F^{1,\delta}_b(k)\ + \mu F^{2,\delta}_b(k)\ +\ \delta\mu F^{3,\delta}_b(k)\ +\ F^{4,\delta}_b(k)\ +\ \delta
F^{5,\delta}_b(k), 
\label{Fb-def}
\end{align}
where
\begin{align}
 F^{1,\delta}_b(k) & \equiv
\inner{\Phi_b(x,k),\left(2\partial_x\partial_X-\kappa(X)W_{\oo}(x)\right)\psi^{(1)}_p(x,X)\Big|_{X=\delta
x}}_{L^2(\R_x)},\nn\\
F^{2,\delta}_b(k) & \equiv \inner{\Phi_b(x,k),\psi^{(0)}(x,\delta x)}_{L^2(\R_x)}, \nn\\
F^{3,\delta}_b(k) & \equiv \inner{\Phi_b(x,k),\psi^{(1)}_p(x,\delta x)}_{L^2(\R_x)},\nn\\
F^{4,\delta}_b(k) & \equiv \inner{\Phi_b(x,k),\left.\D_X^2\psi^{(0)}(x,X)\right|_{X=\delta x}}_{L^2(\R_x)},\nn\\
F^{5,\delta}_b(k) & \equiv \inner{\Phi_b(x,k),\left.\D_X^2\psi^{(1)}_p(x,X)\right|_{X=\delta x}}_{L^2(\R_x)}.
 \label{Fdef}
\end{align}
We shall show that the system \eqref{eta_eqn_floquet_bloch} has a solution,
$(\{\widetilde{\eta}^\delta_b(k)\}_{b\ge1},\mu(\delta))$, such that 
\begin{align}
&k\mapsto\widetilde{\eta}_b(k)\Phi(x,k)\ \ \textrm{is $2\pi-$ periodic for $k\in\R$ a.e. for $x\in\R$}\ ,\label{prop2}\\
&\sum_{b\ge1}\int_0^{2\pi} (1+b^2)^2 |\widetilde{\eta}_b(k)|^2 dk<\infty\ .\label{prop3}
\end{align}
It will then follow that 
\begin{equation*}
\eta(x)\ \equiv \frac{1}{2\pi}\  \sum_{b=1}^{\infty}\int_{\mathcal{B}}\
\widetilde{\eta}_b(k)\ \Phi_b(x,k) dk\ \in\ {\rm domain}(H_\delta)=H^2(\R)
\end{equation*}
and $(\eta^\delta,\mu^\delta)$ is an eigenpair for the eigenvalue problem \eqref{eta_eqn}.

\section{Detailed strategy: Decomposition into near and far energy components}\label{subsec:freq_decomp}
Recall now the definitions and notational conventions associated with the smoothed cutoff functions:
$\chi(|\xi|\le\delta^\tau)$ and $\overline{\chi}(|\xi|\le\delta^\tau)$; see \eqref{chi-def}.
We next decompose $\eta$ into its spectral components, $\eta_{\rm near}$ and $\eta_{\rm far}$,  near and away from
the Dirac point $(E_\star,k_\star)$. We write:
\begin{equation}
\label{eta_near+far}
 \eta(x) = \eta_{\rm near}(x) + \eta_{\rm far}(x),
\end{equation} where 
\begin{align}
\label{eta_near_expression}
 \eta_{\rm near}(x) &=
 \frac{1}{2\pi}\  \sum_{b=\pm}\int_{\mathcal{B}}  \widetilde{\eta}_{b,\rm near}(k)\  \Phi_b(x,k)\ dk, 
\end{align} 
(see Proposition \ref{flo-blo-dirac})
and
\begin{equation}
\label{eta_far_expression}
\eta_{\rm far}(x)\ = \frac{1}{2\pi}\ 
\sum_{b\in\mathbb{Z}}\int_{\mathcal{B}} \widetilde{\eta}_{b,\rm far}(k) \Phi_b(x,k)dk\ .
\end{equation} 
Here, 
\begin{align}
k'&\equiv k-k_\star=k-\pi\label{kprime}\\
 \widetilde{\eta}_{\pm,\rm near}(k) &\equiv \chi\left(\abs{k'}\leq\delta^{\tau}\right)
\widetilde{\eta}_{\pm}(k),\label{teta-near}\\
 \widetilde{\eta}_{b,\rm far}(k) &\equiv
\chi\Big(\abs{k'}\geq(\delta_{b,b_{\star}}+\delta_{b,b_{\star}+1})\delta^{\tau}\Big)
\widetilde{\eta}_{b}(k)\ ,\ \  b\ge1. \label{teta-far}
\end{align}
The parameter $\tau$ is presently free, but will be constrained to
satisfy: $0<\tau<1/2$.

We may rewrite the system of equations \eqref{eta_eqn_floquet_bloch} as two coupled subsystems:\ 
 a pair of equations, which governs the {\bf near energy components}:
\begin{align}
&\left(E_{-}(k)-E_{\star}\right)\widetilde{\eta}_{-,\rm near}(k) \label{near_cpt_2} \\
&\qquad+\delta\chi\left(\abs{k'}\leq\delta^{\tau}\right)\inner{\Phi_{-}(\cdot,k),\kappa(\delta
\cdot)W_{\oo}(\cdot)\left[\eta_{\rm near}(\cdot)+\eta_{\rm far}(\cdot)\right]}_{L^2(\R)}
\nonumber\\
&\qquad\qquad=\delta\chi(\abs{k'}\leq\delta^{\tau})
F_{-}[\mu,\delta](k) + \delta^2\mu\ \widetilde{\eta}_{-,{\rm near}}(k), \nn \\
&\left(E_{+}(k)-E_{\star}\right)\widetilde{\eta}_{+,\rm near}(k) \label{near_cpt_1} \\
&\qquad+\delta\chi\left(\abs{k'}\leq\delta^{\tau}\right)\inner{\Phi_{+}(\cdot,k),\kappa(\delta
\cdot)W_{\oo}(\cdot)\left[\eta_{\rm near}(\cdot)+  \eta_{\rm far}(\cdot)\right]}_{L^2(\R)}
\nonumber\\
&\qquad\qquad=\delta\chi(\abs{k'}\leq\delta^{\tau})
F_{+}[\mu,\delta](k) + \delta^2\mu\ \widetilde{\eta}_{+,{\rm near}}(k), \nn
\end{align} 
coupled to an infinite system governing the {\bf far-energy components}:
\begin{align}
&\left(E_b(k)-E_{\star}\right)\widetilde{\eta}_{b,\rm far}(k) \label{far_cpts} \\
&\qquad +\delta\chi(\abs{k'}\geq(\delta_{b,b_{\star}}+\delta_{b,b_{\star}+1})\delta^{\tau})
\left\langle\Phi_{b}(\cdot,k),\kappa(\delta
\cdot)W_{\oo}(\cdot)\left[\eta_{\rm near}(\cdot)+\eta_{\rm far}
(\cdot)\right] \right\rangle_{L^2(\R)} \nonumber\\
&\qquad\qquad=\delta\chi\Big(\abs{k'}\geq(\delta_{b,b_{\star}}+\delta_{b,b_{\star}+1})\delta^{\tau}\Big)
F_{b}[\mu,\delta](k) + \delta^2\mu\ \widetilde{\eta}_{b,\rm far}(k),\ \ \text{for}\  b\ge1. \nn
\end{align} 
\medskip

Conversely, consider the system \eqref{near_cpt_2}-\eqref{far_cpts}, in which we substitute for $\eta_{\rm near}$ and
 $\eta_{\rm far}$  the $\widetilde{\eta}_{\pm,\rm near}$ and $\widetilde{\eta}_{b,\rm far}\ -$dependent
expressions 
 \eqref{eta_near_expression} and \eqref{eta_far_expression}.  View the resulting system as one governing the unknowns 
$\widetilde{\eta}_{+,\rm near}$, $ \widetilde{\eta}_{-,\rm near}$, and $\widetilde{\eta}_{b,\rm far},\ b\ge1$. 
Suppose now that this system has a solution:  $\mu=\mu(\delta)$  and $\widetilde{\eta}^\delta_{\pm,\rm near}$, 
$\widetilde{\eta}_{b,\rm far}, b\ge1$, for   which  properties \eqref{prop2}-\eqref{prop3} hold. 
Define, using this solution,  $\eta^\delta_{\rm near}(x)$ and $\eta^\delta_{\rm far}(x)$ via \eqref{eta_near_expression}
and \eqref{eta_far_expression} and set $\eta^\delta(x)=\eta^\delta_{\rm near}(x)+\eta^\delta_{\rm far}(x)$. Then,
$\eta^\delta\in H^2(\R)$ and $(\eta^\delta(x),\mu(\delta))$ solves the corrector equation \eqref{eta_eqn}. 
%
\section{Analysis of far energy components\label{subsec:far_freq}}
Using \eqref{far_cpts}, we will solve for 
$\eta_{\rm far} = \eta_{\rm far}[\eta_{\rm near},\mu,\delta]$  and then
  substitute this mapping into \eqref{near_cpt_2} and \eqref{near_cpt_1}, to obtain a closed equation for the near energy
components,
$\eta_{\rm near}$. We begin with the following
\begin{lemma}
 \label{lemma:eval_bounds}
 There exists a $\zeta_0>0$ such that for all $0<\delta<\zeta_0$, the following holds: There exist positive
constants $C_0$, $C_1$ and $C_2$, independent of $\delta$, such that
\begin{align}
\abs{E_b(k)-E_{\star}} &\geq C_0(1+b^2),~~~b\notin\{ +,-\},~~~\textrm{for all}\  k\in\mathcal{B},
\label{eval_bdd_proof_1}\\
 \abs{E_b(k)-E_{\star}} &\geq C_1,~~~b\in\{ +,-\},~~~\zeta_0<\abs{k-k_{\star}}\leq k_{\star},
\label{eval_bdd_proof_2}\\
 \abs{E_b(k)-E_{\star}} &\geq C_2\delta^{\tau},~~~b\in\{ +,-\},~~~\delta^{\tau}\leq
\abs{k-k_{\star}}\leq\zeta_0. \label{eval_bdd_proof_3}
\end{align}
\end{lemma}
{\it Proof of Lemma \ref{lemma:eval_bounds}:}\ \ 
The lower bound \eqref{eval_bdd_proof_1} follows by continuity of $k\mapsto E_b(k)$, that $E_\star$ is a double
eigenvalue, that for a second order ordinary differential operator there are at most two linearly independent
solutions,  and the asymptotics of eigenvalues of self-adjoint second order elliptic equations 
(Weyl's law).  
To prove the lower bound in 
\eqref{eval_bdd_proof_2} note, by continuity of $E_\pm$, that if there were no such strictly positive $C_1$, there would
exist $k_\natural\ne k_\star$, such that  $E_b(k_\natural)=E_\star$. But then the second order ODE: $(H_\delta-E)f=0$ would
have three linearly independent solutions, a contradiction.
 Hence, lower bound \eqref{eval_bdd_proof_2} holds.  Finally, the lower bound \eqref{eval_bdd_proof_3}
follows from the expansions in Proposition \ref{flo-blo-dirac} of $E_\pm(k_\star+k')$ for $|k'|<\zeta$.
 This completes the proof of Lemma \ref{lemma:eval_bounds}. 
\medskip

By Lemma \ref{lemma:eval_bounds}, the far energy system \eqref{far_cpts} may be written equivalently as
\begin{align}
\label{far6}
&\widetilde{\eta}_{b,\rm far}(k)
+\frac{\delta\chi(\abs{k'}\geq(\delta_{b,b_{\star}}+\delta_{b,b_{\star}+1})\delta^{\tau})}{E_b(k)-E_{\star}}
\inner{\Phi_{b}(\cdot,k),\kappa(\delta \cdot)W_{\oo}(\cdot)\left[\eta_{\rm near}(\cdot)+ \eta_{\rm far}
(\cdot)\right]}_{L^2(\R)} \nonumber\\
&~~~=\frac{\delta\chi(\abs{k'}\geq(\delta_{b,b_{\star}}+\delta_{b,b_{\star}+1})\delta^{\tau})}{E_b(k)-E_{\star}}
F_{b}[\mu,\delta](k) +  \delta^2\mu\ \frac{\widetilde{\eta}_{b,\rm far}(k)}{{E_b(k)-E_{\star}}},\qquad  b\ge1\ .
\end{align} 
Here, $k'=k-k_\star, k\in[0,2\pi]$ and we recall that $F_b[\mu,\delta]$ is given by \eqref{Fb-def}-\eqref{Fdef}.
We next view \eqref{far6} as a fixed point system for $\widetilde\eta_{\rm
far}=\{\widetilde{\eta}_{b,\rm far}(k)\}_{b\ge1}$:
\begin{equation}
\widetilde{\mathcal{E}}_b[\widetilde{\eta}_{\rm far};\eta_{\rm near},\mu,\delta]\ =\
\widetilde{\eta}_{b,\rm far}\ ,\qquad  b\ge1, 
\label{fixed-pt1}
\end{equation}
where the mapping $\widetilde{\mathcal{E}}_b$ is defined by:
\begin{align}\label{tE-def}
&\widetilde{\mathcal{E}}_b[\phi;\psi,\mu,\delta](k)  \\
&\quad\equiv 
-\frac{\delta\chi(\abs{k'}\geq(\delta_{b,b_{\star}}+\delta_{b,b_{\star}+1})\delta^{\tau})}{E_b(k)-E_{\star}}
\inner{\Phi_{b}(\cdot,k),\kappa(\delta \cdot)W_{\oo}(\cdot)\left[\psi(\cdot)+ \phi
(\cdot)\right]}_{L^2(\R)}\nn\\
&\quad\quad +\frac{\delta\chi(\abs{k'}\geq(\delta_{b,b_{\star}}+\delta_{b,b_{\star}+1})\delta^{\tau})}{E_b(k)-E_{\star}}
F_{b}[\mu,\delta](k) +  \delta^2\mu\ \frac{\widetilde{\phi}_{b,\textrm{far}}(k)}{{E_b(k)-E_{\star}}}\ , \nn
\end{align}
 and where
\begin{align*}
\phi(x)&=\frac{1}{2\pi}\ \sum_{b\in\mathbb{Z}} \int_{\mathcal{B}} \chi\left(\abs{k'}\geq(\delta_{b,b_{\star}}+\delta_{b,b_{\star}+1})\delta^{\tau}\right)\
\widetilde{\phi}_b(k) \Phi_b(x;k)\ dk \\
&=\ \frac{1}{2\pi}\ 
\sum_{b\in\mathbb{Z}} \int_{\mathcal{B}} \widetilde{\phi}_{b,{\rm far}}(k) \Phi_b(x;k)\ dk\ .
\end{align*}

\nit Equation \eqref{fixed-pt1}  
can be expressed as an  equivalent system for $\eta_{\rm far}$:
\begin{equation}
\mathcal{E}[\eta_{\rm far};\eta_{\rm near},\mu,\delta]\ =\ \eta_{\rm far}\ .
\label{fixed-pt-notilde}
\end{equation}

 For fixed $\mu$, $\delta$ and band-limited $\eta_{\rm near}$, such that 
\begin{equation}
\widetilde{\eta}_{\pm,\rm near}(k)\ =\ \chi\left(\abs{k'}\leq \delta^\tau\right)
\widetilde{\eta}_{\pm,\rm near}(k)
 , \label{near-def}\end{equation}
we shall seek a solution $\{\widetilde{\eta}_{b,\rm far}\}_{b\ge1}$, supported at ``far energies'':
\begin{equation}
\widetilde{\eta}_{b,\rm far}(k)\\ =\
\chi\left(\abs{k'}\ge(\delta_{b,b_{\star}}+\delta_{b,b_{\star}+1})\delta^\tau\right) \widetilde{\eta}_{b,\rm far}(k)
 , ~~~b\ge1 .\label{far-def}\end{equation}

 \nit Introduce the Banach spaces of functions supported in  ``far'' and ``near'' energy regimes:
 \begin{align}
  L^2_{\rm near,\delta^\tau}(\R) &\equiv\
  \left\{ f\in L^2(\R) : \widetilde{f}_b(k)\ \textrm{satisfies \eqref{near-def}}\right\}, \label{L2near} \\
  L^2_{\rm far,\delta^\tau}(\R) &\equiv\
  \left\{ f\in L^2(\R) : \widetilde{f}_b(k)\ \textrm{satisfies \eqref{far-def}}\right\}, \label{L2far}
\end{align}
and the corresponding open balls of radius $\rho$: 
 \begin{align}
  B_{\rm near,\delta^\tau}(\rho) &\equiv\
  \left\{ f\in L^2_{\rm near,\delta^\tau} : \|f\|_{L^2}<\rho \right\}, \label{ball_near} \\
 B_{\rm far,\delta^\tau}(\rho) &\equiv\
  \left\{ f\in L^2_{\rm far,\delta^\tau} : \|f\|_{L^2}<\rho \right\}. \label{ball_far}
\end{align}
Near- and far- energy
 Sobolev spaces $H^s_{\rm far,\delta^\tau}(\R) $ and 
$H^s_{\rm near,\delta^\tau}(\R) $ are analogously defined. 

\begin{proposition}\label{tEmap-pre}
Let $0<\tau<1/2$. 

 \begin{enumerate}
 \item  Pick a positive number $M$.  Choose $(\phi,\psi,\mu,\delta)$ such that
 \[ \phi\in L^2_{\rm far,\delta^\tau},\ \  \psi\in L^2_{\rm near,\delta^\tau},\ |\mu|<M,\ {\rm and}\ 
0<\delta\le1\ . \]
Then,  $\mathcal{E}\left[\phi;\psi,\mu,\delta\right]\in H^2_{\rm far,\delta^\tau}(\R)$. Moreover, there is a constant
$C_M$ depending on $M$, and independent of $\phi$ and $\psi$,   such that 
 \begin{equation}
 \left\|\mathcal{E}\left[\phi;\psi,\mu,\delta\right]\right\|^2_{H^2(\R)}\ \le\ C_M\  \left(\ \delta^{1-2\tau}\ +\
\delta^{2(1-\tau)}
 \left[\norm{\psi}^2_{L^2(\R)} +
\norm{\phi}^2_{L^2(\R)} \right] \ \right) .
 \label{map-est}
 \end{equation}
  \item 
Fix constants $M>0$ and $R>0$. Assume $\psi\in B_{{\rm near},\delta^\tau}(R)$ and $|\mu|<M$. There exists 
$\delta_0\in(0,1]$, such that the following holds:\\
 For $0<\delta<\delta_0$, there is a constant  $\rho_\delta=\mathcal{O}(\delta^{1/2-\tau})$, such that
  \begin{align}
  & \phi\in  B_{\rm far,\delta^\tau}(\rho_\delta)\ \implies\ \mathcal{E}\left[\phi;\psi,\mu,\delta\right]\in 
B_{\rm far,\delta^\tau}(\rho_\delta).\ \label{ball2ball}
  \end{align}
Furthermore,  for any $\phi_1, \phi_2\in B_{\rm far,\delta^\tau}(\rho_\delta)$,\ $\psi_1, \psi_2\in
B_{\rm near,\delta^\tau}(\rho_\delta)$, and $|\mu_1|, | \mu_2|<M$, 
  \begin{align}
&  \left\| \mathcal{E}\left[\phi_1;\psi_1,\mu_1,\delta\right]\ -\  \mathcal{E}\left[\phi_2;\psi_2,\mu_2,\delta\right]\
\right\|_{H^2(\R)}\nn\\
&\qquad \le\ C'_{M,R}\ \delta^{1-\tau}\ \Big(\ \left\|\phi_1-\phi_2\right\|_{L^2(\R)}
\ +\ \left\|\psi_1-\psi_2\right\|_{L^2(\R)}\ +\ |\mu_1-\mu_2|\ \Big) .\label{contraction}
  \end{align}
  \end{enumerate}
 \end{proposition}
 \medskip
 
Applying Proposition \ref{tEmap-pre} to equation \eqref{fixed-pt1} (or equivalently \eqref{fixed-pt-notilde}) we 
have:
\medskip

\begin{corollary}\label{fixed-pt} 
Let $0<\tau<1/2$. 
\begin{enumerate}
\item For any fixed $M>0, R>0$, there exists  $\delta_0\in(0,1]$,  such that for all $0<\delta<\delta_0$, 
equation \eqref{fixed-pt-notilde}, or equivalently, the system \eqref{fixed-pt1}, has a unique fixed point solution, $\eta_{\rm far}=\eta_{\rm far}[\eta_{\rm near},\mu,\delta]$, where 
\begin{align*}
&(\eta_{\rm near},\mu,\delta)\mapsto \eta_{\rm far}[\cdot;\eta_{\rm near},\mu,\delta]=
\mathcal{T}^{-1}\widetilde{\eta}_{\rm far}
\end{align*}
maps from $ B_{{\rm near},\delta^\tau}(R)\times\{|\mu|<M\}\times\{0<\delta<\delta_0\}$
 to $ B_{{\rm far},\delta^\tau}(\rho)$, and 
$\rho=\rho_\delta=\mathcal{O}(\delta^{1-2\tau})$, as in Proposition \ref{tEmap-pre}. 
\item The mapping $(\eta_{\rm near},\mu,\delta)\mapsto \eta_{\rm far}(\cdot;\eta_{\rm near},\mu,\delta)$ is 
Lipschitz in $(\eta_{\rm near},\mu)$ with  values in $H^2(\R)$ and satisfies:
\begin{align}
 \label{far_lipschitz}
\left\|\eta_{\rm far}(\psi_1,\mu_1,\delta) -  \eta_{\rm far}(\psi_2,\mu_2,\delta)\right\|_{H^2(\R)}
&\leq C' \delta^{1-\tau} \Big(\norm{\psi_1-\psi_2}_{L^2(\R)} + \abs{\mu_1-\mu_2} \Big),\\
 \norm{\eta_{\rm far}[\eta_{\rm near};\mu,\delta]}_{H^2(\R)} &\le\ C''\left(\ 
\delta^{1-\tau}\norm{\eta_{\rm near}}_{L^2(\R)}+\delta^{1/2-\tau}\ \right). \label{eta-far-bound}
\end{align}
Here, $C'$ and $C''$ are constants which depend  on $M, R$ and $\tau$.
\item 
In greater detail, the mapping $(\eta_{\rm near},\mu,\delta)\mapsto\eta_{\rm far}[\eta_{\rm near},\mu,\delta]$ 
is affine in $\eta_{\rm near}$ and Lipschitz in $\mu$, with values in $H^2(\R)$ and may be expressed as:
\begin{equation}
\label{eta_far_affine}
\eta_{\rm far}[\eta_{\rm near},\mu,\delta](x) = [A\eta_{\rm near}](x;\mu,\delta) + \mu
B(x;\delta) + C(x;\delta),
\end{equation} 
where 
for $\eta_{\rm near}\in B_{{\rm near},\delta^\tau}(R)$ we have:
\begin{align}
&\left\| [A\eta_{\rm near}](\cdot,\mu_1,\delta) - A[\eta_{\rm near}](\cdot,\mu_2,\delta)\right\|_{H^2(\R)}
\le \ C'_{M,R}\ \delta^{1-\tau}\ |\mu_1-\mu_2| \label{A_lip}\\
 &\norm{[A\eta_{\rm near}](\cdot;\mu,\delta)}_{H^2(\R)} \leq \delta^{1-\tau}
\norm{\eta_{\rm near}}_{L^2(\R)},  \label{A_bdd} \\
 &\norm{B(\cdot;\delta)}_{H^2(\R)} \leq \delta^{1/2-\tau}, ~~~ \text{and} ~~~
 \norm{C(\cdot;\delta)}_{H^2(\R)} \leq \delta^{1/2-\tau}.  \label{B_bdd}
\end{align}
\item  Define the extension of  $\eta_{\rm far}[\cdot;\eta_{\rm near},\mu,\delta]$ to  the half-open interval
$\delta\in[0,\delta_0)$
by setting 
$\eta_{\rm far}[\eta_{\rm near},\mu,0]=0$. Then, by \eqref{eta-far-bound}
$\eta_{\rm far}[\eta_{\rm near},\mu,\delta]$ is continuous at $\delta=0$.
\end{enumerate}
\end{corollary}
We next prove Proposition \ref{tEmap-pre} and Corollary \ref{fixed-pt}.

\nit{\it Proof of Proposition \ref{tEmap-pre}:} Taking absolute values and squaring
\eqref{tE-def} gives
\begin{align*}
&\abs{\widetilde{\mathcal{E}}_b[\phi;\psi,\mu,\delta](k) }^2\ \nn\\
& \qquad \leq
2\delta^2\frac{\chi(\abs{k'}\geq(\delta_{b,b_{\star}}+\delta_{b,b_{\star}+1})\delta^{\tau})}{|E_b(k)-E_{\star}|^2}\\
&\times
\bigg\{\abs{\inner{\Phi_{b}(\cdot,k),\kappa(\delta
\cdot)W_{\oo}(\cdot)\left[(\psi+\phi)(\cdot)\right]}_{L^2(\R)}}^2  + \abs{F_{b}[\mu,\delta](k)}^2
 +\ |\delta\mu|^2\ |\widetilde{\phi}_{b,\rm far}(k)|^2\ \bigg\}.
\end{align*}

\nit By the lower bounds \eqref{eval_bdd_proof_1}-\eqref{eval_bdd_proof_3} we have for
$b=b_{\star},b_{\star}+1$:
\begin{align*}
&\abs{\widetilde{\mathcal{E}}_b[\phi;\psi,\mu,\delta](k) }^2 \nn\\
&\qquad\leq
\frac{2\delta^{2(1-\tau)}}{C_2^2}\bigg\{\abs{\inner{\Phi_{b}(\cdot,k),\kappa(\delta
\cdot)W_{\oo}(\cdot)\left[(\psi+\phi)(\cdot)\right]}_{L^2(\R)}}^2 + \abs{F_{b}[\mu,\delta](k)}^2 \\
&\qquad\qquad+\ |\delta\mu|^2\ |\widetilde{\phi}_{b,\rm far}(k)|^2\ \bigg\}.
\end{align*}
and for $b\neq b_{\star},b_{\star}+1$:
\begin{align*}
&\abs{\widetilde{\mathcal{E}}_b[\phi;\psi,\mu,\delta](k) }^2 \nn\\
&\qquad\le
\frac{2\delta^{2}}{C_0^2(1+b^2)^2}\bigg\{\abs{\inner{\Phi_{b}(\cdot,k),\kappa(\delta
\cdot)W_{\oo}(\cdot)\left[(\psi+\phi)(\cdot)\right]}_{L^2(\R)}}^2 + \abs{F_{b}[\mu,\delta](k)}^2 \\
&\qquad\qquad+\ |\delta\mu|^2\ |\widetilde{\phi}_{b,\rm far}(k)|^2\ \bigg\}.
\end{align*}

Applying Lemma \ref{lemma11} and using that
$\widetilde{\mathcal{E}}_b[\phi;\psi,\mu,\delta](k)=\inner{\Phi_b(\cdot,k),\mathcal{E}[\phi;\psi,\mu,\delta]}$ yields
\begin{align}
  \label{far17}
\norm{\mathcal{E}[\phi;\psi,\mu,\delta]}^2_{H^2(\R)} &\approx
\sum_{b=1}^{\infty} (1+b^2)^2 \int_{\mathcal{B}}
\left|\widetilde{\mathcal{E}}_b[\phi;\psi,\mu,\delta](k)\right|^2\ dk \nonumber\\
&\lesssim \delta^{2(1-\tau)} \bigg[
\sum_{b=1}^{\infty}\int_{\mathcal{B}}
\abs{\inner{\Phi_b(\cdot,k),\kappa(\delta\cdot)W_{\oo}(\cdot)\phi(\cdot)}_{L^2(\R)}}^2 d k
\nonumber\\
&~~~+\sum_{b=1}^{\infty}\int_{\mathcal{B}}
\abs{\inner{\Phi_b(\cdot,k),\kappa(\delta\cdot)W_{\oo}(\cdot)\psi(\cdot)}_{L^2(\R)}}^2
dk 
\nonumber\\
&~~~+ \sum_{b=1}^{\infty}\int_{\mathcal{B}} \abs{F_{b}[\mu,\delta](k)}^2 + \delta^2\mu^2
|\widetilde{\phi}_{b,\rm far}(k)|^2
dk\bigg]
\nonumber\\
&\lesssim \delta^{2(1-\tau)} \bigg[
\norm{\kappa(\delta\cdot)W_{\oo}(\cdot)\psi(\cdot)}^2_{L^2(\R)} +
\norm{\kappa(\delta\cdot)W_{\oo}(\cdot)\phi(\cdot)}^2_{L^2(\R)} \nonumber\\ &~~~+
\delta^2\mu^2  \|\phi\|_{L^2(\R)}^2\ +\  \sum_{b=1}^{\infty}\int_{\mathcal{B}}\abs{F_{b}[\mu,\delta](k)}^2
dk\bigg]\nn\\
&\le\ C_{M}\ \delta^{2(1-\tau)}\ \bigg[  \|\phi\|_{L^2(\R)}^2 + \|\psi\|_{L^2(\R)}^2 + 
 \sum_{b=1}^{\infty}\int_{\mathcal{B}}\abs{F_{b}[\mu,\delta](k)}^2
dk\bigg],
\end{align}
where $C_{M}$ is a constant which depends on $M$.

Furthermore, recalling  that $F_{b}[\mu,\delta]$ is given by \eqref{Fb-def}, \eqref{Fdef}, it is straightforward to show
using Lemma \ref{lemma11} and the boundedness on $\R$ of $W_\oo$ and $\kappa$
  that for $|\mu|\le M$, $|\delta|\le\delta_0\le1$,
  \begin{align*}
\sum_{b=1}^{\infty}\int_{\mathcal{B}} \abs{F_{b}[\mu,\delta](k)}^2\ &\lesssim
\norm{\psi^{(1)}_p(\cdot,\delta\cdot)}^2_{H^2(\R)} +
\norm{\kappa}^2_{L^{\infty}(\R)}\norm{W_{\oo}}^2_{L^{\infty}(\R)}\norm{\psi^{(1)}_p(\cdot,
\delta\cdot)} ^2_ {L^2(\R)} \\ 
&\ \ + 
\norm{\psi^{(0)}(\cdot,\delta\cdot)}^2_{H^2(\R)} +
\delta^2\norm{\psi^{(1)}_p(\cdot,\delta\cdot)}^2_{H^2(\R)}
\le C_{M}\ \delta^{-1}\ .\end{align*} 
Substituting this bound into \eqref{far17} and again using the boundedness of $W_\oo$ and $\kappa$ we have
\begin{align*}
 \norm{\mathcal{E}[\phi;\psi,\mu,\delta]}^2_{H^2(\R)} &\le C'_{M}
\delta^{2(1-\tau)}\left[\norm{\psi}^2_{L^2(\R)} +
\norm{\phi}^2_{L^2(\R)} \right] + C''_{M}\ \delta^{1-2\tau}, \nn
\end{align*} 
which proves part 1 and \eqref{map-est} of Proposition \ref{tEmap-pre}.

To prove \eqref{ball2ball} of part 2, we choose $\tau\in(0,1/2)$ and set 
 \[ \rho_\delta\equiv \sqrt{2 C''_{M}}\ \ \delta^{1/2-\tau}\ {\rm and}\ \  \delta_0\equiv \min\Big\{ 1 ,
\frac{C_M''}{C_M'R^2}\Big\}\ . \]
 Then, for $0<\delta<\delta_0$, if $\phi\in B_{{\rm far},\delta^\tau}(\rho_\delta)$, we have 
$ \mathcal{E}[\phi;\psi,\mu,\delta]\in B_{{\rm far},\delta^\tau}(\rho_\delta)$.

 To prove the Lipschitz estimate \eqref{contraction} (part 3 of the Proposition \ref{tEmap-pre} ), note from
\eqref{tE-def} that for each $\delta>0$,  \[ \tilde{\mathcal E}_b[\phi,\psi,\mu,\delta]  = \tilde{\mathcal
A}_b[\phi,\psi,\mu] + \delta^2\mu\ \frac{\widetilde{\phi}_{b,\textrm{far}}(k)}{{E_b(k)-E_{\star}}}, \]
 where $(\phi,\psi,\mu)\mapsto\tilde{\mathcal A}_b[\phi,\psi,\mu]$ is affine in $(\phi,\psi,\mu)$. Thus,
\eqref{contraction} follows by estimates similar to those in the proof of part 2. This completes the proof of Proposition \ref{tEmap-pre}.
\medskip

\nit{\it Proof of Corollary \ref{fixed-pt}:}
Part 1, the existence of a unique fixed point of solution  of \eqref{fixed-pt-notilde}, is an immediate consequence of \eqref{ball2ball} and \eqref{contraction} of Proposition \ref{tEmap-pre}, and the contraction mapping principle. The Lipschitz estimate \eqref{far_lipschitz} of Part 2 follows from \eqref{contraction} applied to $\phi_1=\phi_2=\eta_{\rm near}$, the fixed point solution of 
\eqref{fixed-pt-notilde}. The bound \eqref{eta-far-bound} of Part 2 follows from \eqref{map-est} with $\phi=\eta_{\rm near}$, the fixed point solution of  \eqref{fixed-pt-notilde}. 

We prove Part 3, using that \eqref{fixed-pt-notilde} is a linear inhomogeneous equation with source terms driven by
$\widetilde{\eta}_{b,\rm near}$ and $F_b[\mu;\delta]$. Let $[\widetilde{A\eta}_{\rm near}]_b(k;\mu,\delta)$ solve
the fixed point system \eqref{fixed-pt1} with $F_b[\mu;\delta](k)$ set equal to zero. Also, let $\mu
\tilde{B}_b(k;\delta)+\tilde{C}_b(k;\delta)$ denote 
the solution of \eqref{fixed-pt1} with $\widetilde{\eta}_{b,\rm near}$ set equal to zero, for all $b$;
recall from \eqref{Fb-def} that $F_b[\mu;\delta](k)$ is affine in $\mu$. Each of the maps
$\{[\widetilde{A\eta}_{\rm near}]_b\}_{b\ge1},\ \{\tilde{B}_b(k;\delta)\}_{b\ge1}$ and
$\{\tilde{C}_b(k;\delta)\}_{b\ge1}$ are fixed points of mappings
comprised of a subset of the terms in the definition of the mapping $\mathcal{E}[\phi,\psi,\mu,\delta]$. Therefore,
the bounds of Proposition \ref{tEmap-pre} apply and the proof of part 3  is complete. Part 4 is
direct consequence of Part 3. This completes the proof of Corollary \ref{fixed-pt}.

\section{Lyapunov-Schmidt reduction to a Dirac system for the near energy components\label{subsec:near_freq}}
Having constructed the mapping 
  $\eta_{\rm far}=\eta_{\rm far}[\eta_{\rm near},\mu,\delta]$ with appropriate bounds, we may view the 
  system \eqref{near_cpt_2}-\eqref{near_cpt_1} as a  closed system for $(\eta_{\rm near},\mu)$, depending on the
parameter $\delta\in(0,\delta_0)$.
Our next step is to rewrite this system  as a perturbed Dirac system. 
\medskip

\nit \textbf{Rescaling the near-energy region}. Set 
\begin{equation*}
 \widetilde{\eta}_{\pm,\rm near}(k) \equiv
\widehat{\eta}_{\pm,\rm near}\left(\frac{k-k_{\star}}{\delta}\right)\  .
\end{equation*}
Recall  that the near energy components  are band limited, \eqref{teta-near}:
\begin{equation*}
\widetilde{\eta}_{\pm,{\rm near}}(k)\ =\ \widehat{\eta}_{\pm,\rm near}\left(\frac{k-k_{\star}}{\delta}\right) =
\chi\left(\frac{\abs{k-k_{\star}}}{\delta}\leq\delta^{\tau-1}\right)
\widehat{\eta}_{\pm,\rm near}\left(\frac{k-k_{\star}}{\delta} \right).
\end{equation*} 
Introduce the recentered and rescaled quasimomentum:
\begin{equation}
\xi\equiv\frac{k-k_{\star}}{\delta}.
\label{xi-def}\end{equation}
Hence, 
\begin{equation}
 \widetilde{\eta}_{\pm,\rm near}(k) = \chi\left(|\xi|\le\delta^{\tau-1}\right)\widehat{\eta}_{\pm,{\rm near}}(\xi).\label{teta-xi}
 \end{equation}

 By
smoothness of $E_{\pm}(k)$ near $k_\star$ (Proposition \ref{flo-blo-dirac}) we have:
\begin{equation*}
E_{\pm}(k)-E_{\star}=\delta E_{\pm}'(k_{\star})\xi+\frac{1}{2}(\delta\xi)^2E_{\pm}''(\widetilde{\xi}_\pm^\delta),\quad\
\end{equation*} 
where $E_{\pm}(k_{\star})=E_{\star}$ and
$\widetilde{\xi}_{\pm}^\delta$ lies between $k_{\star}$ and $k_{\star}+\delta\xi$.
 Thus,
\begin{equation}
 \label{near4}
 (E_{\pm}(k)-E_{\star})\widetilde{\eta}_{\pm,\rm near}(k) = \delta\xi
E_{\pm}'(k_{\star})\widehat{\eta}_{\pm,\rm near}(\xi) + \frac{1}{2}(\delta\xi)^2
E_{\pm}''(\widetilde{\xi}_{\pm}^\delta)\widehat{\eta}_{\pm,\rm near}(\xi).
\end{equation}
 Since  $E_\pm'(k_{\star})=\pm\lamsharp$ (Proposition \ref{flo-blo-dirac}), substituting
\eqref{near4} into the near energy equations \eqref{near_cpt_2} and \eqref{near_cpt_1}, and canceling
a factor of $\delta$ yields:
\begin{align}
&-\lamsharp\ \xi\ \widehat{\eta}_{-,\rm near}(\xi)
+\chi(\abs{\xi}\leq\delta^{\tau-1})\inner{\Phi_{-}(\cdot,k_{\star}+\delta\xi),\kappa(\delta
\cdot)W_{\oo}(\cdot)\eta_{\rm near}(\cdot)}_{L^2(\R)} \label{near6} \\
&\qquad=\chi(\abs{\xi}\leq\delta^{\tau-1}) F_{-}[\mu,\delta](k_\star+\delta\xi) + \delta\mu\ \widehat{\eta}_{-,{\rm
near}}(\xi)\nn\\
&\qquad\qquad - 
\chi(\abs{\xi}\leq\delta^{\tau-1})\inner{\Phi_{-}(\cdot,k_{\star}+\delta\xi),\kappa(\delta
\cdot)W_{\oo}(\cdot)\eta_{\rm far}[\eta_{\rm near},\mu,\delta](\cdot)}_{L^2(\R)} \nn \\
&\qquad\qquad-\frac{1}{2}\delta E_{-}''(\widetilde{\xi}_{-}^\delta)\xi^2\widehat{\eta}_{-,\rm near
}(\xi), \nn \\
&+\lamsharp\ \xi\ \widehat{\eta}_{+,\rm near}(\xi)
+\chi(\abs{\xi}\leq\delta^{\tau-1})\inner{\Phi_{+}(\cdot,k_{\star}+\delta\xi),\kappa(\delta
\cdot)W_{\oo}(\cdot)\eta_{\rm near}(\cdot)}_{L^2(\R)} \label{near5} \\
&\qquad=\chi(\abs{\xi}\leq\delta^{\tau-1}) F_{+}[\mu,\delta](k_\star+\delta\xi) + \delta\mu\ \widehat{\eta}_{+,{\rm
near}}(\xi)\nn\\
&\qquad\qquad - 
\chi(\abs{\xi}\leq\delta^{\tau-1})\inner{\Phi_{+}(\cdot,k_{\star}+\delta\xi),\kappa(\delta
\cdot)W_{\oo}(\cdot)\eta_{\rm far}[\eta_{\rm near},\mu,\delta](\cdot)}_{L^2(\R)}  \nn \\
 &\qquad\qquad-\frac{1}{2}\delta E_{+}''(\widetilde{\xi}_{+}^\delta)\xi^2\widehat{\eta}_{+,\rm near}(\xi). \nn
\end{align}

\nit  We next implement a somewhat lengthy computation resulting in equation \eqref{compacterroreqn} of
Proposition \ref{near_freq_compact}. This equation is a reformulation of \eqref{near6}-\eqref{near5} as a band-limited system
of non-homogeneous Dirac equations, formulated in quasi-momentum space. To this end, we must expand  the inner products in \eqref{near6}-\eqref{near5} for $\delta\ne0$ and small.
\medskip

\noindent \textbf{Simplifying 
$\inner{\Phi_\pm(\cdot,k_{\star}+\delta\xi),\kappa(\delta
\cdot)W_{\oo}(\cdot)\eta_{\rm near}(\cdot)}_{L^2(\R)}$ via Poisson summation}.
 We simplify the inner product 
by expressing $\eta_{\rm near}$  in terms of its spectral components near $E_\star$ ($k'\equiv k-k_\star$ small)
plus  a correction.

By \eqref{eta_near_expression},  using that $\Phi_{\pm}(x,k)=e^{ikx}p_{\pm}(x,k)$,  we have
\begin{align*}
 \eta_{\rm near}(x) &= \frac{1}{2\pi}\ \int_{\abs{k-k_{\star}}\leq\delta^{\tau}} 
\Phi_{+}(x,k)\widetilde{\eta}_{+,\rm near}(k) dk \\
 &\qquad+ \frac{1}{2\pi}\ \int_{\abs{k-k_{\star}}\leq\delta^{\tau}}
 \Phi_{-}(x,k)\widetilde{\eta}_{-,\rm near}(k)dk \\
 &= \frac{1}{2\pi}\ \int_{\abs{k-k_{\star}}\leq\delta^{\tau}}e^{ik_{\star}x}e^{i(k-k_{\star})x}
 p_{+}(x,k)\widehat{\eta}_{+,\rm near}\left(\frac{k-k_{\star}}{\delta}\right)dk \\
&\qquad+ \frac{1}{2\pi}\ \int_{\abs{k-k_{\star}}\leq\delta^{\tau}}e^{ik_{\star}x}e^{i(k-k_{\star})x}
p_{-}(x,k)\widehat{\eta}_{-,\rm near}\left(\frac{k-k_{\star}}{\delta}\right)dk.
\end{align*}
 By Proposition \ref{flo-blo-dirac},  $p_{\pm}(x,k)$ is smooth in $k$, and we may write 
\begin{align}
 \label{pexp}
 p_{\pm}(x,k) &= p_{\pm}(x,k_{\star}) + (k-k_{\star})\partial_kp_{\pm}(x,\widetilde{k}_{\pm}(x,\delta\xi)) \nonumber\\
 &= p_{\pm}(x,k_{\star}) + \delta\xi\partial_k p_{\pm}(x,\widetilde{k}_{\pm}(x;\delta\xi)) \nonumber\\
 &\equiv p_{\pm}(x,k_{\star}) + \Delta p_{\pm}(x;\delta\xi),
\end{align} where $\widetilde{k}_{\pm}(x;\delta\xi)$ lies between $k_{\star}$ and $k_{\star}+\delta\xi$, 
and  
\begin{equation}
 \label{deltapbdd}
  \abs{\Delta p_{\pm}(x,\delta\xi)} \leq \underset{x\in[0,1], ~ \abs{\omega}\leq\delta^{\tau}}{\sup}
\abs{\Delta p_{\pm}(x,\omega)} \leq \delta^{\tau},\ \ \abs{\xi}\leq\delta^{\tau-1}\ .
\end{equation}

\nit By \eqref{pexp} and \eqref{xi-def}-\eqref{teta-xi}
\begin{align}
 \eta_{\rm near}(x) &= e^{ik_{\star}x}p_{+}(x,k_{\star}) \frac{\delta}{2\pi}\ \int_{\abs{\xi}\leq\delta^{\tau-1}}
 e^{i\xi\delta x}\widehat{\eta}_{+,\rm near}(\xi)d\xi 
 + \frac{\delta}{2\pi}\  e^{ik_{\star}x} \rho_{+}(x,\delta x)\nonumber\\
&~~~+ e^{ik_{\star}x}p_{-}(x,k_{\star}) \frac{\delta}{2\pi}\ \int_{\abs{\xi}\leq\delta^{\tau-1}}
 e^{i\xi\delta x}\widehat{\eta}_{-,\rm near}(\xi)d\xi 
 + \frac{\delta}{2\pi}\  e^{ik_{\star}x} \rho_{-}(x,\delta x)\nonumber\\
&= e^{ik_{\star}x} \frac{\delta}{2\pi}\  \left[p_{+}(x,k_{\star})\eta_{+,\rm near}(\delta x) + \rho_{+}(x,\delta x) \right.\nn\\
&\qquad\qquad\qquad \left. + p_{-}(x,k_{\star})\eta_{-,\rm near}(\delta x) 
 + \rho_{-}(x,\delta x)\right], \label{nearinxi}
\end{align} where
\begin{equation}
 \label{rhodefn}
 \rho_{\pm}(x,X) = \int_{\abs{\xi}\leq\delta^{\tau-1}}e^{i\xi X} \Delta p_{\pm}(x,\delta\xi)
\widehat{\eta}_{\pm,\rm near}(\xi)d\xi.
\end{equation} 

\nit We next expand the inner product in \eqref{near5} for small $\delta$; the \eqref{near6} case is treated similarly. Substituting
\eqref{nearinxi} into the inner product in \eqref{near5} yields
\begin{align}
 &\inner{\Phi_{+}(\cdot,k_{\star}+\delta\xi),\kappa(\delta\cdot)W_{\oo}(\cdot)
\eta_{\rm near}(\cdot)}_{L^2(\R)} = \label{full_inner} \\
& \frac{\delta}{2\pi}\  \inner{e^{i\xi\delta \cdot}p_+(\cdot,k_{\star}+\delta\xi), p_{+}(\cdot,k_{\star})
W_{\oo}(\cdot)\kappa(\delta\cdot) {\eta}_{+,\rm near}(\delta \cdot)}_{L^2(\R)} \label{inner1}\\
&+ \frac{\delta}{2\pi}\  \inner{e^{i\xi\delta \cdot}p_{+}(\cdot,k_{\star}+\delta\xi), p_{-}(\cdot,k_{\star})
W_{\oo}(\cdot)\kappa(\delta\cdot) {\eta}_{-,\rm near}(\delta \cdot)}_{L^2(\R)} \label{inner2}\\
&+ \frac{\delta}{2\pi}\  \inner{e^{i\xi\delta \cdot}p_{+}(\cdot,k_{\star}+\delta\xi),
W_{\oo}(\cdot)\kappa(\delta\cdot)\rho_{+}(\cdot,\delta \cdot)}_{L^2(\R)} \label{inner3}\\
&+ \frac{\delta}{2\pi}\ \inner{e^{i\xi\delta \cdot}p_{+}(\cdot,k_{\star}+\delta\xi),
W_{\oo}(\cdot)\kappa(\delta\cdot)\rho_{-}(\cdot,\delta \cdot)}_{L^2(\R)} \ . \label{inner4}
\end{align} 

To obtain a detailed expansion of the inner product terms in \eqref{inner1}-\eqref{inner4} we shall make repeated use
the following lemma, which is  proved using the $L^2_{\rm loc}-$ Poisson summation formula of Theorem \ref{psum-L2}.
{\ }\medskip

\begin{lemma}
 \label{poisson_exp}
 Let $f(x,\xi)$ and $g(x)$ denote smooth functions of $(x,\xi)\in\R\times\R$ that are $1$-periodic in $x$. 
 Let $\Gamma(x,X)$ be defined for $(x,X)\in\R\times\R$, and such that conditions
\eqref{Gamma-cond1}-\eqref{Gamma-cond2} hold:
\begin{align}\label{Gamma-conditions1}
&\Gamma(x+1,X)\ =\ \Gamma(x,X),\\
&\sum_{j=0}^2\ \int_0^1\ \left\|\D_x^j\Gamma(x,X)\right\|_{L^2(\R_X)}^2\ dx\ <\ \infty. 
\label{Gamma-conditions2}\end{align}
Denote by $\widehat{\Gamma}(x,\omega)$ its Fourier transform on $\R$ with respect to the $X$ variable. Then,
\begin{align}
 \label{poisson_app}
 &\frac{\delta}{2\pi}\  \inner{e^{i\xi\delta\cdot}f(\cdot,\delta\xi), g(\cdot)\Gamma(\cdot,\delta\cdot)}_{L^2(\R)} \\
 &\qquad= 
 \sum_{m\in\mathbb{Z}} \int_0^1 e^{2\pi imx} \widehat{\Gamma}\left(x,\frac{2\pi m}{\delta}+\xi\right) g(x)
\overline{f(x,\delta\xi)} dx, \nn
\end{align}
with equality holding in $L^2_{\rm loc}([-\xi_{\rm max},\xi_{max}];d\xi)$, for any fixed $\xi_{\rm max}>0$.
\end{lemma}
\medskip

\begin{remark}\label{applying-poisson_exp}
We shall apply Lemma \ref{poisson_exp} with $f(x,\delta\xi)$ and $g(x)$ arising from $p_\pm(x,k_\star+\delta\xi)$ and
$p_\pm(x;k_\star)$, for $|\delta\xi|\le\delta^\tau,\ \tau>0$. The functions, $\Gamma$, which arise are:
$\Gamma=\Gamma(\delta x)=\kappa(\delta x)\eta_{\pm,{\rm near}}(\delta x)$ and
$\Gamma=\Gamma(x,\delta x)=\kappa(\delta x)\rho_{\pm}(x,\delta x)$, where $\rho_\pm$, which depends on $\eta_{\pm,{\rm
near}}$,  is defined in \eqref{rhodefn}. The function $\eta_{\pm,{\rm near}}$ will be constructed to be band-limited.
Recall also that $p_\pm(x,k)$ is smooth in $x$ and $k$. Hence, 
for these choices of $f, g$ and $\Gamma$, the hypotheses of Lemma \ref{poisson_exp} are easily checked.
\end{remark}
\medskip

To prove Lemma \ref{poisson_exp} we shall use:
\begin{lemma}\label{interchange}
Let $F(x,\zeta)$ and the sequence $F_n(x,\zeta),\ n=1,2,\dots$, belong to
$L^2([0,1]\times[-\zeta_{\rm max},\zeta_{max}];dx d\zeta)$. Assume that
\begin{equation}
\left\| F_n-F \right\|_{L^2([0,1]\times[-\zeta_{\rm max},\zeta_{max}];dxd\zeta)}\ \to\ 0,\ \ {\rm as}\ \ n\to\infty\ .
\label{fntof}\end{equation}
Let $G\in L^2([-\zeta_{\rm max},\zeta_{max}];d\zeta)$. Then, in $L^2([0,1];dx)$
\begin{equation*}
\lim_{n\to\infty}\int^{\zeta_{\rm max}}_{-\zeta_{\rm max}}F_n(x,\zeta)G(\zeta)d\zeta =
\int^{\zeta_{\rm max}}_{-\zeta_{\rm max}}\lim_{n\to\infty}F_n(x,\zeta)G(\zeta)d\zeta = 
\int^{\zeta_{\rm max}}_{-\zeta_{\rm max}} F(x,\zeta)G(\zeta)d\zeta .
\end{equation*}
\end{lemma}
{\it Proof of Lemma \ref{interchange}:}\  Square the difference, apply Cauchy-Schwarz and then integrate $\int_0^1dx$.\\
\medskip

\begin{proofof}{\it Proof of Lemma \ref{poisson_exp}:}
Note first that the inner product on the left hand side of \eqref{poisson_app} is well-defined in $L^2([0,1];d\xi)$ 
by Remark \ref{welldefinedonL2}.
Applying Lemma \ref{interchange} and using that $x\mapsto f(x,\delta\xi)$ and $x\mapsto g(x)$ have period one we obtain 
\begin{align*}
&\frac{\delta}{2\pi}\ \inner{e^{i\xi\delta\cdot}f(\cdot,\delta \xi), g(\cdot)\Gamma(\cdot,\delta\cdot)}_{L^2(\R)} \nn \\
&\qquad=
\frac{\delta}{2\pi}\ \int_{-\infty}^{\infty}e^{-i\xi(\delta x)} \overline{f(x,\delta\xi)} g(x) \Gamma(x,\delta x) dx \nn\\
&\qquad= \frac{\delta}{2\pi}\ \lim_{N\to\infty}\int_{-N}^{N+1}e^{-i\xi(\delta x)} \overline{f(x,\delta\xi)} g(x) \Gamma(x,\delta x)
dx \nn \\
&\qquad=\frac{\delta}{2\pi}\ \lim_{N\to\infty}\sum_{n=-N}^N\ \int_n^{n+1}e^{-i\xi\delta x}\Gamma(x,\delta x) \overline{f(x,\delta\xi)} g(x)
dx \nn\\
&\qquad=\frac{\delta}{2\pi} \lim_{N\to\infty} \int_0^1\ \left[ \sum_{n=-N}^N e^{-i\xi\delta (x+n)}\Gamma\left(x,\delta (x+n)\right)\right]
\overline{f(x,\delta\xi)} g(x)
dx .
\end{align*}
Next, we apply Theorem \ref{psum-L2} with $\Gamma=\Gamma(x,\delta x)$ and $\zeta=\delta\xi$. We obtain
\begin{align*}
 \sum_{m\in\mathbb{Z}}e^{-i\xi\delta(x+m)}\Gamma\left(x,\delta(x+m)\right) &=
\frac{2\pi}{\delta}\sum_{m\in\mathbb{Z}}e^{2\pi imx}\widehat{\Gamma}\left(x,\frac{2\pi m+\delta\xi}{\delta}\right) \\
&=\frac{2\pi}{\delta}\sum_{m\in\mathbb{Z}}e^{2\pi imx}\widehat{\Gamma}\left(x,\frac{2\pi m}{\delta}+\xi\right),
\end{align*}
with equality holding in $L^2([0,1]\times[-\xi_{\rm max},\xi_{\rm max}];dx d\xi)$. 
Again applying Lemma \ref{interchange},
we may interchange limit and integral:
\begin{align*}
&=\frac{\delta}{2\pi}\ \int_0^1\left[\frac{2\pi}{\delta} \lim_{N\to\infty}\sum_{m=-N}^N e^{2\pi imx}\ \widehat{\Gamma}\left(x,\frac{2\pi
m}{\delta}+\xi\right)\right]
\overline{f(x,\delta\xi)} g(x) dx \\
&=  \lim_{N\to\infty} \sum_{m=-N}^N \int_0^1  e^{2\pi imx}\widehat{\Gamma}\left(x,\frac{2\pi m}{\delta}+\xi\right)
\overline{f(x,\delta\xi)} g(x) dx, 
\end{align*} 
This completes the proof of Lemma \ref{poisson_exp}.
\end{proofof}
\nit We next apply Lemma \ref{poisson_exp} to each of the inner products \eqref{inner1}-\eqref{inner4}.\medskip
 
 \nit\underline{\it Expansion of inner product  \eqref{inner1}:}
 Let  $f(x,\delta\xi)=p_+(x,k_{\star}+\delta\xi)$, \\ $g(x)=p_+(x,k_{\star})W_{\oo}(x)$ and
$\Gamma(x,\delta x) = \kappa(\delta x)\eta_{+,\rm near}(\delta x)$. By Lemma \ref{poisson_exp}, 
\begin{align*}
&\frac{\delta}{2\pi}\ \inner{e^{i\xi\delta \cdot}p_+(\cdot,k_{\star}+\delta\xi), p_{+}(\cdot,k_{\star})
W_{\oo}(\cdot)\kappa(\delta\cdot) {\eta}_{+,\rm near}(\delta \cdot)}_{L^2(\R)} \\ &=
\sum_{m\in\mathbb{Z}}\int_0^1e^{2\pi imx}\mathcal{F}_X[\kappa\eta_{+,\rm near}]\left(\frac{2\pi m}{\delta}+\xi
\right) 
\overline{p_{+}(x,k_{\star}+\delta\xi)}p_{+}(x,k_{\star}) W_{\oo}(x)dx.
\end{align*} Using the expansion of $p_{+}(x,k_{\star}+\delta\xi)$ displayed in \eqref{pexp}, we may rewrite this as
\begin{align*}
&\frac{\delta}{2\pi}\ \inner{e^{i\xi\delta \cdot}p_+(\cdot,k_{\star}+\delta\xi), p_{+}(\cdot,k_{\star})
W_{\oo}(\cdot)\kappa(\delta\cdot) {\eta}_{+,\rm near}(\delta \cdot)}_{L^2(\R)} \\
&\qquad \equiv I_+^1(\xi;\eta_{+,\rm near}) + I_+^2(\xi;\eta_{+,\rm near}),
\end{align*}
where
\begin{align}
& I_+^1(\xi;\eta_{+,\rm near}) \label{I_1} \\ 
\qquad &=
\sum_{m\in\mathbb{Z}} \mathcal{F}_X[\kappa\eta_{+,\rm near}]\left(\frac{2\pi m}{\delta}+\xi\right)
\int_0^1e^{2\pi imx} \abs{p_{+}(x,k_{\star})}^2 W_{\oo}(x)dx,  \nn \\
& I_+^2(\xi;\eta_{+,\rm near}) \label{I_2} \\
\qquad &= 
\sum_{m\in\mathbb{Z}}
\mathcal{F}_X[\kappa\eta_{+,\rm near}]\left(\frac{2\pi m}{\delta}+\xi\right)  \int_0^1
e^{2\pi imx} \overline{\Delta p_{+}(x,\delta\xi)}p_{+}(x,k_{\star}) W_{\oo}(x)dx .  \nn
\end{align}
The $m=0$ term in the summation of $I_+^1(\xi;\eta_{+,\rm near})$ in \eqref{I_1} vanishes:
\[\int_0^1\abs{p_{+}(x,k_{\star})}^2W_{\oo}(x)dx=0,\]
 since the integrand is the product of even and odd index
Fourier series. Thus
\begin{equation}
I_+^1(\xi;\eta_{+,\rm near}) =\
\sum_{\abs{m}\geq1} \mathcal{F}_X[\kappa\eta_{+,\rm near}]\left(\frac{2\pi m}{\delta}+\xi\right) 
\int_0^1e^{2\pi imx} \abs{p_{+}(x,k_{\star})}^2W_{\oo}(x)dx. \label{ibstar1}
\end{equation} 

\nit\underline{\it Expansion of the inner product \eqref{inner2}:} By Lemma  \ref{poisson_exp} with 
$f(x,\delta\xi)=p_+(x;k_\star+\delta\xi)$, 
$g(x)=p_{-}(x,k_{\star})W_\oo(x)$ and $\Gamma(\delta x)=
\kappa(\delta x) {\eta}_{-,\rm near}(\delta x)$, 
we have 
\begin{align}
&\frac{\delta}{2\pi}\ \inner{e^{i\xi\delta \cdot}p_+(\cdot,k_{\star}+\delta\xi), p_{-}(\cdot,k_{\star})
W_{\oo}(\cdot)\kappa(\delta\cdot) {\eta}_{-,\rm near}(\delta \cdot)}_{L^2(\R)}\nn\\
&= \sum_{m\in\mathbb{Z}} \mathcal{F}_X[\kappa\eta_{-,\rm near}]\left(\frac{2\pi m}{\delta}+\xi\right) 
\int_0^1e^{2\pi imx}\overline{p_{+}(x,k_{\star})}p_{-}(x,k_{\star}) W_{\oo}(x)dx \nn \\
&\qquad+ \sum_{m\in\mathbb{Z}}
\mathcal{F}_X[\kappa\eta_{-,\rm near}]\left(\frac{2\pi m}{\delta}+\xi\right) \int_0^1 e^{2\pi imx}
\overline{\Delta 
p_{+}(x,\delta\xi)}p_{-}(x,k_{\star}) W_{\oo}(x)dx \nn \\
&\equiv \mathcal{F}_X[\kappa\eta_{-,\rm near}]\left(\xi\right)
\inner{\Phi_1,W_{\oo}\Phi_2}_{L^2([0,1])} + I_+^3(\xi;\eta_{-,\rm near}) +
I_+^4(\xi;\eta_{-,\rm near}),\nn
\end{align}
where
\begin{align}
& I_+^3(\xi;\eta_{-,\rm near})  \label{I_3} \\
& \qquad \equiv \sum_{\abs{m}\geq1} \mathcal{F}_X[\kappa\eta_{-,\rm near}]\left(\frac{2\pi m}{\delta}+\xi\right) 
\int_0^1e^{2\pi imx}
\overline{p_{+}(x,k_{\star})}p_{-}(x,k_{\star}) W_{\oo}(x)dx, \nn \\
& I_+^4(\xi;\eta_{-,\rm near}) \label{I_4} \\
& \qquad \equiv \sum_{m\in\mathbb{Z}}
\mathcal{F}_X[\kappa\eta_{-,\rm near}]\left(\frac{2\pi m}{\delta}+\xi\right) \int_0^1 e^{2\pi imx}
\overline{\Delta
p_{+}(x,\delta\xi)}p_{-}(x,k_{\star}) W_{\oo}(x)dx.  \nn
\end{align} 
The $m=0$ contribution  is nonzero in this case, provided $\left\langle\Phi_1,W_\oo\Phi_2\right\rangle\ne0$.
\medskip

\nit \underline{\it Expansion of  inner products \eqref{inner3} and \eqref{inner4}: }  Note the form of the
dependence of $\rho_{+}(x,\delta x)$ on $x$ and  recall the expansion \eqref{pexp} of $p_{+}(x,k_{\star}+\delta\xi)$.
Applying Lemma \ref{poisson_exp} to  \eqref{inner3}   with 
$f(x,\delta\xi)=p_+(x;k_\star+\delta\xi)$, $g(x)=W_\oo(x)$ and $\Gamma(x,\delta x)=
\kappa(\delta x) \rho_+(x,\delta x)$ we have
\begin{align}
&\frac{\delta}{2\pi}\ \inner{e^{i\xi\delta \cdot}p_+(\cdot,k_{\star}+\delta\xi),
W_{\oo}(\cdot)\kappa(\delta\cdot) {\rho}_{+}(\cdot,\delta\cdot)}_{L^2(\R)}  
\equiv I_+^5(\xi;\eta_{+,\rm near}) + I_+^6(\xi;\eta_{+,\rm near}) \ , \nn\\
&\textrm{where}\nn\\
&I_+^5(\xi;\eta_{+,\rm near}) \equiv \sum_{m\in\mathbb{Z}} \int_0^1e^{2\pi imx}
 \mathcal{F}_X[\kappa\rho_{+}]\left(x,\frac{2\pi m}{\delta}+\xi\right)
\overline{p_{+}(x,k_{\star})} W_{\oo}(x)dx \ , \label{I_5} \\
&I_+^6(\xi;\eta_{+,\rm near}) \equiv \sum_{m\in\mathbb{Z}} \int_0^1 e^{2\pi imx} 
\mathcal{F}_X[\kappa\rho_{+}]\left(x,\frac{2\pi m}{\delta}+\xi\right)
\overline{\Delta p_{+}(x,\delta\xi)} W_{\oo}(x)dx \ . \label{I_6}
\end{align} 

\nit  Further, by  Lemma \ref{poisson_exp} applied to   \eqref{inner4} with
 $f(x,\delta\xi)=p_+(x;k_\star+\delta\xi)$, $g(x)=W_\oo(x)$ and $\Gamma(x,\delta x)=
\kappa(\delta x) \rho_-(x,\delta x)$ we have
\begin{align}
&\frac{\delta}{2\pi}\ \inner{e^{i\xi\delta \cdot}p_+(\cdot,k_{\star}+\delta\xi),
W_{\oo}(\cdot) \kappa(\delta\cdot) {\rho}_{-}(\cdot,\delta\cdot)}_{L^2(\R)} \equiv
 I_+^7(\xi;\eta_{-,\rm near}) + I_+^8(\xi;\eta_{-,\rm near}),\nn\\
 &\textrm{where}\nn\\
& I_+^7(\xi;\eta_{-,\rm near}) \equiv \sum_{m\in\mathbb{Z}} \int_0^1 e^{2\pi imx} 
\mathcal{F}_X[\kappa\rho_{-}]\left(x,\frac{2\pi m}{\delta}+\xi\right)
\overline{p_{+}(x,k_{\star})} W_{\oo}(x)dx \ , \label{I_7}\\
& I_+^8(\xi;\eta_{-,\rm near}) \equiv
 \sum_{m\in\mathbb{Z}} \int_0^1 e^{2\pi imx} 
 \mathcal{F}_X[\kappa\rho_{-}]\left(x,\frac{2\pi m}{\delta}+\xi\right)
\overline{\Delta p_{+}(x,\delta\xi)} W_{\oo}(x)dx \ . \label{I_8}
\end{align}

\nit Assembling the above expansions, we find that the full inner product, \eqref{full_inner}, may be expressed as:
\begin{equation}
 \inner{\Phi_{+}(\cdot,k_{\star}+\delta\xi),\kappa(\delta\cdot)W_{\oo}(\cdot)
\eta_{\rm near}(\cdot)}_{L^2(\R)} = \thetasharp\widehat{\kappa\eta}_{{\rm near},-}(\xi) +
\sum_{j=1}^8I_{+}^j(\xi;\eta_{\rm near})\ ,\label{full-ipplus}
\end{equation} 
where by assumption \eqref{theta-ne0} in Theorem \ref{thm:validity}:
\begin{equation}
\thetasharp=\inner{\Phi_1,W_{\oo}\Phi_2}_{L^2([0,1])}\ne0 \ .
\label{thetasharp-def}
\end{equation}
A very similar calculation yields:
\begin{equation}
 \inner{\Phi_{-}(\cdot,k_{\star}+\delta\xi),\kappa(\delta\cdot)W_{\oo}(\cdot)
\eta_{\rm near}(\cdot)}_{L^2(\R)} = \thetasharp\widehat{\kappa\eta}_{{\rm near},+}(\xi) +
\sum_{j=1}^8I_{-}^j(\xi;\eta_{\rm near}),\ \label{full-ipminus}
\end{equation} 
where the terms $I_{-}^j(\xi;\eta_{\rm near})$ are defined analogously to $I_{+}^j(\xi;\eta_{\rm near})$. 
Note that each term $I_{\pm}^j(\xi;\eta_{\rm near})$ is linear in $\eta_{+,\rm near}$
 or $\eta_{-,\rm near}$.   
\medskip

Noting that $F_\pm[\mu,\delta]$, displayed in \eqref{Fb-def}-\eqref{Fdef}, is an affine function of $\mu$, that
$I^j_\pm(\xi; \widehat{\eta}_{\rm near})$ are linear in $\widehat{\eta}_{+,\rm near}$ and $\eta_{-,\rm near}$,
and that $\eta_{\rm far}$, given by \eqref{eta_far_affine}, is affine in both $\mu$ and $\eta_{\rm near}$, we
may summarize the above calculations in the following:

\begin{proposition}
\label{near_freq_compact}
Let
 \begin{equation}
 \widehat{\beta}(\xi) =
 \begin{pmatrix}
  \widehat{\eta}_{-,\rm near}(\xi) \\ \widehat{\eta}_{+,\rm near}(\xi)\
 \end{pmatrix}.\label{beta-def}
\end{equation}
  The near-energy system \eqref{near6}-\eqref{near5} for the corrector $(\eta,\mu)$ (see \eqref{eta_near+far}) may be written compactly in the form:
\begin{equation}
 \label{compacterroreqn}
 \left(\widehat{\mathcal{D}}^{\delta}+\widehat{\mathcal{L}}^{\delta}(\mu) -\delta \mu\right)\widehat{\beta}(\xi) =
\mu\widehat{\mathcal{M}}(\xi;\delta) + \widehat{\mathcal{N}}(\xi;\delta).
\end{equation} 
Here, $ \widehat{\mathcal{D}}^{\delta}$ denotes the Fourier transform of a ``band-limited Dirac operator'' defined by:
\begin{equation}
 \widehat{\mathcal{D}}^{\delta}\widehat{\beta}(\xi) \equiv -\lamsharp\sigma_3 \xi\widehat{\beta}(\xi)
+ \thetasharp\chi\left(\abs{\xi}\leq\delta^{\tau-1}\right)\sigma_1\widehat{\kappa\beta}(\xi),\ \ 0<\tau<1/2.
\label{bl-dirac-op}
\end{equation} 
Furthermore, $\widehat{\mathcal{L}}^{\delta}(\mu)$ is a linear operator acting on $\widehat{\beta}$ given by
\begin{align}
\widehat{\mathcal{L}}^{\delta}(\mu)\widehat{\beta}(\xi) &\equiv \frac{1}{2}\delta
\chi\left(\abs{\xi}\leq\delta^{\tau-1}\right)
\begin{pmatrix}
 E''_{-}(\widetilde{\xi}^\delta_{-})\\E''_{+}(\widetilde{\xi}^\delta_{+})
\end{pmatrix} \xi^2\widehat{\beta}(\xi) \nn \\
& \quad + \chi\left(\abs{\xi}\leq\delta^{\tau-1}\right)
\sum_{j=1}^8
\begin{pmatrix}
 I_{-}^j(\xi;\widehat{\eta}_{\pm,\rm near}(\xi))\\ I_{+}^j(\xi;\widehat{\eta}_{\pm,\rm near}(\xi))
\end{pmatrix} \label{L_op_1}\\
&\quad\ -\ \chi\left(\abs{\xi}\leq\delta^{\tau-1}\right)
\begin{pmatrix}
\inner{\Phi_{-}(\cdot,k_{\star}+\delta\xi),\kappa(\delta\cdot)W_{\oo}(\cdot)
[A\eta_{\rm near}](\cdot;\mu,\delta)}_{L^2(\R)} \\
\inner{\Phi_{+}(\cdot,k_{\star}+\delta\xi),
\kappa(\delta\cdot)W_{\oo}(\cdot)[A\eta_{\rm near}](\cdot;\mu,\delta)}_{L^2(\R)}
\end{pmatrix}, \label{L_op_2} \\
\widehat{\mathcal{M}}(\xi;\delta) &\equiv \sum_{j=1}^3\widehat{\mathcal{M}}_j(\xi;\delta)\nn\\
&=\chi\left(\abs{\xi}\leq\delta^{\tau-1}\right)
\begin{pmatrix}
 \inner{\Phi_{-}(\cdot,k_{\star}+\delta\xi),\psi^{(0)}(\cdot,\delta\cdot)}_{L^2(\R)} \\
 \inner{\Phi_{+}(\cdot,k_{\star}+\delta\xi),\psi^{(0)}(\cdot,\delta\cdot)}_{L^2(\R)}
\end{pmatrix} \label{M_op_1} \\
&~~~+ \delta \chi\left(\abs{\xi}\leq\delta^{\tau-1}\right)
\begin{pmatrix}
 \inner{\Phi_{-}(\cdot,k_{\star}+\delta\xi),\psi^{(1)}_p(\cdot,\delta\cdot)}_{L^2(\R)} \\
 \inner{\Phi_{+}(\cdot,k_{\star}+\delta\xi),\psi^{(1)}_p(\cdot,\delta\cdot)}_{L^2(\R)}
\end{pmatrix}\label{M_op_2} \\
&~~~+ \chi\left(\abs{\xi}\leq\delta^{\tau-1}\right)
\begin{pmatrix}
 \inner{\Phi_{-}(\cdot,k_{\star}+\delta\xi),\kappa(\delta\cdot)W_{\oo}(\cdot)B(\cdot;\delta)}_{L^2(\R)} \\
 \inner{\Phi_{+}(\cdot,k_{\star}+\delta\xi),\kappa(\delta\cdot)W_{\oo}(\cdot)B(\cdot;\delta)}_{L^2(\R)}
\end{pmatrix},\label{M_op_3}
\end{align}
and $\widehat{\mathcal{N}}(\xi;\delta)$ is independent of $\mu$, given by
\begin{align}
 \widehat{\mathcal{N}}(\xi;\delta)&\equiv \sum_{j=1}^4\widehat{\mathcal{N}}_j(\xi;\delta)\nn\\
 &=\chi\left(\abs{\xi}\leq\delta^{\tau-1}\right) \times \label{N_op_1} \\
&\qquad \begin{pmatrix}
\inner{\Phi_{-}(x,k_{\star}+\delta\xi), \left(2\partial_x\partial_X-\kappa(X)W_{\oo}(x)\right)\psi^{(1)}_p(x,X)
\Big|_{X=\delta x}}_{L^2(\R_x)} \\
\inner{\Phi_{+}(x,k_{\star}+\delta\xi), \left(2\partial_x\partial_X-\kappa(X)W_{\oo}(x)\right)\psi^{(1)}_p(x,X)
\Big|_{X=\delta x}}_{L^2(\R_x)}
\end{pmatrix} \nn \\
&~~~+ \chi\left(\abs{\xi}\leq\delta^{\tau-1}\right)
\begin{pmatrix}
\inner{\Phi_{-}(x,k_{\star}+\delta\xi),\partial_X^2\psi^{(0)}(x,X)\Big|_{X=\delta x}}_{L^2(\R_x)} \\
\inner{\Phi_{+}(x,k_{\star}+\delta\xi),\partial_X^2\psi^{(0)}(x,X)\Big|_{X=\delta x}}_{L^2(\R_x)}
\end{pmatrix} \label{N_op_2}\\
&~~~+ \delta\chi\left(\abs{\xi}\leq\delta^{\tau-1}\right)
\begin{pmatrix}
\inner{\Phi_{-}(x,k_{\star}+\delta\xi),\partial_X^2\psi^{(1)}_p(x,X)\Big|_{X=\delta x}}_{L^2(\R_x)} \\
\inner{\Phi_{+}(x,k_{\star}+\delta\xi),\partial_X^2\psi^{(1)}_p(x,X)\Big|_{X=\delta x}}_{L^2(\R_x)}
\end{pmatrix}\label{N_op_3}\\
&~~~+ \chi\left(\abs{\xi}\leq\delta^{\tau-1}\right)
\begin{pmatrix}
\inner{\Phi_{-}(\cdot,k_{\star}+\delta\xi),\kappa(\delta\cdot)W_{\oo}(\cdot)C(\cdot;\delta)}_{L^2(\R)} \\
\inner{\Phi_{+}(\cdot,k_{\star}+\delta\xi),\kappa(\delta\cdot)W_{\oo}(\cdot)C(\cdot;\delta)}_{L^2(\R)}
\end{pmatrix}.\label{N_op_4}
\end{align}
\end{proposition}

\medskip

We conclude this section with a proposition  asserting that from an appropriate solution:
$( \widehat{\beta}^\delta$, $\mu(\delta) )\in L^{2,1}(\R_\xi)\times\R$
of \eqref{compacterroreqn} one can construct a bound state $\left(\Psi^\delta,E^\delta\right)$ of the eigenvalue
problem.\medskip

\nit{\bf Therefore we can henceforth focus on the solving and estimating the solution of the band-limited Dirac system
\eqref{compacterroreqn}.}
\medskip

\begin{proposition}\label{needtoshow}
Suppose, for $0<\delta<\delta_0$, the band-limited Dirac system \eqref{compacterroreqn} has a solution 
$\left(\widehat{\beta}^\delta(\xi),\mu(\delta)\right)$, where 
$\widehat{\beta}^\delta=(\widehat{\beta}^\delta_-,\widehat{\beta}^\delta_+)^T$,  satisfying:
\begin{align}
 \label{betabound2a1}
 &\norm{\widehat{\beta}(\cdot;\mu,\delta)}_{L^{2,1}(\mathbb{R})} \lesssim \delta^{-1},\ 0<\delta<\delta_0
 \quad ({\rm Proposition~\ref{solve4beta}}), \\
 &\mu(\delta) \text{ bounded and } \mu(\delta) - \mu_0 \to 0 \text{ as } \delta\to0 \quad ({\rm Proposition~\ref{proposition3}}).  \label{betabound2a2}
 \end{align} 
 
Recall from Remark \ref{bandlimited} that $\widehat{\beta}(\xi)$ is supported 
on the set where $|\xi|\le\delta^{\tau-1}$ and define 
\begin{equation}
\widehat{\eta}_{\rm near,+}(\xi)\ =\ \widehat{\beta}_+(\xi),\ \ \widehat{\eta}_{\rm near,-}(\xi)=\widehat{\beta}_-(\xi).
\label{eta-hat-def}
\end{equation} Let
\begin{align}
\eta^\delta_{\rm near}(x)&=\ \frac{1}{2\pi}\  \sum_{b=\pm}\int_{|k-k_\star |\le \delta^\tau} \widehat{\eta}^\delta_{{\rm
near},b}\left(\frac{k-k_\star}{\delta}\right)\ \Phi_b(x;k)\ dk \ , \label{eta_def_beta1} \\
\tilde{\eta}^\delta_{{\rm far},b}(k)&=\ 
\tilde{\eta}_{\rm far,b}[\eta_{\rm near},\mu,\delta](k),\ \ b\ge1\ ;
\quad \textrm{(see Corollary \ref{fixed-pt})} \ , \label{eta_def_beta2} \\
\eta^\delta_{\rm far}(x)&=\ \frac{1}{2\pi}\ \sum_{b=\pm} \int_{\delta^\tau \le |k-k_\star |\le \pi } 
\widetilde{\eta}^\delta_{{\rm far},b}\left(k\right)\ \Phi_b(x;k)\ dk \label{eta_def_beta3} \\
&\qquad +\ \frac{1}{2\pi}\ \sum_{b\ne\pm} \int_\mathcal{B}
\widetilde{\eta}^\delta_{{\rm far},b}\left(k\right)\ \Phi_b(x;k)\ dk \ . \nn
\end{align}
Finally, define
\begin{align}
\eta^\delta(x)&\equiv \eta^\delta_{\rm near}(x) + \eta^\delta_{\rm far}(x),\ \ 
E^\delta\equiv E_\star+\delta^2\mu(\delta),\ \
0<\delta<\delta_0.
\end{align}
Then, for $0<\delta<\delta_0$,\\
\begin{enumerate}
\item[(a)]\ $ \eta^\delta(x)\in H^{2}(\R)$.\\
\item[(b)]\ $\left(\eta^\delta,\mu(\delta)\right)$ solves the corrector equation \eqref{eta_eqn}.\\
\item[(c)]\ Theorem \ref{thm:validity} holds. The pair $(\Psi^\delta,E^\delta)$, defined by (see also
\eqref{main_result_ansatz})
\begin{equation}
 \label{main_result_ansatz1}
 \begin{split}
 \Psi^\delta(x) &= \delta^{1/2}\psi^{(0)}(x,X)+\delta^{3/2}\psi^{(1)}_p(x,X)+\delta^{3/2}\eta^\delta(x),\ \ X=\delta x,\\
 E^\delta &= E_\star+\delta^2E^{(2)}+o(\delta^2),
 \end{split}
\end{equation}
is a solution of the eigenvalue problem \eqref{perturbed_schro_prob1} with corrector estimates asserted in the statement
of Theorem \ref{thm:validity}.
\end{enumerate}
\end{proposition}
\medskip

\nit To prove Proposition \ref{needtoshow} we use the following

\begin{lemma}
 \label{beta_vs_eta}
There exists a $\delta_0>0$ sufficiently small such that, for all $0<\delta<\delta_0$, the following holds:
 Assume $\beta \in L^2(\R)$ and let $\eta^\delta_{\rm near}(x)$ be defined by \eqref{eta_def_beta1}-\eqref{eta_def_beta3}.
 Then, 
 \begin{equation}
 \label{beta_eta_bdd}
 \norm{\eta_{\rm near}}_{H^2(\R)} \lesssim \delta^{1/2} \norm{\beta}_{L^2(\R)}.
 \end{equation}
\end{lemma}

\begin{proofof}{\it Proof of Proposition \ref{needtoshow}:} From $\widehat{\beta}$ we first construct
$\eta^\delta_{\rm near}$, which satisfies the bound:
$ \norm{\eta_{\rm near}}_{H^2(\R)} \lesssim \delta^{1/2} \norm{\beta}_{L^2(\R)}$ (Lemma
\ref{beta_vs_eta}).
Next, part 2 of Corollary \ref{fixed-pt}, \eqref{eta-far-bound}, gives a bound on 
 $\eta_{\rm far}$:  \\
 $\norm{\eta_{\rm far}[\eta_{\rm near};\mu,\delta]}_{H^2(\R)} \le\ C''\left(\ 
\delta^{1-\tau}\norm{\eta_{\rm near}}_{L^2(\R)}+\delta^{1/2-\tau}\ \right)$.
Note that all steps in our derivation of the band-limited Dirac
system 
\eqref{compacterroreqn} are reversible, in particular our application of the Poisson summation formula in
$L^2_{\rm loc}$. Therefore, $(\Psi^\delta,E^\delta)$, given by \eqref{main_result_ansatz1} is an $H^2(\R)$ eigenpair.
\end{proofof}

\medskip

It remains to prove Lemma \ref{beta_vs_eta}.

\begin{proofof}{\it Proof of Lemma \ref{beta_vs_eta}:} Since $\eta_{\rm near}$ is band-limited,  by Remark
\ref{L2H2} it suffices to bound $\eta_{\rm near}$
in $L^2(\R)$.
By \eqref{nearinxi},
\begin{align}
\norm{\eta_{\rm near}}_{L^2(\R)} &\leq \delta \sum_{j\in\{+,-\}} 
\norm{p_j(\cdot,k_{\star})\eta_{j,\rm near}(\delta\cdot)}_{L^2(\R)} + 
\delta\sum_{j\in\{+,-\}} \norm{\rho_j(\cdot,\delta\cdot)}_{L^2(\R)}  \nn \\
&\lesssim \delta^{1/2}\sum_{j\in\{+,-\}} \norm{\eta_{j,\rm near}}_{L^2(\R)} + 
\delta\sum_{j\in\{+,-\}}\norm{\rho_j(\cdot,\delta\cdot)}_{L^2(\R)} \ .
\label{eta_near_bdd}
\end{align} 
We next estimate $\|\rho_j(\cdot,\delta\cdot)\|_{L^2(\R)}$, displayed in
\eqref{rhodefn}:
\begin{equation}
\norm{\rho_{\pm}(\cdot,\delta\cdot)}_{L^2(\R)} = \norm{
\int_{\abs{\xi}\leq\delta^{\tau-1}} e^{i\xi X} \Delta p_{\pm}(x,\delta\xi) \widehat{\eta}_{\pm,\rm near}(\xi)
d\xi \bigg|_{X=\delta x} }_{L^2(\R_x)}. \label{rho_norm} \\
\end{equation} 
By periodicity,  $\Delta p_{\pm}(x,\delta\xi)$ may be expanded as a Fourier series:
\begin{equation}
\Delta p_{\pm}(x,\delta\xi) = \sum_{m\in\mathbb{Z}}(\widehat{\Delta p}_{\pm})_m(\delta \xi) e^{2\pi i mx}\ ,
\label{deltap_fourier}
\end{equation} 
and by smoothness in $x$ and $k$:  
$\left\| (\widehat{\Delta p}_{\pm})_m \right\|_{L^\infty[0,2\pi]}\ \lesssim\ \delta^\tau(1+m^2)^{-1}.$\\
Let
\begin{equation}\label{Jx}
J^\delta_m(\xi)\ \equiv\ (1+m^2)\ \chi\left(\abs{\xi}\leq\delta^{\tau-1}\right)\ (\widehat{\Delta p}_{\pm})_m(\delta
\xi)
\widehat{\eta}_{\pm,\rm near}(\xi)\ . 
\end{equation}
Then
\begin{align*}
\norm{\rho_{\pm}(\cdot,\delta\cdot)}_{L^2(\R)}^2 
&=  \int_{\R_x} \abs{\sum_{m\in\mathbb{Z}}
e^{2\pi imx}\ \int_{\abs{\xi}\leq\delta^{\tau-1}} e^{i\xi \delta x} (\widehat{\Delta p}_{\pm})_m(\delta \xi)
\widehat{\eta}_{\pm,\rm near}(\xi) d\xi }^2dx \\
&= \int_{\R_x} \abs{\sum_{m\in\mathbb{Z}}\ \frac{e^{2\pi imx}}{1+m^2}\ \int_{\R_\xi} e^{i\xi \delta
x}J^\delta_m(\xi)d\xi }^2dx\\
&=\  \int_{\R_x} \abs{\sum_{m\in\mathbb{Z}}\ \frac{e^{2\pi imx}}{1+m^2}\ \widehat{J_m^\delta}(\delta x)}^2 dx \\
& \le\
 \int_{\R_x} \sum_{m\in\mathbb{Z}}\ \frac{1}{1+m^2}\ \left|\widehat{J_m^\delta}(\delta x)\right|^2 dx \\
& \lesssim\delta^{-1}\ \sum_{m\in\mathbb{Z}}\ \frac{1}{1+m^2} \int_\R \left|J_m^\delta(X)\right|^2 dX\ 
\lesssim\delta^{-1+2\tau}\ \int_{\R} \left|\widehat{\eta}_{\pm,\rm near}(\xi)\right|^2 d\xi \ ,
\end{align*}
or 
\begin{align}
\norm{\rho_{\pm}(\cdot,\delta\cdot)}_{L^2(\R)}\ 
\lesssim\delta^{-\frac12+\tau}\ \|\widehat{\eta}_{\pm,\rm near}\|_{L^2(\R)} \ .
\label{rho_pm-bound}\end{align}

\nit Returning to \eqref{eta_near_bdd} we have, for $\delta$ sufficiently small,
\begin{align}
 \label{eta_beta_bdd}
\norm{\eta_{\rm near}}_{L^2(\R)} &\lesssim \delta^{1/2}\sum_{j\in\{+,-\}}
(1+\delta^{\tau})\norm{\widehat{\eta}_{j,\rm near}}_{L^2(\R)}
\lesssim \delta^{1/2} \norm{\widehat{\beta}}_{L^2(\R)}\ .
\end{align} 
This completes the proof of the Lemma \ref{beta_vs_eta}.
\end{proofof}

\section{Analysis of the band-limited Dirac system}\label{analysis-blDirac}

The formal $\delta\downarrow0$ limit of $\widehat{\mathcal{D}}^\delta$, displayed in \eqref{bl-dirac-op}, is a 1D Dirac operator, $\widehat{\mathcal{D}}$, given by
\begin{equation}
 \label{diraclimit}
 \widehat{\mathcal{D}}\widehat{\beta}(\xi) \equiv -\lamsharp\sigma_3\xi \widehat{\beta}(\xi) + \thetasharp\sigma_1\widehat{\kappa\beta}(\xi).
\end{equation} 
Our goal is to solve the  system \eqref{compacterroreqn}.
We therefore  rewrite the linear operator in equation \eqref{compacterroreqn} as a perturbation of   $\widehat{\mathcal{D}}$, and seek $\widehat{\beta}$ as a solution to:
\begin{equation}
 \label{erroreqnfactored}
 \widehat{\mathcal{D}}\widehat{\beta}(\xi) + \left(\widehat{\mathcal{D}}^{\delta}-\widehat{\mathcal{D}} +
\widehat{\mathcal{L}}^{\delta}(\mu) -\delta \mu\right)\widehat{\beta}(\xi) =
\mu\widehat{\mathcal{M}}(\xi;\delta)
+ \widehat{\mathcal{N}}(\xi;\delta).
\end{equation} 

\begin{remark}\label{bandlimited}
We seek a band-limited solution,  $\beta$, {\it i.e.} $\widehat{\beta}=\chi(|\xi|\le\delta^{\tau-1})\widehat{\beta}$.
Note however that our reformulation of \eqref{compacterroreqn} as \eqref{erroreqnfactored} is perturbative about 
$\mathcal{D}$, which does not preserve band-limited functions; $
\widehat{\mathcal{D}}\chi(|\xi|\le\delta^{\tau-1})\ne\chi(|\xi|\le\delta^{\tau-1})\widehat{\mathcal{D}}$. Our strategy
is to first solve \eqref{erroreqnfactored} for $\widehat{\beta}$ in  $L^{2,1}(\R;d\xi)$. We then note that this solution is
band-limited. Indeed, clearly equation \eqref{erroreqnfactored} for $\widehat{\beta}$ may be rewritten as
\eqref{compacterroreqn}. Now applying the projection 
 $\chi(|\xi|\ge\delta^{\tau-1})$ to  \eqref{compacterroreqn}  yields
$\lambda\sigma_3\xi\chi(|\xi|\ge\delta^{\tau-1})\widehat{\beta}(\xi)=0$, and so clearly $\widehat{\beta}(\xi)$ is
supported on the set $\{|\xi|\le\delta^{\tau-1}\}$.
\end{remark}
\medskip

We shall now solve \eqref{erroreqnfactored} using a Lyapunov-Schmidt reduction argument. 
By Theorem \ref{thm:dirac_bound_state}, the null space of $\widehat{\mathcal{D}}$ is spanned by $\widehat{\alpha}_{\star}(\xi)$,  the Fourier transform of the zero energy eigenstate $\alpha_{\star}(X)$ of $\mathcal{D}$; see \eqref{dirac-m-soln}.
Since  $\alpha_{\star}(X)$ is Schwartz class, so too is $\widehat{\alpha}_{\star}(\xi)$ and 
$\widehat{\alpha}_{\star}(\xi)\in H^s(\R)$ for $s\ge1$.
For any $f\in L^2{(\R)}$ introduce the orthogonal projection operators,
\begin{equation}
 \label{projops}
 \widehat{P}_{\parallel}f =
\inner{\widehat{\alpha}_{\star},f}\widehat{\alpha}_{\star},~~~\text{and}~~~\widehat{P}_{\perp}f =
(I-\widehat{P}_{\parallel})f.
\end{equation} 
Since $\widehat{P}_{\parallel}\ \widehat{\mathcal{D}}\ \widehat{\beta}(\xi)=0$ and
$\widehat{P}_{\perp}\widehat{\mathcal{D}}\widehat{\beta}(\xi)=\widehat{\mathcal{D}}\widehat{\beta}(\xi)$, 
equation \eqref{erroreqnfactored} is equivalent to  the system
\begin{align}
&\widehat{P}_{\parallel}\left\{\left(\widehat{\mathcal{D}}^{\delta}-\widehat{\mathcal{D}} +
\widehat{\mathcal{L}}^{\delta}(\mu) -\delta
\mu\right)\widehat{\beta}(\xi) - \mu\widehat{\mathcal{M}}(\xi;\delta) -
\widehat{\mathcal{N}}(\xi;\delta)\right\} = 0, \label{pplleqn}\\
&\widehat{\mathcal{D}}\widehat{\beta}(\xi) +
\widehat{P}_{\perp}\left\{\left(\widehat{\mathcal{D}}^{\delta}-\widehat{\mathcal{D}}
+ \widehat{\mathcal{L}}^{\delta}(\mu) -\delta \mu\right)\widehat{\beta}(\xi)\right\} =
\widehat{P}_{\perp}\left\{\mu\widehat{\mathcal{M}}(\xi;\delta) +
\widehat{\mathcal{N}}(\xi;\delta)\right\}.
\label{pperpeqn}
\end{align} Our strategy will be to first solve \eqref{pperpeqn}  for $\widehat{\beta}=
\widehat{\beta}[\mu,\delta]$, for $\delta>0$ and sufficiently small. This is carried out in Section \ref{proof-of-solve4beta}.
   We then  substitute
 $\widehat{\beta}[\mu,\delta]$ into \eqref{pplleqn} and obtain a closed scalar equation involving  $\mu$ and $\delta$. 
 In Section \ref{final-reduction} below, this equation is solved
   for $\mu=\mu(\delta)$  for $\delta$ small. 

The first step in this strategy is accomplished in
\begin{proposition}\label{solve4beta}
Fix $M>0$. There exists $\delta_0>0$ and a  mapping 
\[(\mu,\delta)\in R_{M,\delta_0}\equiv \{|\mu|<M\}\times (0,\delta_0)\mapsto \widehat{\beta}(\cdot;\mu,\delta)\in
L^{2,1}(\R),\]
which is Lipschitz in $\mu$, such that $\widehat{\beta}(\cdot;\mu,\delta)$ solves
\eqref{pperpeqn} for $(\mu,\delta)\in R_{M,\delta_0}$. Furthermore, we have the bound
\begin{equation}
 \label{betabound2a}
 \norm{\widehat{\beta}(\cdot;\mu,\delta)}_{L^{2,1}(\R)} \lesssim \delta^{-1},\ 0<\delta<\delta_0.
 \end{equation}
\end{proposition}

\section{Proof of Proposition \ref{solve4beta}}\label{proof-of-solve4beta}
  Since $\widehat{\mathcal{D}}$ is
invertible on the range of $\widehat{P}_{\perp}$, equation \eqref{pperpeqn} may be
rewritten as
\begin{equation}
\label{betafactored}
\left(I+\widehat{\mathcal{C}}^{\delta}(\mu)\right)\widehat{\beta}(\xi) =
\widehat{\mathcal{D}}^{-1}\widehat{P}_{\perp}\left\{\mu\widehat{\mathcal{M}}(\xi;\delta) +
\widehat{\mathcal{N}}(\xi;\delta)\right\},
\end{equation} 
where $\widehat{\mathcal{C}}^{\delta}(\mu)$ is the linear operator given by
\begin{equation}
\label{copp}
\widehat{\mathcal{C}}^{\delta}(\mu) =
\widehat{\mathcal{D}}^{-1}\widehat{P}_{\perp}\left(\widehat{\mathcal{D}}^{\delta}-\widehat{\mathcal{D}}+
\widehat{\mathcal{L}}^{\delta}(\mu)-\delta \mu\right).
\end{equation} 

We study $\widehat{\mathcal{C}}^{\delta}(\mu)$ as a operator from 
$L^{2,1}(\R)$ to $L^{2,1}(\R)$. Our immediate goal is to show that  
$\| \widehat{\mathcal{C}}^{\delta}(\mu)\|_{L^{2,1}(\R)\to L^{2,1}(\R)}<1$  and that the right hand side
of \eqref{betafactored} is in $L^{2,1}(\R)$. This will imply the invertibility of
$I+\widehat{\mathcal{C}}^{\delta}(\mu)$ and that we may solve for $\widehat{\beta}\in L^{2,1}(\R;d\xi)$. We
first obtain $L^{2,1}\to L^2$  bounds on the operators $(\widehat{\mathcal{D}}^{\delta}-\widehat{\mathcal{D}})$ and
$\widehat{\mathcal{L}}^{\delta}(\mu)$, and bounds on $L^2$ norms of $\widehat{\mathcal{M}}(\xi;\delta)$ and
$\widehat{\mathcal{N}}(\xi;\delta)$.
To obtain the required $L^2$ bounds, we employ Lemma \ref{beta_vs_eta} along with the following two lemmata. The first lemma is used to bound contributions arising from the
$I^j_{\pm}(\xi)$ terms.
\medskip

\begin{lemma}
 \label{I_bdds}
 Let $f(x,\omega)$, $g(x)$ denote functions satisfying:
 \begin{align}
C_f\ \equiv\ \underset{0\le x\le1,~\abs{\omega}\leq\delta^{\tau}}{\sup} \abs{f(x,\omega)}\ <\ \infty,\ \ 
C_g'\ \equiv\ \|g\|_{L^\infty[0,1]}<\infty \ . \nn
\end{align}
Let $\Gamma:(x,X)\in[0,1]\times\R\mapsto\Gamma(x,X)$ and denote by $\widehat{\Gamma}(x;\zeta)$ the Fourier transform
of $\Gamma(x,X)$ with respect to the variable $X$. Assume further that
 \begin{equation}
   \left\|\ \sup_{0\le x\le1} \widehat{\Gamma}(x,\zeta)\ \right\|_{L^{2,1}(\R_\zeta)}\ <\infty.
   \label{Gammahat-bound}\end{equation}
Define 
\begin{equation*}
 I_m(\xi;\delta) \equiv \int_0^1 e^{2\pi imx} \widehat{\Gamma} \left(x,\frac{2\pi m}{\delta}+\xi\right)
\overline{f(x,\delta\xi)} g(x) dx. \label{I_m_def}
\end{equation*} 
Then, summing over all $m\in\Z$ with $m\ne0$ we have the bound:
\begin{equation}
 \label{I_bdds_1}
  \norm{\chi\left(\abs{\xi}\leq\delta^{\tau-1}\right) \sum_{\abs{m}\geq1} I_m(\xi;\delta) }_{L^2(\R_{\xi})}
\lesssim C_f\ C_g'\ \delta \norm{\sup_{0\le x\le1}\widehat{\Gamma}(x,\zeta)}_{L^{2,1}(\R_\zeta)},
\end{equation} 
and summing over all $m\in\Z$ we have the bound:
\begin{equation}
 \label{I_bdds_2}
 \norm{\chi\left(\abs{\xi}\leq\delta^{\tau-1}\right) \sum_{m\in\mathbb{Z}} I_m(\xi;\delta)}_{L^2(\R_{\xi})}
\lesssim \ C_f\ C_g'\ 
\norm{\sup_{0\le x\le1}\widehat{\Gamma}(x,\zeta)}_{L^{2,1}(\R_\zeta)}.
\end{equation} 
\end{lemma}

\nit {\it Proof:} We first prove \eqref{I_bdds_1}. 
With $I_m(\xi;\delta)$ as defined in \eqref{I_m_def}, 
\begin{align*}
& \norm{\chi\left(\abs{\xi}\leq\delta^{\tau-1}\right) \sum_{\abs{m}\geq1}
I_m(\xi;\delta)}_{L^2(\R_{\xi})}^2\nn\\
&= \norm{\chi\left(\abs{\xi}\leq\delta^{\tau-1}\right) \sum_{\abs{m}\geq1} \int_0^1 e^{2\pi imx}
\widehat{\Gamma}\left(x,\frac{2\pi m}{\delta}+\xi\right) \overline{f(x,\delta\xi)} g(x)
dx}_{L^2(\R_{\xi})}^2 \\
& \le C_f^2\ C_g'^2\ \norm{\chi\left(\abs{\xi}\leq\delta^{\tau-1}\right) \sum_{\abs{m}\geq1}
\sup_{0\le x\le1}\abs{\widehat{\Gamma}\left(x,\frac{2\pi m}{\delta}+\xi\right)}}_{L^2(\R_{\xi})}^2 \\
&= C_f^2\ C_g'^2\ \int_{\R}\chi\left(\abs{\xi}\leq\delta^{\tau-1}\right)
 \sum_{\abs{m}\geq1} \frac{ \abs{ \frac{2\pi m}{\delta}+\xi } }{ \abs{ \frac{2\pi m}{\delta}+\xi }  } 
\sup_{0\le x\le1}\abs{\widehat{\Gamma}\left(x, \frac{2\pi m}{\delta}+\xi \right)}^2 d\xi.
\end{align*} 
Note that $\abs{\xi}\leq\delta^{\tau-1}$ implies that $\abs{ \frac{2\pi m}{\delta}+\xi }\gtrsim
\frac{|m|}{\delta}$ for
$\abs{m}\geq1$. Thus, by the Cauchy-Schwarz inequality,

\begin{align}
 \label{ibdd14}
& \norm{\chi\left(\abs{\xi}\leq\delta^{\tau-1}\right) \sum_{\abs{m}\geq1} I_m(\xi;\delta)}_{L^2(\R_{\xi})}^2
\nn\\
&\lesssim C_f^2\ C_g'^2\ \delta^2 \int_{\abs{\xi}\leq\delta^{\tau-1}} \abs{\sum_{\abs{m}\geq1} \frac{1}{\abs{m}} 
\abs{\frac{2\pi m}{\delta}+\xi} \sup_{0\le x\le1}\abs{\widehat{\Gamma}\left(x,\frac{2\pi m}{\delta}+\xi\right)}}^2
d\xi
\nonumber\\
&\leq C_f^2\ C_g'^2\ \delta^2 \times  \\
&\quad \int_{\abs{\xi}\leq\delta^{\tau-1}} \left(\sum_{\abs{m}\geq1}
\frac{1}{m^2}\right) \left(\sum_{\abs{m}\geq1} \abs{\frac{2\pi m}{\delta}+\xi}^2
\sup_{0\le x\le1}\abs{\widehat{\Gamma}\left(x,\frac{2\pi m}{\delta}+\xi\right)}^2 \right) d\xi. \nn
\end{align} 
Since  $\delta$ is taken to be  small,  for each $m$, the
integrals in \eqref{ibdd14} are over disjoint intervals.  Therefore, 
\begin{align*}
\norm{\chi\left(\abs{\xi}\leq\delta^{\tau-1}\right) \sum_{\abs{m}\geq1} I_m(\xi;\delta)}_{L^2(\R_{\xi})}^2
&\lesssim\ C_f^2\ C_g'^2\ \delta^2 \int_{\R} |\zeta|^2\sup_{0\le x\le1}|\widehat{\Gamma}(x,\zeta)|^2
d\zeta\nn\\
& \leq\ C_f^2\ C_g'^2\ \delta^2 \norm{\sup_{0\le x\le1}\widehat{\Gamma}(x,\cdot)}_{L^{2,1}(\R)}^2.
\end{align*} 
This completes the proof of the first bound, \eqref{I_bdds_1}.

The second bound, \eqref{I_bdds_2}, is for an expression for a sum over $\Z$, which includes the $m=0$ term, is
similarly proved:
\begin{align*}
&\norm{\chi\left(\abs{\xi}\leq\delta^{\tau-1}\right) \sum_{m\in\mathbb{Z}} I_m(\xi;\delta)}_{L^2(\R_{\xi})}^2 \\
&= \norm{\chi\left(\abs{\xi}\leq\delta^{\tau-1}\right) 
 \int_0^1 \sum_{m\in\mathbb{Z}} 
\widehat{\Gamma}\left(x,\frac{2\pi m}{\delta}+\xi\right)\ e^{2\pi imx} \overline{f(x,\delta\xi)}
g(x)dx}_{L^2(\R_{\xi})}^2  \\
& \lesssim\ C_f^2\ C_g'^2\ 
\norm{\chi\left(\abs{\xi}\leq\delta^{\tau-1}\right) \sum_{m\in\mathbb{Z}} 
\sup_{0\le x\le1}\abs{\widehat{\Gamma}\left(x,\frac{2\pi m}{\delta}+\xi\right)}}_{L^2(\R_{\xi})}^2 \\
&\lesssim\  C_f^2\ C_g'^2\ 
\int_{\R}\chi\left(\abs{\xi}\leq\delta^{\tau-1}\right) \abs{ \sum_{m\in\mathbb{Z}} 
\frac{1+\abs{\frac{2\pi m}{\delta}+\xi}}{1+\abs{\frac{2\pi m}{\delta}+\xi}}
\sup_{0\le x\le1}\abs{\widehat{\Gamma}\left(x,\frac{2\pi m}{\delta}+\xi\right)}}^2 d\xi\ .
\end{align*} 
 Noting that $1+\abs{\frac{2\pi m}{\delta}+\xi}^2\gtrsim
1+\abs{m}^2$ for all $m$, because $\abs{\xi}\leq\delta^{\tau}/\delta$, for $\delta>0$ and small, we proceed as in the
proof of \eqref{I_bdds_1} and 
obtain:
\begin{equation*}
 \norm{\chi\left(\abs{\xi}\leq\delta^{\tau-1}\right) \sum_{m\in\mathbb{Z}} I_m(\xi;\delta) }_{L^2(\R)} \lesssim
 C_f\ C_g'
\norm{\sup_{0\le x\le1}\widehat{\Gamma}(x,\cdot)}_{L^{2,1}(\R)},
\end{equation*} completing the proof of \eqref{I_bdds_2} and therewith the lemma.

\medskip

The next lemma offers a formula for bounding a commonly occurring norm.

\begin{lemma}
 \label{parseval_bdd}
 For all $f\in L^2(\R)$ and $b\in\mathbb{N}$,
 \begin{equation}
 \label{mnbddformula}
 \norm{\chi\left(\abs{\xi}\leq\delta^{\tau-1}\right)
\inner{\Phi_b(\cdot,k_\star+\delta\xi),f(\cdot)}_{L^2(\R)}}_{L^2(\R_\xi)}
 \lesssim \delta^{-1/2}\norm{f}_{L^2(\R)}.
\end{equation}
\end{lemma}

{\it Proof.}
 By Parseval's identity, \eqref{parseval},
\begin{equation*}
 \frac{1}{(2\pi)^2}\int_{\abs{k-k_{\star}}\leq\delta^{\tau}}\abs{\inner{\Phi_b(\cdot,k),f(\cdot)}_{L^2(\R)}}^2
dk \leq \norm{f}_{L^2(\R)}^2,
\end{equation*} and therefore, under the change of variables $k-k_{\star}=\delta\xi$, we obtain
\begin{equation*}
 \frac{1}{(2\pi)^2} \int_{\abs{\xi}\leq\delta^{\tau-1}}\abs{\inner{\Phi_b(\cdot,k_\star+\delta\xi),f(\cdot)}_{L^2(\R)}}^2
d\xi \leq \delta^{-1}\norm{f}_{L^2(\R)}^2. \qquad \Box
\end{equation*}

We may now proceed to obtain the $L^2$ bounds on the operators $\widehat{\mathcal{D}}^{\delta}-\widehat{\mathcal{D}}$
and $\widehat{\mathcal{L}}^{\delta}(\mu)$.

\begin{proposition}
 \label{lemma12}
 There exists a $\delta_0>0$ such that for all $0<\delta<\delta_0$ and $0<\tau<1$:
 \begin{equation}
 \label{diracdifflemma}
 \norm{(\widehat{\mathcal{D}}^{\delta}-\widehat{\mathcal{D}})\widehat{\beta}}_{L^{2}(\R)} \lesssim
\delta^{1-\tau}\norm{\widehat{\beta}}_{L^{2,1}(\R)},
\end{equation} and
\begin{equation}
 \label{llemma}
 \norm{\widehat{\mathcal{L}}^{\delta}(\mu)\widehat{\beta}}_{L^{2}(\R)} \lesssim
\delta^{\tau}\norm{\widehat{\beta}}_{L^{2,1}(\R)}\ .
\end{equation} 
\end{proposition}

\begin{proofof} \textit{Proof of Proposition \ref{lemma12}:}
Note, from \eqref{bl-dirac-op} and \eqref{diraclimit} that 
\begin{equation}
(\widehat{\mathcal{D}}^{\delta}-\widehat{\mathcal{D}})\widehat{\beta} = 
\vartheta_\sharp\chi\left(|\xi|>\delta^{\tau-1}\right)\sigma_1\widehat{\kappa\beta} . 
\label{dirac-diff}
\end{equation}
To prove bound \eqref{diracdifflemma} we use that $\widehat{\mathcal{D}}^{\delta}-\widehat{\mathcal{D}}$ is only
supported at high frequencies:
\begin{align*}
& \norm{(\widehat{\mathcal{D}}^{\delta}-\widehat{\mathcal{D}})\widehat{\beta}}_{L^2(\R)}^2 
\approx
 \norm{\chi\left(\abs{\xi}>\delta^{\tau-1}\right)\widehat{\kappa\beta}(\xi)}_{L^2(\R_\xi)}^2 
 \ = \int_{\abs{\xi}>\delta^{\tau-1}}\norm{\widehat{\kappa\beta}(\xi)}_{\mathbb{C}^2}^2 d\xi
\\
 &= \int_{\abs{\xi}>\delta^{\tau-1}}\abs{\xi}^{-2}\norm{\xi\widehat{\kappa\beta}(\xi)}_{\mathbb{C}^2}^2
d\xi \ \leq \delta^{2(1-\tau)}\int_{\R}\norm{\xi\widehat{\kappa\beta}(\xi)}_{\mathbb{C}^2}^2
d\xi \\
 &= \frac{\delta^{2(1-\tau)}}{2\pi}\int_{\R}\norm{\partial_X(\kappa(X)\beta(X))}_{\mathbb{C}^2}^2 dX \\
 & \lesssim \delta^{2(1-\tau)}\int_{\R} \left(\norm{\beta(X)}_{\mathbb{C}^2}^2+
\norm{\partial_X\beta(X)}_{\mathbb{C}^2}^2 \right) dX \ \approx \delta^{2(1-\tau)}
\norm{\widehat{\beta}}_{L^{2,1}(\R)}^2.
\end{align*} 
Here, we have used that $\kappa$ and $\kappa'$ are bounded functions on $\R$.

We now embark of the proof of the bound \eqref{llemma} for the operator 
$\widehat{\beta}\mapsto \widehat{\mathcal{L}}^{\delta}(\mu)\widehat{\beta}$.
 We have from \eqref{L_op_1} and \eqref{L_op_2},
\begin{align}
& \norm{\widehat{\mathcal{L}}^{\delta}(\mu)\widehat{\beta}(\cdot)}_{L^2(\R)}\nn\\
  &\lesssim
 \delta \norm{\chi\left(\abs{\xi}\leq\delta^{\tau-1}\right)\xi^2\widehat{\beta}(\xi)}_{L^2(\R)} + 
 \norm{\chi\left(\abs{\xi}\leq\delta^{\tau-1}\right)\sum_{j=1}^8 
 \begin{pmatrix}
 I^j_{-}(\xi)\\
 I^j_{+}(\xi)
 \end{pmatrix}}_{L^2(\R)} \label{op1_bdd}\\
 & + \norm{\chi\left(\abs{\xi}\leq\delta^{\tau-1}\right)
\begin{pmatrix}
\inner{\Phi_{-}(\cdot,k_{\star}+\delta\xi),\kappa(\delta\cdot)W_{\oo}(\cdot)
[A\eta_{\rm near}](\cdot;\mu,\delta)}_{L^2(\R)} \\
\inner{\Phi_{+}(\cdot,k_{\star}+\delta\xi),
\kappa(\delta\cdot)W_{\oo}(\cdot)[A\eta_{\rm near}](\cdot;\mu,\delta)}_{L^2(\R)}
\end{pmatrix}}_{L^2(\R_{\xi})}. \label{Ldef_bdd}
\end{align}
The first term on the right hand side of \eqref{op1_bdd}-\eqref{Ldef_bdd} is clearly bounded by 
$\delta^{\tau}\|\widehat{\beta}\|_{L^{2,1}(\R)}$. The next term involves a sum over 
 terms $I^j_\pm,\ j=1,\dots,8$, which we bound using Lemma \ref{I_bdds}. We shall explicitly estimate
only the $I_+^j$ terms. The $I_-^j$ terms are similarly controlled.\medskip

\nit{\it  Bound on $I^1_{+}$:} \ Applying the bound \eqref{I_bdds_1} to $I^1_{+}(\xi)$ in \eqref{ibstar1} with 
\[f(x)=p_+(x,k_{\star}),\ g(x)=p_+(x,k_{\star})W_{\oo}(x),\ \textrm{ and}\  \widehat{\Gamma}=
\widehat{\Gamma}(\xi) = \mathcal{F}_X[\kappa\eta_{+,\rm near}]\left(\xi\right)\ , \]  
and using Parseval's identity and the boundedness of $\kappa$ and $\kappa'$ we obtain:
\begin{align}
\norm{\chi\left(\abs{\xi}\leq\delta^{\tau-1}\right)I^1_{+}(\xi)}_{L^2(\R_\xi)}^2 
&\lesssim \delta^2 \norm{\mathcal{F}_X[\kappa\eta_{+,\rm near}](\xi)}_{L^{2,1}(\R_{\xi})}^2 \nn\\
&= \frac{\delta^2}{2\pi} \int_{\R} \abs{\partial_X(\kappa(X)\eta_{+,\rm near}(X))}^2 dX \nn\\
&\lesssim \delta^2\int_{\R} \left( \abs{\eta_{+,\rm near}(X)}^2+
\abs{\partial_X\eta_{+,\rm near}(X)}^2 \right) dX \nn\\
 &\approx \delta^2 \norm{\widehat{\eta}_{+,\rm near}}_{L^{2,1}(\R)}^2\nn \ ,
\end{align}
implying
\begin{equation}
\norm{\chi\left(\abs{\xi}\leq\delta^{\tau-1}\right)I^1_{+}(\xi)}_{L^2(\R_\xi)}\lesssim
\delta \norm{\widehat{\eta}_{+,\rm near}}_{L^{2,1}(\R)}\ .
\label{I1bound}\end{equation}

\nit{\it Bound on $I^3_+$:} We similarly apply \eqref{I_bdds_1} to $I^3_+$, the sum over $|m|\ge1$, given in
\eqref{I_3}, and obtain
 \begin{equation}
\norm{\chi\left(\abs{\xi}\leq\delta^{\tau-1}\right)I^3_{+}(\xi)}_{L^2(\R)} \lesssim
\delta \norm{\widehat{\eta}_{-,\rm near}}_{L^{2,1}(\R)}. \label{I3bound}
\end{equation}
\medskip

\nit{\it  Bound on $I^2_{+}$:} Set
$f=f(x,\delta\xi)=\Delta p_+(x,\delta\xi),\ g(x)=p_+(x,k_{\star})W_{\oo}(x)$, and
$\widehat{\Gamma}=\widehat{\Gamma}(\xi) =
\mathcal{F}_X[\kappa\eta_{+,\rm near}]\left(\xi\right).$
Using  \eqref{I_bdds_2}, the boundedness of $\kappa$ and $\kappa'$,
  and that by \eqref{deltapbdd} $ C_{\Delta p_+}\lesssim\delta^\tau$, we have
 \begin{align}
\norm{\chi\left(\abs{\xi}\leq\delta^{\tau-1}\right)I^2_{+}(\xi)}_{L^2(\R_\xi)} &\lesssim
 C_{\Delta p_+} C_{p_+W_{\rm o}}'\ 
\norm{\mathcal{F}_X[\kappa\eta_{+,\rm near}]}_{L^{2,1}(\R_\xi)} \nn \\
& \lesssim \delta^{\tau} \norm{\widehat{\eta}_{+,\rm near}}_{L^{2,1}(\R)}\ .
\label{I2bound}\end{align} 

 \nit{\it  Bound on $I^4_{+}$:}\ The bound on $I^4_{+}$ is similar to that of $I^2_{+}$:
\begin{equation}
 \norm{\chi\left(\abs{\xi}\leq\delta^{\tau-1}\right)I^4_{+}(\xi)}_{L^2(\R)} \lesssim
 \delta^{\tau} \norm{\widehat{\eta}_{-,\rm near}}_{L^{2,1}(\R)}.
 \label{I4bound}
\end{equation}
 \nit{\it  Bound on $I^5_{+}$:}\ The expression $I^5_{+}$ in \eqref{I_5} can be estimated using \eqref{I_bdds_2}
  of Lemma \ref{I_bdds}, where we set:
  \[ \widehat{\Gamma}(x,\zeta)=\mathcal{F}_{X}[\kappa\rho_{+}](x,\zeta),\ f(x)=p_+(x;k_\star),\ \text{and }
g(x)=W_\oo(x) \ .\]
By \eqref{I_bdds_2} and the boundedness of $\kappa$ and $\kappa'$:
\begin{align*}
 \norm{\chi\left(\abs{\xi}\leq\delta^{\tau-1}\right)I^5_{+}(\xi)}_{L^2(\R)}^2 &\lesssim
\int_{\R} \underset{0\le x\le1}{\sup}\left(
 (1+\abs{\zeta}^2)\abs{\mathcal{F}_{X}[\kappa\rho_{+}](x,\zeta)}^2\right) d\zeta \\
 &\lesssim \int_{\R}\ \underset{0\le x\le1}{\sup}\left(\abs{\rho_{+}(x,X)}^2
 +\abs{\partial_X\rho_{+}(x,X)}^2\right) dX\\
 &\approx \int_{\R_\zeta}\ \underset{0\le x\le1}{\sup}\ (1+|\zeta|^2)\abs{\mathcal{F}_X[\rho_{+}](x,\zeta)}^2
 \ d\zeta \ .
\end{align*} 
From \eqref{rhodefn} we have
\begin{equation*}
 \mathcal{F}_X[\rho_{+}](x,\zeta) = \chi\left(\abs{\zeta}\leq\delta^{\tau-1}\right) \Delta
p_{+}(x,\delta\zeta) \widehat{\eta}_{+,\rm near}(\zeta)\ .
\end{equation*}
Therefore, 
\begin{align*}
 &\norm{\chi\left(\abs{\xi}\leq\delta^{\tau-1}\right)I^5_{+}(\xi)}_{L^2(\R)}^2\nn\\
 &\quad \lesssim\int_{\R}\  \underset{0\le x\le1}{\sup}\ \abs{\Delta p_{+}(x,\delta\zeta)}^2\ 
  \chi\left(\abs{\zeta}\leq\delta^{\tau-1}\right) \ 
 (1+|\zeta|^2)|\widehat{\eta}_{+,\rm near}(\zeta)|^2  d\zeta \\
 &\quad \lesssim  ( C_{\Delta p_{+}} )^2  \|\widehat{\eta}_{+,\rm near}\|_{L^{2,1}(\R)}^2\ \approx 
 \delta^{2\tau}\norm{\widehat{\eta}_{+,\rm near}}_{L^{2,1}(\R)}^2,
\end{align*} where  we have used the bound on $\Delta p_{+}$ in \eqref{deltapbdd}. 
Thus,
\begin{equation}
\norm{\chi\left(\abs{\xi}\leq\delta^{\tau-1}\right)I^5_{+}(\xi)}_{L^2(\R)}\lesssim
\delta^{\tau}\norm{\widehat{\eta}_{+,\rm near}}_{L^{2,1}(\R)} \ .
\label{I5bound}\end{equation}
Similar estimates yield the
bounds: 
\begin{align}
\norm{\chi\left(\abs{\xi}\leq\delta^{\tau-1}\right)I^6_{+}(\xi)}_{L^2(\R)} &\lesssim
\delta^{\tau}\norm{\widehat{\eta}_{+,\rm near}}_{L^{2,1}(\R)}\label{I6bound} \ ,\\
\norm{\chi\left(\abs{\xi}\leq\delta^{\tau-1}\right)I^7_{+}(\xi)}_{L^2(\R)} &\lesssim
\delta^{\tau}\norm{\widehat{\eta}_{-,\rm near}}_{L^{2,1}(\R)}\label{I7bound}\ ,\\
\norm{\chi\left(\abs{\xi}\leq\delta^{\tau-1}\right)I^8_{+}(\xi)}_{L^2(\R)} &\lesssim
\delta^{\tau}\norm{\widehat{\eta}_{-,\rm near}}_{L^{2,1}(\R)}\ .\label{I8bound}
\end{align}

Finally, to complete our bound of $\widehat{\mathcal{L}}(\mu)$, we bound the term displayed in \eqref{Ldef_bdd} using
Lemmas \ref{parseval_bdd} and \ref{beta_vs_eta}. With $f(x)=\kappa(\delta
x)W_{\oo}(x)[A\eta_{\rm near}](x;\mu,\delta)$ apply \eqref{mnbddformula} and the bound \eqref{A_bdd} on the mapping
$\eta_{\rm near}\mapsto
[A\eta_{\rm near}](x;\mu,\delta)$  to obtain
\begin{align*}  
&\norm{\chi\left(\abs{\xi}\leq\delta^{\tau-1}\right)
\begin{pmatrix}
\inner{\Phi_{-}(\cdot,k_{\star}+\delta\xi),\kappa(\delta\cdot)W_{\oo}(\cdot)
[A\eta_{\rm near}](\cdot;\mu,\delta)}_{L^2(\R)} \\
\inner{\Phi_{+}(\cdot,k_{\star}+\delta\xi),
\kappa(\delta\cdot)W_{\oo}(\cdot)[A\eta_{\rm near}](\cdot;\mu,\delta)}_{L^2(\R)}
\end{pmatrix}}_{L^2(\R_{\xi})} \\
&~~~\lesssim \delta^{-1/2}\norm{\kappa(\delta\cdot)W_{\oo}(\cdot)
[A\eta_{\rm near}](\cdot;\mu,\delta)}_{L^2(\R)} \\
&~~~\lesssim \delta^{-1/2} \norm{[A\eta_{\rm near}](\cdot;\mu,\delta)}_{L^2(\R)}
\lesssim \delta^{-1/2} \delta^{1-\tau} \norm{\eta_{\rm near}}_{L^2(\R)}.
\end{align*} 

 Employing \eqref{beta_eta_bdd} and 
the Plancherel theorem then gives
\begin{align}
 &\norm{\chi\left(\abs{\xi}\leq\delta^{\tau-1}\right)
\begin{pmatrix}
\inner{\Phi_{-}(\cdot,k_{\star}+\delta\xi),\kappa(\delta\cdot)W_{\oo}(\cdot)
[A\eta_{\rm near}](\cdot;\mu,\delta)}_{L^2(\R)} \\
\inner{\Phi_{+}(\cdot,k_{\star}+\delta\xi),
\kappa(\delta\cdot)W_{\oo}(\cdot)[A\eta_{\rm near}](\cdot;\mu,\delta)}_{L^2(\R)}
\end{pmatrix}}_{L^2(\R_{\xi})} \nn\\
&~~~ \lesssim \delta^{1-\tau} \norm{\beta}_{L^2(\R)} = 
\delta^{1-\tau} \norm{\widehat{\beta}}_{L^2(\R)}. \label{near_bdd}
\end{align}

Finally, substituting \eqref{near_bdd} and the bounds on $I^j_\pm$ into \eqref{Ldef_bdd} and \eqref{op1_bdd} therefore
yields
\begin{align*}
 \norm{\widehat{\mathcal{L}}^{\delta}(\mu)\widehat{\beta}}_{L^2} &\lesssim
 \delta \delta^{\tau-1} \norm{\xi\widehat{\beta}(\xi)}_{L^2(\R_\xi)} +
\delta^{\tau}\norm{\widehat{\beta}}_{L^{2,1}(\R)} + 
\delta^{1-\tau} \norm{\widehat{\beta}}_{L^2(\R)} \\
& \lesssim \delta^{\tau}\norm{\widehat{\beta}}_{L^{2,1}(\R)},
\end{align*} because $0<\tau<1/2$.
\end{proofof}
\medskip

Next, we seek bounds on the inhomogeneous terms $\widehat{\mathcal{M}}(\xi;\delta)$ and
$\widehat{\mathcal{N}}(\xi;\delta)$ on the right hand side of \eqref{betafactored}.

\begin{proposition}
 \label{lemma13}
$\widehat{\mathcal{M}}(\xi;\delta)$ and $\widehat{\mathcal{N}}(\xi;\delta)$ are bounded in $L^{2}(\R)$,
with bounds
\begin{equation*}
\norm{\widehat{\mathcal{M}}(\cdot;\delta)}_{L^2(\R)} \lesssim \delta^{-1}, ~~~\text{and} ~~~
\norm{\widehat{\mathcal{N}}(\cdot;\delta)}_{L^2(\R)} \lesssim \delta^{-1}.
\end{equation*}
\end{proposition}

\begin{proof}
Proving Proposition \ref{lemma13} reduces to applications of \eqref{mnbddformula} from Lemma \ref{parseval_bdd}.
Applying \eqref{mnbddformula} to $f(x)=\psi^{(0)}(x,\delta x)$, $f(x)=\delta\psi^{(1)}_p(x,\delta x)$  and
$f(x)=\kappa(\delta x)W_{\oo}(x)B(x;\delta)$, respectively, gives
\begin{align*}
\norm{\widehat{\mathcal{M}}(\cdot;\delta)}_{L^2(\R)}&\lesssim \delta^{-1/2}\norm{\psi^{(0)}(\cdot,\delta
\cdot)}_{L^2(\R)} +\delta^{-1/2}\delta \norm{\psi^{(1)}_p(\cdot,\delta \cdot)}_{L^2(\R)} \\
 &\qquad+\delta^{-1/2} \norm{\kappa(\delta \cdot)W_{\oo}(\cdot)B(\cdot;\delta\cdot)}_{L^2(\R)}  \\
&\lesssim \delta^{-1/2}\delta^{-1/2} + \delta^{1/2}\delta^{-1/2} + \delta^{-1/2}\delta^{1/2-\tau}\lesssim \delta^{-1},
\end{align*} using Lemma \ref{lemma:psi_bounds}, bound \eqref{B_bdd} and that $0<\tau<1/2$.

And with $f(x)=(2\partial_x\partial_X-\kappa(X)W_{\oo}(x))\psi^{(1)}_p(x,X)$, 
$f(x)=\kappa(\delta x)W_{\oo}(x)C(x;\delta)$, 
$f(x)=\partial_X^2\psi^{(0)}(x,X)$, and
$f(x)=\delta\partial_X^2\psi^{(1)}_p(x,X)$, respectively, \eqref{mnbddformula} implies that
\begin{align*}
\norm{\widehat{\mathcal{N}}(\cdot;\delta)}_{L^2(\R)} &\lesssim
\delta^{-1/2}\norm{(2\partial_x\partial_X-\kappa(X)W_{\oo}(x))\psi^{(1)}_p(x,X)\Big|_{X=\delta x}}_{L^2(\R_x)} \\
&\qquad + 
\delta^{-1/2}\norm{\kappa(\delta \cdot)W_{\oo}(\cdot)C(\cdot;\delta\cdot)}_{L^2(\R)} \\
&\qquad+ \delta^{-1/2}\norm{\partial_X^2\psi^{(0)}(x,X)\Big|_{X=\delta x}}_{L^2(\R_x)} \\
& \qquad+ 
\delta\delta^{-1/2}\norm{\partial_X^2\psi^{(1)}_p(x,X)\Big|_{X=\delta x}}_{L^2(\R_x)} \\
&\leq \delta^{-1/2} \left[ \norm{2\partial_x\partial_X\psi^{(1)}_p(x,X)\Big|_{X=\delta x}}_{L^2(\R_x)}  \right. \\
&\qquad\qquad+ \norm{\kappa(\delta\cdot)W_{\oo}(\cdot)\psi^{(1)}_p(\cdot,\delta\cdot)}_{L^2(\R)} \\
&\qquad\qquad \left.
+\norm{\kappa(\delta \cdot)W_{\oo}(\cdot)C(\cdot;\delta\cdot)}_{L^2(\R)}
+\norm{\partial_X^2\psi^{(0)}(x,X)\Big|_{X=\delta x}}_{L^2(\R_x)} \right. \\
&\qquad\qquad \left. +\delta\norm{\partial_X^2\psi^{(1)}_p(x,X)\Big|_{X=\delta x}}_{L^2(\R_x)}\right].
\end{align*}
Employing Lemma \ref{lemma:psi_bounds} and bound \eqref{B_bdd} then gives
\begin{equation}
\norm{\widehat{\mathcal{N}}(\cdot;\delta)}_{L^2(\R)} \lesssim \delta^{-1/2} \left[ \delta^{-1/2} +
\delta^{-1/2} +\delta^{1/2-\tau}+\delta^{-1/2} +\delta \delta^{-1/2} \right] \lesssim \delta^{-1}, \label{N_bdd}
\end{equation} because $0<\tau<1/2$.
\end{proof}

\medskip

\nit \textbf{Bounding $\widehat{\mathcal{C}}^{\delta}(\mu)$ on $L^{2,1}(\R_\xi)$ using  the boundedness of wave
operators}. 
We first outline the strategy. Note
that we may express the inverse Dirac operator as $\widehat{\mathcal{D}}^{-1}\widehat{P}_{\perp} =
\widehat{\mathcal{D}}\widehat{\mathcal{D}}^{-2}\widehat{P}_{\perp}$. 
  Therefore, operator bounds on $\mathcal{D}^{-1}{P}_\perp$ can then be reduced to bounds on  $\mathcal{D}^{-2}{P}_\perp$. To prove the latter, we first note (in physical space) that
\begin{align}\label{diracsquared}
 \mathcal{D}^2 &=\left(i\lamsharp\sigma_3\partial_X + \thetasharp\kappa(X)\sigma_1\right)^2 \nonumber\\
 &= -\sigma_3^2\lamsharp^2\partial_X^2 +i\lamsharp\thetasharp\sigma_1\sigma_3\kappa(X)\partial_X
+i\thetasharp\lamsharp\sigma_3\sigma_1\partial_X\kappa(X) +\thetasharp^2\sigma_1^2\kappa^2(X)
\nonumber\\
 &= I(-\lamsharp^2\partial_X^2+\thetasharp^2\kappa^2(X)) +
i\lamsharp\thetasharp(\sigma_1\sigma_3+\sigma_3\sigma_1)\kappa(X)\partial_X
+ i\lamsharp\thetasharp\sigma_3\sigma_1\kappa'(X)
\nonumber\\
 &= I(-\lamsharp^2\partial_X^2+\thetasharp^2\kappa^2(X)) +\lamsharp\thetasharp\sigma_2\kappa'(X) \ .
 \end{align} 
Here, $I$ denotes the $2\times2$ identity matrix, prime denotes differentiation with respect to $X$ and the Pauli
matrices $\sigma_j,\ j=1,2,3$ are displayed in \eqref{Pauli-sigma}. Thus
$\mathcal{D}^2$ is a localized perturbation of a diagonal matrix Schr\"odinger operator. Furthermore,
since $[I,\sigma_2]=0$, $\mathcal{D}^2$ can be conjugated to a diagonal matrix of Schr\"odinger operators, which, in
turn,  can be
conjugated on the range of $P_{\perp}$ to a diagonal matrix of constant coefficient Schr\"odinger operators using
{\it wave operators}, introduced below. The resulting constant coefficient Schr\"odinger operators are then  bounded
using the boundedness properties of wave operators.

Let $V(X)$ denote a function which decays as $X\to\infty$. Wave operators can be used to extend  bounds for functions of
a  constant coefficient operator $H_0$, to bounds for a variable coefficient operator $H=H_0+V(X)$ on its continuous spectral part. The
wave operator $W_+$, and its adjoint, $W_+^{\ast}$, associated with the constant and variable coefficient
Hamiltonians $H_0$ and $H$ are defined by
\begin{align*}
 W_+ &\equiv s- \lim_{t\rightarrow\infty}e^{itH}e^{-itH_0} \\ 
 W_+^{\ast} &\equiv s- \lim_{t\rightarrow\infty}e^{itH_0}e^{-itH}P_{\perp}, 
\end{align*} where $s-$ denotes the strong limit and $P_{\perp}$ is the continuous spectral projection defined
in \eqref{projops}.

For any $f$ Borel on $\R$,
\begin{equation}
 \label{intertwine}
 f(H)\ {P}_{\perp} = W_+f(H_0)W_+^{\ast},~~~f(H_0) = W_+^{\ast}f(H)W_+.
\end{equation} Thus, any bounds on $f(H)\ {P}_{\perp}$ acting between the Sobolev spaces $W^{k_1,p_1}(\R)$
and $W^{k_2,p_2}(\R)$, can be derived from bounds on $f(H_0)$ between these spaces if the wave operator
$W_+$ is bounded between $W^{k_1,p_1}(\R)$ and $W^{k_2,p_2}(\R)$ for $k_j\geq0$ and $p_j\geq1$, $j=1,2$. Boundedness of  $W_+$ is a consequence of the following 
\begin{theorem}
 \label{thm7} {\bf \cite{Weder:99}}.
 Consider the Schr\"odinger operator $H=-\D_y^2+V(y)$,  with a potential, $V$, satisfying
$
  \norm{V}_{L^1_a(\R)}\equiv \int_{\R}(1+\abs{y})^a\abs{V(y)}dy <
\infty,
$
with $a>5/2$. Furthermore, let $k\ge1$, and assume that $\partial^{\ell}_{x}V\in L^{1}(\R)$, for
$\ell=0,1,2,\ldots,k-1$.  Then, the wave operators $W_+$, $W_+^{\ast}$, originally defined on $W^{k,p}\cap L^2$,
$1\leq p\leq\infty$, have extensions to a bounded operators on $W^{k,p}$, $1<p<\infty$. Moreover, there exist positive
constants $c_p$, such that for all  $f\in W^{k,p}(\R)\cap L^2$ ($1<p<\infty$) 
\begin{equation}
 \label{waveopbdd12}
 \norm{W_+f}_{W^{k,p}(\R)}\leq c_p\norm{f}_{W^{k,p}(\R)}, ~~
 \norm{W_+^{\ast}f}_{W^{k,p}(\R)}\leq c_p\norm{f}_{W^{k,p}(\R)} \ .
\end{equation} \end{theorem}
See also \cites{Yajima:95,DF:06,DMW:11}.

Continuing, we next diagonalize $\mathcal{D}^2$, displayed in  \eqref{diracsquared}. It is easy to check that $\sigma_2$
has the
eigenpairs
 $(1,i)^T ~\leftrightarrow~ \Lambda =1$  and 
$(1,-i)^T  ~\leftrightarrow~ \Lambda =-1$.
Therefore,  for any
$g(X)\in\mathbb{C}^2$,
\begin{align}
 \mathcal{D}^2g(X) &= \frac12\mathcal{D}^2\left[g_1(X)\begin{pmatrix}1\\ i\end{pmatrix} 
 +g_{-1}(X)\begin{pmatrix}1\\-i\end{pmatrix}\right] \nonumber\\
 &=\frac12\begin{pmatrix}1\\ i\end{pmatrix}
\left[-\lamsharp^2\partial_X^2+\thetasharp^2\kappa^2(X) + \lamsharp\thetasharp\kappa'(X)\right]g_1(X) \nn \\
 &~~~+\frac12\begin{pmatrix}1\\-i\end{pmatrix}
\left[-\lamsharp^2\partial_X^2+\thetasharp^2\kappa^2(X) - \lamsharp\thetasharp\kappa'(X)\right]g_{-1}(X),
\label{d2_g}
\end{align} where $g_1(X)=\inner{(1,i)^T,g(X)}_{\mathbb{C}^2}$ and $g_{-1}(X)=\inner{(1,-i)^T,g(X)}_{\mathbb{C}^2}$.
Defining
\begin{align*}
 H_0 &\equiv-\lamsharp^2\partial_X^2+\thetasharp^2\kappa_{\infty}^2\ {\rm and}\\\
   H_{\pm} &\equiv -\lambda_\sharp\D_X^2+\thetasharp^2\kappa^2(X)\pm\lamsharp\thetasharp\kappa'(X)\\
   &= H_0+\thetasharp^2(\kappa^2(X)-\kappa_{\infty}^2)\pm\lamsharp\thetasharp\kappa'(X) \ ,\ 
\end{align*}
we may then write \eqref{d2_g} as
\begin{equation}
 \label{d2_of_H}
 \mathcal{D}^2g(X) = 
 \frac12\begin{pmatrix}1\\ i\end{pmatrix} H_+ g_1(X)
+\frac12\begin{pmatrix}1\\-i\end{pmatrix} H_- g_{-1}(X) \ .
\end{equation}

Since $\kappa^2(X)-\kappa_{\infty}^2$ and $\kappa'(X)$ are assumed sufficiently smooth and  rapidly decaying as
$X\to\pm\infty$, we can apply Theorem \ref{thm7} to the wave operators associated with $H_\pm$ to obtain bounds on
functions of 
$\mathcal{D}^2$ in terms of bounds on functions of the constant coefficient operator $H_0$.
 \medskip

Let $W_{+,\kappa'}$ and $W_{+,-\kappa'}$, respectively, denote wave operators associated with the decaying potentials
$\Upsilon_+(X)\equiv \kappa^2(X)-\kappa_{\infty}^2+\kappa'(X)$ and
$\Upsilon_-(X)\equiv\kappa^2(X)-\kappa_{\infty}^2-\kappa'(X)$. Recall that $\Upsilon_\pm(X)$ satisfy the hypotheses of
Theorem \ref{thm7} with $k=0,2$ and $p=2$;
 see \eqref{kappa-hypotheses}.
 By the intertwining property  \eqref{intertwine}, for any $f$ Borel on $\R$,
\begin{align}
 \label{waveopbdd15}
& f(\mathcal{D}^2)P_{\perp}g\nn\\
 &=\frac12\begin{pmatrix}1\\ i\end{pmatrix}
f(H_+)P_{\perp}g_1(X) 
 +\frac12\begin{pmatrix}1\\-i\end{pmatrix}
f(H_-) P_{\perp}g_{-1}(X) \nonumber\\
 &=\frac12\begin{pmatrix}1\\ i\end{pmatrix}
W_{+,+\kappa'}f(H_0 )W_{+,+\kappa'}^{\ast} g_1(X)
+\frac12\begin{pmatrix}1\\-i\end{pmatrix}
 W_{+,-\kappa'} f(H_0) W_{+,-\kappa'}^{\ast} g_{-1}(X)
\nonumber\\ &= 
\frac12\begin{pmatrix} 
W_{+,+\kappa'} f(H_0)
W_{+,+\kappa'}^{\ast}&
 W_{+,-\kappa'} f(H_0) W_{+,-\kappa'}^{\ast}\\
 iW_{+,+\kappa'} f(H_0) W_{+,+\kappa'}^{\ast}&
 -iW_{+,-\kappa'} f(H_0)
W_{+,-\kappa'}^{\ast}\end{pmatrix}
 \begin{pmatrix}g_1\\g_{-1}\end{pmatrix}.
\end{align} 
We shall apply
\eqref{waveopbdd15} below with the choice $f(x)=x^{-1}$.

Returning to the operator $\widehat{\mathcal{C}}^{\delta}(\mu)$, displayed in \eqref{copp}, we separate the bounds as
\begin{align} 
 \label{waveopbdd16}
 \norm{\widehat{\mathcal{C}}^{\delta}(\mu)\widehat{\beta}(\cdot)}_{L^{2,1}(\R)} &\leq
\norm{\widehat{\mathcal{D}}^{-1}\widehat{P}_{\perp}
(\widehat{\mathcal{D}}-\widehat{\mathcal{D}}^{\delta})\widehat{\beta}(\cdot)}_{L^{2,1}(\R)}
+\norm{\widehat{\mathcal{D}}^{-1}\widehat{P}_{\perp}
\widehat{\mathcal{L}}^{\delta}(\mu)\widehat{\beta}(\cdot)}_{L^{2,1}(\R)} \nonumber\\
&~~~+\delta\abs{\mu}\norm{\widehat{\mathcal{D}}^{-1}\widehat{P}_{\perp}
\widehat{\beta}(\cdot)}_{L^{2,1}(\R)},
\end{align} and study each term separately.
From \eqref{waveopbdd16} it is clear that we need to study terms of the form  
$\norm{\widehat{\mathcal{D}}^{-1}\widehat{P}_{\perp} \widehat{g}}_{L^{2,1}(\R)}$.
By Parseval's identity 
\begin{align}
 \norm{\widehat{\mathcal{D}}^{-1}\widehat{P}_{\perp} \widehat{g}}_{L^{2,1}(\R)} &=
 \norm{\widehat{\mathcal{D}}\widehat{\mathcal{D}}^{-2}\widehat{P}_{\perp}
\widehat{g}}_{L^{2,1}(\R)} 
\approx  \norm{\mathcal{D}\mathcal{D}^{-2}P_{\perp} g}_{H^{1}(\R)}\ \lesssim \norm{\mathcal{D}^{-2}P_{\perp}
g}_{H^{2}(\R)}.
\label{wopp_parseval}
\end{align} 
We obtain then from \eqref{wopp_parseval}, \eqref{waveopbdd15} with $f(H_0)=H_0^{-1}$, and the wave operator
bounds of Theorem \ref{thm7}:
\begin{align}
&\norm{\widehat{\mathcal{D}}^{-1}\widehat{P}_{\perp} \widehat{g}}_{L^{2,1}(\R)} \nn \\
&\qquad \lesssim \norm{ 
\begin{pmatrix} W_{+,+\kappa'} H_0^{-1}
W_{+,+\kappa'}^{\ast}&   W_{+,-\kappa'} H_0^{-1} W_{+,-\kappa'}^{\ast}\\
  iW_{+,+\kappa'} H_0^{-1} W_{+,+\kappa'}^{\ast} &   -iW_{+,-\kappa'} H_0^{-1} W_{+,-\kappa'}^{\ast}\end{pmatrix}
\begin{pmatrix}g_1\\g_{-1}\end{pmatrix}}_{H^{2}(\R)} \nn \\
&\qquad \lesssim\ 
\norm{\begin{pmatrix}g_1\\g_{-1}\end{pmatrix}}_{L^{2}(\R)}\ .\label{Dinvgbound}
\end{align}
For the choice $\widehat{g}(\xi) =
(\widehat{\mathcal{D}}-\widehat{\mathcal{D}}^{\delta}) \widehat{\beta}(\xi)$, bound
\eqref{Dinvgbound} gives
\begin{align*}
 \norm{\widehat{\mathcal{D}}^{-1}\widehat{P}_{\perp}
(\widehat{\mathcal{D}}-\widehat{\mathcal{D}}^{\delta})\widehat{\beta}}_{L^{2,1}(\R)}
&= \norm{\widehat{\mathcal{D}}^{-1}\widehat{P}_{\perp} \widehat{g}}_{L^{2,1}(\R)} \\
&\lesssim \norm{(\mathcal{D}-\mathcal{D}^{\delta})\beta}_{L^{2}(\R)}
\lesssim \delta^{1-\tau}\norm{\widehat\beta}_{L^{2,1}(\R)},
\end{align*} 
where the final inequality follows from Proposition \ref{lemma12}.
 Similarly,
\begin{align*}
& \norm{\widehat{\mathcal{D}}^{-1}\widehat{P}_{\perp}
\widehat{\mathcal{L}}^{\delta}(\mu)\widehat{\beta}(\cdot)}_{L^{2,1}(\R)} \lesssim
\delta^{\tau}\norm{\widehat{\beta}}_{L^{2,1}(\R)},\ {\rm and}\ 
\norm{\widehat{\mathcal{D}}^{-1}\widehat{P}_{\perp}\widehat{\beta}}_{L^{2,1}(\R)} \leq
\norm{\widehat{\beta}}_{L^{2,1}(\R)}.
\end{align*} 
Substituting these bounds into \eqref{waveopbdd16}, we have, for $0<\tau<1$, 
\begin{equation*}
 \norm{\widehat{\mathcal{C}}^{\delta}(\mu)\widehat{\beta}(\cdot)}_{L^{2,1}(\R)} \lesssim
 \left(\delta^{1-\tau}+\delta^{\tau}+\delta\abs{\mu}\right)\norm{\widehat{\beta}}_{L^{2,1}(\R)}\ .
\end{equation*} 

\medskip

It follows that, for $\abs{\mu}\le M$ with $M$ fixed, and $0<\delta<\delta_0$ with $\delta_0$ sufficiently small and depending
on $M$,
\begin{equation}
\label{C_op_bdd}
 \norm{\widehat{\mathcal{C}}^{\delta}(\mu)}_{L^{2,1}(\R)\rightarrow L^{2,1}(\R)} < 1 \ .
\end{equation}
 Furthermore, using the same wave operator methods along with the bounds
in Proposition \ref{lemma13} we have that the right hand side of \eqref{betafactored} is in $L^{2,1}(\R)$
and satisfies the bound:
\begin{align}
 \label{rhsl21} 
 \norm{\widehat{\mathcal{D}}^{-1}\widehat{P}_{\perp} \left\{ \mu\widehat{\mathcal{M}}(\cdot;\delta) +
\widehat{\mathcal{N}}(\cdot;\delta)\right\}}_{L^{2,1}(\R)} &\lesssim
\abs{\mu}\norm{\widehat{\mathcal{M}}(\cdot;\delta)}_{L^2(\R)} +
\norm{\widehat{\mathcal{N}}(\cdot;\delta)}_{L^2(\R)} \nonumber\\
&\lesssim \abs{\mu}\delta^{-1} + \delta^{-1} \lesssim \delta^{-1}.
\end{align} 
Thus, for $0<\delta<\delta_0$ we may solve \eqref{betafactored} for $\widehat{\beta}$:
\begin{equation}
 \label{betasoln}
 \widehat{\beta}(\xi;\mu,\delta) =
\left(I+\widehat{\mathcal{C}}^{\delta}(\mu)\right)^{-1}\widehat{\mathcal{D}}^{-1}\widehat{P}_{\perp}\left\{
\mu\widehat{\mathcal{M}}(\xi;\delta) + \widehat{\mathcal{N}}(\xi;\delta)\right\}\ ,
\end{equation} 
which, from \eqref{rhsl21}, satisfies the bound
\begin{equation}
\label{betabound2}
\norm{\widehat{\beta}(\cdot;\mu,\delta)}_{L^{2,1}(\R)} \lesssim
\norm{\widehat{\mathcal{D}}^{-1}\widehat{P}_{\perp}\left\{\mu\widehat{\mathcal{M}}(\xi;\delta)
+\widehat{\mathcal{N}}(\xi;\delta)\right\}}_{L^{2,1}(\R)} \lesssim \delta^{-1}.
\end{equation}

We complete the proof of Proposition \ref{solve4beta} by verifying that $\widehat{\beta}(\xi;\mu,\delta)$ is Lipschitz
in $\mu$. To ease notation, let
\begin{equation*}
\widehat{S}^{\delta}(\mu)=\left(I+\widehat{\mathcal{C}}^{\delta}(\mu)\right)^{-1}\ {\rm and}\ \ \ 
\widehat{T}(\xi;\mu,\delta)= \widehat{\mathcal{D}}^{-1}\widehat{P}_{\perp}\left\{ \mu\widehat{\mathcal{M}}(\xi;\delta) +
\widehat{\mathcal{N}}(\xi;\delta)\right\}.
\end{equation*} 
From \eqref{C_op_bdd} and \eqref{rhsl21}, $\widehat{S}^{\delta}(\mu)$ and
$\widehat{T}(\xi;\mu,\delta)$ are bounded in their respective norms:
\begin{equation}
 \label{S_T_bdd}
 \norm{\widehat{S}^{\delta}(\mu)}_{L^{2,1}(\R)\rightarrow L^{2,1}(\R)} \lesssim 1, ~~~ \text{and}
~~~ \norm{\widehat{T}(\cdot;\mu,\delta)}_{L^{2,1}(\R)} \lesssim \delta^{-1},
\end{equation} and therefore $\widehat{S}^{\delta}(\mu)$
is a mapping from $L^{2,1}(\R)$ to $L^{2,1}(\R)$ and $\widehat{T}(\xi;\mu,\delta)\in
L^{2,1}(\R)$.

Let $\abs{\mu_1}$, $\abs{\mu_2} < M$. Then
\begin{align}
 &\norm{\widehat{\beta}(\cdot;\mu_2,\delta) - \widehat{\beta}(\cdot;\mu_1,\delta)}_{L^{2,1}(\R)} \nn \\
&\qquad = \norm{\widehat{S}^{\delta}(\mu_2)\widehat{T}(\cdot;\mu_2,\delta)-
\widehat{S}^{\delta}(\mu_1)\widehat{T}(\cdot;\mu_1,\delta)}_{L^{2,1}(\R)} \nn \\
&\qquad = \norm{\left(\widehat{S}^{\delta}(\mu_2)-\widehat{S}^{\delta}(\mu_1)\right)\widehat{T}(\cdot;\mu_2,\delta) -
\widehat{S}^{\delta}(\mu_1) \left(\widehat{T}(\cdot;\mu_2,\delta)-\widehat{T}(\cdot;\mu_1,\delta)\right)}_{L^{2,1}(\R)}
\nn \\
&\qquad = \norm{\left(\widehat{S}^{\delta}(\mu_2)-\widehat{S}^{\delta}(\mu_1)\right)
\widehat{T}(\cdot;\mu_2,\delta)}_{L^{2,1}(\R)} \label{beta_lip} \\
&\qquad\qquad+ \norm{
\widehat{S}^{\delta}(\mu_1) \left(\widehat{T}(\cdot;\mu_2,\delta)-\widehat{T}(\cdot;\mu_1,\delta)\right)}_{L^{2,1}(\R)} . \nn
\end{align}
We proceed to bound the various terms occurring in \eqref{beta_lip}, beginning with
$\widehat{S}^{\delta}(\mu_2)-\widehat{S}^{\delta}(\mu_1)$. Since $\widehat{\mathcal{C}}^{\delta}(\mu_1)$ is affine in
$\mu$ (see \eqref{copp}), we have, for $0<\delta<\delta_0$ with $\delta_0$ sufficiently small,
\begin{align}
& \norm{\widehat{S}^{\delta}(\mu_2)-\widehat{S}^{\delta}(\mu_1)}_{L^{2,1}(\R)\rightarrow L^{2,1}(\R)}
\nn \\
  & \ =
\norm{\left(I+\widehat{\mathcal{C}}^{\delta}(\mu_2)\right)^{-1} \left(\widehat{\mathcal{C}}^{\delta}(\mu_1) -
\widehat{\mathcal{C}}^{\delta}(\mu_2)\right) \left(I+\widehat{\mathcal{C}}^{\delta}(\mu_1)\right)^{-1}}_{
L^{2,1}(\R)\rightarrow L^{2,1}(\R)} \nn \\
&\ \lesssim \norm{\widehat{\mathcal{C}}^{\delta}(\mu_1) -
\widehat{\mathcal{C}}^{\delta}(\mu_2)}_{L^{2,1}(\R)\rightarrow L^{2,1}(\R)}\nn\\ 
&\ \lesssim \delta\norm{\widehat{\mathcal{D}}^{-1}\widehat{P}_{\perp}\left(\mu_2-\mu_1\right)}_{L^{2,1}(\R)\rightarrow
L^{2,1}(\R)} \label{S_bdd} \\
& \qquad + \norm{\widehat{\mathcal{D}}^{-1}\widehat{P}_{\perp} \left(\widehat{\mathcal{L}}^{\delta}(\mu_2)
-\widehat{\mathcal{L}}^{\delta}(\mu_1)\right)}_{L^{2,1}(\R)\rightarrow L^{2,1}(\R)}. \nn
\end{align} 
The $\mu$ dependence in $\widehat{\mathcal{L}}^{\delta}(\mu)\widehat{\beta}(\xi)$ arises through
$[A\eta_{\rm near}](\xi;\mu,\delta)$; see \eqref{L_op_2}. Applying the wave operator bounds in \eqref{Dinvgbound} and bound
\eqref{A_lip} from Corollary \ref{fixed-pt} gives 
\begin{align}
& \norm{\widehat{\mathcal{D}}^{-1}\widehat{P}_{\perp} \left(\widehat{\mathcal{L}}^{\delta}(\mu_2)
-\widehat{\mathcal{L}}^{\delta}(\mu_1)\right)\widehat{\beta}(\cdot)}_{L^{2,1}(\R)} \\
&\qquad \lesssim \delta^{-1/2}
\norm{[A\eta_{\rm near}](\cdot;\mu_2,\delta)-[A\eta_{\rm near}](\cdot;\mu_1,\delta)}_{L^{2,1}(\R)} \nn \\
&\qquad \leq C' \delta^{-1/2}\delta^{1-\tau} \abs{\mu_2-\mu_1} = C' \delta^{1/2-\tau} \abs{\mu_2-\mu_1} .
\label{L_bdd}
\end{align}
Similarly, from \eqref{Dinvgbound} we obtain
\begin{equation}
 \label{mu_bdd}
 \delta\norm{\widehat{\mathcal{D}}^{-1}\widehat{P}_{\perp}\left(\mu_2-\mu_1\right)}_{L^{2,1}(\R)\rightarrow
L^{2,1}(\R)} \lesssim \delta \abs{\mu_2-\mu_1}.
\end{equation} Substituting \eqref{L_bdd} and \eqref{mu_bdd} into \eqref{S_bdd} gives the desired bound on $\mathcal{S}^\delta(\mu)$:
\begin{equation}
 \label{S_full_bdd}
 \norm{\widehat{S}^{\delta}(\mu_2)-\widehat{S}^{\delta}(\mu_1)}_{L^{2,1}(\R)\rightarrow L^{2,1}(\R)}
\lesssim \delta^{1/2-\tau} \abs{\mu_2-\mu_1}.
\end{equation}

Next, the wave operator bound \eqref{Dinvgbound} and the $L^2-$ bound for $\widehat{\mathcal{M}}(\xi;\delta)$ in
Proposition \ref{lemma13} imply that
{
\begin{align}
 \label{T_full_bdd}
 \norm{\widehat{T}(\cdot;\mu_2,\delta)-\widehat{T}(\cdot;\mu_1,\delta)}_{L^{2,1}(\R)} & \leq \norm{
\mathcal{D}^{-1}\widehat{P}_{\perp}\widehat{\mathcal{M}}(\cdot;\delta) \left(\mu_2-\mu_1\right)}_{L^{2,1}(\R)}
 \lesssim \abs{\mu_2-\mu_1} \ .
\end{align} 
}

Finally, putting the above together and substituting bounds \eqref{S_full_bdd}, \eqref{T_full_bdd} and \eqref{S_T_bdd} into \eqref{beta_lip} yields
\begin{align}
 \norm{\widehat{\beta}(\cdot;\mu_2,\delta) - \widehat{\beta}(\cdot;\mu_1,\delta)}_{L^{2,1}(\R)} & \lesssim
\delta^{1/2-\tau}\abs{\mu_2-\mu_1}
\norm{\widehat{T}(\cdot;\mu,\delta)}_{L^{2,1}(\R)} \nn \\
&\qquad + \abs{\mu_2-\mu_1}
\norm{\widehat{S}^{\delta}(\mu)}_{L^{2,1}(\R)\rightarrow L^{2,1}(\R)} \nn \\
&\leq C_M \delta^{-1/2-\tau}\abs{\mu_2-\mu_1} \ .\label{beta_lip_mu}
\end{align} This proves, for $(\mu,\delta)\in R_{M,\delta_0}\equiv \{|\mu|<M\}\times (0,\delta_0)$, that
$\widehat{\beta}(\xi;\mu,\delta)$ is Lipschitz in $\mu$, completing the proof of Proposition \ref{solve4beta}.

\section{Final reduction to an equation for $\mu=\mu(\delta)$ and its solution}\label{final-reduction}
 Substituting expression \eqref{betasoln} for $\widehat{\beta}$ into \eqref{pplleqn}, gives an equation  
 relating $\mu$ and $\delta$:
\begin{equation*}
\mathcal{J}_+[\mu,\delta]=0,
\end{equation*} 
where
\begin{align*}
\mathcal{J}_{+}[\mu;\delta] &\equiv
\mu\ \delta\inner{\widehat{\alpha}_{\star}(\cdot),\widehat{\mathcal{M}}(\cdot;\delta)}_{L^2(\R)}
+ \delta\inner{\widehat{\alpha}_{\star}(\cdot),\widehat{\mathcal{N}}(\cdot;\delta)}_{L^2(\R)} \\
&~~~ -\delta\inner{\widehat{\alpha}_{\star}(\cdot),\left(\widehat{\mathcal{D}}^{\delta}-\widehat{\mathcal{D}}\right)
\widehat{\beta}(\cdot;\mu,\delta)}_{L^2(\R)}
-\delta\inner{\widehat{\alpha}_{\star}(\cdot),\widehat{\mathcal{L}}^{\delta}(\mu)
\widehat{\beta}(\cdot;\mu,\delta)}_{L^2(\R)}\nn\\
&~~~+\delta^2 \mu\inner{\widehat{\alpha}_{\star}(\cdot),\widehat{\beta}(\cdot;\mu,\delta)}_{L^2(\R)}. \nonumber
\end{align*}
The function $\mathcal{J}_{+}[\mu,\delta]$ is well defined for  $\abs{\mu}<M$ and $0<\delta<\delta_0$. 
It is also Lipschitz with respect to $\mu$. In the following proposition, we note that $\mathcal{J}_{+}[\mu,\delta]$ can
be
extended to the half-open interval $[0,\delta_0)$ to be continuous at $\delta=0$.
\medskip

\begin{proposition}
 \label{proposition3}
Let $\delta_0>0$ be as above. Define
 \begin{equation}
 \label{J_defn}
  \mathcal{J}[\mu,\delta] \equiv \left\{
\begin{array}{cl}
\mathcal{J}_{+}[\mu,\delta]  & ~~~\text{for}~~ 0<\delta<\delta_0,\\
\mu-\mu_0 & ~~~\text{for}~~ \delta=0\ .
\end{array} \right.\
 \end{equation} 
Here,
\begin{equation}
 \mu_0 \equiv -\inner{\alpha_{\star},\mathcal{G}^{(2)}}_{L^2(\R)} = E^{(2)}\ ,
\label{mu0-def}
\end{equation}
 where $\mathcal{G}^{(2)}$ is given in \eqref{ipG}; see also \eqref{solvability_cond_E2}.
 Fix $M\geq 2 \abs{\mu_0}$. Then,  
 \[ (\mu,\delta)\in\{|\mu|<M,\ 0\le\delta<\delta_0\}\mapsto\mathcal{J}(\mu,\delta)\]
  is well-defined and continuous. 
 \end{proposition} 
\medskip

\begin{proof}  In Appendix \ref{near_freq_limits} it is verified that, for all $0<\delta<\delta_0$ with $\delta_0$
sufficiently small and $0<\tau<1/2$, the following hold for some constant $C_M$:
\begin{align}
 \lim_{\delta\rightarrow0}\delta\inner{\widehat{\alpha}_{\star}(\cdot),
\widehat{\mathcal{M}}(\cdot;\delta)}_{L^2(\R)} &=1; \label{sketch_limit1} \\
 \lim_{\delta\rightarrow0}\delta\inner{\widehat{\alpha}_{\star}(\cdot),
\widehat{\mathcal{N}}(\cdot;\delta)}_{L^2(\R)} &= -\mu_0; \label{sketch_limit2} \\
 \abs{\delta\inner{\widehat{\alpha}_{\star}(\cdot), \left(\widehat{\mathcal{D}}^{\delta}-\widehat{\mathcal{D}}\right)
\widehat{\beta}(\cdot;\mu,\delta)}_{L^2(\R)}} &\le C_M \delta^{1-\tau}; \label{sketch_bound1} \\
 \abs{\delta\inner{\widehat{\alpha}_{\star}(\cdot), \widehat{\mathcal{L}}^{\delta}(\mu)
\widehat{\beta}(\cdot;\mu,\delta)}_{L^2(\R)}} &\le C_M \delta^{\tau}; \label{sketch_bound2} \\
 \abs{\delta^2 \mu\inner{\widehat{\alpha}_{\star}(\cdot), \widehat{\beta}(\cdot;\mu,\delta)}_{L^2(\R)}} 
&\le C_M \delta. \label{sketch_bound3}
\end{align} It follows from limits \eqref{sketch_limit1} and \eqref{sketch_limit2} and bounds \eqref{sketch_bound1} -
\eqref{sketch_bound3} that 
\begin{equation*}
 \mathcal{J}_+(\mu;\delta)=\mu-\mu_0+o(1) ~~~\text{as}~ \delta\rightarrow0, ~\text{uniformly for}~
|\mu|\leq M,
\end{equation*}
and therefore that $\mathcal{J}[\mu,\delta]$, defined in \eqref{J_defn}, is well-defined for
$(\mu,\delta)\in  \{(\mu,\delta)\ :\ |\mu|<M,\ 0\leq\delta<\delta_0\}$ and continuous at $\delta=0$.
\end{proof}
\smallskip

Given $\widehat{\beta}(\cdot,\mu,\delta)$, constructed in Proposition \ref{solve4beta}, to complete the solution of 
the Lyapunov-Schmidt reduced system \eqref{pplleqn}-\eqref{pperpeqn}, it suffices to solve \eqref{pplleqn} for
$\mu=\mu(\delta)$. Moreover, note that \eqref{pplleqn} holds if and only if $\mathcal{J}[\mu;\delta]=0$.

\medskip

\begin{proposition}\label{solveJeq0}
There exists $\delta_0>0$, and a function $\delta\mapsto\mu(\delta)$, defined for $0\le\delta<\delta_0$ such that:
 $|\mu(\delta)|\le M$, $\lim_{\delta\to0}\mu(\delta)=\mu(0)=\mu_0 \equiv E^{(2)}$ and 
 $\mathcal{J}[\mu(\delta),\delta]=0$ for all $0\le\delta<\delta_0$.
\end{proposition}
\medskip

{\it Proof of Proposition \ref{solveJeq0}:}\      Choose
$\mu'\in(\mu_0,M)$. Then, $\mathcal{J}(\mu',0)=\mu'-\mu_0>0$. By continuity at $\delta=0$, there exists
$\delta'\le\delta_0$ such that for all $\delta\in(0,\delta')$, we have $\mathcal{J}(\mu',\delta)>0$. Similarly, let
$\mu''\in(-M,\mu_0)$. Then,  we have $\mathcal{J}(\mu'',0)=\mu''-\mu_0<0$ and for all $\delta''\le\delta'$:
\[ \mathcal{J}(\mu'',\delta)<0\ \ {\rm and}\ \ \mathcal{J}(\mu',\delta)>0\ \ \textrm{for all}\ \  \delta\le\delta''.\]
It follows that there exists some $\mu(\delta)\in(\mu'',\mu')$ such that
$\mathcal{J}\left(\mu(\delta),\delta\right)=0$.

This completes the construction of a solution pair $\left(\widehat{\beta}^\delta(\xi),\mu(\delta)\right)$, with 
$\widehat{\beta}^\delta\in L^{2,1}(\R_\xi)$, which solves the band-limited Dirac system \eqref{compacterroreqn}. Our
main result, Theorem \ref{thm:validity}, now follows directly from  Proposition \ref{needtoshow}.
\bigskip

\appendix

\chapter{A Variant of Poisson Summation}\label{PSF}

\begin{theorem}\label{psum-L2a}
Let $\Gamma(x,X)$ be a function defined for $(x,X)\in\R\times\R$.
 Assume that the mapping  $x\mapsto\Gamma(x,X)$ is $H^2_{\rm per}([0,1])$ with respect to $x$ with values in
$L^2(\R_X)$, {\it i.e.}
\begin{align}\label{Gamma-conditions1a}
&\Gamma(x+1,X)\ =\ \Gamma(x,X) \ ,\\
&\sum_{j=0}^2\ \int_0^1\ \left\|\D_x^j\Gamma(x,\cdot)\right\|_{L^2(\R_X)}^2\ dx\ <\ \infty .
\label{Gamma-conditions2a}\end{align}
We denote this Hilbert space of functions by $\mathbb{H}^2$ with norm-squared, $\|\cdot\|_{\mathbb{H}^2}^2$, given in
\eqref{Gamma-conditions1a}-\eqref{Gamma-conditions2a}.
Denote by $\widehat{\Gamma}(x,\omega)$ the Fourier transform of $\Gamma(x,X)$ with respect to $X$
given by
\begin{equation}
\widehat{\Gamma}(x,\omega)\ =_{\rm def}\ \lim_{N\uparrow\infty}\ \frac{1}{2\pi}\int_{|x|\le N}e^{-i\omega
X}\Gamma(x,X)dX\ ,
\label{Ghat-def}
\end{equation}
where the limit is taken in $L^2([0,1]_x\times\R_\omega)$. \ Then, for any $\zeta_{\rm max}>0$ 
\begin{equation}\label{psum-GofxX} 
\sum_{n\in\Z} e^{- i\zeta(x+n)}\Gamma(x,x+n) = 
2\pi \sum_{n\in\Z} e^{2\pi i n x}\widehat{\Gamma}\left(x,2\pi n+\zeta\right)
\end{equation}
in $L^2([0,1]\times[-\zeta_{\rm max},\zeta_{\rm max}];dxd\zeta)$.
\end{theorem}
\medskip

{\bf Proof of Theorem \ref{psum-L2a}.}\  We begin by recalling the classical Poisson summation formula, applied to
functions $f(y)\in\mathcal{S}(\R)$:
\begin{equation}
\sum_{n\in\Z}f(y+n)\ = \ 2\pi\ \sum_{n\in\Z} \widehat{f}(2\pi n)e^{2\pi i n y}\  .
\label{psum-classic}
\end{equation}
Here, we use the choice of Fourier transform given in \eqref{FT-def}.
Now let $\Gamma(x,X)$ be $C^\infty([0,1])$ and one-periodic with respect to $x$ with values in  Schwartz class,
$\mathcal{S}(\R_X)$. This subspace is dense in $\mathbb{H}^2$ with respect to the norm $\|\cdot\|_{\mathbb{H}^2}$.  
For such $\Gamma$ we can expand $x\mapsto\Gamma(x,X)$ in a rapidly convergent Fourier series:
\begin{equation}
\Gamma(x,X)=\sum_{m\in\Z}\Gamma_m(X)e^{2\pi imx}\ ,
\nn\end{equation}
where $\Gamma_m(X)=\int_0^1e^{-2\pi i m x}\Gamma(x,X)dx$ and  $\Gamma_m(X)\in \mathcal{S}(\R_X)$.
For each fixed $m$, we apply \eqref{psum-classic} to $\Gamma_m(\cdot)$. We have that the Fourier transform of
$\Gamma_m(X)$ is $\widehat{\Gamma}_m(\xi)$ and therefore:
\begin{equation}
\sum_{n\in\Z}\Gamma_m(X+n)\ = \ 2\pi\ \sum_{n\in\Z} \widehat{\Gamma_m}(2\pi n)e^{2\pi i n X} \ .
\label{psum-Gmam}
\end{equation}
Next, multiplying both sides of \eqref{psum-Gmam} by $e^{2\pi i mx}$ and  summing over $m\in\Z$
 (the rapid convergence in $n$ and $m$ ensures that we may interchange summations) yields:
 \begin{equation}
 \sum_{n\in\Z}\Gamma(x,X+n)\ =\ 2\pi \sum_{n\in\Z}e^{2\pi i nX}\widehat{\Gamma}(x,2\pi n)\ .
 \label{interim}\end{equation}
 Next, let's apply \eqref{interim} where $\Gamma(x,X)$ is replaced by $e^{-i\zeta X} \Gamma(x,X)$
 and whose Fourier transform is $\widehat{\Gamma}\left(x,\xi+\zeta\right)$.
 We obtain for every $(x,X)\in\R\times\R$:
 \begin{equation}
 \sum_{n\in\Z}e^{-i\zeta(X+n)}\Gamma(x,X+n)\ =\ 
 2\pi \sum_{n\in\Z}e^{2\pi i nX}\widehat{\Gamma}\left(x,2\pi n+\zeta\right)\ .
 \label{theresult}\end{equation}
 Setting $X=x$ in \eqref{theresult} we obtain that for $\Gamma(x,X)$ in a dense subspace of $\mathbb{H}^2$ we have:
 \begin{equation}
 \sum_{n\in\Z}e^{-i\zeta(x+n)}\Gamma(x,x+n)\ =\ 2\pi\ \sum_{n\in\Z}e^{2\pi i nx}\widehat{\Gamma}\left(x,2\pi
n+\zeta\right)\ ,\ \ x\in\R.
 \label{theresult1}\end{equation}

 We now seek to extend \eqref{theresult1} to all $\mathbb{H}^2$ in the sense of  \eqref{psum-GofxX}.
It will suffice to prove the following two claims:

\nit{\bf Claim 1:} Assume $\Gamma\in\mathbb{H}^2$ and fix $\zeta_{\rm max}>0$.  Then, the  linear mappings defined by
the left and right hand sides
of \eqref{psum-GofxX}:
\begin{align}
\Gamma &\mapsto L[\Gamma](x,\zeta)\equiv \lim_{N\to\infty}\sum_{n=-N}^{N} e^{- i\zeta(x+n)}\Gamma(x,x+n)\
\label{lhs-psum}\\
\Gamma &\mapsto R[\Gamma](x,\zeta)\equiv \lim_{N\to\infty}\ 2\pi\ \sum_{n=-N}^{N} e^{2\pi i n
x}\widehat{\Gamma}\left(x,2\pi n+\zeta\right)\label{rhs-psum}
\end{align}
are well-defined as limits in $L^2([0,1]\times[-\zeta_{\rm max},\zeta_{\rm max}];dxd\zeta)$.
\medskip

\nit{\bf Claim 2:} $\Gamma \mapsto L[\Gamma](x,\zeta)$ and $\Gamma \mapsto R[\Gamma](x,\zeta)$ 
are bounded linear transformations from $\mathbb{H}^2$ to $L^2([0,1]\times[-\zeta_{\rm max},\zeta_{\rm max}];dxd\zeta)$.
That is, there exist
constants $C_L, C_R>0$, depending on $\zeta_{\rm max}$,   such that for all 
$\Gamma\in\mathbb{H}^2$ 
\begin{align}
\int\int_{[0,1]\times[-\zeta_{\rm max},\zeta_{\rm max}];}\left|L[\Gamma](x,\zeta)\right|^2dxd\zeta\le
&C_L\|\Gamma\|_{\mathbb{H}^2}^2\ ,
\label{Lbound}\\
\int\int_{[0,1]\times[-\zeta_{\rm max},\zeta_{\rm max}];}\left|R[\Gamma](x,\zeta)\right|^2dxd\zeta\le
&C_R\|\Gamma\|_{\mathbb{H}^2}^2\ .\label{Rbound}
\end{align}

 We first show that Claim 1 and Claim 2  imply Theorem \ref{psum-L2a} and then give the proofs of these claims.
  Fix a $\Gamma\in \mathbb{H}^2$. Then, there exists a sequence $\Gamma_l(x;X),\ l\ge1$ with  
  $\Gamma\in C^\infty(S^1;\mathcal{S}(\R))$, such that $\|\Gamma_l-\Gamma\|_{\mathbb{H}^2}\to0$ as $l\to\infty$. In the
norm of 
$L^2([0,1]\times[-\zeta_{\rm max},\zeta_{\rm max}];dxd\zeta)$ we have
\begin{align}
&\Big\|L[\Gamma]-R[\Gamma]\Big\| \nn\\ 
&\le\ \  \Big\|L[\Gamma]-L[\Gamma_l]\Big\|+\Big\|L[\Gamma_l]-R[\Gamma_l]\Big\|
+\Big\|R[\Gamma_l]-R[\Gamma]\Big\|
\label{LminusR}
\end{align}
The middle term in \eqref{LminusR} is identically zero since we have proved \eqref{theresult1} for $\Gamma\in
C^\infty(S^1;\mathcal{S}(\R))$. 
Also, by linearity and the bounds \eqref{Lbound} and \ref{Rbound} we have, as $l\to\infty$, 
\[\Big\|L[\Gamma]-R[\Gamma]\Big\|_{L^2([0,1]\times[-\zeta_{\rm max},\zeta_{\rm max}];; dxd\zeta)} \le
(C_L+C_R)\|\Gamma_l-\Gamma\|_{\mathbb{H}^2}\to0
.\]
 Thus, $L[\Gamma](x,\zeta)=R[\Gamma](x,\zeta)$ in $L^2([0,1]\times[-\zeta_{\rm max},\zeta_{\rm max}];; dxd\zeta)$, as
asserted in Theorem
\ref{psum-L2a}. \medskip

\nit{\bf Proof of Claim 1:}\  Pick $K\in\Z_+$ such that $2\pi K>\zeta_{\rm max}$. We first prove that
$L[\Gamma](x,\zeta)$, the limit  in \eqref{lhs-psum} exists. For any
positive integer $N$ define:
\[L_{_{N}}[\Gamma](x,\zeta)=\sum_{n=-N}^{N} e^{- i\zeta(x+n)}\Gamma(x,x+n).\]
Let $N_1, N_2$ be positive integers with $N_1<N_2$. Then
\begin{align*}
&\int\int_{[0,1]\times[-\zeta_{\rm max},\zeta_{\rm max}]}
\left|L_{_{N_1}}[\Gamma](x,\zeta)-L_{_{N_2}}[\Gamma](x,\zeta)\right|^2dxd\zeta\nn\\
&\quad = \int_0^1 dx\int_{-\zeta_{\rm max}}^{\zeta_{\rm max}}d\zeta \Big| \sum_{N_1< |n|\le N_2} e^{-
i\zeta(x+n)}\Gamma(x,x+n) \Big|^2\nn\\
&\quad \le \int_0^1 dx\int_{-2\pi K}^{2\pi K}d\zeta \Big| \sum_{N_1 <|n|\le N_2} e^{- i\zeta n}\Gamma(x,x+n)
\Big|^2\nn\\
&\quad = 4\pi K\  \int_0^1 dx  \sum_{N_1 < |n|\le N_2} \Big|\Gamma(x,x+n) \Big|^2
 \end{align*}
 The last equality follows from Parseval's relation, since \\
 $\zeta\mapsto\sum_{N_1<|n|\le N_2} e^{- i\zeta
n}\Gamma(x,x+n)$ is the Fourier series of a $2\pi$ periodic function.
Note now that 
\begin{equation}
\Big|\Gamma(x,x+n)\Big|^2\le \sup_{0\le y\le 1} \Big|\Gamma(y,x+n)\Big|^2\lesssim 
 \sum_{j=0}^2 \int_0^1\Big|\D^j_y\Gamma(y,x+n))|^2 dy.\label{traceGamma}\end{equation}
 Therefore,
 \begin{align}
&\int\int_{[0,1]\times[-\zeta_{\rm max},\zeta_{\rm max}]}
\left|L_{_{N_1}}[\Gamma](x,\zeta)-L_{_{N_2}}[\Gamma](x,\zeta)\right|^2dxd\zeta\nn\\
&\quad \lesssim 
4\pi K\ \sum_{j=0}^2 \int_0^1dy \int_0^1 dx  \sum_{N_1<|n|\le N_2} \Big|\D^j_y\Gamma(y,x+n))|^2 dy\nn\\
&\quad \lesssim 4\pi K\  
\sum_{j=0}^2 \int_0^1dy \int_{[-N_2,N_2+1]\setminus[-N_1,N_1+1]}  \Big|\D^j_y\Gamma(y,X)|^2 dX\ \to0
\nn\end{align}
as $ N_2>N_1\to\infty$ by the assumptions on $\Gamma(x,X)$: \eqref{Gamma-conditions1a}-\eqref{Gamma-conditions2a}. It
follows that $\{ L_{_{N}} \}_{N\in\mathbb{N}}$ is a Cauchy sequence in $L^2([0,1]\times[-\zeta_{\rm max},\zeta_{\rm
max}];dxd\zeta)$. We denote this
limit by $L[\Gamma](x,\zeta)$ and also write:
\begin{equation}
\label{Ldef}
  L[\Gamma](x,\zeta) = \sum_{n\in\Z} e^{- i\zeta(x+n)}\Gamma(x,x+n)\ .
\end{equation}

 Now we turn to the proof that the limit $R[\Gamma](x,\zeta)$, the limit  in \eqref{rhs-psum} exists. For a positive
integer $N$  define: 
 \begin{equation}
 R_{_{N}}[\Gamma](x,\zeta)=
2\pi\ \sum_{n=-N}^{N} e^{2\pi i n x}\widehat{\Gamma}\left(x,2\pi n+\zeta\right).\label{RN-def}
\end{equation} 
For  $N_2>N_1$, we have
\begin{align}
&\int\int_{[0,1]\times[-\zeta_{\rm max},\zeta_{\rm max}]}
\left|R_{_{N_1}}[\Gamma](x,\zeta)-R_{_{N_2}}[\Gamma](x,\zeta)\right|^2dxd\zeta\nn\\
&\quad = (2\pi)^2\ \int_{-\zeta_{\rm max}}^{\zeta_{\rm max}} d\zeta\int_0^1dx \Big|  
\sum_{N_1<|n|\le N_2}  e^{2\pi i n x}\widehat{\Gamma}\left(x,2\pi n+\zeta\right)\Big|^2 \ .
\label{Rint}
\end{align}

We view the $dx-$ integral in  \eqref{Rint} as an expression of the general form:
\begin{equation}
\int_0^1 |\Delta_I(x,\zeta)|^2 dx,\ \textrm{where}\ \  
\Delta_I(x,\zeta)\ \equiv\ 
\sum_{n\in I}  e^{2\pi i n x}\widehat{\Gamma}\left(x,2\pi n+\zeta\right) \ ,
\label{RintI}
\end{equation}
and $I$ denotes a finite subset of $\Z$. 

\begin{proposition}\label{psum-I}
For $\zeta\in[-\zeta_{\rm max},\zeta_{\rm max}]$ a.e.,
\begin{equation}
\int_0^1 \left|\Delta_I[\Gamma](x,\zeta)\right|^2\ dx\lesssim\ \sum_{j=0}^2 \sum_{m\in I}\int_0^1
\Big|\D_x^j\widehat{\Gamma}\left(x,2\pi m+\zeta\right)\Big|^2 dx .\label{almost1I}
\end{equation}
\end{proposition}
\medskip

\noindent{\it Proof of Proposition \ref{psum-I}:}\  Note first that
 
 \begin{align}
 \int_0^1dx \left|\Delta_I[\Gamma](x,\zeta)\right|^2\ &=\  \int_0^1dx\
\Delta_I[\Gamma](x,\zeta)\overline{\Delta_I[\Gamma](x,\zeta)}\nn\\
 & =\  \int_0^1dx \sum_{m,m'\in I} e^{2\pi i (m-m')x}\ \widehat{\Gamma}\left(x,2\pi m +\zeta\right)\
\overline{\widehat{\Gamma}\left(x,2\pi m' +\zeta\right)}\ .
 \end{align}

 Consider the contribution from oscillatory-in-$x$ terms, {\it i.e.} the sum over all $m\in I$ and $m'\in I$ for which 
$m\ne m'$. Integration by parts yields:
 \begin{align*}
& \sum_{m\ne m'} \int_0^1\ e^{2\pi i(m-m')x}\ \widehat{\Gamma}\left(x,2\pi m +\zeta\right)\
\overline{\widehat{\Gamma}\left(x,2\pi m' +\zeta\right)}\ dx\\
 &\ =\sum_{m\ne m'}\left(\frac{1}{2\pi i (m-m')}\right)^2\nn\\
 &\qquad\qquad\times \int_0^1e^{2\pi i(m-m')x} 
 \frac{\D^2}{\D x^2}\left[\widehat{\Gamma}\left(x,2\pi m +\zeta\right)
  \overline{\widehat{\Gamma}\left(x,2\pi m' +\zeta\right)}\right] dx\\
  &=\sum_{j+j'=2}\sum_{m\ne m'}\left(\frac{1}{2\pi i (m-m')}\right)^2\nn\\
 &\qquad\qquad\times \int_0^1e^{2\pi i(m-m')x} 
{\rm coeff}_{_{jj'}}\D_x^j\widehat{\Gamma}\left(x,2\pi m +\zeta\right)
  \overline{\D_x^{j'}\widehat{\Gamma}\left(x,2\pi m' +\zeta\right)} dx,\\
  &\qquad\qquad\qquad\qquad (0\le j,j'\le2)\\
  &\lesssim \sum_{m\ne m'}\left(\frac{1}{2\pi (m-m')}\right)^2\nn\\
 &\ \ \ \times 
 \sum_{j=0}^2\left[\int_0^1
\left| \D_x^j\widehat{\Gamma}\left(x,2\pi m +\zeta\right)\right|^2 dx
\right]^{\frac12}\ 
\sum_{j'=0}^2\left[\int_0^1\left|\D_x^{j'}\widehat{\Gamma}\left(x,2\pi m' +\zeta\right)\right|^2 dx\right]^{\frac12} \ .
 \end{align*}
 The contribution from $m,m'\in I$ with $m=m'$ is:
 \[ \sum_{m\in I}
 \int_0^1\  \left|\widehat{\Gamma}\left(x,2\pi m +\zeta\right)\right|^2 \ dx\ .\]
 It follows that
\begin{align}
  \int_0^1 \left|\Delta_I[\Gamma](x,\zeta)\right|^2 dx &\lesssim\ 
  \sum_{m'\in I}\sum_{m\in I}\frac{1}{1+(m'-m)^2} z_m(\zeta)\ z_{m'}(\zeta) ,\nn\\
 &= \sum_{m'\in\Z}{\bf 1}_I(m')\sum_{m\in\Z}\frac{1}{1+(m'-m)^2}\ {\bf 1}_{I}(m)z_m(\zeta) {\bf
1}_I(m')z_{m'}(\zeta) ,\label{zz-sum}
 \end{align}
  where ${\bf 1}_I(m)=1$ for $m\in I$ and zero otherwise, and 
\begin{equation}  z_m(\zeta) \equiv \sum_{j=0}^2\left[\int_0^1
\left| \D_x^j\widehat{\Gamma}\left(x,2\pi m +\zeta\right)\right|^2 dx
\right]^{\frac12}\ .
\label{zm-def}
\end{equation}
To bound the sum in \eqref{zz-sum}, first apply the Cauchy-Schwarz inequality to the sum over $m'$:
\begin{align}
&\int_0^1 \left|\Delta_I[\Gamma](x,\zeta)\right|^2 dx\nn\\
&\qquad \lesssim\ \left( \sum_{m'\in\Z}{\bf 1}_I(m')\left|\sum_{m\in \Z}\frac{1}{1+(m'-m)^2}
 {\bf 1}_I(m) z_m(\zeta)\right|^2\right)^{\frac12} \times \\
 &\qquad\qquad \left(\sum_{m'\in \Z} |{\bf
1}_I(m')z_{m'}(\zeta)|^2\right)^{\frac12}.\nn
\end{align}
Next, by Young's inequality ($K\in l^1(\Z)\ \implies \|K\star z\|_{l^2(\Z)}\le \|K\|_{l^1(\Z)}\ \|z\|_{l^2(\Z)}$)
\begin{equation}
\int_0^1 \left|\Delta_I[\Gamma](x,\zeta)\right|^2\ dx\lesssim\ \sum_{m'\in I} |z_{m'}(\zeta)|^2 .\label{almost}
\end{equation}
Therefore by \eqref{almost} and \eqref{zm-def} we have
\begin{equation}
\int_0^1 \left|\Delta_I[\Gamma](x,\zeta)\right|^2\ dx\lesssim\ \sum_{j=0}^2 \sum_{m\in I}\int_0^1
\Big|\D_x^j\widehat{\Gamma}\left(x,2\pi m+\zeta\right)\Big|^2 dx .
\end{equation}
This completes the proof of Proposition \ref{psum-I}.
\medskip

We next apply Proposition \ref{psum-I} for the choice of $I=\{n\in\Z: N_1< |n|\le N_2\}$. 
For this choice of $I$, we integrate \eqref{almost1I} with respect to $\zeta$ and obtain:
\begin{align}
 &\int_{-\zeta_{\rm max}}^{\zeta_{\rm max}}d\zeta \int_0^1 dx \left|\Delta_I[\Gamma](x,\zeta)\right|^2 \\
 &\qquad \lesssim\
\sum_{j=0}^2 \int_0^1\ dx \sum_{N_1<|m|\le
N_2} \int_{-\zeta_{\rm max}}^{\zeta_{\rm max}} d\zeta  \Big|\D_x^j\widehat{\Gamma}\left(x,2\pi m+\zeta\right)\Big|^2 \nn
\\
&\qquad \le \sum_{j=0}^2\int_0^1 dx \int_{-\infty}^\infty 
\Big[\sum_{N_1<|m|\le N_2} {\bf 1}_{_{[2\pi m-\zeta_{\rm max},2\pi m+\zeta_{\rm max} ]}}(\xi)
|\D_x^j\widehat{\Gamma}\left(x,\xi\right)|^2\Big]d\xi 
\label{almost2}
\end{align}
which tends to zero as $N_2>N_1\to\infty$ by \eqref{Gamma-conditions2a} and the Plancherel identity.
Since, $\Delta_I(x,\zeta)\ =\ R_{N_2}[\Gamma](x,\zeta)-R_{N_1}[\Gamma](x,\zeta)$, the sequence
$\{R_{N}[\Gamma]\}_{N\ge1}$ is Cauchy in $L^2([0,1]\times[-\zeta_{\rm max},\zeta_{\rm max}])$ and has a limit which we
denote by
\begin{equation}
R[\Gamma](x,\zeta) = 2\pi\ \sum_{n\in\Z} e^{2\pi i n x}\widehat{\Gamma}\left(x,2\pi n+\zeta\right) .\label{Rdef}
\end{equation}
This completes the proof of Claim 1.
\medskip

\nit{\bf Proof of Claim 2:}\  Pick $K\in\Z_+$ such that $2\pi K>\zeta_{\rm max}$. 
 From the definition of $L[\Gamma](x,\zeta)$ \eqref{Ldef}, the fact that
 $\zeta\mapsto \sum_{n\in\Z} e^{- i\zeta n}\Gamma(x,x+n)$ is the Fourier series of a $2\pi-$ periodic function
 and the Parseval relation, we have:
 \begin{align*}
 \int_{-\zeta_{\rm max}}^{\zeta_{\rm max}}d\zeta | L_N[\Gamma](x,\zeta) |^2 &= 
 \int_{-\zeta_{\rm max}}^{\zeta_{\rm max}}d\zeta \Big| e^{- i\zeta x}\sum_{|n|\le N} e^{- i\zeta n}\Gamma(x,x+n)
\Big|^2\nn\\
  & \le  \int_{-2\pi K}^{2\pi K}d\zeta \Big| \sum_{|n|\le N} e^{- i\zeta n}\Gamma(x,x+n) \Big|^2 \\
  & \le 4\pi K\
\sum_{n\in\Z}
\Big|\Gamma(x,x+n) \Big|^2.
 \end{align*}
  Therefore, by \eqref{traceGamma} 
\begin{align} \int_{-\zeta_{\rm max}}^{\zeta_{\rm max}}d\zeta | L_N[\Gamma](x,\zeta) |^2\
  &\lesssim 4\pi K\sum_{j=0}^2\sum_{n\in\Z}\int_0^1 \left|\D_y^j\Gamma(y,x+n)\right|^2 dy.
    \nn\end{align}
    Integration over $x$ implies the bound \eqref{Lbound}. Namely,
    \begin{align} \int_0^1dx \int_{-\zeta_{\rm max}}^{\zeta_{\rm max}}d\zeta  | L_N[\Gamma](x,\zeta) |^2\
  &\lesssim 4\pi K\ \sum_{j=0}^2\sum_{n\in\Z}\int_0^1 \int_0^1dx\left|\D_y^j\Gamma(y,x+n)\right|^2 dy\nn\\
  &= 4\pi K\ \sum_{j=0}^2\sum_{n\in\Z}\int_n^{n+1}\left|\D_y^j\Gamma(y,X)\right|^2dX dy\nn\\
  &=4\pi K\  \sum_{j=0}^2\int_\R\left|\D_y^j\Gamma(y,X)\right|^2dX dy\ =\ 4\pi K\  \|\Gamma\|_{\mathbb{H}^2}^2.
    \nn\end{align}
    Since $L_N[\Gamma]\to L[\Gamma]$ in $L^2([0,1]\times[-\zeta_{\rm max},\zeta_{\rm max}])$ it follows that 
   $L:\Gamma\mapsto L[\Gamma]$ is bounded from 
    $\mathbb{H}^2$ to $L^2([0,1]\times[-\zeta_{\rm max},\zeta_{\rm max}]);dxd\zeta)$.
   \medskip
   
   Finally, we turn to the bound on $\Gamma\mapsto R[\Gamma]$. By \eqref{Rdef}, \eqref{RintI} and Proposition
\ref{psum-I}, with $I=\{-N\le |m|\le N\}$, we have
   \begin{align}
   \int_0^1 |R_N[\Gamma](x,\zeta)|^2 dx\ &=\ \int_0^1 |\Delta_I(x,\zeta)|^2 dx\nn\\
   &\lesssim \sum_{j=0}^2\sum_{|m|\le N} \int_0^1dx |\D_x^j\widehat{\Gamma}(x,2\pi m+\zeta)|^2\ .
   \nn\end{align}
Integration over $[-\zeta_{\rm max},\zeta_{\rm max}]$ with respect to $\zeta$ and applying the Plancherel identity gives
for every $N\in\Z_+$:
   \begin{align*}
&\int_{-\zeta_{\rm max}}^{\zeta_{\rm max}}d\zeta \int_0^1 dx \left|R_N[\Gamma](x,\zeta)\right|^2 \\
&\qquad\lesssim \sum_{j=0}^2\ \int_0^1 dx \int_\R {\bf 1}_{_{[2\pi m-\zeta_{\rm max},2\pi m+\zeta_{\rm max}]}}(\xi)\
\Big|\D_x^j\widehat{\Gamma}\left(x,\xi\right)\Big|^2 d\xi \\
&\qquad\lesssim C(K)\ \sum_{j=0}^2\int_0^1 dx  \int_\R  \Big|\D_x^j \Gamma\left(x,X\right)\Big|^2 dX 
= C(K)\ \|\Gamma\|_{\mathbb{H}^2}^2.
\end{align*}
Since $R_N[\Gamma]\to R[\Gamma]$ in $L^2([0,1]\times[-\zeta_{\rm max},\zeta_{\rm max}]);dxd\zeta)$ we obtain the desired
bound \eqref{Rbound}. 
This completes the proof of   Theorem \ref{psum-L2a} (Theorem \ref{psum-L2}), the Poisson summation formula in $L^2_{\rm
loc}$.

\chapter{1D Dirac points and Floquet-Bloch Eigenfunctions}
\label{linear-band-crossing}

\section{Conditions ensuring a 1D Dirac point; proof of Theorem \ref{conditions-for-dirac}}
\label{conds4dirac-appendix}

Let $V_\ee\in L^2_\ee$. We consider the Floquet-Bloch eigenvalue problem \eqref{FB-evp}
\begin{align}
H_{V_\ee}(k) p\ &=\ E p,\quad  p(x+1;k)=  p(x;k),\ \ {\rm where}\label{fb-p}\\
H_{V_\ee}(k)\ &\equiv\ -(\D_x+ik)^2+V_\ee(x) \ , \label{HVofk-def}
\end{align}
for $k$ near $k_\star=\pi$. We write
 \[ k=k_{\star}+{k'}\ ,\]
 where $ \abs{{k'}}$  will be taken to be small and rewrite the Floquet-Bloch eigenvalue problem as\begin{align}
 \label{dis_rel_prob}
& \left(-(\partial_x+i(k_\star+k'))^2+  V_{\ee}(x)\right)p(x;k_{\star}+{k'})\ =\ E(k_\star+k')\ p(x;k_\star+k'),\\
& p(x+1;k_{\star}+{k'}) = p(x;k_{\star}+{k'}).\nn
\end{align} 
\medskip

Recall that $E_\star$ is a double eigenvalue with corresponding eigenspace spanned by $\{\Phi_1,\Phi_2\}$. We study
\eqref{dis_rel_prob} for $E$ near $E_\star$ and  $k'$ small as a problem in the perturbation theory of a degenerate
eigenvalue. Let 
\begin{equation}
p_j(x)=e^{-ik_\star x}\Phi_j(x)
\label{pj-def1}\end{equation}
and seek a solution  of the form:
\begin{align}
E(k_\star+k')&=E_\star+E^{(1)},\label{E-ansatz}\\
 p(x;k_{\star}+{k'})&= p^{(0)} +\ p^{(1)}, \qquad p^{(0)}\equiv\alpha p_1+\beta p_2\ , \label{p-ansatz}\end{align}
where
\begin{equation}
\label{p1-orthog}
\left\langle p_j,p^{(1)}\right\rangle_{L^2[0,1]}=0,\quad j=1,2 .
\end{equation}

Substituting \eqref{E-ansatz}-\eqref{p-ansatz} into \eqref{dis_rel_prob} we obtain the inhomogeneous problem for $
p^{(1)}\in L^2[0,1]$:
\begin{align}
\left(H_{V_\ee}(k_\star)-E_\star\right)p^{(1)} &= \left(2ik'(\D_x+ik_\star)-(k')^2+E^{(1)}\right)p^{(1)}\nn\\
&\qquad +\left(2ik'(\D_x+ik_\star)-(k')^2+E^{(1)}\right)p^{(0)} \nn\\
&\equiv J(\alpha,\beta,{k'},E^{(1)},p^{(1)})\ .
\label{p1-eqn}\end{align}

Introduce the orthogonal projections $P_{\parallel}$ and $P_{\perp}$ defined by:
\begin{align*}
 P_{\parallel}f(x)&=\inner{p_1,f}_{L^2([0,1])}p_1(x)+\inner{p_2,f}_{L^2([0,1])}
p_2(x),  \\ 
 P_{\perp}f(x) &= (I-P_{\parallel})f(x).
\end{align*} 
Equation \eqref{p1-eqn} may be rewritten as  the following system for the unknowns $p^{(1)}=p^{(1)}(x;{k'})$ and
$E^{(1)}=E^{(1)}({k'})$:
\begin{align}
\left(H_{V_\ee}(k_\star)-E_{\star}\right) p^{(1)} &= P_{\perp}J(\alpha,\beta,{k'},E^{(1)},p^{(1)}),\label{p1-perp}\\
0 &= P_{\parallel}J(\alpha,\beta,{k'},E^{(1)},p^{(1)}). \label{p1-para}
\end{align} 
\medskip

In detail, system \eqref{p1-perp}-\eqref{p1-para} reads
\begin{align}
&\left(H_{V_\ee}(k_\star)-E_{\star}\right)p^{(1)} =
P_{\perp}\left(2i{k'}(\D_x+ik_\star)-(k')^2+E^{(1)}\right)p^{(1)}\nn\\
&\qquad\qquad\qquad \qquad+ P_\perp\ \left(2ik'(\D_x+ik_\star) \right)p^{(0)},\label{p1-perp1}\\
 &P_{\parallel}\left(2i{k'}(\partial_x+ik_\star) -{k'}^2 +E^{(1)} )\right)p^{(0)}
+P_{\parallel}\left(2i{k'}(\partial_x+ik_\star\right)p^{(1)}\ = 0 \label{p1-para1}.
\end{align} 

Introduce the resolvent
\begin{equation*}
 R_{k_\star}(E_{\star}) = \left(H_{V_\ee}(k_\star) - E_{\star}\right)^{-1},
\end{equation*} defined as a bounded map from $P_{\perp}L^2[0,1]$ to $P_{\perp}H^2[0,1]$.
 Equation \eqref{p1-perp1} for $p^{(1)}$ can be rewritten as
 \begin{equation}
 \left( I + A \right) p^{(1)}\ =\ R_{k_\star}(E_\star)\ P_\perp\ \left(2ik'(\D_x+ik_\star \right)p^{(0)},
 \label{p1-perp2}
 \end{equation}
 where
 \begin{equation}\label{A-def}
 f\mapsto Af\ \equiv\ R_{k_\star}(E_\star)\ P_{\perp}\left(-2i{k'}(\D_x+ik_\star)+(k')^2-E^{(1)}\right)\ f
 \end{equation}
 is a bounded operator on $H^2_{\rm per}[0,1]$ for any $s$. Furthermore, for $|k'|+|E^{(1)}|$ sufficiently small, the
operator norm of $A$ is less than one, $(I+A)^{-1}$ exists, and hence \eqref{p1-perp2} is uniquely solvable in $P_\perp
H^2_{\rm per}[0,1]$:
 \begin{align}
 p^{(1)}\ &=\ \left( I\ +\ R_{k_\star}(E_\star)\ P_{\perp}\left(-2i{k'}(\D_x+ik_\star)+(k')^2-E^{(1)}\right)\
\right)^{-1}\nn\\
 &\qquad \circ\ R_{k_\star}(E_\star)\ P_\perp\ \left(2ik'(\D_x+ik_\star) \right)p^{(0)}.
 \label{p1-perp3}
 \end{align}
 
Recall that $p^{(0)}=\alpha p_1(x)+\beta p_2(x)$ (equation \eqref{p-ansatz}) and therefore that $p^{(1)}$ is linear
in $\alpha$ and $\beta$. We may therefore write
\begin{equation}
\label{p1-perp4}
 p^{(1)} = p^{(1)}(x;k',E^{(1)}) = k'\ g^{(1)}[{k'},E^{(1)}](x)\alpha +  k'\ g^{(2)}[{k'},E^{(1)}](x)\beta,
\end{equation} 
where $(k',E^{(1)})\mapsto g^{(j)}(k',E^{(1)})$ is a smooth mapping from a neighborhood of
$(0,0)\in\R\times\mathbb{C}$ into $H^2_{\rm per}([0,1])$, which satisfies the bound
\begin{equation*}
 \norm{g^{(j)}(k',E^{(1)})}_{H^2([0,1])} \lesssim\ 1+|k'| + |E^{(1)}|\ ,~~~j=1,2.
\end{equation*}
Note also that 
\begin{equation}
P_\parallel g^{(j)}(k',E^{(1)})=0,\ j=1,2.
\label{Ppar}
\end{equation}

Thus we can substitute \eqref{p1-perp4} into \eqref{p1-para1} and obtain a system of two homogeneous linear equations
for $\alpha$ and $\beta$. To express this system in a compact form we note the following relations:
\begin{align}
\left(\D_x+ik_\star\right)p_j &= e^{-ik_\star x}\ \D_x\ e^{ik_\star x}p_j= e^{-ik_\star x}\D_x\Phi_j\nn,\\
\left\langle p_i,p_j\right\rangle_{L^2[0,1]} &=\left\langle \Phi_i,\Phi_j\right\rangle_{L^2[0,1]}=\delta_{ij},\ j=1,2,
\end{align}
and define
\begin{equation}
G^{(j)}[k',E^{(1)}](x)= e^{ik_\star x}\ g^{(j)}[k',E^{(1)}](x),\ \textrm{ and note}\  \left\langle
\Phi_i,G^{(j)}\right\rangle=0,\ i,j=1,2
\nn\end{equation}
by \eqref{Ppar}. Furthermore, since 
\begin{align}
&\D_x: H^s_{k_\star,\sigma}\to H^{s-1}_{k_\star,\sigma},\ \ \sigma=\ee,\oo\nn,\\
&R_{k_\star}(E_\star): P_\perp H^s_{k_\star,\sigma}\to P_\perp H^{s+2}_{k_\star,\sigma},\ \ \sigma=\ee,\oo,
\label{relations}\end{align}
we have that 
\[ G^{(1)}[k',E^{(1)}]\in H^2_\ee,\ \ \text{and } \ \ G^{(2)}[k',E^{(1)}]\in H^2_\oo\ .\]

Equation \eqref{p1-para1} for $(\alpha, \beta)$ can now be written in the form
\begin{equation}
\mathcal{M}(E^{(1)},{k'})
\begin{pmatrix}
 \alpha \\ \beta
\end{pmatrix} =0,
\label{Mab0}\end{equation}
 with (inner products in \eqref{M-def} are over $L^2([0,1])$)
\begin{align}
\mathcal{M}(E^{(1)},{k'}) &\equiv
\begin{pmatrix}
 E^{(1)}+2i{k'}\inner{\Phi_1,\partial_x\Phi_1}-{k'}^2 &
2i{k'}\inner{\Phi_1,\partial_x\Phi_2} \\
2i{k'}\inner{\Phi_2,\partial_x\Phi_1}&
E^{(1)}+2i{k'}\inner{\Phi_2,\partial_x\Phi_2}-{k'}^2
\end{pmatrix} &\nn\\
&~~+ 
2i({k'})^2\begin{pmatrix}
\inner{\Phi_1,\partial_x G^{(1)}[k',E^{(1)}]} &
\inner{\Phi_1,\partial_x G^{(2)}[k',E^{(1)}]} \\
\inner{\Phi_2,\partial_x G^{(1)}[k',E^{(1)}]} &
\inner{\Phi_2,\partial_x G^{(2)}[k',E^{(1)}]}
\end{pmatrix},
\label{M-def}
\end{align} where $(k',E^{(1)}) \mapsto G^{(j)}[k',E^{(1)}], \D_xG^{(j)}[k',E^{(1)}]$  are  smooth functions of
$(k',E^{(1)})$ in a neighborhood  of $(0,0)$ and
\begin{align}
 &  \|G^{(j)}[k',E^{(1)}]\|_{L^2[0,1]}\ +\
\|\D_xG^{(j)}[k',E^{(1)}]\|_{L^2[0,1]}=\mathcal{O}\left(1+|k'|+|E^{(1)}|\right).
 \label{Gest}\end{align}

\medskip

Thus, $E=E_{\star}+E^{(1)}(k')$ is an eigenvalue for the spectral problem \eqref{dis_rel_prob} if and
only if $E^{(1)}=E^{(1)}(k')$ solves
\begin{equation}
 \text{det}\mathcal{M}\left(E^{(1)},{k'}\right)=0.
 \label{detM0}
\end{equation}

Symmetries can be used to simplify \eqref{M-def}. Indeed, by the relations \eqref{relations} and the fact
that
\begin{equation}
f\in L^2_{k_\star,\ee},\ g\in L^2_{k_\star,\oo}\ \ \implies\ \ \left\langle f,g\right\rangle=0,
\label{oo-ee-orthog}
\end{equation}
 the matrix $\mathcal{M}_{ij}\left(E^{(1)},{k'}\right)$, displayed in \eqref{M-def}, has vanishing off-diagonal entries:
\begin{equation}
\textrm{ $\mathcal{M}_{ij}\left(E^{(1)},{k'}\right)=0$ for $i, j=1,2$ and $i\ne j$}.
\label{diff-ij0}
\end{equation} 

Therefore, $ \text{det}\ \mathcal{M}\left(E^{(1)},{k'}\right)=0$, \eqref{detM0},  implies that either
\begin{equation}
E^{(1)}+2i{k'}\inner{\Phi_1,\partial_x\Phi_1}-{k'}^2+2i{k'}^2\inner{\Phi_1,\partial_x G^{(1)}[k',E^{(1)}]}
=0,\label{M110}
\end{equation}
or 
\begin{equation}
E^{(1)}+2i{k'}\inner{\Phi_2,\partial_x\Phi_2}-{k'}^2+2i{k'}^2\inner{\Phi_2,\partial_x G^{(2)}[k',E^{(1)}]}=0.
\label{M220}
\end{equation}

Furthermore, since $\Phi_1(x)= \mathcal{I}\left[\Phi_2\right](x)=\Phi_2(-x)$, it is easily seen that
\begin{equation}
\left\langle\Phi_1,\partial_x\Phi_1\right\rangle= -{\inner{\Phi_2,\partial_x\Phi_2}}.
\label{11m22}\end{equation} 

Thus, using \eqref{Gest}, we have the two branches of the dispersion locus:
\begin{align}
&\mathcal{F}_-(E^{(1)},k')\ \equiv\ E^{(1)}+\lambda_\sharp\ {k'}\ -{k'}^2+{k'}^2\gamma_1(k',E^{(1)}) =0\label{M110-A},\\
&\mathcal{F}_+(E^{(1)},k')\ \equiv\ E^{(1)}-\lambda_\sharp\ {k'}\ -{k'}^2+{k'}^2\gamma_2(k',E^{(1)})=0,
\label{M220-A}
\end{align}
where $\gamma_j(k',E^{(1)}),\ j=1,2$, are smooth and
\[ \gamma_j(k',E^{(1)})\equiv i\inner{\Phi_j,\partial_x G^{(j)}[k',E^{(1)}]}=\mathcal{O}\left(1+|k'|+|E^{(1)}|\right),\
j=1,2,\ \]
 for $|k'|+|E^{(1)}|$ small.

The constant (``Fermi velocity''), $\lambda_\sharp$, may be expressed in terms of the Fourier coefficients of
$\Phi_1\in L^2_{k_\star,\ee}$:
\begin{align} 
\lamsharp &\equiv 2i\inner{\Phi_1,\partial_x\Phi_1}_{L^2([0,1])}\label{lamsharp_defn}\\
&= 2i\int^{1}_{0}\sum_{m\in2\mathbb{Z}}\overline{c(m)}e^{-ik_{\star} x}e^{-2\pi imx} 
\sum_{p\in2\mathbb{Z}}i(2p\pi+k_{\star}){c}(p)e^{ik_{\star} x}e^{2\pi ipx}dx  \nonumber\\
&=-2\pi\left\{2\sum_{m\in2\mathbb{Z}}m\abs{c_1(m)}^2+1\right\} \ \ \ ({\rm Recall}\ k_\star=\pi).
\label{lamsharp_fourier_exp}
\end{align} 
Thus,  $\lambda_{\sharp}$ is real.

\nit {\bf Nondegeneracy Hypothesis:} $\lambda_\sharp$, given by \eqref{lamsharp_defn}, \eqref{lamsharp_fourier_exp}\ is
non-zero. In Lemma \ref{lam_nonzero_eps_small} below, it is shown that $\lambda_\sharp$ is always nonzero for $\eps$
sufficiently small.
\medskip

Since $\lambda_\sharp\ne0$, the implicit function theorem can be applied to show the existence of a positive constant,
$\zeta_0$, and smooth maps $k'\mapsto\eta_\pm(k')=\mathcal{O}(k')$, defined for $|k'|<\zeta_0$,   such that
\begin{align}
\mathcal{F}_-(E^{(1)}_-(k'),k') = 0\ \ {\rm and}\ \ \mathcal{F}_+(E^{(1)}_+(k'),k') &= 0,\nn
\end{align}
where
\begin{align}
E^{(1)}_-(k)\ &=\ -\lambda_\sharp\ k'\ \left(1+\eta_-(k')\right),\label{E+def}\\
E^{(1)}_+(k)\ &=\ +\lambda_\sharp\ k'\ \left(1+\eta_+(k')\right),\label{E-def}
\end{align}
are defined for $|k'|<\zeta_0$. Note that $\pm\lambda_{\sharp}$ is the linear order term in the expansion of
$E_{\pm}(k)$ around $k_{\star}$ - that is, $E'_{\pm}(k_{\star}) = \pm\lambda_{\sharp}$.

This completes the proof of Theorem \ref{conditions-for-dirac}.
\medskip

\section{Expansion of Floquet-Bloch modes near a 1D Dirac; proof of Proposition
\ref{flo-blo-dirac}\label{flo-blo-appendix}}

Corresponding to the Floquet-Bloch eigenvalue $E_-(k')$, given by \eqref{E-def}, we have the corresponding null-vector,
$(\alpha_-,\beta_-)=(1,0)$ of $\mathcal{M}(E_-(k'),k')$. Via \eqref{p-ansatz} this generates the Floquet-Bloch
eigenstate:
\begin{equation}
\Phi_-(x;k') = c_{-}(k')\left(\ \Phi_1(x) + \xi_-(x;k')\ \right)\ \in\ L^2_{k_\star,\ee},
\label{Phi-}
\end{equation}
where $c_{-}(k')=1+\mathcal{O}_\pm(k')$ is a normalization constant
and $k'\mapsto\xi_-(x;k'),\ c_-(k')$ are smooth for $x\in[0,1],\ |k'|<\zeta_0$.
 Similarly, the Floquet-Bloch eigenvalue $E_+(k')$, given by \eqref{E+def},
and corresponding null vector $(\alpha_+,\beta_+)=(0,1)$ of $\mathcal{M}(E_+(k'),k')$,
 generates, via \eqref{p-ansatz}, the Floquet-Bloch eigenstate:
\begin{align}
\Phi_+(x;k') &= c_+(k')\left(\ \Phi_2(x) + \xi_+(x;k')\ \right)\ \in\ L^2_{k_\star,\oo},
\label{Phi+}
\end{align}
where $c_{+}(k')=1+\mathcal{O}_\pm(k')$ is a normalization constant
and $k'\mapsto\xi_+(x;k'),\ c_{+}(k')$ are smooth for $x\in[0,1],\ |k'|<\zeta_0$.

\chapter{Dirac Points for Small Amplitude Potentials}\label{dirac-pt-small-eps}
Consider the Floquet-Bloch eigenvalue problem \eqref{FB-evp} 
\begin{equation}
\label{eps_nonzero_prob}
 \begin{split}
 &H^{(\eps)}\Phi = E\Phi, ~~~ \Phi(x+1) = e^{i k_{\star}}\Phi(x), ~~~ \text{where} \\
&H^{(\eps)} \equiv -\partial_x^2+\varepsilon V_{\ee}(x).
\end{split}
\end{equation} 

We first show that for $|\eps|$ small there exist quasi-momentum / energy pairs  
$(k_\star,E^{(\eps)}_{\star,n})$ with $(\pi, E^{(\eps)}_{\star,n}))|_{\eps=0}=(\pi,\pi^2(2n+1)^2)$, which  are Dirac
points in the sense of Definition \ref{dirac-pt-gen}. Fix
$n$ and to simplify notation set $E^{(\eps)}_{\star,n}\equiv E_{\star}$. By Theorem \ref{conditions-for-dirac} we
need to prove that, for $\eps$ sufficiently small, the following conditions hold.
\medskip
\begin{enumerate}[I.]
 \item $E_{\star}$ is a simple $L^2_{k_{\star},\ee}$- eigenvalue of $H^{(\eps)}$ with 
1-dimensional
eigenspace \[\textrm{span}\{\Phi_1(x)\}\subset L^2_{k_{\star},\ee}.\]
 \item $E_{\star}$ is a simple $L^2_{k_{\star},\oo}$- eigenvalue of $H^{(\eps)}$ with 
1-dimensional
eigenspace \[\textrm{span}\Big\{\Phi_2(x)=\mathcal{I}\left[\Phi_1\right](x)=\Phi_1(-x)\Big\}\subset
L^2_{k_{\star},\oo}.\]
 \item Non-degeneracy condition:
 \begin{equation}
  0\ne\lamsharp \equiv 2i\left\langle\Phi_1,\D_x\Phi_1\right\rangle\ = 
-2\pi\left\{2\sum_{m\in2\mathbb{Z}}m\abs{c_1(m)}^2+1\right\},\ \label{lambda-sharp-eps}
 \end{equation} where $\{c_1(m)\}_{m\in\Z}$ denote the $L^2_{k_\star,\ee}-$ Fourier coefficients of
$\Phi_1(x)$.
\end{enumerate}
\medskip
 Since $\eps V_\ee(x)\in L^2_\ee$, it follows from Proposition \ref{H-action-on-L_e} that property I implies property
II. Thus to prove Theorem \ref{thm:dirac-pt}, it is sufficient verify properties I and III:
\medskip
\begin{lemma}
 \label{property1-proof} (Property I)
 There exists an $\eps_0>0$ such that for all $\eps\in(-\eps_0,\eps_0)$, $E_{\star}$ is a simple $L^2_{k_{\star},\ee}$-
eigenvalue of $H^{(\eps)}$ with 1-dimensional eigenspace
 \[\textrm{span}\{\Phi_1(x)\}\subset L^2_{k_{\star},\ee}.\]
\end{lemma}

\begin{lemma}
 \label{lam_nonzero_eps_small} (Property III)
 There exists an $\eps_0>0$ such that for all $\eps\in(0,\eps_0)$, the ``Fermi velocity'', $\lambda_{\sharp}$, defined
in \eqref{lambda-sharp-eps} for $H^{(\eps)}$ is nonzero.
\end{lemma}

\medskip

\textit{Proof of Lemma \ref{property1-proof}}: Recall that $E^{(0)}\equiv E^{(0)}_{\star,n}=\pi^2(2n+1)^2$ is a simple 
$L^2_{k_\star,e}$ eigenvalue of \eqref{eps_nonzero_prob} for $\eps=0$, with eigenspace 
\begin{equation}
\label{even-espace}
{\rm span}\left\{ \Phi^{(0)}_n(x)=e^{i\pi x}\ e^{2\pi i n x}\right\}\subset L^2_{k_\star,\ee},
\quad \textrm{if $n$ is even},
\end{equation}
and eigenspace
\begin{equation}
\label{odd-espace}
{\rm span}\left\{ \Phi^{(0)}_n(-x)=e^{i\pi x}\ e^{2\pi i (-n-1) x}\right\}\subset L^2_{k_\star,\ee},
\quad \textrm{ if $n$ is odd.}
\end{equation}
Without loss of generality, we may assume that $n$ is even and denote the corresponding eigenfunction as
$\Phi^{(0)}(x)$.

We seek a solution to \eqref{eps_nonzero_prob} of the form
\begin{equation}
 \label{eps_exp}
 \begin{split}
\Phi^{(\eps)}(x) &= \Phi^{(0)}(x) + \varepsilon \Phi^{(1)}(x),\\
E^{(\eps)} &= E^{(0)} + \varepsilon E^{(1)},
\end{split}
\end{equation} where $\Phi^{(1)}$ and $E^{(1)}$ are $\eps-$ dependent corrections and
$|\eps|\ne0$ and small. Substituting \eqref{eps_exp} into \eqref{eps_nonzero_prob}, we obtain the inhomogeneous problem
for $\Phi^{(1)}\in L^2_{k_\star,\ee}$:
\begin{equation}
 \label{eps_nonzero_prob_expanded}
\left(-\partial_x^2 - E^{(0)}\right)\Phi^{(1)} + \left(V_{\ee} -
E^{(1)}\right)\left(\Phi^{(0)}+\varepsilon \Phi^{(1)}\right)=0.
\end{equation}

We proceed with a Lyapunov-Schmidt reduction argument. Define the orthogonal projection operators
$Q_{\parallel}$ and $Q_{\perp}$:
\begin{equation*}
 \label{epssmlbig10}
 Q_{\parallel}f(x)=\inner{\Phi^{(0)},f}_{L^2[0,1]}\Phi^{(0)}(x),~~~\text{and}~~~  Q_{\perp}f(x)
= (I-Q_{\parallel})f(x).
\end{equation*} 
Equation \eqref{eps_nonzero_prob_expanded} may be rewritten as the following system for the unknowns
$\Phi^{(1)}=\Phi^{(1)}(x)$ and $E^{(1)}$:
\begin{align}
 \label{eps_orthog_proj}
 &Q_{\parallel}\left(V_{\ee} - E^{(1)}\right)\left(\Phi^{(0)}+\varepsilon \Phi^{(1)}\right) = 0, \\
 \label{eps_projected}
 &\left(-\partial_x^2 - E^{(0)}\right)\Phi^{(1)} =  -Q_{\perp}\left(V_{\ee} - E^{(1)}\right)
\left(\Phi^{(0)}+\varepsilon \Phi^{(1)}\right).
\end{align} 

Next, we introduce the bounded resolvent operator
\begin{equation*}
 R(E^{(0)}) = \left(-\partial_x^2 - E^{(0)}\right)^{-1}: 
 Q_{\perp}L^2_{k_{\star},\ee} \to Q_{\perp}H^2_{k_{\star},\ee},
\end{equation*} 
and rewrite equation \eqref{eps_projected} as 
\begin{equation}
 \label{eps_projected_cpt}
 \left(I+\mathcal{A}(\varepsilon)\right)\Phi^{(1)}(x) = -R(E^{(0)})Q_{\perp} \left(V_{\ee}(x) - E^{(1)}\right)
\Phi^{(0)},
\end{equation} where
\begin{equation*}
 \mathcal{A}(\varepsilon) = \varepsilon R(E^{(0)})Q_{\perp} \left(V_{\ee} - E^{(1)}\right),
\end{equation*} is a bounded operator on $H^2_{k_\star,\ee}$. 
Furthermore, for $\eps$ sufficiently small, $\mathcal{A}(\varepsilon)$ has norm less than one. Therefore $(I+\mathcal{A}(\varepsilon))^{-1}$ exists and
\eqref{eps_projected} is uniquely solvable in $Q_\perp H^2_{k_\star,\ee}$:
\begin{equation}
\label{phi1_soln}
\Phi^{(1)} = -\left(I+\varepsilon R(E^{(0)})Q_{\perp} \left(V_{\ee}(x) - E^{(1)}\right)\right)^{-1} \circ 
R(E^{(0)})Q_{\perp} \left(V_{\ee}(x) - E^{(1)}\right)\Phi^{(0)}.
\end{equation}

Substituting \eqref{phi1_soln} into \eqref{eps_orthog_proj}, yields algebraic problem with $E^{(1)}$ as
the only unknown:
\begin{equation}
\label{G-def}
G(E^{(1)};\varepsilon) \equiv \inner{\Phi^{(0)},\left(V_{\ee} - E^{(1)}\right)
\left(\Phi^{(0)}+\varepsilon\Phi^{(1)}\right)}_{L^2[0,1]}=0.
\end{equation} We solve $G(E^{(1)};\varepsilon)=0$ for $E^{(1)}$ using the implicit function theorem.

For $\varepsilon=0$, equation \eqref{G-def}, the Fourier description of the potential $V_\ee$ in \eqref{Vxs} and the definition
of $\Phi^{(0)}$ in \eqref{even-espace} imply
\begin{equation}
 \begin{split}
 \label{epssmlbig19}
 E^{(1)} &= \inner{\Phi^{(0)},V_{\ee}\Phi^{(0)}}_{L^2[0,1]}=v_{0}.
 \end{split}
\end{equation} Moreover, $G(E^{(1)};\varepsilon)$ is analytic in a neighborhood  of $(E^{(1)},\varepsilon)=(v_{0},0)$,
and it is straightforward to check directly that $\partial_{E^{(1)}} G(v_{0},0)=-1$. Therefore,  by the implicit
function
theorem, there exists a $\varepsilon_0>0$ such that for $\varepsilon\in(-\eps_0,\varepsilon_0)$,
there is a continuous function $\varepsilon\mapsto E^{(1)}(\varepsilon)$, satisfying $E^{(1)}(0)=v_0$ and 
\begin{equation*}
 \label{epssmlbig20}
 G(E^{(1)}(\varepsilon),\varepsilon) = 0~~~\text{for}~\varepsilon\in(-\eps_0,\varepsilon_0).
\end{equation*}

It follows from the expansions in \eqref{eps_exp} that, for $\varepsilon\in(-\eps_0,\varepsilon_0)$, $E_{\star}\equiv
E^{\eps}$ is a simple $L^2_{k_\star,\ee}-$ eigenvalue with corresponding eigenfunction
$\Phi_1(x)\equiv\Phi^{(\eps)}(x)$. This completes the proof of Lemma \ref{property1-proof}.

\medskip

\textit{Proof of Lemma \ref{lam_nonzero_eps_small}}: We continue to assume, without loss of generality, that $n$ is
even. From \eqref{eps_exp} and \eqref{even-espace}, $\Phi_1(x)$ may be expanded:
\begin{align*}
 \Phi_1(x)\equiv\Phi^{(\eps)}(x) = e^{ik_{\star} x}\ e^{2i\pi nx} + \mathcal{O}{(\varepsilon)}.
\end{align*}
Calculating $\lamsharp$ we have:
\begin{align}
 \label{lamsharp_small_eps}
 \lamsharp&=
2i\inner{\Phi_1,\partial_x\Phi_1}_{L^2([0,1])}\nonumber\\
 &= 2i\int^{1}_{0}e^{-ik_{\star} x}e^{-2\pi inx}\partial_x\left(e^{ik_{\star} x}e^{2\pi inx}\right)dx +
\mathcal{O}(\varepsilon)\nonumber\\
&=-2\left(2n\pi+k_{\star}\right) + \mathcal{O}{(\varepsilon)} = 
-2\pi\left(2n+1\right) + \mathcal{O}{(\varepsilon)}.
\end{align} 
Therefore, for $\eps$ sufficiently small, $\lamsharp\ne0$.
This completes the proof of Lemma \ref{lam_nonzero_eps_small} and therewith the proof of Theorem \ref{thm:dirac-pt}.

\chapter{Genericity of Dirac Points - 1D and 2D cases}\label{tC-discrete!}

Theorem \ref{thm:dirac-pt} establishes the existence of a sequence of Dirac points of $H^{(\varepsilon)}=-\D_x^2+ V_\ee(x)$   in the sense of Definition \ref{dirac-pt-gen} for all $\varepsilon$ sufficiently small and non-zero. That  is, 
 for each fixed   $n\ge1$, there exists $\eps_0=\eps_0(n)>0$ and  a real-valued  continuous function, defined on $(-\eps_0,\eps_0)$:
 \[ \eps\mapsto E^{(\eps)}_{\star,n},\ \ {\rm with}\
\left.E^{(\eps)}_{\star,n}\right|_{\eps=0}=\pi^2(2n+1)^2, \]
  such that
  $(k_\star=\pi,E^{(\eps)}_{\star,n})$ is a Dirac point (linear band crossing) in the sense of Definition \ref{dirac-pt-gen}
 with symmetry $\mathcal{S}=\mathcal{I}$ and decomposition $H^2_{k_\star}=H^2_{k_\star,\ee}\oplus H^2_{k_\star,\oo}$.
 
The article \cite{FW:12}  studies the operators $H_h^{(\eps)}=-\Delta +\eps V_h$, where $V_h$ is a
 a honeycomb lattice potential on $\R^2$, satisfying a simple non-degeneracy condition. For $H^{(\eps)}_h$ there is an analogous notion of Dirac point (linear band crossing), a quasi-momentum / energy pair at which a conical touching of  adjacent dispersion surfaces occurs; see Theorem 4.1 of \cite{FW:12}.  Part (3) of Theorem 5.1 of \cite{FW:12} states that 
  that for all $\varepsilon\in(-\eps_0,\eps_0)\setminus\{0\}$ 
  the operator $H^{(\eps)}_h$, has a Dirac point $({\bf k}={\bf K}_\star,E=E^\eps_\star)$ where ${\bf K}$ denotes a high-symmetry quasi-momentum (a vertex of the regular hexagonal Brillouin zone) provided   $V_h$ is a two-dimensional honeycomb lattice potential with an appropriate non-degeneracy hypothesis. 
  
  In \cite{FW:12} it is further proved that $H^{(\eps)}_h$ has Dirac points for all real $\eps$ with $|\eps|\ge\eps_0$, except possibly for an exceptional set, $\widetilde{\mathcal{C}}$, which is countable and closed; see Parts (1) and (2) of Theorem 5.1 of \cite{FW:12}.
    The method of \cite{FW:12} can be applied to the operator  $H^{(\varepsilon)}=-\D_x^2+ V_\ee(x)$ to show that 
each map $\eps\mapsto E^{(\eps)}_{\star,n}$ can be continued extended from $(-\eps_0(n),\eps_0(n))$ to all $\R$ minus an exceptional set, $\widetilde{\mathcal{C}}_n$,  which is a countable and closed subset of $\R\setminus(-\eps_0(n),\eps_0(n))$, and such that for all $\eps\in\R\setminus\widetilde{\mathcal{C}}_n$ the quasi-momentum / energy pair 
$(k_\star=\pi,E^{(\eps)}_{\star,n})$ is a Dirac point in the sense of Definition \ref{dirac-pt-gen}
 with symmetry $\mathcal{S}=\mathcal{I}$ and decomposition $H^2_{k_\star}=H^2_{k_\star,\ee}\oplus H^2_{k_\star,\oo}$.
 \medskip
 
\nit {\sl The purpose of this section is to prove:\smallskip

\begin{proposition}\label{discrete!}
For $H^{(\varepsilon)}=-\D_x^2+ V_\ee(x)$, the exceptional set of $\eps-$ values is a discrete subset of $\R\setminus(-\eps_0,\eps_0)$.
 \end{proposition}
\smallskip

 \nit We carry out the proof 
for $H^{(\varepsilon)}=-\D_x^2+ V_\ee(x)$. Discreteness of the exceptional set holds as well for $H^{(\eps)}_h$ and in Remark \ref{honey-discrete} we explain how to adapt the proof in this section to obtain this result.  Without loss of generality, assume $\eps\ge0$. }
 \medskip
 
The global study of the eigenvalue problem for $H^{(\eps)}$ is facilitated by the following:\smallskip

 \begin{proposition}\label{Eee-def}
 There exists a function  $(E,\eps)\mapsto\mathcal{E}_\ee(E,\eps)$, defined and analytic on $\C^2$ and such that for $\eps$ real,  $E$ is an $L^2_\ee-$ eigenvalue of geometric multiplicity $m$ if and only if $E$ is a root of 
$\mathcal{E}_\ee(E,\eps)=0$ of multiplicity $m$. 
\end{proposition}

\nit In analogy with the construction in \cite{FW:12}, we take  $\mathcal{E}_\ee(E,\eps)= \det[I-(\mu+1)T(\eps)]$, where $T(\eps)$ is a trace class operator, obtained when formulating the eigenvalue problem for $-\D_x^2+ V_\ee(x)$ on $L^2_\ee$ as a Lippmann-Schwinger (nonlocal / integral ) equation. See Section 7 of  \cite{FW:12} for details for the case of $H^{(\eps)}_h=-\Delta+\eps V_h$. In \cite{FW:12} the analogue of $T(\eps)$ is not trace class and hence a renormalized (modified) determinant is required.  
\smallskip

We require the following two results on complex-analytic functions on $\C^2$.
 
 \begin{lemma}\label{lemma1} Let $F(z,t)$ be analytic on a neighborhood of $(0,0)\in\C^2$ with $F(0,0)=0$. Suppose $F(z,t)=0$ and $ t\in\R$ implies $z\in\R$. Then there exists an analytic function $g(t)$ on a neighborhood of $0\in\C$ such that $g(0)=0$ and $F(g(t),t)=0$ for all small $t$.
 \end{lemma}
 
\smallskip
 
 \begin{lemma}\label{lemma2} 
  Let $F$ be as in Lemma \ref{lemma1}. Then we may write 
\begin{equation}\label{F-factored}
  F(z,t)\ =\ G(z,t)\times \Pi_{\nu=1}^m\left(z-g_\nu(t)\right),
  \end{equation}
  where $m\ge1$, 
  \begin{enumerate}
  \item $G(z,t)$ and $g_\nu(t)$, $\nu=1,\dots,m$,  are analytic in a neighborhood of the origin,
  \item  $g_\nu(0)=0$ for each $\nu=1,\dots,m$,\ and
  \item $G(0,0)\ne0$.
  \end{enumerate}
  \end{lemma}
  \medskip
  
  \nit{\bf Proof of Lemma \ref{lemma1}:}\ By hypothesis, for $t\in\R$, $F$ can only vanish for  $z\in\R$. Therefore $F$  is not identically zero. Therefore, the equation $F(z,t)=0$ can be solved for $z$ in terms of $t$ by a Puisieux series
\cite{Knopp-volume2:47}, {\it i.e.} a convergent series of the form
\begin{equation}
z\ =\ \sum_{l\ge1}a_l (t^{\frac1m})^l,
\label{puisieux}
\end{equation}
where $m\ge1$ is an integer and any $m^{th}$ root of $t$ produces a solution of $F(z,t)=0$.
Let $t$ be small and positive and use the positive $m^{th}$ root of $t$. By hypothesis,  $z$ must be real. Hence, all $a_l$ are real. If all non-zero $a_l$ are such that $l\equiv0\ ({\rm mod}$ $ m$), then we are done. Otherwise, let $n$ denote the least positive integer such that $n$ is not a positive integer multiple of $m$ and $a_n\ne0$. We may then write our solution of $F(z,t)=0$ as
\begin{equation}
z= \left[\ \widetilde{\sum}_{l\ge1}\ a_l t^{l/m}\ \right]\ +\ a_n \left( t^{1/m} \right)^n\ +\ o(|t|^{\frac{n}{m}}),
\label{puisieux1}
\end{equation}
where $ \widetilde{\sum}_{l\ge1}$ indicates a sum over all $l\ge1$, such that  $l\equiv0$ (mod $m$). 
Recall again that \eqref{puisieux1} holds with any choice of $m^{th}$ root of $t$. 

\nit Whenever $t$ is real, we know that $z$ is real, and therefore the expression in \eqref{puisieux1} within brackets is real. Therefore, 
$a_n \left( t^{1/m} \right)^n\ =\ \textrm{real quantity}\ +\ o\left( |t|^{n/m} \right)$
for $t$ real and small. Since $a_n$ is real and non-zero, we must have 
\begin{equation}
 \left( t^{1/m} \right)^n\ =\ \textrm{real quantity}\ +\ o\left( |t|^{n/m} \right)\ \ \textrm{for $t$ real};
 \label{tmn}\end{equation}
 here again $t^{1/m}$ may be any of the $m^{th}$ roots of $t$.  It follows that $(e^{2\pi i/m})^n\in\R$, {\it i.e.} $n/m\in\Z$ or $n/m\in \Z+1/2$. By the definition of $n$, we have $n/m\notin\Z$. Therefore, 
 $n/m\in\Z+1/2$.
 
   Now take $t$ small and negative and we obtain $ \left( t^{1/m} \right)^n\in i\R$. This contradicts \eqref{tmn}. It follows that  all non-zero $a_l$ are real and are such that $l\equiv0$ mod $m$.
   In other words, the Puisieux series \eqref{puisieux}, which we now denote $g(t)$,  is in fact a convergent power series in the variable $t$ in a neighborhood of $0\in\C$. This completes the proof of Lemma \ref{lemma1}.
 \medskip
 
 \nit{\bf Proof of Lemma \ref{lemma2}:}  Since $z\mapsto F(z,0)$ is analytic and  $F(z,0)\ne0$ for $z\notin\R$, we know that $F(z,0)$ vanishes to finite order at $z=0$.  We proceed by induction on, $m\ge0$,  the order of vanishing.  We cannot have $m=0$ since $F(0,0)=0$. If $m=1$, then the desired conclusion follows from the implicit function theorem.
 
 Suppose that $F(z,0)$ vanishes at $z=0$ to order $m\ge2$ and assume the validity of Lemma \ref{lemma2} for $\widetilde{F}(z,t)$, such that $\widetilde{F}(z,0)$ vanishes to order $\le m-1$ at $z=0$.  We next establish the conclusion of Lemma \ref{lemma2} for the function $F$, completing our induction on $m$. 
 
 Applying the Weierstrass preparation theorem \cite{Krantz:92}, we may assume without loss of generality that $F(z,t)$ is a Weierstrass polynomial, {\it i.e.}
 $F(z,t)=z^m\ +\ \sum_{l=0}^{m-1}f_l(t)z^l$ in a neighborhood of $(0,0)$ with $f_l(t)$  analytic and $f_l(0)=0$.  By Lemma \ref{lemma1}, there is an analytic function $g(t)$ such that $g(0)=0$ and $F(g(t),t)=0$ .
  Dividing the polynomial in $z$, $F(z,t)$, by the polynomial in $z$, $z-g(t)$ we may write
  \begin{equation}
  F(z,t)=  (z-g(t))\cdot \tilde{F}(z,t), \label{double-star}
  \end{equation}
   where $\tilde{F}(z,t)$ satisfies the hypotheses of Lemma \ref{lemma2} and $\tilde{F}(z,0)$ vanishes to order $m-1$ at $z=0$. Applying the induction hypothesis to $\tilde{F}$ and using 
   \eqref{double-star}, we obtain the conclusion of Lemma \ref{lemma2} for the function $F(z,t)$.
  \medskip
  
Proposition \ref{discrete!}, discreteness of the exception set of $\eps$ values, will now be deduced from:\medskip
 
 \begin{proposition}\label{implies-discrete!}
 For all $\eps\in(0,\infty)$ outside of a discrete set, there exists a Floquet-Bloch eigenpair $(E,\Phi_1)=(E_\star^\eps,\Phi^\eps_1)$ with $\Phi^\eps_1\in L^2_\ee$ for $\left.(-\D_x^2+\eps V_\ee)\right|_{L^2_\ee}$ with the following properties:
 \begin{itemize}
 \item[(a)] $|E|\le \overline{C}_0\eps + \overline{C}_1$, where 
 $\overline{C}_0$ and $ \overline{C}_1$ depend only on $V_\ee$.
 \item[(b)] $E$ is a multiplicity one eigenvalue of $\left.(-\D_x^2+\eps V_\ee)\right|_{L^2_\ee}$ 
 \item[(c)] Non-degeneracy condition: $\lambda_\sharp=\lambda_\sharp^\eps\ne0$ holds, where $\lambda_\sharp$  defined in terms of $\Phi_1$, is given in  \eqref{lambda-sharp-eps}.
 \end{itemize}
 \end{proposition}
 \medskip
 
 \nit{\bf Proof of Proposition \ref{implies-discrete!}:} Our analysis for the small $\eps$ regime implies  that there exist positive constants $\overline{\eps}_1$ and $\overline{C}_1$ (depending only on $V_\ee$, such that every $\eps\in[0,\overline{\eps}_1)$ admits an eigenpair $(E,\Phi_1)$ for $\left.(-\D_x^2+\eps V_\ee)\right|_{L^2_\ee}$ satisfying (a), (b) and (c), with $|E|\le \overline{C}_1$.  We choose $\overline{C}_0>\max_{x\in\R}|V_\ee(x)|$.
 \medskip
 
 We say that a given $\eps$ is ``good'' if there exists an eigenpair 
 $(E,\Phi_1)$ for \\
 $\left.(-\D_x^2+\eps V_\ee)\right|_{L^2_\ee}$ satisfying (a)-(c); otherwise we say that $\eps$ is ``bad''. All $\eps\in[0,\overline{\eps}_1)$ are good.
 \medskip
 
  Suppose now that Proposition \ref{implies-discrete!} fails. Then for some $\tilde\eps>0$ there are infinitely many bad $\eps$ in $(0,\tilde\eps)$. We will derive a contradiction. 
    Let 
  \begin{equation*}
   \eps_c\ \equiv\ \inf\{\ \tilde\eps>0: \textrm{infinitely many $\eps$ in $(0,\tilde\eps)$\ are bad.}\ \}
  \end{equation*}
  
 \nit By definition of $\eps_c$,  $0<\overline{\eps}_1\le\eps_c<\infty$ and  we have the following
\begin{align}
&\textrm{For any $\tilde\eps\in(0,\eps_c)$, there are only finitely many bad $\eps$ in $(0,\tilde\eps)$.}\label{eqn1}\\
&\textrm{For any $\tilde\eps>\eps_c$, there are infinitely many bad $\eps$ in $(0,\tilde\eps)$.}
\label{eqn2}\end{align}
  By \eqref{eqn1}, we can find a strictly increasing sequence $\{\eps_\nu\}_{\nu\ge1}$ such that each $\eps_\nu$ is good, and $\eps_\nu\uparrow\eps_c$ as $\nu\to\infty$. 
  
 Corresponding to $\eps_\nu$, let  $(E_\nu,\Phi_{1,\nu})$ be an eigenpair for $\left.(-\D_x^2+\eps V_\ee)\right|_{L^2_\ee}$ satisfying (a)-(c). By (a), for each $\nu$ we have $|E_\nu|\le \overline{C}_0\eps_c+\overline{C}_1$. Hence, by passing to a subsequence we have that there exists $E_c$ such that
\begin{align*}
&\textrm{$E_\nu\to E_c$ as $\nu\to\infty$,\ \  with\  $|E_c|\le \overline{C}_0\eps_c+\overline{C}_1$.}
\end{align*}
 
\nit For $z_0\in\C$ and $r, \eta >0$, introduce the notation
\begin{equation*}
D(z_0,r) \equiv\ \{z\in\C:|z-z_0|<r\}\ \ \ {\rm and}\ \  I(\eta)\equiv(\eps_c-\eta,\eps_c+\eta)\ .
\end{equation*}
\medskip

By Proposition \ref{Eee-def} there exists an analytic mapping $\mathcal{E}_\ee:D(E_c,\eta_1)\times D(\eps_c,\eta_1)\to\C$ such that
\begin{align} \label{eqn4}
&\textrm{$E$ is an eigenvalue of $\left.(-\D_x^2+\eps V_\ee)\right|_{L^2_\ee}$ with multiplicity $m$}\\
&\textrm{if and only if $E$ is a zero of  $\mathcal{E}_\ee(\cdot,\eps)$ of order $m$,} \nn
\end{align}
where we take $0<\eta_1<\eps_c$. Property \eqref{eqn4} holds for all $(E,\eps)\in D(E_c,\eta_1)\times D(\eps_c,\eta_1)$ and all $m\ge1$.
\medskip
 
 Next, in a manner which is  parallel to the analysis of Section 8 in \cite{FW:12} we have that for some $0<\eta_2<\eps_c$
 \begin{align*}
&\textrm{there exist finitely many analytic functions  $(E,\eps)\mapsto\Phi_{1,jk}(x;E,\eps),\ 1\le j,k\le N$}\\
&\textrm{ mapping
$D(E_c,\eta_2)\times D(\eps_c,\eta_2)$ to $H^2_{k_\star,\ee}$ with the following property:}\nn
\end{align*}
For any  $(E,\eps)\in  D(E_c,\eta_1)\times D(\eps_c,\eta_1)$ such that $E$ is a multiplicity one eigenvalue of $\left.(-\D_x^2+\eps V_\ee)\right|_{L^2_\ee}$, all $\Phi_{1,jk}(x)$ are in the corresponding eigenspace. Furthermore, there is at least one choice of $jk$ with $1\le j,k\le N$,  such that the function $\Phi_{1,jk}(x;E,\eps)$ is not identically zero. 

Denote by $\lambda_{\sharp,jk}(E,\eps)$ the quantity arising from the  function $\Phi_{1,jk}(\cdot;E,\eps)\in H^2_{k_\star,\ee}$ via the formula  \eqref{lambda-sharp-eps}. We have that 
\begin{equation*}
\textrm{  each mapping $(E,\eps)\mapsto \lambda_{\sharp,jk}(E,\eps)$ is analytic 
 on $D(E_c,\eta_2)\times D(\eps_c,\eta_2)$}
 \end{equation*}
and the following holds:

 \nit  Let $(E,\eps)\in D(E_c,\eta_2)\times D(\eps_c,\eta_2)$, and
suppose $E$ is a multiplicity one eigenvalue of $\left.(-\D_x^2+\eps V_\ee)\right|_{L^2_\ee}$.
Then,
 \footnote{ The expression for $\lambda_\sharp$ in 
 \eqref{lambda-sharp-eps} depends on the choice of eigenfunction, but condition (c) does not.}
\begin{align}
\label{eqn7}
&\textrm{
(c) holds for $(E,\eps)$
 if and only if $\lambda_{\sharp,jk}(E,\eps)\ne0$ for some $j,k=1,\dots,N$.
  }
\end{align}
By taking $\eta_1$ and $\eta_2$ smaller we may take $\eta_1=\eta_2$.
\medskip

From \eqref{eqn4} we have that $\eps\in\R$ and $\mathcal{E}_\ee(E,\eps)=0$ implies $E\in\R$.
 By Lemma \ref{lemma2}  we find that $\mathcal{E}_\ee$ may be expanded as:
 \begin{equation}
 \mathcal{E}_\ee(E,\eps)\ =\ \Theta(E,\eps)\cdot \Pi_{j=1}^J\left(\ E-g_j(\eps)\ \right)^{m_j},
  \ \ {\rm on}\ \ D(E_c,\eta_3)\times D(\eps_c,\eta_3)\ ,
  \label{Eee-factored}\end{equation}
  where $\Theta(E,\eps)$ is analytic and non-vanishing on $D(E_c,\eta_3)\times D(\eps_c,\eta_3)$,   and 
  \begin{align*}
&\textrm{ $g_1,\dots, g_J:D(\eps_c,\eta_2)\to D(E_c,\eta_1)$ are distinct analytic functions of $\eps$,} 
\end{align*}
and $m_1,\dots, m_J$ are positive integers.
  We may suppose $\eta_3<\eta_2=\eta_1$.
  
  Since $E_\nu\to E_c$ and $\eps_\nu\to\eps_c$ as $\nu\to\infty$, we may pass to a subsequence and assume that $(E_\nu,\eps_\nu)\in  D(E_c,\eta_3)\times D(\eps_c,\eta_3)$
   for all $\nu$.  Since each $(E_\nu,\eps_\nu)$ satisfies (b), from \eqref{eqn4} and \eqref{Eee-factored} we have that  for each $\nu$ we have $\mu_\nu=g_{j_\nu}(\eps_\nu)$ for some $j_\nu\in\{1,\dots, J\}$, with $m_{j_\nu}=1$.  By passing to a subsequence
 and by permuting $g_1,\dots, g_J$ and $m_1,\dots, m_J$, we may assume without loss of generality that $j_\nu=1$ for each $\nu$. Thus, for each $\nu\ge1$,
   \begin{align}
   \label{eqn11-12} E_\nu\ =\ g_1(\eps_\nu)\ \ {\rm with}\ \ m_1=1 .
   \end{align}
   Since the $\eps\mapsto g_1(\eps),\dots,g_J(\eps)$ are \underline{distinct} analytic functions on $D(\eps_c,\eta_3)$, there are only finitely many $\eps\in I\left(\eta_3/2\right)=(\eps_c-\eta_3/2,\eps_c+\eta_3/2)$ such that $g_1(\eps)=g_j(\eps)$ for some $j\ne1$. Therefore, by \eqref{Eee-factored} and since $m_1=1$, it follows that $E=g_1(\eps)$ is a simple zero of $\mathcal{E}_\ee(\cdot,\eps)$ for all but finitely many $\eps\in I\left(\eta_3/2\right)$.
   
   Thanks to \eqref{eqn4} we conclude that
   \begin{align}\label{eqn13}
 &  \textrm{$E=g_1(\eps)$ is a multiplicity one eigenvalue of $\left.(-\D_x^2+\eps V_\ee)\right|_{L^2_\ee}$}
   \end{align}
   for all but finitely many $\eps\in I(\eta_3/2)$.
   
   Note that $\frac{d}{d\eps}\left.(-\D_x^2+\eps V_\ee)\right|_{L^2_\ee}=V_\ee$ has norm at most $\overline{C}_0$ as an operator on $L^2_\ee$. Hence, \eqref{eqn13} and perturbation theory of simple eigenvalues  together show that $\left|g_1'(\eps)\right|\le\overline{C}_0$ for all but finitely many $\eps\in I\left(\eta_3/2\right)$. Recalling that $g_1(\eps)$ is analytic in $D(\eps_c,\eta_3)$ we conclude that $\left|g_1'(\eps)\right|\le\overline{C}_0$ for all $\eps\in I\left(\eta_3/2\right)$.
   
   Recalling that $E_\nu=g_1(\eps_\nu)$, that the $\eps_\nu$ strictly increase to $\eps_c$
    and that each $(E_\nu,\eps_\nu)$ satisfies (a), we conclude that 
    \begin{equation}
  \textrm{  $E=g_1(\eps)$ satisfies $|E|\le \overline{C}_0\eps+\overline{C}_1$ }
 \label{eqn14} \end{equation}
 for all $\eps\in (\eps_{\overline{\nu}},\eps_c+\eta_3/2)$, 
 where we pick $\overline{\nu}$ so that 
 \begin{equation}
\textrm{ $\eps_{\overline{\nu}}$ belongs to $I\left(\eta_3/2\right)$, and $\eps_{\overline{\nu}}<\eps_c$}
\label{eqn15}
\end{equation}

From  \eqref{eqn13},  \eqref{eqn14},  \eqref{eqn15}
  we see that 
  \begin{equation}\label{eqn16}
\textrm{  $(E,\eps)\equiv (g_1(\eps),\eps)$\ satisfies (a) and (b)}
\end{equation}
for all but finitely many $ \eps\in \widehat{I}\equiv (\eps_{\overline{\nu}},\eps_c+\eta_3/2)$, 
where
 \begin{equation}\label{eqn17}
\textrm{  $\eps_c$ is an interior point of $\widehat{I}$ and $\widehat{I}\subset I\left(\eta_3/2\right)$}\ .
\end{equation}

Next, we address (c). Since $(E_\nu,\eps_\nu)$ satisfies (c), we learn from \eqref{eqn7} that the
 vector in $\C^{N^2}$: $\left(\lambda_{\sharp,jk}(E_\nu,\eps_\nu)\right)$,\  $ j,k=1,\dots,N$ is nonzero for each $\nu$. Recalling \eqref{eqn11-12}, we see that the function
  $\eps\mapsto\left(\lambda_{\sharp,jk}(E_\nu,\eps_\nu)\right)$,\  $ j,k=1,\dots,N$, which maps
   $D(\eps_c,\eta_3)$ into $\C^{N^2}$, is analytic and not identically zero. Hence, there are at most finitely many $\eps\in\widehat{I}$ for which $\lambda_{\sharp,jk}(E_\nu,\eps_\nu)=0$
    for all $j,k=1,\dots,N$. Another appeal to \eqref{eqn7} tells us that (c) holds for $(E,\eps)=(g_1(\eps),\eps)$ for all but finitely many $\eps\in\widehat{I}$.
    
 Together with   \eqref{eqn16}-\eqref{eqn17}, this in turn tells us that 
 \begin{equation}
 \label{eqn19} 
 \textrm{all but finitely many $\eps\in\widehat{I}=
 \left(\eps_{\overline{\nu}},\eps_c+\eta_3/2\right)$ are good.}
 \end{equation}
 On the other hand, \eqref{eqn1} tells us that 
 \begin{equation}
  \label{eqn20}
 \textrm{all but finitely many $\eps\in\widehat{I}=
 \left(0,\eps_{\overline{\nu}}\right)$ are good.}
 \end{equation}
 since $\eps_{\overline{\nu}}<\eps_c$. From \eqref{eqn19} and \eqref{eqn20} we conclude that there are at most finitely many bad $\eps$ in $(0,\eps_c+\eta_3/2)$, contradicting \eqref{eqn2}. Thus, our assumption that Proposition \ref{implies-discrete!} fails must be false.

\medskip
  
  \begin{remark}[Honeycomb lattice potentials in 2D; the exceptional set of $\eps$ for $H^{(\eps)}_h$ is discrete]
  \label{honey-discrete}
 The above strategy can be applied to $H^{(\eps)}_h=-\Delta+\eps V_h$, where $V_h$ is a honeycomb lattice potential in the sense of \cite{FW:12}. In particular, it can be shown that outside of a possible discrete set of real $\eps-$ values, the operator  $H^{(\eps)}_h$ has Dirac points for all $\eps\in\R\setminus\{0\}$. 
 
 We remark briefly on how to adapt the argument of this section to  $H^{(\eps)}_h$. For definitions and notation see \cite{FW:12}.  The essential differences are:
 \begin{itemize}
 \item[(i)] A Dirac point is an energy / quasimomentum pair $\left({\bf K}_\star,E_\star\right)$, with 
  ${\bf K}_\star\in\mathcal{B}_h\subset\R^2_{\bf k}$, the 2D Brillouin zone, such that:\\
  - $E_\star$ is a simple $L^2_{{{\bf K}_\star},\tau}$ eigenvalue of $H^{(\eps)}_h$,\\
  -  $E_\star$ is a simple $L^2_{{\bf K}_\star,\overline{\tau}}$ eigenvalue of $H^{(\eps)}_h$, and \\
  -   $E_\star$ is not a $L^2_{{\bf K}_\star,1}$ eigenvalue of $H^{(\eps)}_h$.
    We note that $E_\star$ is a simple $L^2_{{\bf K}_\star,\tau}$ eigenvalue implies, by a symmetry argument,  that $E_\star$ is a simple $L^2_{{\bf K}_\star,\overline{\tau}}$ eigenvalue.\\
    - Let $E_\star$ denote such an eigenvalue. Then,  $\lambda_\sharp$ (see  (4.1) of \cite{FW:12} ), defined in terms of any choice of normalized $L^2_{{\bf K}_\star,\tau}$  eigenstate corresponding to the eigenvalue $E_\star$,
     does not vanish.
 \item[(ii)] There is an analytic function, $\mathcal{E}_\tau(E,\eps)$ which vanishes to order $m$ 
 if and only if $H^{(\eps)}_h$ has an eigenfunction with geometric multiplicity eigenvalue $m$ in the space $L^2_{{\bf K},\tau}$. 
 \item[(iii)] There is an analytic function, $\mathcal{E}_1(E,\eps)$ which vanishes to order $m$ 
 if and only if $H^{(\eps)}_h$ has an eigenfunction with geometric multiplicity eigenvalue $m$ in the space $L^2_{{\bf K},1}$.
 \end{itemize}
Discreteness follows from the formulation and proof of an assertion analogous  to Proposition \ref{implies-discrete!}: For all $\eps\in\R\setminus(-\eps_0,\eps_0)$, outside of a possible  discrete set of exceptional values, there exists a real eigenvalue of $H^{(\eps)}_h$, $E=E_\star^\eps$ satisfying the following conditions:
 \begin{itemize}
 \item[(a)] $|E|\le \overline{C}_0\eps + \overline{C}_1$, where 
 $\overline{C}_0$ and $ \overline{C}_1$ depend only on $V_h$.
 \item[(b.$\tau$)] $E$ is a multiplicity one $L^2_{{\bf K},\tau}-$ eigenvalue of $H^{(\eps)}_h$. 
  \item[(b.$1$)] $E$ is \underline{not} an $L^2_{{\bf K},1}-$ eigenvalue of $H^{(\eps)}_h$
 \item[(c)] Non-degeneracy condition: $\lambda_\sharp=\lambda_\sharp(E,\eps)\ne0$ holds.
 \end{itemize}
   \end{remark}
   \medskip
  
\begin{remark}\label{TexasA&M}
In a preprint (arXiv.org/abs/1412.8096) and follow-up ongoing work to revise it,
G. Berkolaiko and A. Comech have been 
working to simplify the treatment of Dirac points [11]
and to clarify the nature of the exceptional set of $\varepsilon-$ values, defined therein.
They have recently communicated their ideas to us.
Those ideas have some overlap with our
use of Lemma \ref{lemma1} and Lemma \ref{lemma2} .
The above works are independent of one another.
\end{remark}
\medskip

\chapter{Degeneracy Lifting at Quasi-momentum Zero}\label{degeneracies_lifted}

In this section we prove Remark \ref{k0_remark}. 
 Fix $k=0$ and consider the Floquet eigenvalue problem \eqref{FB-evp}:
\begin{equation}
\label{k0_eps_prob}
 \begin{split}
  &H^{(\eps)}\Phi = E\Phi, ~~~ \Phi(x+1) = \Phi(x), ~~~ \text{where} \\
&H^{(\eps)} \equiv -\partial_x^2+\varepsilon V_{\ee}(x).
\end{split}
\end{equation} 
For $\varepsilon=0$, \eqref{k0_eps_prob} has, for $n\in\Z$,  doubly degenerate eigenvalues at $E_n^{(0)}=(2n\pi)^2$
with corresponding eigenfunctions $\Psi_{n}(x)=e^{2\pi inx}$. We now study how these eigenvalues perturb for
$\varepsilon$ small but nonzero.

Fix $n\in\mathbb{Z}$. To simplify notation, we shall drop the $n$ subscripts and label the eigenvalue as
$E^{(0)} = E_n^{(0)}$ and corresponding eigenfunctions as $\Psi_1(x)=\Psi_{n}(x)=e^{2\pi inx}$ and
$\Psi_2(x)=\Psi_{-n}(x)=e^{-2\pi inx}$. 
For $|\eps|\ne0$ and small,  expand the solution of the eigenvalue problem \eqref{k0_eps_prob} as 
\begin{equation}
 \label{k0_eps_exp}
 \begin{split}
\Phi(x;\varepsilon) &= \Phi^{(0)} + \varepsilon \Phi^{(1)},~~~\Phi^{(0)}=\alpha\Psi_1+\beta\Psi_2,\\
E(\varepsilon) &= E^{(0)} + \varepsilon E^{(1)},
\end{split}
\end{equation} 
where $\Phi^{(1)}$ and $E^{(1)}$ are $\eps-$ dependent corrections. Substituting
\eqref{k0_eps_exp} into \eqref{k0_eps_prob} yields an inhomogeneous problem for $\Phi^{(1)}\in L^2_{\rm per}[0,1]$:
\begin{equation}
 \label{k0_eps_prob_expanded}
\left(-\partial_x^2 -E^{(0)}\right)\Phi^{(1)}(x) + \left(V_{\ee}(x) - E^{(1)}\right)
\left(\Phi^{(0)}(x)+\varepsilon \Phi^{(1)}(x)\right)=0.
\end{equation}

We now proceed by a Lyapunov-Schmidt reduction strategy. Define the orthogonal projection operators
$Q_{\parallel}$ and $Q_{\perp}$:
\begin{align*}
 Q_{\parallel}f(x) &= \inner{\Psi_1,f}_{L^2[0,1]}\Psi_1(x) +
\inner{\Psi_2,f}_{L^2[0,1]}\Psi_2(x),~~~\text{and}~~~\\
 Q_{\perp}f(x) &= (I-Q_{\parallel})f(x).
\end{align*} Equation \eqref{k0_eps_prob_expanded} may then be rewritten as the equivalent system
\begin{align}
  Q_{\parallel}\left(V_{\ee} - E^{(1)}\right)\left(\Phi^{(0)}+\varepsilon \Phi^{(1)}\right) &= 0,
\label{k0_qpara} \\
    Q_{\perp}\left(V_{\ee} - E^{(1)}\right)\left(\Phi^{(0)}+\varepsilon \Phi^{(1)}\right) &=
\left(-\partial_x^2-E^{(0)}\right) \Phi^{(1)}. \label{k0_qperp}
\end{align} 

Introducing the resolvent operator
\begin{equation*}
 R(E^{(0)}) = \left(-\partial_x^2 - E^{(0)}\right)^{-1}: 
 Q_{\perp}L^2 \to Q_{\perp}H^2 \ ,
\end{equation*}
 equation \eqref{k0_qperp} for $\Phi^{(1)}$ may be equivalently written as
 \begin{equation}
 \left( I + \mathcal{A}_\eps \right) \Phi^{(1)}\ =\ R(E^{(0)})\ {Q}_\perp\ 
 \left(V_{\ee}(x) - E^{(1)}\right) \Phi^{(0)},
 \label{D_p1-perp2}
 \end{equation}
 where
 \begin{equation}\label{D_A-def}
 f\mapsto \mathcal{A}_\eps f\ \equiv\ -\eps R(E^{(0)})\ {Q}_{\perp} 
 \left(V_{\ee}(x) - E^{(1)}\right) \ f
 \end{equation}
 is a bounded operator on $H^2_{\rm per}[0,1]$. Furthermore, for $\eps$ sufficiently small,
the operator norm of $\mathcal{A}_\eps$ is less than one, $(I+\mathcal{A}_\eps)^{-1}$ exists, and hence
\eqref{D_p1-perp2} is uniquely solvable in ${Q}_\perp H^2_{\rm per}[0,1]$:
 \begin{align}
 \Phi^{(1)}\ &= \left( I - \eps \ R(E^{(0)}) {Q}_{\perp}
 \left(V_{\ee}(x) - E^{(1)}\right)  \right)^{-1}
  R(E^{(0)}) {Q}_\perp\
 \left(V_{\ee}(x) - E^{(1)}\right)) \Phi^{(0)}.
 \label{D_p1-perp3}
 \end{align}

Recall that $\Phi^{(0)}=\alpha \Psi_1(x)+\beta \Psi_2(x)$ and therefore that $\Phi^{(1)}$ is linear in $\alpha$ and
$\beta$. We may therefore write
\begin{equation}
  \Phi^{(1)} = g^{(1)}[\varepsilon,E^{(1)}](x)\alpha +  g^{(2)}[\varepsilon,E^{(1)}](x)\beta, \label{k0_psi1}
\end{equation}
where $(\varepsilon,E^{(1)})\mapsto g^{(j)}(\varepsilon,E^{(1)})$ is a smooth mapping from a
neighborhood of $(0,0)\in\R\times\mathbb{C}$ into $H^2_{\rm per}([0,1])$, which satisfies the bound
\begin{equation}
\label{g_bounds}
 \norm{g^{(j)}(\eps,E^{(1)})}_{H^2([0,1])} \lesssim\ 1+ \abs{\eps} + |E^{(1)}|\ ,~~~j=1,2.
\end{equation} 
Note also that
\begin{equation}
 \label{k0_g_perp}
 Q_{\parallel}g^{(j)}(\varepsilon,E^{(1)})=0,~~~j=1,2.
\end{equation}

We may now substitute \eqref{k0_psi1} into \eqref{k0_qpara}, projecting onto $\Psi_1$ and $\Psi_2$ to obtain a
system of two homogeneous linear equations for $\alpha$ and $\beta$. Noting the relations
\begin{align*}
 \inner{\Psi_i,\Psi_{j}}_{L^2[0,1]} &= \delta_{ij},~~~i,j=1,2,\\
 \inner{\Psi_i,V_{\ee}\Psi_i}_{L^2[0,1]} &= v_0, ~~~i=1,2,\\
\inner{\Psi_i,V_{\ee}\Psi_j}_{L^2[0,1]} &= v_{2n}, ~~~i,j=1,2,~i\neq j,
\end{align*}
and 
\begin{equation*}
 \inner{\Psi_{i},g^{(j)}(\varepsilon,E^{(1)})}_{L^2[0,1]}=0,~~~i,j=1,2,
\end{equation*} from \eqref{k0_g_perp}, this system may be compactly written as
\begin{equation}
 \label{k0_matrix_sys}
 \mathcal{M}(\varepsilon,E^{(1)}) \begin{pmatrix} \alpha\\\beta \end{pmatrix} =0,
\end{equation} with (inner products in \eqref{Mk0-def} are over $L^2([0,1])$)
\begin{align}
&\mathcal{M}(\varepsilon,E^{(1)}) \label{Mk0-def} \\
&\quad = 
 \begin{pmatrix} 
 v_0-E^{(1)} +\varepsilon \inner{\Psi_{1},V_{\ee}g^{(1)}[\varepsilon,E^{(1)}]} &
 v_{2n} +\varepsilon \inner{\Psi_{1},V_{\ee}g^{(2)}[\varepsilon,E^{(1)}]}\\
 v_{2n} +\varepsilon \inner{\Psi_{2},V_{\ee}g^{(1)}[\varepsilon,E^{(1)}]} & 
 v_0-E^{(1)} +\varepsilon \inner{\Psi_{2},V_{\ee}g^{(2)}[\varepsilon,E^{(1)}]}
\end{pmatrix}.
\nn
\end{align}

Therefore $E(\varepsilon) = E^{(0)} + \varepsilon E^{(1)}$ is an eigenvalue for the spectral problem
\eqref{k0_eps_prob} if and only if $E^{(1)}= E^{(1)}(\varepsilon)$ solves
\begin{equation}
 \label{k0_det}
 \text{det}\mathcal{M}\left(\varepsilon,E^{(1)}\right) =0,
\end{equation} or equivalently
\begin{align}
 \label{k0_eval_corrections}
 \mathcal{J}(E^{(1)},\eps) \equiv (E^{(1)})^2 - 2 E^{(1)} v_0 + v_0^2 - v_{2n}^2 + \eps\rho(\eps,E^{(1)}) \ ,
\end{align}
where $(\eps,E^{(1)}) \mapsto \rho(\eps,E^{(1)})$ is a smooth mapping from a neighborhood of
$(0,0)\in\R\times\mathbb{C}$ given by
\begin{align*}
\rho(\eps,E^{(1)}) &\equiv (v_0-E^{(1)}) \left(\inner{\Psi_{1},V_{\ee}g^{(1)}} +
\inner{\Psi_{2},V_{\ee}g^{(2)}} \right) \\
&\qquad- v_{2n} \left(\inner{\Psi_{1},V_{\ee}g^{(2)}} +
\inner{\Psi_{2},V_{\ee}g^{(1)}} \right) \\
&\qquad + \eps 
\left(\inner{\Psi_{1},V_{\ee}g^{(1)}}\inner{\Psi_{2},V_{\ee}g^{(2)}}
-\inner{\Psi_{2},V_{\ee}g^{(1)}}\inner{\Psi_{1},V_{\ee}g^{(2)}} \right) ,
\end{align*} 
and satisfies the bound
\begin{equation*}
\rho(\eps,E^{(1)}) = \mathcal{O}\left(1+\abs{\eps}+\abs{E^{(1)}}\right) ,
\end{equation*}
for $\abs{\eps}+\abs{E^{(1)}}$ small. We solve $\mathcal{J}(E^{(1)},\eps)=0$, \eqref{k0_eval_corrections}, for
$E^{(1)}$ with $\eps$ small using the implicit function theorem.

Since, for
\begin{equation*}
 \nu_{\pm} = \nu_{\pm,n} = v_0 \pm v_{2n} \ ,
\end{equation*} $\mathcal{J}(E^{(1)},\eps)$ is analytic in a neighborhood of $(E^{(1)},\eps)=(\nu_{\pm},0)$, 
$\mathcal{J}(\nu_\pm,0)=0$ and $\partial_{E^{(1)}}\mathcal{J}(\nu_\pm,0) = \pm 2 v_{2n}$,
it follows from the implicit function theorem that for $v_{2n}\neq0$, there exists an
$\eps_0>0$ and continuous functions $\eps\mapsto E_\pm^{(1)}(\eps)$, satisfying
\begin{equation*}
 \mathcal{J}(E_\pm^{(1)}(\eps),\eps) = 0, \quad \text{for $\eps\in(0,\eps_0)$} \ .
\end{equation*} Thus, for $v_{2n}\ne0$ and $\varepsilon$ sufficiently small and nonzero, the double
eigenvalue at $k=0$
splits and the degeneracy is ``lifted'': 
\begin{align*}
 E_{\pm}(\varepsilon) &=  E^{(0)} + \varepsilon E_\pm^{(1)}(\varepsilon) \\
 &= (2n\pi)^2 + \varepsilon\left(v_{0} \pm v_{2n}\right) +
\mathcal{O}\left(\varepsilon^2\right).
\end{align*}

\chapter{Gap Opening Due to Breaking of Inversion Symmetry}\label{band_splitting}

We consider the Floquet-Bloch eigenvalue problem \eqref{FB-evp}
\begin{align}
\left( H_{V_\ee}(k) + \delta \kappa_\infty W_\oo(x) \right) p\ &=\ E p,\quad  p(x+1;k)=  p(x;k),\ \ {\rm
where}\label{E_fb-p}\\
H_{V_\ee}(k)\ &\equiv\ -(\D_x+ik)^2+V_\ee(x) \ , \label{E_HVofk-def}
\end{align}
for $k$ near $k_\star-\pi$ and $\delta$ small and nonzero. Here we allow $\delta$ to take both positive and negative
values: $\delta>0$ corresponds to the $\kappa_\infty$ case while $\delta<0$ corresponds to the
$-\kappa_\infty$ case. We let 
 \[ k=k_{\star}+\delta {\zeta}\ ,\]
 where $\zeta$ is bounded and rewrite the Floquet-Bloch eigenvalue problem as
 \begin{align}
 \label{E_dis_rel_prob}
& \left(-(\partial_x+i(k_\star+ \delta \zeta))^2+  V_{\ee}(x) +\delta \kappa_\infty W_\oo(x)
\right)p(x;k_{\star}+\delta {\zeta}) \nn \\
&\qquad\qquad = E(k_\star+ \delta \zeta)\ p(x;k_\star+ \delta \zeta),\\
& p(x+1;k_{\star}+\delta {\zeta}) = p(x;k_{\star}+\delta {\zeta}).\nn
\end{align} 
\medskip

Recall that for $\delta=0$, $E_{\star}$ is a doubly degenerate eigenvalue with corresponding eigenspace
spanned by $\{\Phi_1,\Phi_2\}$. We study how $E$ in \eqref{E_dis_rel_prob} deforms away from $E_\star$ for $\delta$
small. Let
\begin{equation*}
p_j(x)=e^{-ik_\star x}\Phi_j(x)
\end{equation*}
and seek a solution  of the form:
\begin{align}
E(k_\star+\delta \zeta)&=E_\star+ \delta E^{(1)},\label{E_E-ansatz}\\
 p(x;k_{\star}+\delta {\zeta})&= p^{(0)} + \delta p^{(1)}, \qquad p^{(0)}\equiv\alpha p_1+\beta p_2\ ,
\label{E_p-ansatz}
\end{align}
where
\begin{equation*}
\inner{ p_j,p^{(1)} }_{L^2[0,1]}=0,\quad j=1,2 .
\end{equation*}

Substituting \eqref{E_E-ansatz}-\eqref{E_p-ansatz} into \eqref{E_dis_rel_prob}, we obtain
an inhomogeneous problem for $p^{(1)}\in L^2_{\rm per}[0,1]$:
\begin{align}
\left(H_{V_\ee}(k_\star)-E_\star\right)p^{(1)} &= \left(2i\delta \zeta(\D_x+ik_\star)-\delta^2 \zeta^2 +\delta E^{(1)} -
\delta\kappa_\infty W_\oo(x) \right)p^{(1)}\nn\\
&\qquad +\left(2i\zeta(\D_x+ik_\star)-\delta\zeta^2 + E^{(1)} -\kappa_\infty W_\oo(x) \right)p^{(0)}
\nn\\
&\equiv J(\alpha,\beta,\delta,{\zeta},E^{(1)},p^{(1)})\ .
\label{E_p1-eqn}
\end{align}

We proceed to follow a Lyapunov-Schmidt reduction argument. Introduce the orthogonal projections
$\widetilde{Q}_{\parallel}$ and $\widetilde{Q}_{\perp}$ defined by:
\begin{align*}
 \widetilde{Q}_{\parallel}f(x)&=\inner{p_1,f}_{L^2([0,1])}p_1(x)+\inner{p_2,f}_{L^2([0,1])}
p_2(x),  \\ 
 \widetilde{Q}_{\perp}f(x) &= (I-\widetilde{Q}_{\parallel})f(x).
\end{align*} 
Equation \eqref{E_p1-eqn} may then be rewritten as the equivalent system for the unknowns
$p^{(1)}=p^{(1)}(x;\delta,{\zeta})$ and $E^{(1)}=E^{(1)}(\delta,{\zeta})$:
\begin{align}
\left(H_{V_\ee}(k_\star)-E_{\star}\right) p^{(1)} &=
\widetilde{Q}_{\perp}J(\alpha,\beta,\delta,{\zeta},E^{(1)},p^{(1)}),\label{E_p1-perp}\\
0 &= \widetilde{Q}_{\parallel}J(\alpha,\beta,\delta,{\zeta},E^{(1)},p^{(1)}).\label{E_p1-para}
\end{align} 
\medskip

In detail, system \eqref{E_p1-perp}-\eqref{E_p1-para} reads as
\begin{align}
\left(H_{V_\ee}(k_\star)-E_{\star}\right)p^{(1)} &=
\widetilde{Q}_{\perp} \left(2i\delta{\zeta}(\D_x+ik_\star)-\delta^2\zeta^2+\delta E^{(1)} -\delta\kappa_\infty
W_\oo(x) \right) p^{(1)} \nn\\
&\qquad + \widetilde{Q}_{\perp} \Big(2i{\zeta}(\D_x+ik_\star) -\kappa_\infty W_\oo(x) \Big) p^{(0)}
,\label{E_p1-perp1} \\
0&= \widetilde{Q}_{\parallel} \left(2i{\zeta}(\D_x+ik_\star)-\delta\zeta^2+E^{(1)} -\kappa_\infty W_\oo(x)
\right) p^{(0)} \nn \\
&\qquad + \widetilde{Q}_{\parallel} \Big(2i\delta{\zeta}(\D_x+ik_\star) -\delta\kappa_\infty W_\oo(x) \Big)
p^{(1)} \label{E_p1-para1} .
\end{align}

Introduce the resolvent
\begin{equation*}
 R_{k_\star}(E_{\star}) = \left(H_{V_\ee}(k_\star) - E_{\star}\right)^{-1},
\end{equation*} defined as a bounded map from $\widetilde{Q}_{\perp}L^2[0,1]$ to $\widetilde{Q}_{\perp}H^2[0,1]$.
 Equation \eqref{E_p1-perp1} for $p^{(1)}$ may then be rewritten as
 \begin{equation}
 \left( I + A_\delta \right) p^{(1)}\ =\ R_{k_\star}(E_\star)\ \widetilde{Q}_\perp\ 
 \big(2i{\zeta}(\D_x+ik_\star)-\kappa_\infty W_\oo(x) \big) p^{(0)},
 \label{E_p1-perp2}
 \end{equation}
 where
 \begin{equation*}
 f\mapsto A_\delta f\ \equiv\ -R_{k_\star}(E_\star)\ \widetilde{Q}_{\perp} 
 \left(2i\delta{\zeta}(\D_x+ik_\star)-\delta^2\zeta^2+\delta E^{(1)} -\delta\kappa_\infty W_\oo(x) \right) \ f
 \end{equation*}
 is a bounded operator on $H^2_{\rm per}[0,1]$. Furthermore, for $\delta$ sufficiently small,
the operator norm of $A_\delta$ is less than one, $(I+A_\delta)^{-1}$ exists, and hence \eqref{E_p1-perp2} is uniquely
solvable in $\widetilde{Q}_\perp H^2_{\rm per}[0,1]$:
 \begin{align*}
 p^{(1)}\ &=\ \left( I\ -\ R_{k_\star}(E_\star)\ \widetilde{Q}_{\perp}
 \left(2i\delta{\zeta}(\D_x+ik_\star)-\delta^2\zeta^2+ \delta E^{(1)} -\delta\kappa_\infty W_\oo(x) \right) \
\right)^{-1}\nn\\
 &\qquad \circ\ R_{k_\star}(E_\star)\ \widetilde{Q}_\perp\ 
 \left(2i{\zeta}(\D_x+ik_\star) -\kappa_\infty W_\oo(x) \right) p^{(0)}.
 \end{align*}

Recall that $p^{(0)}=\alpha p_1(x)+\beta p_2(x)$ (equation \eqref{E_p-ansatz}) and therefore that $p^{(1)}$ is 
linear in $\alpha$ and $\beta$. We may therefore write
\begin{equation}
\label{E_p1-perp4}
 p^{(1)} = p^{(1)}(x;\delta,\zeta,E^{(1)}) = g^{(1)}[\delta,{\zeta},E^{(1)}](x)\alpha +
g^{(2)}[\delta,{\zeta},E^{(1)}](x)\beta,
\end{equation} 
where $(\delta,\zeta,E^{(1)})\mapsto g^{(j)}(\delta,\zeta,E^{(1)})$ is a smooth mapping from a neighborhood of
$(0,0,0)\in\R\times\R\times\mathbb{C}$ into $H^2_{\rm per}([0,1])$, which satisfies the bound
\begin{equation*}
 \norm{g^{(j)}(\delta,\zeta,E^{(1)})}_{H^2([0,1])} \lesssim\ 1 + \zeta + |\delta|\left(1+|E^{(1)|}\right)  \ ,~~~j=1,2.
\end{equation*}
Note also that 
\begin{equation}
\widetilde{Q}_\parallel g^{(j)}(\delta,\zeta,E^{(1)})=0,\ j=1,2.
\label{E_Ppar}
\end{equation}

We may now substitute \eqref{E_p1-perp4} into \eqref{E_p1-para1} and obtain a system of two homogeneous linear
equations for $\alpha$ and $\beta$. To express this system in a compact form we employ several simplifying
relations. Recall that
\begin{align}
\left(\D_x+ik_\star\right)p_j &= e^{-ik_\star x}\ \D_x\ e^{ik_\star x}p_j= e^{-ik_\star
x}\D_x\Phi_j\label{E_simple_0},\\
\left\langle p_i,p_j\right\rangle_{L^2[0,1]} &=\left\langle \Phi_i,\Phi_j\right\rangle_{L^2[0,1]}=\delta_{ij},\ j=1,2,
\end{align}
along with the definitions of $\lamsharp$ and $\thetasharp$ in \eqref{thetasharp_defn}:
\begin{align}
 &2i\inner{p_i,(\partial_x + ik_\star) p_i}_{L^2[0,1]} = 2i\inner{\Phi_i,\partial_x \Phi_i}_{L^2[0,1]} =
(-1)^{i+1} \lamsharp, ~~~i=1,2,\\
&\inner{p_i,W_\oo(x) p_i}_{L^2[0,1]} = \inner{\Phi_i,W_\oo(x) \Phi_j}_{L^2[0,1]} = \thetasharp,
~~~~i\neq j.
\end{align}
Define
\begin{equation}
G^{(j)}[\zeta,E^{(1)}](x)= e^{ik_\star x}\ g^{(j)}[\zeta,E^{(1)}](x),\ \textrm{ and note}\  \left\langle
\Phi_i,G^{(j)}\right\rangle_{L^2[0,1]}=0,\ i,j=1,2
\end{equation}
by \eqref{E_Ppar}. 
Furthermore, since
\begin{equation*}
f\in L^2_{k_\star,\ee},\ g\in L^2_{k_\star,\oo}\ \ \implies\ \ \left\langle f,g\right\rangle=0,
\end{equation*}
we obtain
\begin{align}
&\inner{p_i,(\partial_x +ik_\star) p_j}_{L^2[0,1]} = \inner{\Phi_i,\partial_x \Phi_j}_{L^2[0,1]} = 0, ~~~i\neq j, \\
&\inner{p_i,W_{\oo}p_i}_{L^2[0,1]} = \inner{\Phi_i,W_{\oo}\Phi_i}_{L^2[0,1]} = 0, ~~~i=1,2, \label{E_simple_1}
\end{align} 

Using the above relations \eqref{E_simple_0}-\eqref{E_simple_1}, equation \eqref{E_p1-para1} for $(\alpha, \beta)$ can
now be written in the form
\begin{equation}
\mathcal{M}(E^{(1)},\delta,{\zeta})
\begin{pmatrix}
 \alpha \\ \beta
\end{pmatrix} =0,
\label{E_Mab0}
\end{equation}
 with (inner products in \eqref{E_M-def} are over $L^2([0,1])$)
\begin{align}
\mathcal{M}(E^{(1)},\delta,{\zeta}) &\equiv
\begin{pmatrix}
 E^{(1)}+{\zeta}\lamsharp &
-\kappa_\infty\thetasharp \\
-\kappa_\infty\thetasharp &
E^{(1)}-{\zeta}\lamsharp 
\end{pmatrix} \ + \nn\\
&\ \delta
\begin{pmatrix}
-{\zeta}^2 + \inner{\Phi_1,(2i \zeta\partial_x -\kappa_\infty W_\oo) G^{(1)}} &
\inner{\Phi_1,(2i \zeta\partial_x -\kappa_\infty W_\oo) G^{(2)}} \\
-{\zeta}^2 + \inner{\Phi_2,(2i \zeta\partial_x -\kappa_\infty W_\oo) G^{(1)}} &
\inner{\Phi_2,(2i \zeta\partial_x -\kappa_\infty W_\oo) G^{(2)}}
\end{pmatrix},
\label{E_M-def}
\end{align} where $(\delta,\zeta,E^{(1)}) \mapsto G^{(j)}[\delta,\zeta,E^{(1)}], \D_x G^{(j)}[\delta,\zeta,E^{(1)}]$ 
are smooth
functions of \\
$(\delta,\zeta,E^{(1)})$ in a neighborhood  of $(0,0,0)$ and
\begin{align*}
 &  \|G^{(j)}[\delta,\zeta, E^{(1)}]\|_{L^2[0,1]} +
\|\D_x G^{(j)}[\delta,\zeta, E^{(1)}]\|_{L^2[0,1]}=
\mathcal{O}\left(1 + |\zeta| + |\delta|\left( 1+\abs{E^{(1)}} \right)\right).
\end{align*}

\medskip

Thus $E=E_{\star}+E^{(1)}(\delta,\zeta)$ is an eigenvalue for the spectral problem \eqref{E_dis_rel_prob} if and
only if $E^{(1)}=E^{(1)}(\delta,\zeta)$ solves
\begin{equation*}
 \text{det}\mathcal{M}\left(E^{(1)},\delta,{\zeta}\right)=0,
\end{equation*}
 or equivalently
\begin{equation*}
\mathcal{J}(E^{(1)},\delta,\zeta) = 0, 
\end{equation*}
where
\begin{align*}
 \mathcal{J}(E^{(1)},\delta,\zeta) &\equiv (E^{(1)})^2 - \zeta^2\lamsharp^2 - \kappa_\infty^2\thetasharp^2 
 + \delta \rho(\delta,\zeta,E^{(1)}) \ ,
\end{align*}
 and
 \begin{align*}
 \norm{\rho(\delta,\zeta,E^{(1)})}_{L^2[0,1]} &\lesssim 1 + \zeta + |\delta|\left(1+|E^{(1)}\right) \ .
\end{align*}

Note that $\mathcal{J}(E^{(1)},0,\zeta)$ has a solution
\begin{equation*}
 E^{(1)} = \nu_{\pm}(\zeta) \equiv \pm \sqrt{\kappa_\infty^2\thetasharp^2 + \zeta^2\lamsharp^2}.
\end{equation*}
We will now apply the implicit function theorem to prove that there exists a $\delta_0>0$ and smooth functions
$E^{(1)}_\pm(\delta,\zeta)$  given by
\begin{equation*}
 E^{(1)}_\pm(\delta,\zeta) = \nu_\pm(k)\left(1+\mathcal{O}(\delta)\right) = \pm \sqrt{\kappa_\infty^2\thetasharp^2 +
\zeta^2\lamsharp^2} + \mathcal{O}(\delta) \ ,
\end{equation*} and defined for $|\delta|<\delta_0$ and $|\zeta|\leq1$ such that
\[ \mathcal{J}(E^{(1)}_\pm(\delta,\zeta),\delta,\zeta) = 0 \ . \]
To confirm this it suffices to check that 
\[\partial_{E^{(1)}}\mathcal{J}(\nu_{\pm},0,\zeta) = 2 \nu_{\pm}(\zeta) \neq 0, \]
which holds since by assumption $\thetasharp\neq0$.

Altogether then, for all $|\delta|<\delta_0$ and nonzero and quasi-momentums $k$ such that 
\begin{equation}
\label{E_k_cond}
 |k-k_\star| < \delta ,
\end{equation}
the eigenvalue $E$ in \eqref{E_dis_rel_prob} splits:
\begin{equation}
\label{E_Epm_defn}
 E_{\delta,\pm}(k) =  E^{(0)} + \delta E_\pm^{(1)}(\delta,k/\delta) = E_{\star} \pm
\sqrt{\delta^2\kappa_\infty^2\thetasharp^2 + k^2\lamsharp^2} +\mathcal{O}\left(\delta^2\right).
\end{equation}

\chapter{Bounds on Leading Order Terms in Multiple Scale Expansion}\label{psi_bound_proof}

In this section we prove bounds on $x\mapsto\psi^{(0)}(x,\delta x)$ and $\psi^{(1)}_p(x,\delta x)$ in order to prove Lemma \ref{lemma:psi_bounds}.
 For any $s\ge0$ the bound 
\[
\norm{\psi^{(0)}(\cdot,\delta\cdot)}_{H^s(\R)}\ +\  
\norm{\left.\D_X^2\psi^{(0)}(x,X)\right|_{X=\delta x}}_{L^2(\R)}\ \lesssim\ \delta^{-\frac12}
\]
is a straightforward consequence of the expression for $\psi^{(0)}(x,\delta x)$ and 
 the boundedness of
 $\Phi_j,\ j=1,2$ and its derivatives in $L^\infty(\R)$. 

Turning now to the  bounds on $\psi^{(1)}_p(x,X)$, recall
from \eqref{perturbed_schro_delta1} and \eqref{G1def} that $\psi^{(1)}_p$ satisfies:
\begin{align}
 &(-\partial_x^2+V_{\ee}(x)-E_{\star})\psi^{(1)}_p(x,X)\nn\\
  &=
\sum_{j=1}^2\left.\Big[2\partial_{x}\Phi_j(x)\partial_{X}\alpha_{\star,j}(X)
+\left(-\kappa(X)W_{\oo}(x)\right)\alpha_{\star,j}(X)\Phi_j(x)\Big]\right|_{X=\delta x} \nn \\ 
 & \equiv G^{(1)}(x,X)\Big|_{X=\delta x} \equiv \sum_{i=1}^4f_i(x)g_i(X)\Big|_{X=\delta x}\ ,\ \ \psi^{(1)}_p(x+1,X)
=e^{ik_\star}\psi^{(1)}_p(x,X)
 \label{psi1_eqn}
\end{align} 
Recall further that  $W_{\oo}(x)$, $\Phi_j(x,k)$ are smooth and have derivatives with respect to $x$ which are uniformly bounded
on $\R$, and furthermore that  $\alpha_{\star}=(\alpha_{\star,1},\alpha_{\star,2})\in\mathcal{S}(\R)$ (Remark
\ref{regularity}), and therefore that
$f_i(x)$ is bounded, smooth and $k_{\star}$-pseudo periodic, and that $g_{i}(X)\in\mathcal{S}(\R_X)$.

The solution of \eqref{psi1_eqn}, $\psi^{(1)}_p(x,X)$, can be expanded in terms of a complete set of  $k_\star-$
pseudo-periodic states: $\Phi_b(x,k_\star)=e^{ik_{\star}x}p_b(x;k_{\star})$ , $p_b(x+1;k_{\star})=p_b(x;k_{\star})$,\ 
$b\geq1$. Here, we recall that the Dirac point occurs at the intersection of the $b_\star$-band and $(b_\star+1)$-band.
\begin{equation}
 \label{lemma_psi1_exp}
 \psi^{(1)}_p(x,\delta x) = \sum_{\underset{b\neq b_\star,b_{\star+1}}{b\geq1}}
\frac{1}{(E_b(k_{\star})-E_{\star})}\inner{\Phi_b(\cdot;k_{\star}),G^{(1)}(\cdot,X)}_{L^2([0,1])}
\bigg|_{X=\delta x}\Phi_b(x;k_{\star})\ .
\end{equation} 
Squaring and integrating over $\R$ yields:
\begin{align}
 &\norm{\psi^{(1)}_p(\cdot,\delta\cdot)}_{L^2(\R)}^2 =
 \sum_{\underset{b\neq b_\star,b_{\star+1}}{b\geq1}}
\frac{1}{(E_b(k_{\star})-E_{\star})\ (E_{b'}(k_{\star})-E_{\star})} \times \nonumber\\&~~~ 
\int_{\R} \inner{\Phi_b(\cdot;k_\star),G^{(1)}(\cdot,\delta x)}
\ \overline{\inner{\Phi_{b'}(\cdot;k_\star),G^{(1)}(\cdot,\delta x)}}\ 
\cdot \Phi_b(x;k_\star)\overline{\Phi_{b'}(x;k_\star)}\ dx
\nonumber\\&
= \sum_{\underset{b\neq b_\star,b_{\star+1}}{b\geq1}} \frac{1}{ (E_b(k_{\star})-E_{\star})\ (E_{b'}(k_{\star})-E_{\star}) }
 \int_0^1\ p_b(x;k_{\star}) p_{b'}(x;k_{\star}) \ \times 
\nn\\&\qquad\qquad
 \Big[\ \sum_{m\in\mathbb{Z}}
\left\langle\Phi_b(\cdot;k_\star),G^{(1)}(\cdot,\delta(x+m))\right\rangle
\overline{\left\langle \Phi_{b'}(\cdot;k_\star),G^{(1)}(\cdot,\delta(x+m))\right\rangle}\ \Big] dx \label{lemma_psi1_exp2}\end{align}  
Consider now the sum in square brackets in \eqref{lemma_psi1_exp2}.  By \eqref{psi1_eqn},
\begin{align*}
&\Big|\ \sum_{m\in\mathbb{Z}}
\left\langle\Phi_b(\cdot;k_\star),G^{(1)}(\cdot,\delta(x+m))\right\rangle_{L^2([0,1])}
\overline{\left\langle \Phi_{b'}(\cdot;k_\star),G^{(1)}(\cdot,\delta(x+m))\right\rangle_{L^2([0,1])}}\ \Big|\nn\\
&\qquad = 
\Big|\ \sum_{1\le i,j\le4} \Big[
\left\langle\Phi_b(\cdot;k_\star),f_i(\cdot)\right\rangle_{L^2([0,1])}
\overline{\left\langle \Phi_{b'}(\cdot;k_\star),f_j(\cdot)\right\rangle_{L^2([0,1])}} \ \times \\
&\qquad\qquad \sum_{m\in\mathbb{Z}}
g_i(\delta(x+m))\overline{g_j(\delta(x+m))}\Big]\ \Big|\nn\\
&\qquad \le\ C\ \max_{1\le i,j\le4}\ \Big|\ \sum_{m\in\mathbb{Z}} g_i(\delta(x+m))\overline{g_j(\delta(x+m))}\ \Big|\ ,
\end{align*}
where $C$ is a constant depending only on bounds of $f_j,\ 1\le j\le4$.

Since $g_{i}\in\mathcal{S}(\R_X)$,  we may apply the
Poisson summation formula to the sum just above and obtain:
\begin{align}
 \label{lemma_F_sum}
 \sum_{m\in\mathbb{Z}}\overline{g_i(\delta(x+m))}g_j(\delta(x+m)) &=
\frac{2\pi}{\delta}\sum_{m\in\mathbb{Z}}\widehat{ \overline{g_i}\ g_j}\left(\frac{2\pi m}{\delta}\right)e^{2\pi imx}
\nonumber\\&
=\frac{2\pi}{\delta}\left[\widehat{\overline{g_i}g_j}\left(0\right) +
\sum_{\underset{m\neq0}{m\in\mathbb{Z}}}\widehat{\overline{g_i}g_j}\left(\frac{2\pi m}{\delta}\right)e^{2\pi
imx}\right]\ .
\end{align} 
Furthermore, since  $g_{i}\in\mathcal{S}(\R)$ we have for $\delta$ sufficiently small that the sum over nonzero $m$ is
bounded:
\begin{align*}
\max_{1\leq
i,j\leq4}\abs{\sum_{\underset{m\neq0}{m\in\mathbb{Z}}}\widehat{\overline{g_i}g_j}\left(\frac{2\pi
m}{\delta}\right)e^{2\pi imx}}
&\leq C \sum_{m=1}^{\infty} \frac{1}{1+\left(\frac{2\pi m}{\delta}\right)^{100}} \lesssim \sum_{m=1}^{\infty}
\left(\frac{\delta}{m}\right)^{100} \lesssim \delta^{100}\ .
\end{align*} Therefore, the full sum \eqref{lemma_F_sum} is of order of magnitude $\delta^{-1}$.
 Substitution of the above bounds into \eqref{lemma_psi1_exp2} and using the Weyl asymptotics: $E_b(k)\approx b^2,\
b\gg1$, we obtain
\begin{equation*}
 \norm{\psi^{(1)}_p(\cdot,\delta\cdot)}_{L^2(\R)} \lesssim \delta^{-1/2}. 
\end{equation*}
The general  $H^s$ bound: $\norm{\psi^{(1)}_p(\cdot,\delta\cdot)}_{H^2(\R)} \lesssim \delta^{-1/2}$ and the
bounds
\[\norm{\partial_X^2\psi^{(1)}_p(x,X)\Big|_{X=\delta x}}_{L^2(\R_x)} \lesssim \delta^{-1/2}\ \ \textrm{and}\ \ 
\norm{\partial_x\partial_X\psi^{(1)}_p(x,X)\Big|_{X=\delta x}}_{L^2(\R_x)} \lesssim \delta^{-1/2}\ ,\]
 are all obtained very similarly. We omit the details.

\medskip

\chapter{Derivation of Key Bounds and Limiting relations in the Lyapunov-Schmidt Reduction}\label{near_freq_limits}

The proof of our main theorem is reduced to the search for solutions, $\mu=\mu(\delta)$,  of the algebraic equation $\mathcal{J}_+[\mu,\delta]=0$; see Section \ref{final-reduction} of Chapter \ref{sec:proof-exists-mode} .
The construction of  $(\mu,\delta)\mapsto\mathcal{J}_+[\mu,\delta]$ requires the following:
\begin{proposition}
Let $0<\tau<1/2$. There exist constants $\delta_0, C_M>0$, such that for all $0<\delta<\delta_0$:
\begin{align}
 \lim_{\delta\rightarrow0}\delta\inner{\widehat{\alpha}_{\star}(\cdot),
\widehat{\mathcal{M}}(\cdot;\delta)}_{L^2(\R)} =
\lim_{\delta\rightarrow0}\sum_{j=1}^3\delta\inner{\widehat{\alpha}_{\star}(\cdot),
\widehat{\mathcal{M}}_j(\cdot;\delta)}_{L^2(\R)}   &=1; \label{limit1} \\
 \lim_{\delta\rightarrow0}\delta\inner{\widehat{\alpha}_{\star}(\cdot),
\widehat{\mathcal{N}}(\cdot;\delta)}_{L^2(\R)} =
\lim_{\delta\rightarrow0}\sum_{j=1}^4\delta\inner{\widehat{\alpha}_{\star}(\cdot),
\widehat{\mathcal{N}}_j(\cdot;\delta)}_{L^2(\R)} &= -\mu_0; \label{limit2} \\
 \abs{\delta\inner{\widehat{\alpha}_{\star}(\cdot), \left(\widehat{\mathcal{D}}^{\delta}-\widehat{\mathcal{D}}\right)
\widehat{\beta}(\cdot;\mu,\delta)}_{L^2(\R)}} &\leq C_M \delta^{1-\tau}; \label{bound1} \\
 \abs{\delta\inner{\widehat{\alpha}_{\star}(\cdot), \widehat{\mathcal{L}}^{\delta}(\mu)
\widehat{\beta}(\cdot;\mu,\delta)}_{L^2(\R)}} &\leq C_M \delta^{\tau}; \label{bound2} \\
 \abs{\delta^2 \mu\inner{\widehat{\alpha}_{\star}(\cdot), \widehat{\beta}(\cdot;\mu,\delta)}_{L^2(\R)}} 
&\leq C_M \delta. \label{bound3}
\end{align}
\end{proposition}
We first recall from \eqref{psi0-soln} that
$
 \psi^{(0)}(x,X) =  \alpha_{\star,1}(X)\Phi_1(x) +\alpha_{\star,2}(X)\Phi_2(x) 
$. It is convenient to introduce the notation:
\begin{equation*}
\alpha_{\star,-}(X) \equiv \alpha_{\star,1}(X)\ \ {\rm and}\ \ \alpha_{\star,+}(X)\equiv \alpha_{\star,2}(X) .
\end{equation*}
It follows that
\begin{align}
 \psi^{(0)}(x,X) 
 &= \alpha_{\star,-}(X)\Phi_-(x;k_\star) +\alpha_{\star,+}(X)\Phi_+(x;k_\star)\ . \label{LS-leading}
\end{align}

\medskip

\nit \textbf{Proof of limit \eqref{limit1}.} There are three terms in $\widehat{\mathcal{M}}(\xi;\delta)$ that
contribute to the limit \eqref{limit1}. These are  the expressions $\widehat{\mathcal{M}}_j(\xi;\delta), j=1,2,3$,
displayed in 
\eqref{M_op_1}-\eqref{M_op_3}. We first check that the contributions from the $j=2$ and $j=3$ terms tend to zero in
the limit as
$\delta\rightarrow0$. We then evaluate the contribution from $\widehat{\mathcal{M}}_1(\xi;\delta)$.\medskip

Applying \eqref{mnbddformula} in Lemma \ref{parseval_bdd} and the bound \eqref{B_bdd}
 for $B(x;\delta)$, we
observe: 
\begin{align*}
&\delta\inner{\widehat{\alpha}_{\star}(\xi),\widehat{\mathcal{M}}_3(\xi;\delta)}_{L^2(\R_\xi)}\nn\\
 &\quad =\delta\ \abs{\inner{\widehat{\alpha}_{\star}(\xi),
\begin{pmatrix}
 \chi\left(\abs{\xi}\leq\delta^{\tau-1}\right) \inner{\Phi_{-}(\cdot,k_{\star}+\delta\xi),
\kappa(\delta\cdot)W_{\oo}(\cdot)B(\cdot;\delta)}_{L^2(\R)} \\
 \chi\left(\abs{\xi}\leq\delta^{\tau-1}\right) \inner{\Phi_{+}(\cdot,k_{\star}+\delta\xi),
\kappa(\delta\cdot)W_{\oo}(\cdot)B(\cdot;\delta)}_{L^2(\R)}
\end{pmatrix} }_{L^2(\R_\xi)} }\\ 
&\quad \lesssim
\delta\norm{\widehat{\alpha}_{\star}}_{L^2(\R)}
\delta^{-1/2}\norm{\kappa(\delta\cdot)W_{\oo}(\cdot)B(\cdot;\delta)}_{L^2(\R)} \\
&\quad  \lesssim \delta^{1/2}\norm{B(\cdot;\delta)}_{L^2(\R)} 
\lesssim \delta^{1/2}\delta^{1/2-\tau} \lesssim \delta^{1-\tau} \ .
\end{align*}
 Similarly,  %
\begin{align*}
&\delta\inner{\widehat{\alpha}_{\star}(\xi),
\widehat{\mathcal{M}}_2(\xi;\delta)}_{L^2(\R_\xi)}\\
&\qquad \le \norm{\widehat{\alpha}_\star}_{L^2(\R)} \norm{\delta^2 \chi\left(\abs{\xi}\leq\delta^{\tau-1}\right)
\begin{pmatrix}
 \inner{\Phi_{-}(\cdot,k_{\star}+\delta\xi),\psi^{(1)}_p(\cdot,\delta\cdot)}_{L^2(\R)} \\
 \inner{\Phi_{+}(\cdot,k_{\star}+\delta\xi),\psi^{(1)}_p(\cdot,\delta\cdot)}_{L^2(\R)}
\end{pmatrix}}_{L^2(\R_\xi)}\nn\\
 &\qquad  \leq \delta^2\delta^{-1/2}\norm{\psi^{(1)}_p(\cdot,\delta\cdot)}_{L^2(\R)}
\lesssim (\delta^2\delta^{-1/2})\delta^{-1/2}= \delta.
\end{align*}  
Therefore $\lim_{\delta\to0}\sum_{j=2}^3\delta\inner{\widehat{\alpha}_{\star}(\xi),
\widehat{\mathcal{M}}_j(\xi;\delta)}_{L^2(\R_\xi)}=0$ and it follows that
\[\lim_{\delta\to0}\delta\inner{\widehat{\alpha}_{\star}(\cdot),
\widehat{\mathcal{M}}(\cdot;\delta)}_{L^2(\R)}=
\lim_{\delta\to0}\delta\inner{\widehat{\alpha}_{\star}(\cdot),
\widehat{\mathcal{M}}_1(\cdot;\delta)}_{L^2(\R)}.\]
To compute this limit we begin with the following:\medskip

\noindent{\bf Claim A:} The components of $\widehat{\mathcal{M}}_1(\xi;\delta)$ satisfy 
\begin{align}
 \label{limita_claim}
&\chi\left(|\xi|\le \delta^{\tau-1}\right)\left[
\delta\inner{\Phi_{\pm}(\cdot,k_{\star}+\delta\xi),\psi^{(0)}(\cdot,\delta\cdot)}_{L^2(\R)} -
2\pi \widehat{\alpha}_{\star,\pm}(\xi) \right]
 =  o(1) \\
 &\qquad\qquad\qquad \text{as} ~ \delta\rightarrow 0\ \textrm{in}\ L^2(\R_\xi). \nn
\end{align} 
Assuming \eqref{limita_claim}, it  follows that
\begin{align*}
 \delta
\inner{\widehat{\alpha}_{\star}(\xi),\widehat{\mathcal{M}}_1(\xi;\delta)}_{L^2(\R_\xi)}
&= 2\pi\ \|\widehat{\alpha}_{\star}\|_{L^2(|\xi|\le\delta^{\tau-1})} + o(1) \\
&= \norm{\alpha_{\star}}_{L^2(\R)}^2 + o(1)=1+o(1)\ .
\end{align*} 

We now prove the claim \eqref{limita_claim}.
Recall from \eqref{LS-leading} that \\ $\psi^{(0)}(x,X)= \alpha_{\star,+}(X)p_+(x,k_\star)+\alpha_{\star,-}(X)p_-(x,k_\star)$ and that \\
$\Phi_\pm(x,k)=e^{ikx}p_\pm(x,k)$.
Therefore, 
\begin{align}
\label{limita_inner}
&\delta\inner{\Phi_{+}(\cdot,k_{\star}+\delta\xi),\psi^{(0)}(\cdot,\delta\cdot)}_{L^2(\R)}  \nn \\
&=\delta\inner{e^{i\xi\delta \cdot}p_{+}(\cdot,k_{\star}+\delta\xi), \alpha_{\star,+}(\delta
\cdot)p_{+}(\cdot,k_{\star})}_{L^2(\R)}\nn\\ 
&\quad + \delta\inner{e^{i\xi\delta \cdot}p_{+}(\cdot,k_{\star}+\delta\xi), \alpha_{\star,-}(\delta
\cdot)p_{-}(\cdot,k_{\star})}_{L^2(\R)},
\end{align} 
We now apply the expansion Lemma \ref{poisson_exp}, based on the Poisson summation, to expand the inner products in
\eqref{limita_inner}.
With the choices:
$f(x,\delta\xi)=p_{+}(x,k_{\star}+\delta\xi)$, $g(x)=p_{\pm}(x,k_{\star})$ and
$\Gamma(x,X)=\alpha_{\star,\pm}(X)$ we obtain, in $L^2_{\rm loc}(d\xi)$,
\begin{align}
&\delta\inner{\Phi_{+}(\cdot,k_{\star}+\delta\xi),\psi^{(0)}(\cdot,\delta\cdot)}_{L^2(\R)}  \nn \\
&= 2\pi\ \sum_{m\in\mathbb{Z}} \int_0^1 e^{imx} \overline{p_{+}(x,k_{\star}+\delta\xi)} \ \times \nn \\
&\qquad \qquad\left[ \widehat{\alpha}_{\star,+}
\left(\frac{2\pi m}{\delta}+\xi\right) p_{+}(x,k_{\star}) +
\widehat{\alpha}_{\star,-} \left(\frac{2\pi m}{\delta}+\xi\right)
p_{-}(x,k_{\star})\right] dx \nn\\
&= 2\pi\ \int_0^1 \overline{p_{+}(x,k_{\star}+\delta\xi)} \left[ \widehat{\alpha}_{\star,+}
\left(\xi\right) p_{+}(x,k_{\star}) +\widehat{\alpha}_{\star,-} \left(\xi\right)
p_{-}(x,k_{\star})\right] dx \nn\\
&\qquad+
2\pi\ \sum_{\abs{m}\geq1} \int_0^1 e^{imx} \overline{p_{+}(x,k_{\star}+\delta\xi)}\ \times  \nn \\
&\qquad\qquad \left[ \widehat{\alpha}_{\star,+}
\left(\frac{2\pi m}{\delta}+\xi\right) p_{+}(x,k_{\star}) +
\widehat{\alpha}_{\star,-} \left(\frac{2\pi m}{\delta}+\xi\right)
p_{-}(x,k_{\star})\right] dx \nn\\
&\equiv \widehat{a}_{+}(\xi;\delta) + \widehat{A}_{+}(\xi;\delta), \label{aA-def}
\end{align} 
where $\widehat{a}_{+}$ refers to the first line and $\widehat{A}_{+}$ to the second. An analogous expression holds for
the inner product in \eqref{limita_inner} with $\Phi_-$ in place of $\Phi_+$.

Since  $p_{\pm}(x,k)$ is uniformly bounded in $(x,k)\in[0,1]\times[0,2\pi]$ we may  pass to the limit as
$\delta\rightarrow0$ in $\widehat{a}_{\pm}(\xi;\delta)$ and using the orthonormality of $p_\pm(x,k)$ we obtain 
\begin{equation*}
\widehat{a}_{\pm}(\xi;\delta) = 2\pi\ \widehat{\alpha}_{\star,\pm}(\xi) + o(1),\ \ {\rm as}\ \ \delta\to0.
\end{equation*} 

With a view toward proving \eqref{limita_claim} we use \eqref{aA-def}. Subtracting $2\pi\widehat{\alpha}_{\star,+}(\xi)$
from  \eqref{aA-def} we obtain:   
\begin{equation}
\textrm{ LHS of}\ \eqref{limita_claim}\ =\ \chi(|\xi|\le\delta^{\tau-1})\left[\ \left(\ \widehat{a}_+(\xi;\delta) -
2\pi\ \widehat{\alpha}_{\star,+}(\xi)\ \right)
 + \widehat{A}_{+}(\xi;\delta)\ \right] \ ,
 \label{aA-diff}
\end{equation}
and we seek to show that this expression tends to zero in $L^2(\R_\xi)$. An analogous argument applies with $-$ in place
of $+$.
By orthonormality of $p_\pm(x,k)$, the  first term within the brackets of \eqref{aA-diff} can be written as
\begin{align}
\widehat{a}_+(\xi;\delta) - 2\pi\ \widehat{\alpha}_{\star,+}(\xi) &= 
2\pi\int_0^1
\left[\overline{p_+(x,k_\star+\delta\xi)}-\overline{p_+(x,k_\star)}\right]
\widehat{\alpha}_{\star,+}(\xi)p_+(x,k_\star)\ dx\nn\\
&+ 2\pi\int_0^1
\left[\overline{p_+(x,k_\star+\delta\xi)}-\overline{p_+(x,k_\star)}\right]
\widehat{\alpha}_{\star,-}(\xi)p_-(x,k_\star)\ dx \ .
\label{e11}\end{align}
By  \eqref{deltapbdd} we have :
\begin{align}
&\chi\left(\abs{\xi}\leq\delta^{\tau-1}\right)\Big| \widehat{a}_+(\xi;\delta) - 2\pi\ \widehat{\alpha}_{\star,+}(\xi) 
\Big|\nn\\
&\qquad \lesssim \chi\left(\abs{\xi}\leq\delta^{\tau-1}\right) \delta|\xi|\ \left( |\alpha_{\star,+}(\xi)|+
|\alpha_{\star,-}(\xi)|\right)\ \lesssim\ \delta^\tau |\alpha_\star(\xi)| \ .
\label{e12}\end{align}
Squaring and integrating over $\R_\xi$ we obtain:
\[\Big\|\chi\left(\abs{\xi}\leq\delta^{\tau-1}\right)\ \left[\widehat{a}_+(\xi;\delta) - 2\pi\
\widehat{\alpha}_{\star,+}(\xi) \right]
\Big\|_{L^2(\R_\xi)}\ \lesssim\delta^\tau .\]

To complete the proof of Claim A we study the $L^2(\R_\xi)-$ limit, as $\delta\to0$ of 
$\chi\left(\abs{\xi}\leq\delta^{\tau-1}\right)\widehat{A}_{\pm}(\xi;\delta)$.  Applying Lemma \ref{I_bdds} with  the
identifications
$f(x,\delta\xi)=p_+(x,k_{\star}+\delta\xi)$, $g(x)=p_{\pm}(x,k_{\star})$ and $\widehat{\Gamma}(x,\xi) =
\widehat{\alpha}_{\star,\pm}(\xi)$, \eqref{I_bdds_1} implies 
\begin{equation}
\norm{\chi\left(\abs{\xi}\leq\delta^{\tau-1}\right)\widehat{A}_{+}(\xi;\delta)}_{L^2(\R)}
\lesssim \delta \norm{\widehat{\alpha}_{\star}}_{L^{2,1}(\R)}
\lesssim \delta .
\label{A+bound}\end{equation} 
A similar bound applies to  $\widehat{A}_{-}(\xi;\delta)$.
 This completes the proof of the Claim A, \eqref{limita_claim}, and therewith the evaluation of the limit, 
\eqref{limit1}. \medskip

\nit \textbf{Proof of limit \eqref{limit2}.} We now turn to evaluating 
\[\lim_{\delta\rightarrow0} \delta\inner{\widehat{\alpha}_{\star}(\cdot),\widehat{\mathcal{N}}(\cdot;\delta)}
=\lim_{\delta\rightarrow0}\sum_{j=1}^4
\delta\inner{\widehat{\alpha}_{\star}(\cdot),\widehat{\mathcal{N}}_j(\cdot;\delta)},\]
where $\widehat{\mathcal{N}}_j(\xi;\delta),\ j=1,\dots,4$, are displayed in \eqref{N_op_1}-\eqref{N_op_4}.
 We first show that the contributions from the $j=3$ and $j=4$ terms vanish in the
limit as $\delta$ tends to zero.

Using the expression for \eqref{N_op_4} in $\widehat{\mathcal{N}}_4$, the bound   \eqref{mnbddformula} (Lemma
\ref{parseval_bdd}) and the bound \eqref{B_bdd} for
$C(x;\delta)$, we obtain:
\begin{align*}
&\delta\inner{\widehat{\alpha}_{\star}(\xi),\widehat{\mathcal{N}}_4(\xi;\delta)}_{L^2(\R_\xi)}\nn\\
 &\  =\delta\ \abs{\inner{\widehat{\alpha}_{\star}(\xi),
\begin{pmatrix}
 \chi\left(\abs{\xi}\leq\delta^{\tau-1}\right) \inner{\Phi_{-}(\cdot,k_{\star}+\delta\xi),
\kappa(\delta\cdot)W_{\oo}(\cdot)C(\cdot;\delta)}_{L^2(\R)} \\
 \chi\left(\abs{\xi}\leq\delta^{\tau-1}\right) \inner{\Phi_{+}(\cdot,k_{\star}+\delta\xi),
\kappa(\delta\cdot)W_{\oo}(\cdot)C(\cdot;\delta)}_{L^2(\R)}
\end{pmatrix} }_{L^2(\R_\xi)} }\\ 
&\quad  \lesssim
\delta\norm{\widehat{\alpha}_{\star}}_{L^2(\R)}
\delta^{-1/2}\norm{\kappa(\delta\cdot)W_{\oo}(\cdot)C(\cdot;\delta)}_{L^2(\R)} \\
&\quad \lesssim \delta^{1/2}\norm{C(\cdot;\delta)}_{L^2(\R)} 
 \lesssim \delta^{1/2}\delta^{1/2-\tau}=\delta^{1-\tau} \ .
\end{align*}
Similarly, the contribution to the limit \eqref{limit2} from
$\delta\inner{\widehat{\alpha}_{\star}(\xi),\widehat{\mathcal{N}}_3(\xi;\delta)}_{L^2(\R_\xi)}$ vanishes, by Lemma
\ref{lemma:psi_bounds}, since:
\begin{align*} 
& \delta\ \norm{\delta \chi\left(\abs{\xi}\leq\delta^{\tau-1}\right)
\begin{pmatrix}
 \inner{\Phi_{-}(\cdot,k_{\star}+\delta\xi), \partial_X^2\psi^{(1)}_p(x,X)\Big|_{X=\delta x} }_{L^2(\R)} \\
 \inner{\Phi_{+}(\cdot,k_{\star}+\delta\xi), \partial_X^2\psi^{(1)}_p(x,X)\Big|_{X=\delta x} }_{L^2(\R)}
\end{pmatrix}}_{L^2(\R)}\\
&\qquad \lesssim \delta^2\delta^{-1/2}\|\D_X^2\psi_p^{(1)}\|_{L^2}\lesssim \delta^2\delta^{-1/2}\delta^{-1/2}=\delta.
\end{align*} 
Therefore 
\[\lim_{\delta\rightarrow0} \delta\inner{\widehat{\alpha}_{\star}(\cdot),\widehat{\mathcal{N}}(\cdot;\delta)}
=\lim_{\delta\rightarrow0}\sum_{j=1}^2
\delta\inner{\widehat{\alpha}_{\star}(\cdot),\widehat{\mathcal{N}}_j(\cdot;\delta)}.\]

\noindent{\bf Claim B:} In  $L^2(|\xi|\le\delta^{\tau-1};d\xi)$,  we have the following as $\delta\to0$:
\begin{enumerate}
\item The component inner products  of  the vector $\widehat{\mathcal{N}}_1(\xi;\delta)$
are given by:\medskip

\begin{align}
&\delta\inner{\Phi_{\pm}(\cdot,k_{\star}+\delta\xi),
\left(2\partial_x\partial_X-\kappa(X)W_{\oo}(x)\right) \psi^{(1)}_p(x,X)\Big|_{X=\delta x}}_{L^2(\R_x)} \nn \\ 
 &= 2\pi\ \int_0^1\overline{\Phi_{\pm}(x,k_{\star})}
\mathcal{F}_X\left[\left(2\partial_x\partial_X-\kappa(X)W_{\oo}(x)\right)\psi^{(1)}_p(x,X) 
\right](\xi) dx + o(1) . \label{limitb_claim}
\end{align} 
\item The component inner products of  the vector 
$\widehat{\mathcal{N}}_2(\xi;\delta)$ are given by:
\begin{align}
\delta
& \inner{\Phi_{\pm}(x,k_{\star}+\delta\xi),
\partial_X^2\psi^{(0)}(x,X)\Big|_{X=\delta x}}_{L^2(\R_x)} \nn \\
&~~~ = 2\pi\ \int_0^1\overline{\Phi_{\pm}(x,k_{\star})}
\mathcal{F}_X\left[\partial_X^2\psi^{(0)}(x,X)\right](\xi) dx + o(1) .
\label{limitb_claim2}
\end{align}
\end{enumerate}
Below we give the proof of Claim B, \eqref{limitb_claim}-\eqref{limitb_claim2}, but we first show how to use these
assertions,
together with the formal multi-scale asymptotic calculation of 
Section \ref{subsec:multiscale_analysis} of Chapter \ref{sec:bound_states},  to evaluate the limit \eqref{limit2}. 
\medskip

 Using the asymptotics 
 \eqref{limitb_claim}-\eqref{limitb_claim2}, and recalling that  $\widehat{\mathcal{N}}(\xi;\delta)$ is localized to 
the set $|\xi|\le\delta^{\tau-1}$, we obtain
{\small
\begin{align*}
& \lim_{\delta\rightarrow0} \delta \inner{\widehat{\alpha}_{\star}(\xi),
\widehat{\mathcal{N}}(\xi;\delta)}_{L^2(\R_{\xi})}\ =\ 
\lim_{\delta\rightarrow0} \int_{|\xi|\le\delta^{\tau-1}}\ d\xi\ \  \overline{\widehat{\alpha}_{\star}(\xi)}^T
\sum_{j=1}^2\delta\widehat{\mathcal{N}}_j(\xi;\delta) \\
&= \lim_{\delta\rightarrow0}\ \int_{|\xi|\le\delta^{\tau-1}}\ d\xi\ \  \overline{\widehat{\alpha}_{\star}(\xi)}^T \ \times \\
&\qquad \begin{pmatrix}
 \inner{\delta \Phi_{-}(x,k_{\star}+\delta\xi),
\left(2\partial_x\partial_X-\kappa(X)W_{\oo}(x)\right) \psi^{(1)}_p(x,X)+\partial_X^2\psi^{(0)}(x,X)\Big|_{X=\delta
x}}_{L^2(\R_x)} \\
\inner{\delta \Phi_{+}(x,k_{\star}+\delta\xi),
\left(2\partial_x\partial_X-\kappa(X)W_{\oo}(x)\right) \psi^{(1)}_p(x,X)+\partial_X^2\psi^{(0)}(x,X)\Big|_{X=\delta
x}}_{L^2(\R_x)}
\end{pmatrix} \\
&= 2\pi\ \int_{L^2(\R_{\xi})} \ d\xi\ \ \overline{\widehat{\alpha}_{\star}(\xi)}^T \ \times \\
&\qquad \begin{pmatrix}
 \inner{\Phi_{-}(x,k_{\star}),\mathcal{F}_X\left[
\left(2\partial_x\partial_X-\kappa(X)W_{\oo}(x)\right)
\psi^{(1)}_p(x,X)+\partial_X^2\psi^{(0)}(x,X)\right](\xi)}_{L^2_x[0,1]}\\
\inner{\Phi_{+}(x,k_{\star}),\mathcal{F}_X\left[
\left(2\partial_x\partial_X-\kappa(X)W_{\oo}(x)\right)
\psi^{(1)}_p(x,X)+\partial_X^2\psi^{(0)}(x,X)\right](\xi)}_{L^2_x[0,1]}
\end{pmatrix} \\
& = 2\pi\ \inner{\mathcal{F}_X[\alpha_{\star}](\xi),
\begin{pmatrix}
 \inner{\Phi_{-}(x,k_{\star}),\mathcal{F}_X[G^{(2)}](x,\xi)}_{L^2_x[0,1]} \\
 \inner{\Phi_{+}(x,k_{\star}),\mathcal{F}_X[G^{(2)}](x,\xi)}_{L^2_x[0,1]}
\end{pmatrix}}_{L^2(\R_{\xi})} \ ,
\end{align*} 
}
where $G^{(2)}(x,X)$ is as defined in \eqref{G2def}.
The last equality bridges the current quasi-momentum domain calculation to the formal multi-scale calculation of Section \ref{subsec:multiscale_analysis}
 of Chapter \ref{sec:bound_states}. \medskip

We now evaluate the limiting inner product, by interchanging the order of integration and an application of the
Plancherel Theorem:
\begin{align}
&2\pi\ \inner{\mathcal{F}_X[\alpha_{\star}](\xi),
\begin{pmatrix}
 \inner{\Phi_{-}(x,k_{\star}),\mathcal{F}_X[G^{(2)}](x,\xi)}_{L^2_x[0,1]} \\
 \inner{\Phi_{+}(x,k_{\star}),\mathcal{F}_X[G^{(2)}](x,\xi)}_{L^2_x[0,1]}
\end{pmatrix}}_{L^2(\R_{\xi})}\nn\\
&\quad = 2\pi\ \int_{\R_\xi} \sum_{j=\pm} \overline{\mathcal{F}_X[\alpha_{\star,j}](\xi)}\
\inner{\Phi_j(x;k_\star),\mathcal{F}_X[G^{(2)}](x,\xi)}_{L^2_x[0,1]}\ d\xi\nn\\
&\quad = 2\pi\ \int_0^1 \sum_{j=\pm}\ \overline{\Phi_j(x;k_\star)} \
\inner{\mathcal{F}_X[\alpha_{\star,j}](\xi),\mathcal{F}_X[G^{(2)}](x,\xi)}_{L^2(\R_\xi)}\ dx\nn\\
&\quad =  \int_0^1 \sum_{j=\pm}\ \overline{\Phi_j(x;k_\star)} \ \inner{\alpha_{\star,j}(X),G^{(2)}(x,X)}_{L^2(\R_X)}\ dx\nn\\
&\quad =  \int_{\R_X} \sum_{j=\pm}\ \overline{\alpha_{\star,j}(X)} \ \inner{\Phi_j(x;k_\star),G^{(2)}(x,X)}_{L_x^2([0,1])}\
dX\nn\\ 
&\quad =  \inner{\alpha_{\star}(X),
\begin{pmatrix}
 \inner{\Phi_-(\cdot;k_\star),G^{(2)}(\cdot,X)}_{L^2[0,1]} \\
 \inner{\Phi_+(\cdot;k_\star),G^{(2)}(\cdot,X)}_{L^2[0,1]}
\end{pmatrix}}_{L^2(\R_X)} \nn \\
&\quad =\inner{\alpha_{\star}(X),\mathcal{G}^{(2)}(X)}_{L^2(\R_X)}\nn\\
&\quad = -E^{(2)} =_{\rm def} -\mu_0,
\end{align}
where $E^{(2)}\equiv\mu_0$ is given in  \eqref{solvability_cond_E2}
 and $\mathcal{G}^{(2)}(X)$ is defined in \eqref{ipG}; see also \eqref{mu0-def}. That $\mathcal{G}^{(2)}(X)$ and $E^{(2)}$ do indeed match the multiscale definitions follows from \eqref{LS-leading} and \eqref{G2def}.
\medskip

\noindent{\bf Proof of Claim B,  \eqref{limitb_claim} and \eqref{limitb_claim2}:}\  We begin with \eqref{limitb_claim2}
which is proved in a  manner similar to the proof of \eqref{limita_claim}. 
By the definition of $\psi^{(0)}(x,X)$ and since $\Phi_\pm(x,k)=e^{ikx}p_\pm(x,k)$, the left hand side of
\eqref{limitb_claim2}
may be written as
\begin{align}
\label{limitb2_inner}
&\delta\inner{\Phi_{+}(\cdot,k_{\star}+\delta\xi),\partial_X^2\psi^{(0)}(\cdot,\delta\cdot)}_{L^2(\R)}\nn\\
 &=
\delta\inner{e^{i\xi\delta \cdot}p_{+}(\cdot,k_{\star}+\delta\xi),
\partial_X^2\alpha_{\star,+}(X)p_{+}(x,k_{\star})\Big|_{X=\delta x}}_{L^2(\R_x)} \nn \\
&~~~+ \delta\inner{e^{i\xi\delta \cdot}p_{+}(\cdot,k_{\star}+\delta\xi),
\partial_X^2\alpha_{\star,-}(X)p_{-}(x,k_{\star})\Big|_{X=\delta x}}_{L^2(\R_x)}.
\end{align}
We next apply  Lemma \ref{poisson_exp} to \eqref{limitb2_inner} with the identifications:
$f(x,\delta\xi)=p_{+}(x,k_{\star}+\delta\xi)$, $g(x)=p_{\pm}(x,k_{\star})$ and
$\Gamma(x,\delta x)=\partial_X^2\alpha_{\star,\pm}(X)|_{X=\delta x}$. 
In $L^2_{\rm loc}(d\xi)$ we have (recall that $\mathcal{F}_X$ denotes the Fourier transform with respect to the variable
$X$):
\begin{align}
&\delta\inner{\Phi_{+}(\cdot,k_{\star}+\delta\xi),\partial_X^2\psi^{(0)}(\cdot,\delta\cdot)}_{L^2(\R)} 
\equiv \widehat{d}_{+}(\xi;\delta) + \widehat{D}_{+}(\xi;\delta),
\label{cb-1}\end{align}
where
\begin{align}
\widehat{d}_{+}(\xi;\delta) &= 2\pi \int_0^1 \overline{p_{+}(x,k_{\star}+\delta\xi)} \ \times \label{d+def} \\
&\qquad \Big[ \mathcal{F}_X\left[\partial_X^2\alpha_{\star,+}\right](\xi)p_+(x;k_\star)+
\mathcal{F}_X\left[\partial_X^2\alpha_{\star,-}\right](\xi)p_-(x;k_\star) \Big] dx , \nn \\
\widehat{D}_{+}(\xi;\delta) &= 2\pi \sum_{\abs{m}\geq1} \int_0^1e^{imx}\ \overline{p_{+}(x,k_{\star}+\delta\xi)}
e^{-ik_\star x} \mathcal{F}_X\left[\partial_X^2\psi^{(0)}(x,X) \right]\left(\frac{2\pi m}{\delta}+\xi\right) .
\label{D+def}
\end{align}
The terms $\widehat{d}_{+}(\xi;\delta)$ and $\widehat{D}_{+}(\xi;\delta)$ are, respectively, the $m=0$ and $|m|\ge1$
terms in the Poisson sum.  An analogous expression holds for when $\Phi_+$ is replaced by $\Phi_-$.

The expression for $\widehat{D}_{+}$ is of the same type as $\widehat{A}_{+}$, displayed in \eqref{aA-def}. Thus, in a
manner similar to the proof of the bound \eqref{A+bound} we can apply Lemma \ref{I_bdds} to obtain
 \[ \lim_{\delta\to0}\ \Big\|\chi(|\xi|\le\delta^{\tau-1})\widehat{D}_{\pm}(\xi;\delta)\Big\|_{L^2(\R_\xi)}\ =\
o(\delta)\ .\]  
So to prove claim \eqref{limitb_claim2} of Claim B it suffices to show that as $\delta\to0$:
\begin{equation*}
\Big\| \chi(|\xi|\le\delta^{\tau-1})\Big[\ \widehat{d}_{\pm}(\xi;\delta) -
2\pi\ \int_0^1 \overline{\Phi_\pm(x,k_{\star})} 
\mathcal{F}_X\left[\partial_X^2\psi^{(0)}(x,X)\right]\left(\xi\right) dx\ \Big] \Big\|_{L^2(\R_\xi)}\to0 .
\end{equation*} 
To see this, we have from \eqref{d+def}:
\begin{align*}
&\widehat{d}_{+}(\xi;\delta) \nn \\
&\quad = 2\pi \int_0^1 \overline{p_{+}(x,k_{\star})}
\Big[ \mathcal{F}_X\left[\partial_X^2\alpha_{\star,+}\right](\xi)p_+(x;k_\star)+
\mathcal{F}_X\left[\partial_X^2\alpha_{\star,-}\right](\xi)p_-(x;k_\star) \Big] dx \nn\\
&\qquad + 2\pi \int_0^1 \left[ \overline{p_{+}(x,k_{\star}+\delta\xi)}-\overline{p_{+}(x,k_{\star})} \right] \ \times \nn \\
&\qquad \qquad
\Big[ \mathcal{F}_X\left[\partial_X^2\alpha_{\star,+}\right](\xi)p_+(x;k_\star)+
\mathcal{F}_X\left[\partial_X^2\alpha_{\star,-}\right](\xi)p_-(x;k_\star) \Big] dx \nn\\
&\quad = 2\pi\ \int_0^1 \overline{\Phi_+(x,k_{\star})} 
\mathcal{F}_X\left[\partial_X^2\psi^{(0)}(x,X)\right]\left(\xi\right) dx + \widehat{e}_+(\xi,\delta),
\end{align*}
where (as in \eqref{e11}-\eqref{e12}) we have 
\begin{align*}
\Big\|\ \chi(|\xi|\le\delta^{\tau-1}) \widehat{e}_+(\xi;\delta)\ \Big\|_{L^2(\R_\xi)}\ &\lesssim\ 
\Big\|\chi(|\xi|\le\delta^{\tau-1})\delta\xi\cdot \xi^2\widehat{\alpha_\star}(\xi)\Big\|_{L^2(\R_\xi)} \\ 
& \le\delta^\tau
\Big\| \xi^2\widehat{\alpha_\star}(\xi)\Big\|_{L^2(\R_\xi)} ,
\end{align*}
which tends to zero as $\delta\to0$. Similar results hold for $\Phi_+$ replaced by $\Phi_-$ in \eqref{cb-1}. This completes the proof of 
\eqref{limitb_claim2} of Claim B.
\medskip

We conclude the proof of Claim B by now verifying the limit \eqref{limitb_claim}.
Recall from \eqref{main_result_psi1}, \eqref{Fdefn} 
and the relation $\Phi_j(x,k_\star)=e^{ik_\star x}p_j(x,k_\star)$, $j=\pm$, that
\begin{align}
 \psi^{(1)}_p(x,X) 
 &= R(E_{\star})\circ  e^{ik_{\star}x}
 \mathcal{A}\ P^{(0)}(x,X)\ \equiv e^{ik_\star x}\  R_{k_\star}(E_\star)
 \mathcal{A}\ P^{(0)}(x,X) .
\label{psi1p-star}
\end{align} 
Here, 
\begin{align}
P^{(0)}(x,X)&= e^{-ik_\star x}\psi^{(0)}(x,X)=\sum_{j\in\{+,-\}}  \alpha_{\star,j}(X)p_j(x,k_{\star}) \
,\label{psi0-def1}\\
\mathcal{A}&=\mathcal{A}(x,X,\D_x,\D_X)\ \equiv\ 2 \left(\partial_x +
ik_{\star}\right)\partial_X -\kappa(X)W_{\oo}(x) .
\label{Adef} \end{align}
The operator $R_{k_\star}(E_{\star})$ is given by 
\begin{align}
R_{k_\star}(E_{\star}) &= e^{-ik_\star x}\ R(E_{\star})\ e^{ik_{\star}x}=(H_{V_\ee}(k_\star)-E_\star)^{-1}\label{Rkstar}
\end{align}
and is bounded from $\Pi_\perp L^2_{\rm per}(dx)$ to $ H^2_{\rm per}(dx)$. Here, 
$\Pi_\perp$  projects onto the orthogonal complement of the nullspace of $H_{V_\ee}(k_\star)-E_\star$.
Note that  $\alpha_\star(X)$ was chosen to ensure solvability, which is equivalent to
$x\mapsto\mathcal{A}P^{(0)}(x,X)\in\Pi_\perp L^2_{\rm per}([0,1];dx)$; see \eqref{G1orthog}-\eqref{dirac_eqn1}.

 Next, applying the operator $\left(2\partial_x\partial_X-\kappa(X)W_{\oo}(x)\right)$ to $\psi^{(1)}_p$ we obtain:
\begin{align}
 &\left(2\partial_x\partial_X-\kappa(X)W_{\oo}(x)\right)\psi^{(1)}_p(x,X) =e^{ik_{\star}x} \mathcal{K}(x,X) 
,\ \ \textrm{where } \label{nonzero_limit} \\
& \mathcal{K}(x,X) \equiv\ 
 \Big(\mathcal{A}\circ
R_{k_\star}(E_{\star})\ \circ\mathcal{A}\Big)\ P^{(0)}(x,X) \ .
\label{K-def}\end{align}
We next substitute \eqref{nonzero_limit} into the inner product on the left hand side of \eqref{limitb_claim}
 and find 
\begin{align}
 &\delta\inner{\Phi_{+}(x,k_{\star}+\delta\xi),\left(2\partial_x\partial_X-\kappa(X)W_{\oo}(x)\right)
\psi^{(1)}_p(x,X)\Big|_{X=\delta x}}_{L^2(\R_x)}\nn\\
 &\qquad = 
\delta\left\langle e^{i\delta\xi x} p_{+}(x,k_{\star}+\delta\xi),
\mathcal{K}(x,\delta x) \right\rangle_{L^2(\R_x)} \ . \label{ip-limitb}
\end{align}
We  study the limit of \eqref{ip-limitb} as $\delta\to0$ by applying  Lemma
\ref{poisson_exp} with the choices $f(x,\delta \xi)=p_+(x,k_\star+\delta\xi)$, $g(x)\equiv1$ and
$\Gamma(x,X)=\mathcal{K}(x,X)$. 
To verify the hypotheses of Lemma \ref{poisson_exp}  we shall check that
\begin{equation}
\|\mathcal{K}\|_{\mathbb{H}^{a,b}}^2\equiv \sum_{j=0}^a\int_0^1 \|\D_y^j\mathcal{K}(x,X)\|_{H^b(\R_X)}^2\ dx < \infty\ ,
\ {\rm for} \  a=2,\ b=2\ .\label{moreGamma}
\end{equation}
The condition \eqref{moreGamma}, which is more stringent  than 
the condition \eqref{Gamma-conditions2} of Lemma \ref{poisson_exp},  will be used below.
 Denote by $\mathbb{H}^{a,b}$ the Hilbert space of functions obtained via completion of the linear space of functions
$x\mapsto\Gamma(x,X)$ which are $C^\infty_{\rm per}([0,1]_x)$ with values in $\mathcal{S}(\R_X)$;
 compare with \eqref{Gamma-conditions2a}.
Note that $\mathcal{A}:\mathbb{H}^{a,b}\to \mathbb{H}^{a-1,b-1}$ is bounded and 
$R_{k_\star}(E_\star)\Pi_\perp:\mathbb{H}^{a,b}\to\mathbb{H}^{a+2,b}$ is bounded. 

Therefore, 
\begin{align}
\|\mathcal{K}(x,X)\|_{\mathbb{H}^{2,2}}^2 &=
\Big\|\Big(\mathcal{A}\circ
R_{k_\star}(E_{\star})\ \circ\mathcal{A}\Big)\ P^{(0)}(x,X)  \Big\|_{\mathbb{H}^{2,2}}^2\nn\\
&\lesssim \Big\|\left( R_{k_\star}(E_{\star})\ \circ\mathcal{A} \right)\ P^{(0)}(x,X) \Big\|_{\mathbb{H}^{3,3}}^2\nn\\
&\lesssim \Big\| \mathcal{A} P^{(0)}(x,X) \Big\|_{\mathbb{H}^{1,3}}^2 
\lesssim \Big\| P^{(0)}(x,X) \Big\|_{\mathbb{H}^{2,4}}^2\le C, \label{Knorm}
\end{align}
where $C$ is an order one constant which depends on $p_j,\ j=\pm$, and $\alpha_\star(X)$; see \eqref{psi0-def1}.
\medskip

Applying Lemma \ref{poisson_exp} to \eqref{ip-limitb},
the inner product on the left hand side of \eqref{limitb_claim}, we obtain
\begin{align}
\eqref{ip-limitb}  \equiv\ 
\widehat{f}_{+}(\xi;\delta) + \widehat{F}_{+}(\xi;\delta),
 \label{ip+}
\end{align} 
where
 \begin{align}
 \widehat{f}_{+}(\xi;\delta) = 2\pi\ \int_0^1 \overline{p_{+}(x,k_{\star}+\delta\xi)}
 \mathcal{F}_X[\mathcal{K}(x,X)](\xi) dx\ ,
\label{fhat+}\end{align} and
\begin{align}
\widehat{F}_{+}(\xi;\delta) = 2\pi \sum_{\abs{m}\geq1} \int_0^1 e^{imx} \overline{p_{+}(x,k_{\star}+\delta\xi)}
\mathcal{F}_X[\mathcal{K}(x,X)]\left(\frac{2\pi m}{\delta}+\xi\right) dx \ .
\label{Fhat+}\end{align}
Analogous expressions can be obtained for the inner product \eqref{ip+}, with $\Phi_+$ replaced by $\Phi_-$, leading to 
expressions:  $ \widehat{f}_{-}(\xi;\delta)$ and $ \widehat{F}_{-}(\xi;\delta)$.
\medskip

We next apply the bounds Lemma \ref{I_bdds} to show
 $\chi\left(\abs{\xi}\leq\delta^{\tau-1}\right)\widehat{F}_{+}(\xi;\delta)\to0$ in $L^2(d\xi)$. The hypotheses of Lemma
\ref{I_bdds} require, with 
 $\widehat{\Gamma}(x,\zeta)=\mathcal{F}_X[\mathcal{K}(x,X)](\zeta)\equiv \widehat{\mathcal{K}}(x,\zeta)$, the bound 
\eqref{Gammahat-bound}:
\begin{equation}
   \left\|\ \sup_{0\le x\le1} |\widehat{\mathcal{K}}(x,\zeta)|\ \right\|_{L^{2,1}(\R_\zeta)}\ <\infty \ .
   \label{Khat-bound}\end{equation}
  To prove \eqref{Khat-bound}, first   note 
 \begin{align*}
\left| \sup_{0\le x\le1}   |\widehat{\mathcal{K}}(x,\zeta)| \right|^2
\le\ \int_0^1 |\widehat{\mathcal{K}}(x,\zeta)|^2+
|\D_x\widehat{\mathcal{K}}(x,\zeta)|^2\ dx \ .
  \end{align*}
 Multiplication by $1+\zeta^2$ and  integration $d\zeta$ over $\R$ and applying the Plancherel identity, yields
  \begin{align*}
   &\left\|\ \sup_{0\le x\le1} |\widehat{\mathcal{K}}(x,\zeta)|\ \right\|^2_{L^{2,1}(\R_\zeta)} \\
   &\qquad \le  \int_\R (1+|\zeta|^2)\ d\zeta\  \int_0^1 |\widehat{\mathcal{K}}(x,\zeta)|^2+
|\D_x\widehat{\mathcal{K}}(x,\zeta)|^2\ dx\nn\\
&\qquad= \int_0^1 dx\ \int_\R (1+|\zeta|^2)\ \left( |\widehat{\mathcal{K}}(x,\zeta)|^2+
|\D_x\widehat{\mathcal{K}}(x,\zeta)|^2 \right)\  d\zeta\\
&\qquad= \int_0^1 dx\ \int_\R  |(I-\D_X^2)^{\frac12}\mathcal{K}(x,X)|^2+
|\D_x (I-\D_X^2)^{\frac12}\mathcal{K}(x,X)|^2\  dX\nn\\
&\qquad= \sum_{j=0}^1 \int_0^1 \|\D_x^j\mathcal{K}(x,X)\|_{H^1(\R_X)}^2\ dx\ =\ \|\mathcal{K}\|^2_{\mathbb{H}^{1,1}} \ .
\end{align*}
 The norm $\|\mathcal{K}\|_{\mathbb{H}^{1,1}}$ is finite by  \eqref{Knorm} and therefore \eqref{Khat-bound}  holds.
 
Now apply Lemma \ref{I_bdds} and we obtain, using \eqref{K-def},  that
\begin{align*}
\norm{\chi\left(\abs{\xi}\leq\delta^{\tau-1}\right)\widehat{F}_{+}(\xi;\delta)}_{L^2(\R_\xi)} &\lesssim 
\delta \norm{\sup_{0\le x\le1} \widehat{\mathcal{K}}(x,\zeta)}_{L^{2,1}(\R_\zeta)} \\
&\lesssim \delta \norm{\left(\abs{\zeta}^2+1\right)
\widehat{\alpha}_{\star}\left(\zeta\right)}_{L^{2,1}(\R_\zeta)} \\
&\approx \delta
\norm{\widehat{\alpha}_{\star}}_{L^{2,3}(\R)} \lesssim \delta .
\end{align*}
The same bound holds with $ \widehat{F}_{+}$ replaced by $\widehat{F}_{-}$. 
Therefore, 
\begin{align*}
&\Big|\ \inner{\widehat{\alpha}_{\star}(\xi),
\begin{pmatrix}
 \chi\left(\abs{\xi}\leq\delta^{\tau-1}\right)\widehat{F}_{-}(\xi;\delta) \\
 \chi\left(\abs{\xi}\leq\delta^{\tau-1}\right)\widehat{F}_{+}(\xi;\delta)
\end{pmatrix}}_{L^2(\R_\xi)}\ \Big| \\
&\qquad \lesssim\ 
\norm{\widehat{\alpha}_{\star}}_{L^2(\R)} \norm{
\begin{pmatrix}
 \chi\left(\abs{\xi}\leq\delta^{\tau-1}\right)\widehat{F}_{-}(\xi;\delta) \\
 \chi\left(\abs{\xi}\leq\delta^{\tau-1}\right)\widehat{F}_{+}(\xi;\delta)
\end{pmatrix}}_{L^2(\R_\xi)}
\lesssim \delta.
\end{align*}
\medskip

\nit To complete the proof of Claim B,  we study the limits of 
$\lim_{\delta\to0}\widehat{f}_\pm(\xi;\delta)$. 

\begin{align}
\widehat{f}_{+}(\xi;\delta) &= 2\pi\ \int_0^1 \overline{p_{+}(x,k_{\star})}
 \widehat{\mathcal{K}}(x,\xi) dx\nn\\
 &\quad + 2\pi\ \int_0^1 \Big[\overline{p_{+}(x,k_{\star}+\delta\xi)}-\overline{p_{+}(x,k_{\star})}\Big]
 \widehat{\mathcal{K}}(x,\xi) dx \ . \label{e38}
 \end{align}
 Multiplying \eqref{e38} by $\chi\left(|\xi|\le\delta^{\tau-1}\right)$ and estimating in $L^2(\R_\xi)$ we have: 
 \begin{align}
&\Big\|\chi\left(|\xi|\le\delta^{\tau-1}\right)\Big[ \widehat{f}_{+}(\xi;\delta) - 2\pi\ \int_0^1
\overline{p_{+}(x,k_{\star})}
 \widehat{\mathcal{K}}(x,\xi) dx\Big]\Big\|_{L^2(\R_\xi)}\nn\\
 &\quad \lesssim 
 \Big\|\chi\left(|\xi|\le\delta^{\tau-1}\right) (\delta\xi) \sup_{0\le x\le1}|
\widehat{\mathcal{K}}(x,\xi)|\Big\|_{L^2(\R_\xi)}\nn\\ &\quad \lesssim \delta^\tau\ \Big\|\sup_{0\le x\le1}|
\widehat{\mathcal{K}}(x,\xi)|\Big\|_{L^2(\R_\xi)}=o(\delta^\tau),\ \ \delta\to0
 \end{align}
using \eqref{deltapbdd} and \eqref{Khat-bound}. Finally observe from \eqref{nonzero_limit} that
\begin{align}
  \label{implicit16}
 & 2\pi \int_0^1 \overline{p_{+}(x,k_{\star})} \mathcal{F}_X[\mathcal{K}](x,\xi) dx  = 2\pi \int_0^1
\overline{\Phi_{+}(x,k_{\star})} \mathcal{F}_X[e^{ik_\star x}\mathcal{K}](x,\xi) dx \nn \\
 &= 2\pi \ \int_0^1\overline{\Phi_+(x,k_{\star})}
\mathcal{F}_X\left[\left(2\partial_x\partial_X-\kappa(X)W_{\oo}(x)\right) \psi^{(1)}_p(x,X)  \right](\xi) dx \
.
\end{align} 
The last equality holds by \eqref{nonzero_limit}. This proves \eqref{limitb_claim}
and therewith 
 Claim B .
\medskip

\nit \textbf{Proof of bound \eqref{bound1}.} To prove \eqref{bound1}, we use that
\[(\widehat{\mathcal{D}}^{\delta}-\widehat{\mathcal{D}})=\vartheta_\sharp\chi\left(|\xi|>\delta^{\tau-1}
\right)\sigma_1\widehat{\kappa\beta}\] (see \eqref{dirac-diff}) acts only on high frequencies components and that
$\widehat{\alpha}_{\star}$, being a Schwartz class function, is dominated by low frequencies.
 We obtain, using the Cauchy-Schwarz inequality, Plancherel identity and the uniform boundedness 
 of $\kappa(X)$:
\begin{align*}
& \abs{\delta\inner{\widehat{\alpha}_{\star}(\cdot),\left(\widehat{\mathcal{D}}^{\delta}-\widehat{\mathcal{D}}\right)
\widehat{\beta}(\cdot;\mu,\delta)}_{L^2(\R)}} \\
&\qquad \lesssim \delta
\left(\int_{\abs{\xi}>\delta^{\tau-1}}\abs{\widehat{\alpha}_{\star}(\xi)}^2 d\xi\right)^{1/2}
\left(\int_{\abs{\xi}>\delta^{\tau-1}}\abs{\widehat{\kappa\beta}(\xi;\mu,\delta)}^2 d\xi\right)^{1/2} \\
&\qquad\lesssim \delta \left(\int_{\abs{\xi}>\delta^{\tau-1}}\frac{\abs{\xi}^{2}}{\abs{\xi}^{2}}
\abs{\widehat{\alpha}_{\star}(\xi)}^2 d\xi\right)^{1/2}
\norm{\kappa\beta(\cdot;\mu,\delta)}_{L^2(\R)} \\
&\qquad\lesssim \delta\frac{1}{\delta^{(\tau-1)}}
\left(\int_{\abs{\xi}>\delta^{\tau-1}}\abs{\xi}^{2}\abs{\widehat{\alpha}_{\star}(\xi)}^2
d\xi\right)^{1/2} \norm{\widehat{\beta}(\cdot;\mu,\delta)}_{L^{2,1}(\R)} \\
&\qquad\lesssim \delta^{2-\tau} \norm{\xi\widehat{\alpha}_{\star}(\xi)}_{L^2(\R_\xi)}
\delta^{-1} \lesssim \delta^{1-\tau}.
\end{align*}
Here, we have used the bound $\|\beta(\cdot,\mu,\delta)\|_{L^{2,1}}\lesssim\delta^{-1}$ in \eqref{betabound2a}.
\medskip

\nit \textbf{Proof of bound \eqref{bound2}.} The bound \eqref{bound2} follows from the Cauchy-Schwarz
inequality, bound \eqref{llemma} (Proposition \ref{lemma12}) and  \eqref{betabound2a}:
\begin{align*}
 \abs{\delta\inner{\widehat{\alpha}_{\star}(\cdot),\widehat{\mathcal{L}}^{\delta}(\mu)
\widehat{\beta}(\cdot;\mu,\delta)}_{L^2(\R)}} &\leq
\delta\norm{\widehat{\alpha}_{\star}}_{L^2(\R)}
\norm{\widehat{\mathcal{L}}^{\delta}(\mu) \widehat{\beta}(\cdot;\mu,\delta)}_{L^2(\R)} \\
& \lesssim \delta^{1}\delta^{\tau}\norm{\widehat{\beta}(\cdot;\mu,\delta)}_{L^{2,1}(\R)}\ \lesssim
\delta^{1+\tau}\delta^{-1}  = \delta^{\tau} .
\end{align*}
\medskip

\nit \textbf{Proof of bound \eqref{bound3}.} The bound \eqref{bound3} follows directly from the Cauchy-Schwarz
inequality and \eqref{betabound2a}:
\begin{align*}
 \abs{\delta^2 \mu\inner{\widehat{\alpha}_{\star}(\cdot),\widehat{\beta}(\cdot;\mu,\delta)}_{L^2(\R)}} 
 &\leq \delta^2\abs{\mu} \norm{\widehat{\alpha}_{\star}}_{L^2(\R)}
\norm{\widehat{\beta}(\cdot;\mu,\delta)}_{L^2(\R)} \ \lesssim \delta^2 \delta^{-1} = \delta.
\end{align*}

\backmatter

\bibliographystyle{amsplain}
\bibliography{1d-edge}

\begin{bibdiv}
\begin{biblist}

\bib{Avila-etal:13}{article}{
      author={Avila, J.~C.},
      author={Schulz-Baldes, H.},
      author={Villegas-Blas, C.},
       title={Topological invariants of edge states for periodic
  two-dimensional models},
        date={2013},
     journal={Mathematical Physics, Analysis and Geometry},
      volume={16},
       pages={136\ndash 170},
}

\bib{Borisov-Gadylshin:08}{article}{
      author={Borisov, D.~I.},
      author={Gadyl'shin, R.~R.},
       title={On the spectrum of a periodic operator with a small localized
  perturbation},
        date={2008},
     journal={Izvestiya: Mathematics},
      volume={72},
      number={4},
       pages={659\ndash 688},
}

\bib{BR:11}{article}{
      author={Bronski, J.},
      author={Rapti, Z.},
       title={Counting defect modes in periodic eigenvalue problems},
        date={2011},
     journal={SIAM J. Math. Anal.},
      volume={43},
      number={2},
       pages={803\ndash 827},
}

\bib{DF:06}{article}{
      author={D'Ancona, D.},
      author={L., Fanelli},
       title={{$L^p$}-boundedness of the wave operator for the one dimensional
  {S}chr\"odinger operator},
        date={2006},
     journal={Commun. Math. Sci.},
      volume={268},
      number={415},
}

\bib{Deift-Hempel:86}{article}{
      author={Deift, P.~A.},
      author={Hempel, R.},
       title={On the existence of eigenvalues of the schr{\"o}dinger operator
  h-$\lambda$w in a gap of $\sigma$(h)},
        date={1986},
     journal={Commun. Math. Phys.},
      volume={103},
      number={3},
       pages={461\ndash 490},
}

\bib{DMW:11}{article}{
      author={Duch\^{e}ne, V.},
      author={Marzuloa, J.~L.},
      author={Weinstein, M.~I.},
       title={Wave operator bounds for one-dimensional {S}chr\"odinger
  operators with singular potentials and applications},
        date={2011},
     journal={J. Math. Phys.},
      volume={52},
      number={013505},
}

\bib{DVW:12}{article}{
      author={Duch\^ene, V.},
      author={Vuki{\'c}evi{\'c}, I.},
      author={Weinstein, M.~I.},
       title={Scattering and localization properties of highly oscillatory
  potentials},
        date={2014},
     journal={Comm. Pure Appl. Math.},
      volume={1},
       pages={83\ndash 128},
}

\bib{DVW-CMS:14}{article}{
      author={Duch\^ene, V.},
      author={Vuki{\'c}evi{\'c}, I.},
      author={Weinstein, M.~I.},
       title={Homogenized description of defect modes in periodic structures
  with localized defects},
        date={to appear},
     journal={Commun. Math. Sci.},
}

\bib{Eastham:74}{book}{
      author={Eastham, M.},
       title={Spectral theory of periodic differential equations},
   publisher={Hafner Press},
        date={1974},
}

\bib{FLW-PNAS:14}{article}{
      author={Fefferman, C.~L.},
      author={Lee-Thorp, J.~P.},
      author={Weinstein, M.~I.},
       title={Topologically protected states in one-dimensional continuous
  systems and dirac points},
        date={2014},
     journal={Proceedings of the National Academy of Sciences},
       pages={201407391},
}

\bib{FW:12}{article}{
      author={Fefferman, C.~L.},
      author={Weinstein, M.~I.},
       title={Honeycomb lattice potentials and {D}irac points},
        date={2012},
     journal={J. Amer. Math. Soc.},
      volume={25},
      number={4},
       pages={1169\ndash 1220},
}

\bib{FW:14}{article}{
      author={Fefferman, C.~L.},
      author={Weinstein, M.~I.},
       title={Wave packets in honeycomb lattice structures and two-dimensional
  {D}irac equations},
        date={2014},
     journal={Commun. Math. Phys.},
      volume={326},
       pages={251\ndash 286},
         url={DOI: 10.1007/s00220-013-1847-2},
}

\bib{Figotin-Klein:97}{article}{
      author={Figotin, A.},
      author={Klein, A.},
       title={Localized classical waves created by defects},
        date={1997},
     journal={J. Statistical Phys.},
      volume={86},
      number={1/2},
       pages={165\ndash 177},
}

\bib{Graf-Porta:13}{article}{
      author={Graf, J.-M.},
      author={Porta, M.},
       title={Bulk-edge correspondence for two-dimensional topological
  insulators},
        date={2012},
     journal={Comm. Math. Phys.},
      volume={324},
      number={3},
       pages={851\ndash 895},
}

\bib{H:82}{article}{
      author={Halperin, B.~I.},
       title={Quantized {H}all conductance, current-carrying edge states, and
  the existence of extended states in a two-dimensional disordered potential},
        date={1982},
     journal={Phys. Rev. B},
      volume={25},
      number={4},
       pages={2185\ndash 2190},
}

\bib{HK:10}{article}{
      author={Hasan, M.~Z.},
      author={Kane, C.~L.},
       title={Colloquium: {T}opological {I}nsulators},
        date={2010},
     journal={Reviews of Modern Physics},
      volume={82},
       pages={3045},
}

\bib{Hatsugai:93}{article}{
      author={Hatsugai, Y.},
       title={The {C}hern number and edge states in the integer quantum hall
  effect},
        date={1993},
     journal={Phys. Rev. Lett.},
      volume={71},
       pages={3697\ndash 3700},
}

\bib{JR:76}{article}{
      author={Jackiw, R.},
      author={Rebbi, C.},
       title={Solitons with fermion number $1/2$},
        date={1976},
     journal={Phys. Rev. D},
      volume={13},
      number={12},
       pages={3398\ndash 3409},
}

\bib{Boas:46}{article}{
      author={Jr., R. P.~Boas},
       title={{P}oisson's summation formula in {$L^2$ }},
        date={1946},
     journal={J. London Math. Soc.},
      volume={21},
       pages={102\ndash 105},
}

\bib{Kane-Mele:05}{article}{
      author={Kane, C.~L.},
      author={Mele, E.~J.},
       title={$\mathbb{Z}_2$ topological order and the quantum spin {H}all
  efect},
        date={2005},
     journal={Phys. Rev. Lett.},
      volume={95},
       pages={146802},
}

\bib{Katsnelson:12}{book}{
      author={Katsnelson, M.},
       title={Graphene: Carbon in two dimensions},
   publisher={Cambridge University Press},
        date={2012},
}

\bib{Knopp-volume2:47}{book}{
      author={Knopp, K.},
       title={Theory of functions - part 2; applications and further
  development of the general theory},
   publisher={Dover},
        date={1947},
}

\bib{Krantz:92}{book}{
      author={Krantz, S.},
       title={Function theory of several complex variables,},
   publisher={AMS Chelsea Publishing},
     address={Providence, Rhode Island},
        date={1992},
}

\bib{experiment:14}{article}{
      author={Lee-Thorp, J.~P.},
      author={Vuki\'{c}evi\'{c}, I.},
      author={Yang, J.},
      author={Xu, X.},
      author={Fefferman, C.~L.},
      author={Wong, C.~W.},
      author={Weinstein, M.~I.},
       title={Photonic realization of topologically protected bound states in
  waveguide arrays},
     journal={submitted.},
}

\bib{magnus2013hill}{book}{
      author={Magnus, W.},
      author={Winkler, S.},
       title={Hill's equation},
   publisher={Courier Dover Publications},
        date={2013},
}

\bib{Chen-etal:09}{article}{
      author={Malkova, N.},
      author={Hromada, I.},
      author={Wang, X.},
      author={Bryant, G.},
      author={Chen, Z.},
       title={Observation of optical shockley-like surface states in photonic
  superlattices},
        date={2009},
     journal={Opt. Express},
      volume={34},
      number={11},
       pages={1633\ndash 1635},
}

\bib{Meixner-Schaefke:54}{book}{
      author={Meixner, J.},
      author={Sch\"afke, R.W.},
       title={Mathieusche funktionen und sph\"aroidfunktionen},
   publisher={Springer-Verlag},
     address={Berline-G\"ottingen-Heidelberg},
        date={1954},
}

\bib{Moser:81}{article}{
      author={Moser, J.},
       title={An example of a {S}chr\"odinger equation with almost periodic
  potential and nowhere dense spectrum},
        date={1981},
     journal={Comment. Math. Helvetici},
      volume={56},
       pages={198\ndash 224},
}

\bib{RMP-Graphene:09}{article}{
      author={Neto, A. H.~Castro},
      author={Guinea, F.},
      author={Peres, N. M.~R.},
      author={Novoselov, K.~S.},
      author={Geim, A.~K.},
       title={The electronic properties of graphene},
        date={2009},
     journal={Reviews of Modern Physics},
      volume={81},
       pages={109\ndash 162},
}

\bib{RH:08}{article}{
      author={Raghu, S.},
      author={Haldane, F. D.~M.},
       title={Analogs of quantum-{H}all-effect edge states in photonic
  crystals},
        date={2008},
     journal={Phys. Rev. A},
      volume={78},
      number={3},
       pages={033834},
}

\bib{Rechtsman-etal:13a}{article}{
      author={Rechstman, M.~C.},
      author={Zeuner, J.~M.},
      author={Plotnik, Y.},
      author={Lumer, Y.},
      author={Podolsky, D.},
      author={Dreisow, F.},
      author={Nolte, S.},
      author={Segev, M.},
      author={Szameit, A.},
       title={Photonic {F}loquet topological insulators},
        date={2013},
     journal={Nature},
      volume={496},
       pages={196},
}

\bib{Rechtsman-etal:13}{article}{
      author={Rechtsman, M.~C.},
      author={Plotnik, Y.},
      author={Zeuner, J.~M.},
      author={Song, D.},
      author={Chen, Z.},
      author={Szameit, A.},
      author={Segev, M.},
       title={Topological creation and destruction of edge states in photonic
  graphene},
        date={2013},
     journal={Phys. Rev. Lett.},
      volume={111},
       pages={103901},
}

\bib{RS4}{book}{
      author={Reed, M.},
      author={Simon, B.},
       title={Methods of modern mathematical physics: Analysis of operators,
  volume iv},
   publisher={Academic Press},
        date={1978},
        ISBN={9780125850049},
}

\bib{RofeBeketov-Kholkin:05}{book}{
      author={Rofe-Beketov, F.~S.},
      author={Kholkin, A.~M.},
       title={Spectral analysis of differential operators - interplay between
  spectral and oscillatory properties},
      series={World Scientific Monograph Series in Mathematics},
   publisher={World Scientific},
     address={Singapore},
        date={2005},
      volume={7},
}

\bib{Shockley:39}{article}{
      author={Shockley, W.},
       title={On the surface states associated with a periodic potential},
        date={1939},
     journal={Phys. Rev.},
      volume={56},
}

\bib{Simon:76}{article}{
      author={Simon, B.},
       title={The bound state of weakly coupled {S}chr{\"o}dinger operators in
  one and two dimensions},
        date={1976},
     journal={Annals of Physics},
      volume={97},
      number={2},
       pages={279\ndash 288},
}

\bib{Su-Schrieffer-Heeger:79}{article}{
      author={Su, W.~P.},
      author={Schrieffer, J.~R.},
      author={Heeger, A.~J.},
       title={Solitons in polyacetylene},
        date={1979},
     journal={Phys. Rev. Lett.},
      volume={42},
       pages={1698\ndash 1701},
}

\bib{Tamm:32}{article}{
      author={Tamm, I.~E.},
       title={A possible kind of electron binding on crystal surfaces},
        date={1932},
     journal={Phys. Z. Sowjetunion},
      volume={1},
      number={733},
}

\bib{TKNN:82}{article}{
      author={Thouless, D.~J.},
      author={Kohmoto, M.},
      author={Nightgale, M.~P.},
      author={Nijs, M.~Den},
       title={Quantized hall conductance in a two-dimensional periodic
  potential},
        date={1982},
     journal={Phys. Rev. Lett.},
      volume={49},
       pages={405},
}

\bib{Soljacic-etal:08}{article}{
      author={Wang, Z.},
      author={Chong, Y.~D.},
      author={Joannopoulos, J.~D.},
      author={Soljacic, M.},
       title={Reflection-free one-way edge modes in a gyromagnetic photonic
  crystal},
        date={2008},
     journal={Phys. Rev. Lett.},
      volume={100},
       pages={013905},
}

\bib{Weder:99}{article}{
      author={Weder, R.},
       title={The {$W_{k,p}$}-continuity of the {S}chr\"odinger wave operators
  on the line},
        date={1999},
     journal={Commun. Math. Phys.},
      volume={208},
      number={507},
}

\bib{WK:87}{article}{
      author={Weinstein, M.~I.},
      author={Keller, J.~B.},
       title={Asymptotic behavior of stability regions for {H}ill's equation},
        date={1987},
     journal={SIAM J. Appl. Math.},
      volume={47},
      number={5},
       pages={941\ndash 958},
}

\bib{W:81}{article}{
      author={Wen, X.-G.},
       title={Gapless boundary excitations in the quantum {H}all states and in
  the chiral spin states},
        date={1991},
     journal={Phys. Rev. B},
      volume={43},
      number={13},
       pages={11025},
}

\bib{Yajima:95}{article}{
      author={Yajima, K.},
       title={The {$W_{k,p}$}-continuity of wave operators for {S}chr\"odinger
  operators},
        date={1995},
     journal={J. Math. Soc. Jpn.},
      volume={47},
      number={551},
}

\bib{Fan-etal:08}{article}{
      author={Yu, Z.},
      author={Veronis, G.},
      author={Wang, Z.},
      author={Fan, S.},
       title={One-way electromagnetic waveguide formed at the interface between
  a plasmonic metal under a static magnetic field and a photonic crystal},
        date={2008},
     journal={Phys. Rev. Lett.},
      volume={100},
       pages={023902},
}

\end{biblist}
\end{bibdiv}

\end{document}